\documentclass[a4paper,11pt]{scrartcl}
\usepackage[pdfpagelabels]{hyperref}
\usepackage[utf8]{inputenc}
\usepackage{amsmath,amsfonts,amssymb,amsthm,mathtools}
\usepackage{amsmath,verbatim}
\usepackage{amsfonts}
\usepackage{amssymb}
\usepackage{dsfont}
\usepackage{mathrsfs}
\usepackage[T1]{fontenc}
\usepackage{lmodern}
\usepackage[ngerman,english]{babel}
\usepackage[inline]{enumitem}
\usepackage{latexsym,graphicx}
\usepackage[all]{xy}
\usepackage{mathtools}
\usepackage{color}
\usepackage{authblk}
\usepackage{amscd}
\usepackage{microtype}
\usepackage{mathrsfs,cite}
\usepackage{upgreek}
\numberwithin{equation}{section}
\usepackage{txfonts}

\DeclareMathOperator{\tr}{Tr}
\DeclareMathOperator{\trs}{tr}

\newtheorem{theorem}{Theorem}
\newtheorem{proposition}{Proposition}[section]
\newtheorem{lemma}[proposition]{Lemma}
\newtheorem{corollary}[proposition]{Corollary}
\theoremstyle{definition}
\newtheorem{remark}[proposition]{Remark}

\newcommand{\dda}{\mathrm{d}}
\newcommand{\de}{\,\dda}

\renewcommand{\Re}{\textrm{Re}}
\renewcommand{\Im}{\textrm{Im}}
\newcommand{\cH}{\mathcal{H}}

\newcommand{\cN}{\mathcal{N}}
\newcommand{\Tr}{\mathrm{Tr}}
\newcommand{\gesssim}{\gtrsim}
\newcommand{\gesim}{\gtrsim}
\newcommand{\lesim}{\lesssim}
\renewcommand{\Re}{\textrm{Re}}
\renewcommand{\Im}{\textrm{Im}}
\renewcommand\rho\varrho
\newcommand{\eps}{\varepsilon}
\renewcommand{\rho}{\varrho}
\renewcommand{\epsilon}{\varepsilon}
\newcommand{\nn}{\nonumber}

\newcommand{\R} {\mathbb{R}}

\newcommand{\ou}{%
	\mathrel{%
		\vcenter{\offinterlineskip
			\ialign{##\cr$\lesssim$\cr\noalign{\kern-1.5pt}$\gtrsim$\cr}%
		}%
	}%
}

\theoremstyle{definition}

\newcommand{\beq}{\begin{equation}}
\newcommand{\eeq}{\end{equation}}
\usepackage{pdfsync}

\begin{document}

\title{The Gibbs state of the mean-field Bose gas}

\author{Andreas Deuchert, Phan Thành Nam, Marcin Napiórkowski}

\date{\today}

\maketitle

\begin{abstract} 
We consider the homogeneous mean-field Bose gas at temperatures proportional to the critical temperature of its Bose–Einstein condensation phase transition. We prove a trace norm approximation for the grand canonical Gibbs state in terms of a reference state, which is given by a convex combination of products of coherent states and Gibbs states associated with certain temperature-dependent Bogoliubov Hamiltonians. The convex combination is expressed as an integral over a Gibbs distribution of a one-mode $\Phi^4$-theory describing the condensate. This result justifies an analogue of Lee and Yang’s extension of Bogoliubov theory to positive temperatures, and it allows us to derive various limiting distributions for the number of particles in the condensate, as well as precise formulas for the one- and two-particle density matrices of the Gibbs state. Key ingredients of our proof, which are of independent interest, include two novel abstract correlation inequalities. The proof of one of them is based on an application of an infinite-dimensional version of Stahl's theorem.
\end{abstract}

\setcounter{tocdepth}{2}
\tableofcontents

\section{Introduction and main results} \label{sec:intro}

\subsection{Background and summary}

The theory of Bose gases began about a hundred years ago with the introduction of Bose–Einstein statistics in the pioneering works of Bose and Einstein \cite{Bose1924,Einstein1924}. Based on an analysis of the ideal (non-interacting) gas, Einstein predicted in \cite{Einstein1925} the condensation of a macroscopic number of particles into a single quantum state (Bose–Einstein condensation (BEC)), a phenomenon that was also conjectured to occur in interacting bosonic many-particle systems at sufficiently low temperatures. In 1947, Bogoliubov \cite{Bogoliubov-47} realized that the occurrence of BEC in such systems allows one, at least formally, to replace certain operators in the Hamiltonian associated with the condensate by complex numbers (the so-called c-number substitution). By additionally neglecting terms in the Hamiltonian that are expected to be small, one obtains an explicitly solvable model. This model predicts that the low-energy excitations of the Hamiltonian are sound waves, which implies superfluid behavior of the system (flow without friction). These considerations are now collectively referred to as Bogoliubov theory. It has become crucial in various fields of physics, including condensed matter physics, quantum field theory, and nuclear physics, as it provides a powerful framework for understanding the collective behavior of interacting quantum systems.

Inspired by Bogoliubov's work, Lee and Yang \cite{LeeYan-58} predicted in 1958 an extension of Bogoliubov theory for systems at positive temperature. One prediction of their theory is that elementary excitations have a dispersion relation of the form 
\begin{align}\label{eq:LY-ek}
e_{\xi}(p) =  \sqrt{|p|^4+ g \varrho \xi(T) |p|^2}
\end{align}
with the momentum $p$, an interaction-dependent coupling constant $g > 0$, the density $\varrho$, and the temperature-dependent condensate fraction $\xi(T)$. For temperatures that are comparable to the critical temperature of the BEC phase transition, the Gibbs state displays \textit{incomplete occupation} of the $p=0$ mode with a condensate fraction $\xi \in [0,1]$. In this case a $c$-number substitution yields a modified Bogoliubov Hamiltonian, whose eigenvalues are given by sums of the elementary excitations in \eqref{eq:LY-ek}. The main contribution of our work is a justification of Lee and Yang's version of Bogoliubov theory at positive temperature.

The proof of Bose–Einstein condensation (BEC) and the justification of Bogoliubov theory from microscopic principles in various parameter regimes are major challenges in modern mathematical physics. In 2002 a proof of BEC was obtained by Lieb and Seiringer \cite{LieSei-02} in the Gross--Pitaevskii  limit, which is a scaling regime used to describe dilute trapped Bose gases. Moreover, in the breakthrough article \cite{BocBreCenSch-19}, Boccato, Brennecke, Cenatiempo, and Schlein recently demonstrated that in the same limit Bogoliubov theory correctly predicts the low-lying excitation spectrum of a dilute Bose gas. Shortly afterward, the Lee–Huang–Yang formula \cite{LeeHuaYan-57} for the ground-state energy (zero-temperature case) of the dilute Bose gas in the more challenging thermodynamic limit was established by Fournais and Solovej in \cite{FouSol-20,FouSol-23}. Their work confirmed the predictions of Bogoliubov theory in this setting and concluded a series of remarkable contributions made over the last seven decades \cite{Dyson-57, LieYng-98, YauYin-09, ErdSchYau-08, GiuSei-09, BriFouSol-20, BasCenSch-21, BasCenGiuOlgPasSch-24}. 
A related conjecture for the free energy at low temperatures, where the entropy contribution is proportional to that of the Bogoliubov excitation spectrum, has been proved recently in \cite{HHNST-23, HHST-24}. Some earlier results on the free energy of the dilute Bose gas in the thermodynamic limit can be found in \cite{Seiringer-08,Yin-10,DeuMaySei-20,MaySei-20}. See also  \cite{Baletal2017} for an overview of an ambitious long-term project aimed at proving BEC with renormalization group techniques.

Despite this impressive progress, a proof of the BEC phase transition in the thermodynamic limit appears to be beyond the reach of current mathematical techniques, let alone a detailed study of the properties of the Gibbs state. From a conceptual perspective, the thermodynamic limit is difficult to analyze because the kinetic energy operator looses its spectral gap, which prevents the use of coercivity of the relevant energy functionals to extract information about the states. In particular, it remains unclear, even heuristically, when Bogoliubov theory breaks down and what types of corrections should be expected. 

In this paper, we are interested in the homogeneous Bose gas in \textit{the mean-field limit}. This parameter regime is more tractable than the thermodynamic limit but still exhibits a rich mathematical structure. For this model, BEC and the validity of Bogoliubov theory for the low-lying excitation spectrum were established in Seiringer's seminal paper \cite{Seiringer-11}. His proof inspired much of the later work on the subject and laid the basis for our current understanding of Bogoliubov theory, which eventually led to a proof of the Lee--Huang--Yang formula. Building on insights from \cite{Seiringer-11}, Lewin, Nam, Serfaty, and Solovej derived in \cite{LewNamSerSol-15} a trace norm approximation for the Gibbs state of the system at sufficently low temperatures. As in the case of the ground state problem, their approximating state is \textit{quasi-free (Gaussian)}. However, understanding the system for temperatures on the scale of the critical temperature of the BEC phase transition remains a challenging open problem. Key open questions include characterizing Gibbs states, determining how the critical temperature depends on particle interactions, and computing critical exponents.

When the interaction, temperature, and chemical potential of the model are chosen such that the system approaches the critical point of the BEC phase transition from above, Bogoliubov theory is inapplicable. As shown independently by Lewin, Nam, and Rougerie in \cite{LewNamRou-21} and by Fröhlich, Knowles, Schlein, and Sohinger in \cite{FroKnoSchSoh-22}, the correct description of the system in this parameter regime is given by a nonlinear classical field theory. Observing these effects requires a stronger interaction than that considered in \cite{Seiringer-11,LewNamSerSol-15}. 

In this paper, we aim to provide a full description of the Gibbs state of the mean-field Bose gas for all temperatures on the scale of the critical temperature for BEC. We prove that the Gibbs state of the interacting model can be approximated in trace norm by 
\begin{equation}
    \Gamma_{\beta,N} = \frac{1}{Z} \int_{\mathbb{C}} |z \rangle \langle z | \otimes G^{\mathrm{Bog}}(z) \exp\left(-\beta\left(\frac{\hat{v}(0)}{2N} |z|^4 - \mu |z|^2 \right)\right) \de z,
    \label{eq:normApproximationMainIntro}
\end{equation}
 a convex combination of products of coherent states $|z \rangle \langle z|$ and Gibbs states $G^{\mathrm{Bog}}(z)$ associated with certain \textit{temperature-dependent} Bogoliubov Hamiltonians. Here $N$ denotes the expected particle number and $\beta > 0$ the inverse temperature. 
 The integral is taken over a Gibbs distribution of a one-mode $\Phi^4$-theory, which depends on the mean-field interaction potential $v/N$. It describes the condensate and was recently introduced by Boccato, Deuchert, and Stocker in \cite{BocDeuSto-24} to construct a trial for the free energy in a related model. Unlike in the low-temperature results in \cite{Seiringer-11,LewNamSerSol-15}, our approximating state is \textit{non-quasi-free}. Other main results, derived from our norm approximation of the Gibbs state, include the computation of the various limiting distributions for the number of particles in the condensate and precise formulas for the one- and two-particle density matrices of the Gibbs state. We also obtain an asymptotic expansion of the free energy. An analogue of \eqref{eq:LY-ek} appears in all of these results.

Justifying this approximation requires the development of general techniques that go substantially \textit{beyond the use of the coercivity} of the Gibbs free energy functional. More precisely, while we start with a variational formulation of the problem, the application of two new abstract \textit{correlation inequalities} allows us to derive results for the Gibbs state that are not shared by approximate minimizers of the free energy functional. One of them states that for two self-adjoint operators $A$ and $B$, 
\begin{equation}
\Tr [ B^2 \Gamma_0 ] \leq a e^{a} + \frac{1}{4} \Tr( [[B,A],B] \Gamma_0)
\label{eq:StahlA2-intro-intro}
\end{equation}
provided $|\Tr (B \Gamma_t )| \le a$ with $\Gamma_t =Z_t^{-1} \exp(-A+tB)$, $t\in [-1,1]$. This is inspired by the work of Lewin, Nam and Rougerie \cite{LewNamRou-21} and provides a very efficient way to control second-order correlations in the interacting Gibbs state using a first-order bound as input. The proof of \eqref{eq:StahlA2-intro-intro} relies on the key observation that Stahl's theorem \cite{Stahl2013} elegantly implies convexity of the Duhamel two-point function. Using this inequality we prove the norm approximation in \eqref{eq:normApproximationMainIntro} and obtain information about the two-point correlation function. However, gaining access to the four-point function is significantly harder. To solve this problem we establish a second new abstract correlation inequality for higher moments.   

We expect that the general tools developed in this paper are of independent interest and will pave the way for a better understanding of the Bose gas in other scaling regimes, as well as other many-body quantum systems at positive temperature. In fact, our two abstract correlation inequalities have already been used in very recent work of Nam, Zhu, and Zhu \cite{NamZhuZhu-25} on the derivation of the $\Phi^4_3$ theory, which substantially extends the earlier derivations of classical field theories with regular nonlocal interactions \cite{LewNamRou-21,FroKnoSchSoh-22} as well as the related results on the $\Phi^4_2$ theory \cite{FroKnoSchSoh-24}.

In the following, we introduce our model and highlight a few key results and ideas. Other main results, which require more notation to be stated, will be presented in Section~\ref{sec:MainResultsPart2}. 

\subsubsection*{Notation} 
We write $a \lesssim b$ to say that there exists a constant $C>0$ independent of the relevant parameters (for instance, the particle number or  the inverse temperature) such that $a \leq C b$ holds. If $a \lesssim b$ and $b \lesssim a$ we write $a \sim b$, and $a \simeq b$ means that the leading orders of $a$ and $b$ are equal in the limit considered. In case the constant depends on a parameter $k$, we write $a \lesssim_k b$. By $\Vert \cdot \Vert_p$ we denote the standard $L^p$-norm. We always interpret $\tr [ A \Gamma ]$ as $\tr[\Gamma^{1/2} A \Gamma^{1/2}]\in [0,\infty]$ for self-adjoint nonnegative operators $A,\Gamma$. 

\subsection{The many-body quantum model} \label{sec:many-body-model}

We consider a system of bosonic particles confined to the torus $\Lambda = [0,1]^3$. The one-particle Hilbert space is given by $L^2(\Lambda)$. We are interested in a system with a fluctuating particle number (the grand canonical ensemble), and hence the Hilbert space of the entire system is given by the bosonic Fock space
\begin{equation}
	\mathscr{F}=\mathscr{F}(L^2(\Lambda)) = \bigoplus_{n=0}^{\infty} L^2_{\mathrm{sym}}(\Lambda^n)= \mathbb{C} \oplus L^2(\Lambda)\oplus L^2_{\mathrm{sym}}(\Lambda^2) \oplus \cdots
    \label{eq:FockSpace}
\end{equation}
Here $L^2_{\mathrm{sym}}(\Lambda^n)$ denotes the subspace of $L^2(\Lambda^n)$ consisting all functions $\Psi(x_1,...,x_n)$ that are invariant under any permutation of the particle coordinates $x_1, ..., x_n \in \Lambda$. The Hamiltonian of the system reads 
\begin{equation}    \label{eq:FockSpaceHamiltonian1}
    \mathcal{H}_N = \bigoplus_{n=0}^{\infty} \mathcal{H}_N^n=  \bigoplus_{n=0}^{\infty} \left( \sum_{i=1}^n - \Delta_i + \frac{1}{N} \sum_{1 \leq i < j \leq n} v(x_i-x_j) \right)
\end{equation}
with the Laplacian $-\Delta \ge 0$ on $L^2(\Lambda)$ with periodic boundary conditions and a periodic function $v:\Lambda\to \R_+$ satisfying 
$$0\le \hat v\in \ell^1(\Lambda^*),\quad \text{ where } \quad \hat v(p)=\langle \varphi_p,v\rangle_{L^2(\Lambda)},\quad \varphi_p(x) = e^{\mathrm{i}p \cdot x},\quad \Lambda^* = 2\pi \mathbb{Z}^3.$$
In particular, since $v$ is bounded, $\mathcal{H}_N^n$ is a self-adjoint operator on $L^2_{\mathrm{sym}}(\Lambda^n)$ with the same domain $H^2_{\mathrm{sym}}(\Lambda^n)$ as the Laplacian $\sum_{i=1}^n - \Delta_i$. Moreover, $ \mathcal{H}_N$ is bounded from below on $\bigcup_{m \ge 2} \bigoplus_{n=0}^{m}  H^2_{\mathrm{sym}}(\Lambda^n)$ and can be defined as a self-adjoint operator on $\mathscr{F}(L^2(\Lambda))$ by the Friedrichs' extension. 

For given parameters $\beta,N\in (0,\infty)$, we are interested in the {\em grand canonical Gibbs state} 
\begin{equation}
	G_{\beta,N} = \frac{\exp\left( -\beta (\mathcal{H}_N - \mu_{\beta,N} \mathcal{N}) \right) }{\tr \exp\left( -\beta (\mathcal{H}_N - \mu_{\beta,N} \mathcal{N}) \right) }
	\label{eq:interactingGibbsstate}
\end{equation}
acting on on $\mathscr{F}(L^2(\Lambda))$. Here $\mathcal{N} = \bigoplus_{n=0}^{\infty} n$ denotes the number operator and the chemical potential $\mu_{\beta,N}\in \mathbb{R}$ is chosen such that $\tr [\mathcal{N} G_{\beta,N}] = N$ holds.  In this setting, the factor $1/N$ in front of the interaction term in $\mathcal{H}_N^n$ implements a mean-field scaling.

\subsubsection*{Second quantization} 
Grand canonical systems are most naturally discussed in the language of second quantization, which we introduce next. The annihilation operator $a_p$ and the creation operator $a^*_p$ of a particle described by the function $\varphi_p(x) = e^{\mathrm{i p \cdot x}}$ with $x \in \Lambda$ and $p \in \Lambda^*$ are defined by
\begin{align}
    (a_p \psi)(x_1,...,x_{n-1}) &= \sqrt{n} \int_{\Lambda} e^{-\mathrm{i}p \cdot x} \psi(x_1,...,x_{n-1},x) \de x,\\
        \label{eq:annihilationOperator}
        (a^*_p \psi)(x_1,...,x_{n+1}) &= \frac{1}{\sqrt{n+1}} \sum_{j=1}^n e^{\mathrm{i} p \cdot x_j} \psi(x_1,...,x_{j-1}, x_{j+1}, ..., x_n),
\end{align}
and extended to $\mathscr{F}(L^2(\Lambda))$ by linearity. As the notation suggest, $a^*(\psi)$ is the adjoint of $a(\psi)$. This family of operators satisfies the canonical commutation relations (CCR)
\begin{equation}
	[a_p,a_q^*] = \delta_{p,q}, \quad \quad [a_p,a_q] = 0 = [a^*_p,a^*_q],\quad \forall p,q\in \Lambda^*. 
	\label{eq:CCR}
\end{equation}

The creation and annihilation operators allow us to represent many operators acting on the Fock space in a convenient way. In particular, if $h$ is a self-adjoint operator on $L^2(\Lambda)$, then its second quantization can be written as 
$$
\de \Upsilon(h)=\bigoplus_{n=0}^{\infty} \left( \sum_{i=1}^n h_i \right) = \sum_{p,q \in \Lambda^*} \langle \varphi_p, h \varphi_q\rangle a^*_p a_q.  
$$
For the number operator in \eqref{eq:interactingGibbsstate} this implies $\mathcal{N}=  \de \Upsilon(1)=\sum_{p \in \Lambda^*} a^*_p a_p$ and the second quantization of the Laplacian $-\Delta$ reads $ \de \Upsilon(-\Delta)=\sum_{p \in \Lambda^*} p^2 a^*_p a_p$. More generally, the grand canonical Hamiltonian $\cH_N$ in \eqref{eq:FockSpaceHamiltonian1} can be written as 
\begin{equation}  \label{eq:Hamiltonian}
	\mathcal{H}_N = \sum_{p\in \Lambda^*} p^2 a_p^* a_p + \frac{1}{2 N} \sum_{p,u,v \in \Lambda^*} \hat{v}(p) a_{u+p}^* a_{v-p}^* a_u a_v.	
	\end{equation}

The mean-field Hamiltonian in \eqref{eq:FockSpaceHamiltonian1} and \eqref{eq:Hamiltonian}, along with its variants, such as the canonical version $ \mathcal{H}_N^N$ on $L^2_{\mathrm{sym}}(\Lambda^N)$ and its analogue on $L^2_{\mathrm{sym}}(\mathbb{R}^{3N})$ with a trapping potential, has been extensively studied over the past 60 years. Regarding proofs of BEC for the ground state and low-lying eigenfunctions, we refer to \cite{LieLin-63,SeiYngZag-12} for the analysis of the Lieb--Liniger model, \cite{BenLie-83,Solovej-90,Bach-91,BacLewLieSie-93,Kiessling-12,BosLeoMitPet-24} for bosonic atoms, \cite{LieYau-87} for stars, \cite{VdBLewPul-88,RagWer-89} for confined systems, and \cite{LewNamRou-14} for a general result. Concerning the justification of the Bogoliubov excitation spectrum, Seiringer's result \cite{Seiringer-11} has inspired several further developments, including \cite{GreSei-13} for trapped systems, \cite{LewNamSerSol-15} for abstract results including a norm approximation for eigenstates and Gibbs states, \cite{DerNap-13} for systems in large volumes, \cite{NamSei-15} for collective excitations with multiple condensate modes, \cite{BosPetSei-21,NamNap-21} for higher-order corrections, \cite{BroSei-2022} for translation-invariant systems in $\mathbb{R}^3$, and \cite{DeuSei-21} for a study of the mean-field Bose gas at positive temperature.

The techniques developed in the mean-field model have also been widely used to investigate the Gross--Pitaevskii limit, where the interaction potential $v$ is replaced by the  
$N$-dependent potential $N^3 v(Nx)$ that converges to a delta function when $N\to \infty$. For this limit, we refer to \cite{LiebSeiYng-2000,LieSei-02,LieSei-06,NamRouSei-16,BocBreCenSch-18,BocBreCenSch-20,NamNapRicTri-21,Hainzl-21,HaiSchTri-22,BreBroCarOld-24} for the ground state problem, \cite{BocBreCenSch-19,NamTri-23,BreSchSch-22a,BreSchSch-22,BaCenOlgPasqSchl-2022,CarOlgSAubSchl-2024,Brooks-25} for the excitation spectrum, and \cite{DeuSeiYng-19,DeuSei-20,BreLeeNam-24,BocDeuSto-24,CapDeu-23} for studies of systems at positive temperatures. There is also a vast literature on quantum dynamics; for a comprehensive introduction, we refer to the book \cite{BenPorSch-16}, and for a recent review, to \cite{Napiorkowski-23}.

So far, existing studies in the mean-field regime have mostly focused on the ground state problem or the Gibbs state at temperatures of order 1. In these cases, most particles occupy the condensate, with only a finite number remaining in the thermal cloud; see \cite{NamRad-23,MitPic-23} for a robust justification of this fact with exponential estimates. Therefore, the contribution of thermally excited particles can, in principle, be predicted using a second-order perturbation method. In contrast, the present paper is devoted to the case where the temperature is proportional to the critical temperature of the BEC phase transition. In this parameter regime, the number of thermally excited particles is always comparable to the total number of particles. This behavior can already be observed in the ideal gas, which we introduce below.
 
\subsubsection*{The ideal Bose gas and the BEC phase transition} 
If $v = 0$, we obtain a non-interacting model called the {\em ideal Bose gas} with Gibbs state
\begin{equation}
    G_{\beta,N}^{\mathrm{id}} = \frac{\exp\left( -\beta (\de \Upsilon(-\Delta) - \mu_0(\beta,N) \mathcal{N}) \right) }{\tr \exp\left( -\beta (\de \Upsilon(-\Delta) - \mu_0(\beta,N)\mathcal{N}) \right) },
    \label{eq:GibbsStateIdealGas}
\end{equation}
which is exactly solvable. Here the chemical potential $\mu_0=\mu_0(\beta,N) < 0$ is the unique solution to the equation
\begin{equation}
	N =  \Tr [\cN G_{\beta,N}^{\mathrm{id}}] = \sum_{p \in \Lambda^*} \frac{1}{\exp(\beta \left( p^2 -  \mu_0 ) \right) - 1}.
	\label{eq:idealgase1pdmchempot}
\end{equation} 

As realized by Bose and Einstein in \cite{Bose1924,Einstein1924}, the ideal Bose gas displays a phase transition:  in the limit $N \to \infty$, the number of particles in the zero-momentum mode, $N_0(\beta,N) = \Tr [a_0^*a_0   G_{\beta,N}^{\mathrm{id}}]$, behaves as
\begin{equation}
	N_0(\beta,N) =  \frac{1}{\exp(-\beta \mu_0)-1} \simeq N \left[ 1 - \left( \frac{\beta_{\mathrm{c}}}{\beta} \right)^{3/2} \right]_+, \quad \text{ where } \quad \beta_{\mathrm{c}} = \beta_{\mathrm{c}}(N) = \frac{1}{4 \pi} \left( \frac{N}{\upzeta(3/2)} \right)^{-2/3}. 
	\label{eq:crittemp}
\end{equation}
Here $\upzeta$ denotes the Riemann zeta function and $[x]_+ = \max\{ 0,x \}$. We highlight that the inverse critical temperature $\beta_c$ is proportional to $N^{-2/3}$, which is a consequence of the fact that we work in a fixed volume. 

More precisely, if $\beta$ is chosen as $\beta = \kappa N^{-2/3}$ with some fixed $\kappa\in (0,\infty)$, the following sharp transition occurs as $\kappa$ crosses the critical point at $\kappa=\frac{1}{4 \pi}[\upzeta(3/2)]^{2/3}$: 
\begin{equation}
\begin{cases}\label{eq:BEC-phase-transition-simple}
N_0(\beta,N) \sim N, \quad \mu_0 \simeq -(\beta N_0)^{-1} \sim -N^{-1/3} &\quad \text{if}\quad  \kappa > \frac{1}{4 \pi}[\upzeta(3/2)]^{2/3} \quad \text {(condensed phase)}, \\
N_0(\beta,N) \sim 1, \quad \mu_0 \sim -\beta^{-1} \sim -N^{2/3} &\quad \text{if}  \quad \kappa < \frac{1}{4 \pi}[\upzeta(3/2)]^{2/3} \quad \text {(non-condensed phase)}. 
\end{cases}
\end{equation}

In the present paper, we are interested in the interacting Gibbs state \eqref{eq:interactingGibbsstate}, with $v \ne 0$, in the temperature regime $\beta/\beta_c \sim 1$, which includes both phases described in \eqref{eq:BEC-phase-transition-simple} (as we will see below the interacting model displays a similar phase transition). In particular, we pay special attention to the case $\beta N^{2/3} \to \frac{1}{4 \pi}[\upzeta(3/2)]^{2/3}$, where the order of magnitude of $N_0(\beta, N)$ varies between $O(N)$ and $O(1)$, depending on the rate at which the critical point is approached. In other words, when $\beta=\frac{1}{4\pi}[\zeta(\frac32)^{\frac23}] N^{- \frac23}+o(N^{-\frac23})$, it is the exact  form of the $o(N^{-\frac23})$ term that determines $N_0(\beta,N)$.

\subsection{Coherent states and Bogoliubov theory}
Unlike the ideal gas, the interacting Bose gas is not exactly solvable. In fact, the presence of interactions gives rise to intriguing quantum phenomena, whose understanding is, in general, mathematically very challenging. A cornerstone of the theory of interacting Bose gases is Bogoliubov's seminal paper \cite{Bogoliubov-47}, in which he introduced an effective, explicitly solvable model (the Bogoliubov Hamiltonian) to explain superfluidity in such system. The main idea proposed by Bogoliubov is to replace $a_0$ and $a_0^*$ in the Hamiltonian of a system in the condensed phase with a complex number, whose absolute value is proportional to $\sqrt{\Tr[a_0^* a_0 G_{\beta,N}]} \gg 1$. A rigorous way to implement this idea is the c-number substitution using coherent states \cite{KlauSka-85,LieSeiYng-05}. In the following we explain this in some more detail. We start by considering the case of zero temperature.
\subsubsection*{Zero temperature}
Let us recall the exponential property $\mathcal{F}(\mathfrak{h}_1 \oplus \mathfrak{h}_2) \cong \mathscr{F}(\mathfrak{h}_1) \otimes \mathscr{F}(\mathfrak{h}_2)$ of the bosonic Fock space, where $\cong$ denotes unitary equivalence. Using this we write our Fock space as
\begin{equation}
    \mathscr{F}(L^2(\Lambda)) \cong \mathscr{F}_0 \otimes \mathscr{F}_+
    \label{eq:exponentialProperty}
\end{equation}
with $\mathscr{F}_0 = \mathcal{F}(\text{span}\{\varphi_0\})$, $\varphi_0(x) = 1$ and the excitation Fock space
\begin{equation}
    \mathscr{F}_+ = \mathscr{F}( Q L^2(\Lambda) ), \quad \text{ where } \quad Q = \mathds{1}(-\Delta \neq 0).
    \label{eq:excitationFockSpace}
\end{equation}
In the space $\mathscr{F}_0$, we introduce the coherent states 
\begin{equation}
	| z \rangle = \exp( z a_0^* - \overline{z} a_0 ) | \Omega_0 \rangle, \quad \quad z \in \mathbb{C}
	\label{eq:coherentstate}
\end{equation}
with the vacuum vector $\Omega_0 \in \mathscr{F}_0$. Heuristically speaking, the coherent state $|z\rangle$ describes a BEC in the constant function $z/|z| \in L^2(\Lambda)$ with an expected number of $|z|^2$ particles. One of its most important properties is that it is an eigenvector of the annihilation operator $a_0$ with eigenvalue $z$, i.e. 
\begin{equation}
    a_0 | z \rangle = z | z \rangle.
    \label{eq:coherentstateEigenfunction}
\end{equation}

A state on the bosonic Fock space $\mathscr{F}$ is a positive operator $\Gamma$ acting on $\mathscr{F}$ with $\Tr[\Gamma] =1$. If we compute the expectation of the Hamiltonian in \eqref{eq:Hamiltonian} in a state of the form $\Gamma = |z \rangle \langle z | \otimes G(z)$, where $|z \rangle \langle z |$ denotes the orthogonal projection onto the vector $|z \rangle$ and $G(z)$ is a state on $\mathscr{F}_+$, \eqref{eq:coherentstateEigenfunction} allows us to replace the operators $a_0$ and $a_0^*$ by $z$ and $\overline{z}$, respectively. At zero temperature one expects all except for $o(N)$ particles to reside in the condensate, that is, $|z|^2 \simeq N$. In this case it is reasonable to neglect all terms in the Hamiltonian with only one or no factors of $z$ or $\overline{z}$. The result of this heuristic computation reads
\begin{align}
    \Tr[ \mathcal{H}_N  \Gamma ] \approx& \frac{\hat{v}(0) |z|^4}{2N} + \Tr_+\left[ \left( \sum_{p \in \Lambda^*_+} p^2 a_p^* a_p + \frac{1}{2N} \sum_{p \in \Lambda_+^*} \hat{v}(p) \left( 2 |z|^2 a_p^* a_p  + z^2 a_p^* a_{-p}^* + \overline{z}^2 a_p a_{-p} \right) \right) G(z) \right],
    \label{eq:firstBogoliubovHamiltonian}
\end{align}
where $\Tr_+$ denotes the trace over $\mathscr{F}_+$ and $\Lambda_+^* = \Lambda^* \backslash \{ 0 \}$. 

The quadratic Hamiltonian in \eqref{eq:firstBogoliubovHamiltonian} can be explicitly diagonalized with a unitary transformation $\mathcal{U}_z$ on $\mathscr{F}_+$ which satisfies 
\begin{equation}
	\mathcal{U}_z^* a_{p} \mathcal{U}_z = u_{p,z} a_{p} + \frac{z^2}{|z|^2} v_{p,z} a_{-p}^*, \quad \quad p \in \Lambda^*_+,
	\label{eq:Bogtrafoz}
\end{equation}
with the coefficients $u_{p,z} = \sqrt{1+v_{p,z}^2}$ and
\begin{align}
	v_{p,z} = \frac{1}{2} \left( \frac{p^2}{p^2 + 2 \hat{v}(p) |z|^2/N} \right)^{1/4} - \frac{1}{2} \left( \frac{p^2}{p^2 + 2 \hat{v}(p) |z|^2/N} \right)^{-1/4}.
	\label{eq:coefficientsBogtrafo}
\end{align}
More precisely, denoting the Hamiltonian in the second line of \eqref{eq:firstBogoliubovHamiltonian} by $\widetilde{\mathcal{H}}^{\mathrm{Bog}}$, we have
\begin{equation}
		\mathcal{U}_z \widetilde{\mathcal{H}}^{\mathrm{Bog}} \mathcal{U}_z^* = -\frac{1}{2} \sum_{p\in \Lambda^*_+} \left[ p^2 + \hat{v}(p) |z|^2/N - \epsilon_z(p) \right] +  \sum_{p \in \Lambda^*_+} \epsilon_z(p) a_{p}^* a_p 
		\label{eq:diagonalform}
	\end{equation}
with the Bogoliubov dispersion relation $\epsilon_z(p) = |p| \sqrt{ p^2 + 2 \hat{v}(p) |z|^2/N }$. The first term on the right-hand side of \eqref{eq:diagonalform} is the correction to the leading order contribution $\hat{v}(0) |z|^4/(2N)$ and the second term describes the excitations above the ground state. Since the dispersion relation is linear they can be interpreted as sound waves traveling through the BEC. Thus, up to a constant shift of the energy, the Bogoliubov Hamiltonian $\widetilde{\mathcal{H}}^{\mathrm{Bog}}$ is unitarily equivalent to the Hamiltonian of a non-interacting Bose gas with dispersion relation $\epsilon_z(p)$. In this way, Bogoliubov theory can be interpreted as a \textit{quasi-free approximation}. In the zero temperature case there is complete BEC, and we therefore have $|z|^2 = N$. To approximately minimize the energy in \eqref{eq:firstBogoliubovHamiltonian}, we have to choose $G(z)$ as a projection onto the ground state of $\widetilde{\mathcal{H}}^{\mathrm{Bog}}$. The expected number of particles in this vector is uniformly bounded in the limit $N \to \infty$. A rigorous version of the above arguments can be found in \cite{Seiringer-11,LewNamSerSol-15}. Let us now come to the discussion of the positive temperature case, which is the main concern of the present paper.
\subsubsection*{Positive temperature}
We introduce a positive temperature to the system that is proportional to the critical temperature of the BEC phase transition. In this case we can encounter a situation, where $\Tr[a_0^* a_0 G_{\beta,N}] \sim N$ (order $N$ condensed particles) and $N - \Tr[a_0^* a_0 G_{\beta,N}] \sim N$ (order $N$ particles are thermally excited). Following an idea that has been proposed by Lee and Yang in \cite{LeeYan-58} we note that for \textit{typical eigenstates of the Hamiltonian}, we expect to have $|z|^2 \simeq \Tr[a_0^* a_0 G_{\beta,N}]$. Thus, this part of the c-number substitution clearly continues to make sense.

However, one crucial assumption of Bogoliubov theory is violated in the positive temperature setting: namely, that all terms in the Hamiltonian involving no or only one factor of $z$ or $\overline{z}$ are small. In fact, it is not difficult to see that some of these terms are of order $N$ and, therefore, contribute to the leading order of the interaction energy. This relates to a question raised in the physics literature about the temperature at which the validity of Bogoliubov theory breaks down; see, e.g., \cite{Andersen2004}. A key insight in our work is that Bogoliubov theory remains valid in the presence of a macroscopic number of thermally excited particles, as long as the temperature stays on the order of the critical temperature. The reason is that the interaction terms, which can be neglected at zero temperature, do not alter the structure of the Gibbs state. For some terms, this is due to their particular form (in the case of density-density interactions), while for others, it is because they are small enough. 
Proving these claims requires a delicate analysis of the properties of the interacting Gibbs state, which constitutes the main contribution of the present article. A key ingredient for this analysis is our new set of abstract correlation inequalities. We emphasize that some terms in the Hamiltonian are small only when expectations are taken with respect to the Gibbs state. However, if the Gibbs state is replaced by an approximate minimizer of the relevant energy functional (the Gibbs free energy functional), these same terms may yield substantially larger contributions.

These considerations, along with an additional ingredient related to the condensate (which we will explain in a moment), lead to the following heuristic picture for the energy of the positive-temperature system, which is approximately given by 
\begin{align}
    \Tr[ \mathcal{H}_N  \Gamma ] \approx& \frac{\hat{v}(0) N}{2} + \mu_0(\beta,N) \int_{\mathbb{C}} \Tr_+[ \mathcal{N}_+ G(z)] g(z) \de z + \frac{\hat{v}(0)}{2N} \left( \int_{\mathbb{C}} |z|^4 g(z) \de z - \left( \int_{\mathbb{C}} |z|^2 g(z) \right)^2 \right) \de z \nn \\
    &+ \int_{\mathbb{C}} \Tr_+\left[ \mathcal{H}^{\mathrm{Bog}}(z) G(z) \right] g(z) \de z. \label{eq:firstBogoliubovHamiltonianb}
\end{align}
Here the thermally excited particles are described by the Bogoliubov Hamiltonian
\begin{equation}
    \mathcal{H}^{\mathrm{Bog}}(z) = \left( \sum_{p \in \Lambda^*_+} (p^2-\mu_0(\beta,N)) a_p^* a_p + \frac{N_0(\beta,N)}{2N} \sum_{p \in \Lambda_+^*} \hat{v}(p) \left( 2 a_p^* a_p  + \frac{z^2}{|z|^2} a_p^* a_{-p}^* + \frac{\overline{z}^2}{|z|^2} a_p a_{-p} \right) \right).
    \label{BogoliubovHamiltonianAndi}
\end{equation}
Note that we added and subtracted the chemical potential $\mu_0$. Moreover, $\de z = \de x \de y/\pi$, where $x = \Re z, y = \Im z$ and the condensate is described by a one-mode $\Phi^4$-theory with the Gibbs distribution
\begin{equation}\label{eq:GibbsDistributionDiscrete}
	g^{\mathrm{BEC}}(z) = \frac{\exp\left( -\beta\left( \frac{\hat{v}(0)}{2N} |z|^4 - \mu^{\mathrm{BEC}} |z|^2 \right) \right)}{\int_{\mathbb{C}} \exp\left( -\beta\left( \frac{\hat{v}(0)}{2N} |w|^4 - \mu^{\mathrm{BEC}} |w|^2 \right) \right) \de w} 
\end{equation}
and an appropriate choice for $\mu^{\mathrm{BEC}} \in \mathbb{R}$. The function $g^{\mathrm{BEC}}$ has been introduced recently in a similar context in \cite{BocDeuSto-24}, see also \cite{CapDeu-23}. Similar expressions also appeared earlier in \cite{BruZag1999,BruZag2008}. Note also the relation between $g^{\mathrm{BEC}}$ and the classical field theory in \cite{LewNamRou-21,FroKnoSchSoh-22}. This effective condensate theory will be derived rigorously in Theorem~\ref{thm:norm-approximation} below.

The rationale behind this approximation is as follows. The terms in \eqref{eq:firstBogoliubovHamiltonianb} are all contributions to the energy up to terms of the order $o(N^{2/3})$. Since the inverse temperature scales as $\beta \sim N^{-2/3}$ and we are interested in the operator $\exp(-\beta ( \mathcal{H} - \mu \mathcal{N}))$ this accuracy is sufficient to determine the structure of the Gibbs state. The first term on the right-hand side of \eqref{eq:firstBogoliubovHamiltonianb} is of order $N$ but it does not alter the structure of the Gibbs state, which is determined by the Bogoliubov Hamiltonian in \eqref{BogoliubovHamiltonianAndi}. However, this is not entirely true because we still have an integral over $g^{\mathrm{BEC}}(z) \de z$. The function $g^{\mathrm{BEC}}$ arises because the condensate exhibits entropy-induced particle number fluctuations, which have a standard deviation of order $N^{5/6}$ and depend on the particle interactions. The energy associated with these fluctuations is given by the third term on the right-hand side of \eqref{eq:firstBogoliubovHamiltonianb}. Since the function $g^{\mathrm{BEC}}$ is peaked around its mean we can replace the quantity $|z|^2$ in the Bogoliubov Hamiltonian in \eqref{eq:firstBogoliubovHamiltonian} with $N_0(\beta,N)$ in \eqref{eq:crittemp}, which approximates the mean of $g^{\mathrm{BEC}}$. Note that the condensate's particle number fluctuations are much smaller at zero temperature (of order $1$). Note also that the condensate of the ideal gas exhibits fluctuations of the number of condensed particle that are of order $N$. As a result of the condensate fluctuations, the state that approximates the Gibbs state in \eqref{eq:interactingGibbsstate} is not quasi-free. It turns out that the chemical potential $\mu_0$ in \eqref{BogoliubovHamiltonianAndi} has to be chosen as $\mu_0(\beta,N)$ in \eqref{BogoliubovHamiltonianAndi}. It is irrelevant in the condensed phase but ensures the validity of our approximation across the critical point.

The Gibbs state related to the Bogoliubov Hamiltonian in \eqref{BogoliubovHamiltonianAndi} will be denoted by 
\begin{equation}
	G^{\mathrm{Bog}}(z) = \frac{\exp\left( -\beta \mathcal{H}^{\mathrm{Bog}}(z)  \right)}{\tr_{\mathscr{F}_+} \left[ \exp\left( -\beta \mathcal{H}^{\mathrm{Bog}}(z) \right) \right]}. \label{eq:BogoliubovGibbsState}
\end{equation}
It satisfies $\Tr[\mathcal{N}_+ G^{\mathrm{Bog}}(z)] = \sum_{p \in \Lambda_+^*} \gamma_p$
with
\begin{equation}
\gamma_p :=     \Tr [ a^*_p a_{p} G^{\mathrm{Bog}}(z) ]   = \frac{u^2_p + v^2_p}{\exp(\beta \epsilon(p)) - 1} + v^2_p, \quad \forall p\in \Lambda^*_+,\quad \forall z\in \mathbb{C}. 
    \label{eq:gammap}
\end{equation}
Here and in the following we denote by $u_p, v_p$ the related quantities in \eqref{eq:coefficientsBogtrafo} with $p^2$ replaced by $p^2-\mu_0$ and $|z|^2$ replaced by $N_0(\beta,N)$. In particular,
\begin{equation}
        \epsilon(p) = \sqrt{p^2 - \mu_0(\beta,N) } \sqrt{ p^2 - \mu_0(\beta,N) + 2 \hat{v}(0) N_0(\beta,N)/N }
        \label{eq:BogoliubovDispersion}
\end{equation}
is the analogue of \eqref{eq:LY-ek} for our model. In the next section we state our first main result.

\subsection{Main results, part I}
\label{sec:mainResultsPartI}

Our first result provides a trace norm approximation of the Gibbs state $G_{\beta,N}$ in \eqref{eq:interactingGibbsstate}. Note that we use $N_0(\beta,N)$ in \eqref{eq:crittemp} instead of $\beta$ to distinguish between the two parameter regimes in parts (a) and (b). Since the map $\beta \mapsto N_0(\beta,N)$ is strictly monotone increasing, this is mathematically sound. It is also convenient, because the size of $N_0$ can vary from being of order $N$ to being of order $1$ when $|\beta - \beta_{\mathrm{c}}(N)|/\beta_{\mathrm{c}}(N)$ with $\beta_{\mathrm{c}}(N)$ in \eqref{eq:crittemp} is small (cf. the discussion at the end of Section \ref{sec:many-body-model}), and computing the exact dependence is a tedious task.

\begin{theorem}[Trace norm approximation of the Gibbs state] \label{thm:norm-approximation} 
Let the interaction potential $v \in L^1(\Lambda)$ be a nonnegative, even, periodic function (with period $1$), whose Fourier coefficients $\hat{v}$ are nonnegative and satisfy $\sum_{p \in \Lambda^*} (1+|p|) \hat{v}(p) < + \infty$. We consider the limit $N \to \infty$, $\beta N^{2/3} \to \kappa \in (0,\infty)$. Then the Gibbs state $G_{\beta,N}$ in \eqref{eq:interactingGibbsstate} satisfies
\begin{equation}
\Vert G_{\beta,N} - \Gamma_{\beta,N} \Vert_1  \lesssim N^{-1/48} 
\label{eq:norm-approximation-intro}
\end{equation}
with the state $\Gamma_{\beta,N}$, which is defined as follows. 
\begin{enumerate}[label=(\alph*)]
    \item If $N_0(\beta,N) \geq N^{2/3}$ with $N_0(\beta,N)$ in \eqref{eq:crittemp} we have 
\begin{equation}
    \Gamma_{\beta,N} = \int_{\mathbb{C}} |z \rangle \langle z| \otimes G^{\mathrm{Bog}}(z) g^{\mathrm{BEC}}(z) \de z 
    \label{eq:referenceState-intro1}
\end{equation}
with the coherent state $|z \rangle$ in \eqref{eq:coherentstate},  $G^{\mathrm{Bog}}$ in \eqref{eq:BogoliubovGibbsState} and $g^{\mathrm{BEC}}$ in \eqref{eq:GibbsDistributionDiscrete}. The chemical potential $\mu^{\mathrm{BEC}}$ related to $g^{\mathrm{BEC}}$ is chosen such that $\int_{\mathbb{C}} |z|^2 g^{\mathrm{BEC}}(z) \de z = N - \sum_{p \in \Lambda_+^*} \gamma_p$ holds with $\gamma_p$ in \eqref{eq:gammap}.
\item If $N_0(\beta,N) < N^{2/3}$ we have
\begin{equation}
    \Gamma_{\beta,N} = \frac{\exp\left( -\beta (\de \Upsilon(-\Delta) - \mu_0(\beta,N) \mathcal{N}) \right) }{\tr \exp\left( -\beta (\de \Upsilon(-\Delta) - \mu_0(\beta,N)\mathcal{N}) \right) }
    \label{eq:referenceState-intro2}
\end{equation}
with the chemical potential $\mu_0$ of the ideal gas in \eqref{eq:idealgase1pdmchempot}. 
\end{enumerate}
\end{theorem}

We interpret Theorem~\ref{thm:norm-approximation} as a justification of Bogoliubov theory for the mean-field Bose gas at positive temperature. The state $|z \rangle \langle z| \otimes G^{\mathrm{Bog}}(z)$ appearing under the integral in \eqref{eq:referenceState-intro1} is quasi-free. However, integrating this state with respect to the measure $g^{\mathrm{BEC}}(z) \de z$ over $\mathbb{C}$ destroys this property. That is, if $N_0(\beta,N) \geq N^{2/3}$ the state $\Gamma_{\beta,N}$ is not quasi-free, and our result therefore goes beyond the standard quasi-free approximation. More details on this can be found in part (a) of Remark~\ref{rem:2pdm} in Section~\ref{sec:mainReducedDensityMatrices}. This should be contrasted with the validity of the quasi-free approximation at zero and sufficiently low temperatures, see \cite{Seiringer-11,GreSei-13,LewNamSerSol-15,DerNap-13,NamSei-15,BosPetSei-21,NamNap-21,BroSei-2022}. 

If $N_0(\beta,N) < N^{2/3}$, the Gibbs state $G_{\beta,N}$ can be approximated in trace norm by the one of the ideal gas. This is a remarkable stability result: although an interaction is added to the system, it does not change the Gibbs state with high accuracy. We expect this to hold only for systems with a mean-field interaction.

Concerning the above theorem we have the following additional remarks.

\begin{remark}
\label{rem:Theorem1}
\begin{enumerate}[label=(\alph*)]
    \item Since each term in the Hamiltonian $\mathcal{H}_N$ in \eqref{eq:Hamiltonian} consists of an even number of creation and annihilation operators the Gibbs state $G_{\beta,N}$ satisfies $[G_{\beta,N},\mathcal{N}] = 0$. The same is obviously not true for the state $|z \rangle \langle z| \otimes G^{\mathrm{Bog}}(z)$ appearing under the integral in \eqref{eq:referenceState-intro1}. However, the integration over the function $g^{\mathrm{BEC}}(z)$ in \eqref{eq:GibbsDistributionDiscrete}, which only depends on $|z|$, restores this property. This can be seen from the fact that $[\Gamma_{\beta,N},\mathcal{N}] = 0$ is equivalent to the property that all expectations of products of creation and annihilation operators in the state vanish unless the number of $a$'s and $a^*$'s is the same. When we use Lemma~\ref{lem:1pdmAndPairingFunction} below to compute the relevant expectations, we see that every operator $a_p$ is accompanied by a factor $z/|z|$ and every operator $a_p^*$ by a factor $\overline{z}/|z|$. All remaining expressions only depend on $|z|$. If the number of $a$'s and $a^*$'s does not match the integration over the phase of $z$ gives zero. Thus, we conclude that $[\Gamma_{\beta,N},\mathcal{N}]=0$. 
    \item From a physics point of view, the state $\Gamma_{\beta,N}$ in \eqref{eq:referenceState-intro1} is given as a convex combination over all pure phases. Here, a pure phase is represented by a state in which each complex number $z$ appearing in the integration has a fixed phase. Since we did not add a symmetry-breaking perturbation to the Hamiltonian, all pure phases appear with the same weight in the integration. The convex combination is given by an integral over the Gibbs distribution $g^{\mathrm{BEC}}$ in \eqref{eq:GibbsDistributionDiscrete}. This function has been recently been used in \cite{BocDeuSto-24,CapDeu-23} to derive an upper bound for the free energy of the Bose gas in the Gross--Pitaevskii limit and it is conceptually related to the classical field theory in \cite{LewNamRou-21,FroKnoSchSoh-22}. Theorem~\ref{thm:norm-approximation} provides the first rigorous derivation of this effective condensate theory.
    \item The function $g^{\mathrm{BEC}}$ describes the particle number fluctuations of the condensate. As we will see later, its variance is of the order $N^{5/3}$ if $N_0(\beta,N) \gesssim N^{5/6}$. This should be compared to the variance of the distribution of the number of particles in the coherent state $|z \rangle$, which if of the order $|z|^{2}$. Since $g^{\mathrm{BEC}}(z)$ is peaked around $|z|^2 \sim N_0(\beta,N) \sim N$ if $\kappa > \frac{1}{4 \pi}[\upzeta(3/2)]^{2/3}$ this variance is roughly of order $N$. We conclude that the main contribution to the variance of the number of particles in the condensate comes from $g^{\mathrm{BEC}}$.
    \item Theorem~\ref{thm:norm-approximation} and our other main results in Section~\ref{sec:MainResultsPart2} are stated with the assumption $\beta N^{2/3} \to \kappa \in (0,\infty)$ with $\beta_{\mathrm{c}}$ in \eqref{eq:crittemp}. However, our techniques also allow us to treat the zero temperature limit $\beta N^{2/3} \to \infty$. For the sake of simplicity, we prefer to not provide the details.
    \end{enumerate}
\end{remark}

Using the result of Theorem~\ref{thm:norm-approximation} together with our two new abstract correlation inequalities, we obtain approximate expressions for the one- and two-particle density matrices, compute various limiting distributions for the number of particles in the condensate, and derive a free energy expansion. These additional results will be discussed in detail in Section~\ref{sec:MainResultsPart2}.

\subsection{Discussion of the proof of Theorem~\ref{thm:norm-approximation} and correlation inequalities}
\label{sec:DiscussionOfProof}

Let us continue our discussion by explaining a few selected ideas of the proof of Theorem~\ref{thm:norm-approximation}. As mentioned already earlier, our proof is based on a variational formulation of the problem, which we introduce first. 

The Gibbs variational principle states that the Gibbs state $G_{\beta,N}$ in \eqref{eq:interactingGibbsstate} is the unique minimizer of the Gibbs free energy functional  
\begin{equation}
	\mathcal{F}(\Gamma) = \tr[\mathcal{H}_N \Gamma ] - \frac{1}{\beta} S(\Gamma) \quad \text{ with the von-Neumann entropy } \quad S(\Gamma) = - \tr[\Gamma \ln(\Gamma)]
    \label{eq:GibbsFreeEnergyFunctional}
\end{equation}
in the set
\begin{equation}
	\mathcal{S}_N = \left\{ \Gamma \in \mathcal{B}(\mathscr{F}) \ | \ 0 \leq \Gamma, \tr \Gamma = 1, \tr [ \mathcal{N} \Gamma ] = N \right\}.
	\label{eq:states}
\end{equation}
The (grand canonical version of the) free energy is defined as the minimum of $\mathcal{F}$:
\begin{equation}
	F(\beta,N) = \min_{\Gamma \in \mathcal{S}_N} \mathcal{F}(\Gamma) = -\frac{1}{\beta} \ln\left( \tr \exp\left( -\beta (\mathcal{H}_N - \mu_{\beta,N} \mathcal{N}) \right) \right) + \mu_{\beta,N} N
	\label{eq:freeenergy}
\end{equation}
with the chemical potential $\mu_{\beta,N}$ in \eqref{eq:interactingGibbsstate} related to the Gibbs state $G_{\beta,N}$. The goal is to obtain sufficiently precise upper and lower bounds for the free energy, which is the main difficulty in the proof of Theorem~\ref{thm:norm-approximation}. With these bounds at hand, one can obtain \eqref{eq:norm-approximation-intro} with a simple argument that uses Pinsker's inequality for the relative entropy. We refer to Section~\ref{sec:traceNormBoundGibsState} for more details. An upper bound for $F(\beta,N)$ can be obtained by using $\Gamma_{\beta,N}$ in \eqref{eq:norm-approximation-intro} as a trial state. A similar upper bound has been obtained recently in the more challenging Gross--Pitaevskii limit in \cite{BocDeuSto-24,CapDeu-23}. The main difficulty is therefore to prove a corresponding lower bound. As we explain in more detail now, one needs to exploit the specific structure of the Gibbs state to achieve this goal. 

As explained above \eqref{eq:firstBogoliubovHamiltonianb}, Bogoliubov theory is based on two ingredients: the c-number substitution and the proof that all terms in the energy (apart from those captured in \eqref{BogoliubovHamiltonianAndi}) that can change the structure of the Gibbs state are small. While the first step can be easily implemented with a c-number substitution in the spirit of \cite{LieSeiYng-05}, the second step is highly nontrivial. As an example, let us consider the term
\begin{align}\label{eq:2-moment-intro}
\frac{1}{N} &\sum_{p,k,p+k \in \Lambda_+^*} \hat{v}(p)  \Tr [ (a^*_{k+p} a^*_{-p} a_k a_0 + \mathrm{h.c.}) G_{\beta,N}]. 
\end{align}
Since it does not appear in the definition of $G^{\mathrm{Bog}}(z)$ in \eqref{eq:BogoliubovGibbsState} we need to show that it is of the order $o(N^{2/3})$. In Section~\ref{sec:lowerBoundSimpliefiedHamiltonian} we show that, by using the translation-invariance of $G_{\beta,N}$ and the summability of $\hat{v}$, this term can be bounded in terms of the quantities $\sup_{p \in \Lambda_+^*} \Tr[(a_p^*a_p)^2 G_{\beta,N}]$ and $\sup_{p \in \Lambda_+^*} \Tr[B_p^2 G_{\beta,N}]$, where 
\begin{equation}
    B_p = \frac{1}{2} \left(\sum_{r,r+p \in \Lambda^*_+} a_{r+p}^* a_r + {\mathrm{h.c.}}\right) = \de \Upsilon ( Q \cos(p\cdot x) Q) 
\end{equation}
with $Q = \mathds{1}(-\Delta \neq 0)$. To prove that the term in \eqref{eq:2-moment-intro} is of order $o(N^{2/3})$, we show 
\begin{equation}
    \sup_{p \in \Lambda_+^*} \Tr[(a_p^*a_p)^2 G_{\beta,N}] + \sup_{p \in \Lambda_+^*} \Tr[B_p^2 G_{\beta,N}] \lesssim N^{4/3}. 
\end{equation}
These bounds are motivated by the fact that they hold for the Gibbs state of the ideal gas $G_{\beta,N}^{\mathrm{id}}$ in \eqref{eq:GibbsStateIdealGas}. In the following, we only explain how to prove the bound for $B_p$. It is interesting to note that the term in \eqref{eq:2-moment-intro} equals zero when we replace $G_{\beta,N}$ by $G_{\beta,N}^{\mathrm{id}}$. It is also interesting to note that for an approximate minimizer $\Gamma$ of $\mathcal{F}$ with a remainder of the order $O(N^{2/3})$ we only have $\sup_{p \in \Lambda_+^*} \Tr[B_p^2 \Gamma] \lesssim N^{5/3}$, which is not good enough to obtain a lower bound for the free energy.

In the first step of our proof we derive bounds for the free energy up to a remainder of the order $O(N^{2/3})$. Using these bounds and a Griffith (or Hellmann--Feynman) argument, we obtain the bound $|\Tr[B_p G_{\beta,N}]| \lesssim N^{2/3}$. To be more precise, we obtain this bound with $G_{\beta,N}$ replaced by the perturbed Gibbs state 
	\begin{equation}
    G_{\beta,N,t} = \frac{\exp(-\beta( \mathcal{H}_{N} - t B_p))}{\Tr[\exp(-\beta( \mathcal{H}_{N} - tB_p))]}, \quad t\in [-1,1],
    \label{eq:GeneralGibbsState-intro}
\end{equation}
that is, we have
\begin{equation}
    \sup_{t\in [-1,1]} |\Tr [ B_p G_{\beta,N,t} ]| \lesssim N^{2/3}.
          \label{eq:A-G-intro}
\end{equation}
To derive the claimed second moment estimate from the first moment bound in \eqref{eq:A-G-intro}, we use the following novel abstract correlation inequality, which is of independent interest. 

\begin{theorem}[Second order correlation inequality] \label{thm:correlation-intro} Let $A$ be a self-adjoint operator on a separable complex Hilbert space and assume that $\Tr [e^{-sA} ] <+\infty$ holds for all $s> 0$. Let $B$ be a symmetric operator that is $A$-relatively bounded with a relative bound strictly smaller than $1$. We also assume that the Gibbs state
\begin{equation}\label{eq:Gibbs-t-intro}
\Gamma_t = \frac{\exp(-A+tB)}{\Tr [ \exp(-A+tB) ]} , \quad t\in [-1,1]
\end{equation}
satisfies
\begin{align}\label{eq:CRI-condition-intro}
\sup_{t\in [-1,1]}|\Tr (B \Gamma_t )| \le a.
\end{align}
Then we have
\begin{equation}
\Tr [ B^2 \Gamma_0 ] \leq a e^{a} + \frac{1}{4} \Tr( [[B,A],B] \Gamma_0). 
\label{eq:StahlA2-intro}
\end{equation}
\end{theorem}

The idea of using the first moment estimate for a family of perturbed Gibbs states to deduce a second moment estimate for the original Gibbs state is inspired by the recent work of Lewin, Nam, and Rougerie \cite{LewNamRou-21}. Here such bounds have been a crucial input for implementing a semiclassical analysis in infinite dimensions. Theorem~\ref{thm:correlation-intro} is a considerably improved version of the abstract correlation inequality in \cite[Theorem~7.1]{LewNamRou-21}. The main new ingredient in the proof of Theorem~\ref{thm:correlation-intro} is the observation that \textit{Stahl's theorem} \cite{Stahl2013}, which was formerly known as the Bessis--Moussa--Villani (BMV) conjecture, implies an elegant convexity of the Duhamel two point function. Since we expect this to find applications also in other contexts, we now explain these ideas in some more detail. 

We start by noting that
\begin{equation}
    Z'(t) = \Tr [B \exp(-A+tB)], \quad \text{ with }\quad Z(t)=   \Tr [\exp(-A+tB)],     
\end{equation}
which follows from the cyclicity of the trace. If we take another derivative we find
\begin{align}\label{eq:Duhamel-2-point-intro}
Z''(t) = \partial_t \Tr [B \exp(A-tB)] = \int_0^1 \Tr [B \exp((A-tB)s) B \exp((A-tB)(1-s)))] \de s.
\end{align}
The term on the right-hand side of \eqref{eq:Duhamel-2-point-intro}, when divided by $Z(t)$, is called the Duhamel two-point function. It satisfies the following bound relating it to the second moment of $B$:
\begin{align} \label{eq:LNR}
0\le \Tr [B^2\Gamma_t ] - \int_0^1 \Tr [ B \Gamma_t^s B \Gamma_t^{1-s}] ds \le \frac{1}{4} \Tr ( [B,[A,B]] \Gamma_t).
\end{align} 
The above inequality is taken from \cite[Theorem 7.2]{LewNamRou-21}, which is a variant of  the Falk--Bruch inequality \cite{FalBru-69} (see also \cite[Theorem 3.1]{DysLieSim-78}).  

It therefore remains to control the Duhamel two-point function. 

The new observation in the present paper is that Stahl's theorem, see \cite{Stahl2013,Eremenko2015}, implies that $Z''(t)$ is convex. More precisely, an infinite-dimensional version of Stahl's theorem that we prove in Section~\ref{sec:StahlsTheorem} guarantees under the assumptions of Theorem~\ref{thm:correlation-intro} the existence of a nonnegative Borel measure $\mu$ on $\mathbb{R}$ such that  
    \begin{equation}
       Z(t) = \tr[ \exp(-A+tB) ] = \int_{-\infty}^{\infty} e^{ts} \de \mu(s)
        \label{eq:stahl5-intro}
    \end{equation}
holds. In other words, $Z(t)$ is the Laplace transform of a positive measure. Consequently, 
     \begin{equation}
    Z^{(2n)}(t) =  \int_{-\infty}^{\infty} s^{2n} e^{ts} \de \mu(s) \ge 0
        \label{eq:stahl5-6-intro}
    \end{equation}
holds for all $n \in \mathbb{N}$ even if the operator $B$ has no sign. In particular, $t \mapsto Z''(t)$ is convex. This allows us to estimate
\begin{equation}
    Z''(0) \le \frac{1}{2}\int_{-1}^1 Z''(s) \de s = \frac{1}{2} ( Z'(1) - Z'(-1) ) \leq \sup_{t \in [-1,1]} | \Tr [B \Gamma_t] | Z(t)  \leq a e^{a} Z(0).
    \label{eq:introAndi1}
\end{equation}
The obtain the last bound we additionally used $e^{-a}\le Z(t)/Z(0) \le e^a$ for all $t\in [-1,1]$, which follows from the assumption in \eqref{eq:CRI-condition-intro} and an application of Grönwall's inequality. Putting \eqref{eq:LNR} and \eqref{eq:introAndi1} together, we obtain \eqref{eq:StahlA2-intro}.
  
The main difference between our proof of Theorem~\ref{thm:correlation-intro} and the proof of \cite[Theorem~7.1]{LewNamRou-21} is that we use the exact convexity of $Z''(t)$, which follows from Stahl's theorem, while the analysis in  \cite{LewNamRou-21} is based on an approximate convexity of the function $t \mapsto \Tr [ B^2 \Gamma_t ]$ (this function is not known to be convex) as 
\begin{align}\label{eq:B2B3}
\partial_t \Tr [B^2 \Gamma_t] \approx \Tr [B^3 \Gamma_t], \quad  \partial_t \Tr [B^3 \Gamma_t] \approx \Tr [B^4 \Gamma_t] \ge 0. 
\end{align}
The approximation in \eqref{eq:B2B3} can be justified by \eqref{eq:LNR}, but  this requires one to control expectations of the form $\Tr ( [B^2,[B^2, A]] \Gamma_t)$ and $\Tr ( [B^3,[B^3, A]] \Gamma_t)$, which is clearly more involved than controlling the simpler term $\Tr ([B,[B,A]]\Gamma_0)$ appearing in \eqref{eq:StahlA2-intro}. Because of this, the abstract statement in \cite[Theorem 7.1]{LewNamRou-21} is more complicated than our Theorem~\ref{thm:correlation-intro}.

This concludes our discussion of the elements of the proof of Theorem~\ref{thm:norm-approximation}. Before providing more details on our other main results, we note that the correlation inequality in Theorem~\ref{thm:correlation-intro} is also sufficient to establish bounds for the 1-pdm of the Gibbs state and the condensate distributions. However, it is not sufficient to obtain our pointwise bounds for the 2-pdm of the Gibbs state in Fourier space and for the variances of $a^*_0 a_0$ and $\mathcal{N} - a_0^*a_0$. To achieve this, we require an extension of Theorem~\ref{thm:correlation-intro} to higher moments. Since this is another key novelty of our paper, we now explain it in more detail.

Except when $B$ commutes with $A$, obtaining higher-moment bounds is significantly more challenging than proving second-moment bounds. The main difficulty is that no extension of the Falk–Bruch inequality is known that relates the Duhamel $n$-point function to the $n$-th moment $\Tr[B^{n}\Gamma_t]$ when $n>2$. If such an extension existed, it would allow us to extend the proof strategy of Theorem~\ref{thm:correlation-intro} to this case. To overcome this issue, we introduce a second new abstract correlation inequality. 
\begin{theorem}[Higher order correlation inequality]\label{thm:higher-moments}
    Let $A$ and $B$ satisfy the assumptions of Theorem~\ref{thm:correlation-intro}, including that the Gibbs state $\Gamma_t$ in \eqref{eq:Gibbs-t-intro} satisfies \eqref{eq:CRI-condition-intro}. We assume in addition that there is a self-adjoint operator $X \geq 1$ such that $X$ is $A$-relatively bounded and 
\begin{align}\label{eq:corr-thm-ass2}
[A,X]=[B,X]=0, \quad \pm B\le X,\quad \pm [[B,A],B] \le b X^{\alpha}
\end{align}
hold with some constants $b > 0$ and $\alpha \in \mathbb{R}$. Then      for all even $k\in \mathbb{N}$ we have
    \begin{equation}\label{eq:corr-thm-conclusion-Gamma}
        \Tr [B^k \Gamma_0] \lesssim_{k} e^{2a}  \sup_{t\in [-1,1]} \Big\{ 1 + b^2 \Tr[ X^{k-2+2\alpha} \Gamma_t ] \Big\}.
    \end{equation}
Moreover, if $B\ge 0$, then for all $k\in \mathbb{N}$ we have
      \begin{equation}\label{eq:corr-thm-conclusion-B>0-Gamma}
       \Tr[ B^{k} \Gamma_0 ] \lesssim_{k} e^{2a} \sup_{|t| \le 1}  \Big(  1  + \sum_{\ell=1}^k b |\Tr ( [B, [B, A]] X^ {\alpha+\ell-3}\Gamma_t )| \Big). 
  \end{equation}
\end{theorem}
The main idea here is to introduce the additional operator $X$ that commutes with both, $A$ and $B$, and can be used to dominate $B$ as well as $[[A,B],B]$. In applications, it is chosen as $1 + \mathcal{N}$, which allows us to obtain bounds also in this more complicated setting. The trade-off, however, is that if the expectation of $B$ is substantially smaller than $X$ (which occurs, for example, when $B=a_p^* a_p$ with $p\in \Lambda_+^*$) then the bounds we obtain are meaningful only for lower moments. Two examples of bounds obtained using Theorem~\ref{thm:higher-moments} are
\begin{equation}
    \sup_{p \in \Lambda_+} \Tr[(a_p^*a_p)^4 G_{\beta,N}] \lesssim N^{8/3} \quad \text{ and } \quad \Tr[(\cN_+- \Tr[\cN_+ G_{\beta,N} ] )^4 G_{\beta,N}] \lesssim N^{8/3}
    \end{equation}
with $\cN_+=\sum_{q\in \Lambda_+^*} a_q^* a_q$, which are crucial to obtain information on the 2-pdm of the Gibbs state and the variance of the number of excited particles. 

In order to put the abstract correlation inequalities to good use, we need first-order bounds for suitably perturbed Gibbs states on the optimal scale. Already this step requires a substantial amount of new ideas. For more information we refer to Section~\ref{sec:organizationOfArticle}.  

In the next section, we will discuss our other main results in more detail.

\section{Main results, part II} \label{sec:MainResultsPart2}

In this section we discuss our other main results that have already been mentioned after Theorem~\ref{thm:norm-approximation}. We start with the results that concern the reduced density matrices.
\subsection{Reduced density matrices}
\label{sec:mainReducedDensityMatrices}
Before we provide our main results for the 1- and the 2-pdms of the Gibbs state, we introduce some notation. 

Let $\Gamma \in \mathcal{S}_N$ with $\mathcal{S}_N$ in \eqref{eq:states} be a state satisfying $\tr[ \mathcal{N}^{k} \Gamma ] < + \infty$ for $k \in \mathbb{N}$. We define the $k$-particle reduced density matrix ($k$-pdm) $\gamma_{\Gamma}^{(k)} \in \mathcal{B}(L^2)(\Lambda^{k}) $ of $\Gamma$ via its integral kernel  
\begin{equation}
    \gamma_{\Gamma}^{(k)}(p_1,...,p_{k};q_1,...,q_{k}) = \tr[ a_{q_1}^* ... a_{q_{k}}^* a_{p_1} ... a_{p_{k}} \Gamma]
    \label{eq:k1k2-pdm}
\end{equation} 
in Fourier space and note that it satisfies 
\begin{equation}
    \tr_{L^2(\Lambda^k)} [ \gamma_{\Gamma}^{(k)} ] = \tr[ \mathcal{N} (\mathcal{N} - 1) ... (\mathcal{N} - k + 1) \Gamma ].
    \label{eq:Normalizationk-pdm}
\end{equation}

We say that a sequence of states $\Gamma_N \in \mathcal{S}_N$ with $\mathcal{S}_N$ in \eqref{eq:states} indexed by the particle number displays Bose--Einstein condensation (BEC) iff
\begin{equation}
    \liminf_{N \to \infty} \sup_{\Vert \psi \Vert_2 = 1} \frac{\langle \psi, \gamma^{(1)}_{\Gamma_N} \psi \rangle}{N} > 0,
    \label{eq:definitionBEC}
\end{equation}
that is, iff the largest eigenvalue of $\gamma^{(1)}_{\Gamma_N}$ growths proportionally to $N$. The largest eigenvalue of $\gamma^{(1)}_{\Gamma_N}$ divided by $N$ and the corresponding eigenvector are called the condensate fraction and the condensate wave function, respectively. 

Our first main result in this section concerns the 1-pdm of the Gibbs state. 

\begin{theorem}[1-pdm]
\label{thm:1-pdm}
    Let $v$ satisfy the assumptions of Theorem~\ref{thm:norm-approximation}. We consider the limit $N \to \infty$, $\beta N^{2/3} \to \kappa \in (0,\infty)$. Then we have the following statements for the 1-pdm $\gamma_{\beta,N}$ of the Gibbs state $G_{\beta,N}$ in \eqref{eq:interactingGibbsstate}.
    \begin{enumerate}[label=(\alph*)]
    \item The 1-pdm satisfies
    \begin{equation}
          \bigg\Vert \ \gamma_{\beta,N} - \widetilde{N}_0 | \varphi_0 \rangle \langle \varphi_0 | - \sum_{p \in \Lambda_+^*} \gamma_p | \varphi_p \rangle \langle \varphi_p | \ \bigg\Vert_1 \lesssim N^{2/3-1/48}
            \label{eq:theorem21pdm}
        \end{equation} 
        with $\varphi_p(x) = e^{\mathrm{i} p \cdot x}$, $\gamma_p$ in \eqref{eq:gammap}, and $\widetilde{N}_0 = N - \sum_{p \in \Lambda_+^*} \gamma_p$. By $\Vert \cdot \Vert_1$ we denoted the trace norm. 
    \item If $\kappa > \frac{1}{4 \pi}[\upzeta(3/2)]^{2/3}$, then the eigenvalues $\gamma_{\beta,N}(p)$, $ p \in \Lambda^*$ of $\gamma_{\beta,N}$ satisfy
    \begin{align}
        | \gamma_{\beta,N}(0) - \widetilde{N}_0(\beta,N) | \lesssim N^{2/3-1/48}, \qquad 
        | \gamma_{\beta,N}(p) - \gamma_p | \lesssim \left( \frac{1}{\beta p^2} + 1 \right) N^{-1/96} \quad \forall p \in \Lambda_+^*. 
        \label{eq:theorem21pdmb}
    \end{align}
    \item If $\kappa < \frac{1}{4 \pi}[\upzeta(3/2)]^{2/3}$, then for $\mu_0$ in \eqref{eq:idealgase1pdmchempot} we have
    \begin{equation}
        \gamma_{\beta,N}(p) = \frac{1}{\exp(\beta(p^2 - \mu_0(\beta,N)))-1} + O(N^{-1/96}).
        \label{eq:theorem21pdmc}
    \end{equation}
    \end{enumerate}
\end{theorem}

Concerning the above theorem we have the following remarks.

\begin{remark}
\label{rem:1pdm}
    \begin{enumerate}[label=(\alph*)]
    \item The approximations in \eqref{eq:theorem21pdm} and \eqref{eq:theorem21pdmb} allow us to see the influence of the Bogoliubov modes, which is of the order $\exp(\beta(\epsilon(p))-1)^{-1} \simeq (\beta \epsilon(p))^{-1} \sim N^{2/3}$. This also concerns the effect of the Bogoliubov modes on the expected number of particle in the condensate because $\widetilde{N}_0(\beta,N) = N - \sum_{p \in \Lambda_+^*} \gamma_p = N_0(\beta,N) + O(N^{2/3})$ with $N_0$ in \eqref{eq:crittemp} and $\gamma_p$ in \eqref{eq:gammap}, see Lemma~\ref{lem:BoundN0} in Section~\ref{sec:upperBoundFreeEnergy}. 
    \item The pointwise bounds in \eqref{eq:theorem21pdmb} and \eqref{eq:theorem21pdmc} allow us to detect the leading order behavior of every single eigenvalue of the 1-pdm of the Gibbs state until the exponential decay of $(\exp(\beta p^2)-1)^{-1}$ starts for momenta with $|p| \gg N^{1/3}$. In this regime we only know that the eigenvalues vanish in the limit $N \to \infty$. 
    \item In combination, \eqref{eq:definitionBEC}, \eqref{eq:theorem21pdm} and $\widetilde{N}_0(\beta,N) = N_0(\beta,N) + O(N^{2/3})$, yield a proof of the BEC phase transition with critical temperature given by that of the ideal gas in \eqref{eq:crittemp} to leading order. Note that the BEC phase transition has been established for more challenging scaling limits in \cite{DeuSeiYng-19,DeuSei-20,DeuSei-21}.  
    \item The bound in \eqref{eq:theorem21pdmc} cannot be obtained with techniques that only rely on the use of coercivity. For instance, while we know the (grand canonical version of the) free energy up to a remainder of order $N^{5/8}$, see Theorem~\ref{thm:main1} below, this alone only allows control over the number of particles with energy $e \sim 1$ up to the same accuracy. However, using our correlation inequalities, we demonstrate that the expectation $\Tr[a_p^* a_p G_{\beta,N}]$ of the number of particles with momentum $p \in \Lambda^*$ is of order $1$ in $N$ for $\kappa < \frac{1}{4 \pi}[\upzeta(3/2)]^{2/3}$, and we compute the corresponding constant explicitly. We highlight the remarkable accuracy of this result.   
    \end{enumerate}
\end{remark}

Let us define the pairing function
\begin{equation}
    \alpha_p = u_{p} v_{p} \left\{ 2  \gamma^{\mathrm{Bog}}_{p} + 1 \right\} \quad \text{ with } \quad \gamma^{\mathrm{Bog}}_p = \frac{1}{\exp(\beta \epsilon(p))-1}
    \label{eq:alphap}
\end{equation}
and $u_p$, $v_p$, $\epsilon(p)$ defined below \eqref{eq:gammap}. It is related to the Bogoliubov Gibbs state in \eqref{eq:BogoliubovGibbsState}, see Lemma~\ref{lem:1pdmAndPairingFunction} in Section~\ref{sec:upperBoundFreeEnergy}. 

While the 1-pdm provides insights into the BEC phase transition and the formation of quasi-particles, the 2-pdm is essential for capturing all other interaction-induced correlations between the particles. For the 2-pdm of the Gibbs state we have the following statement. For the sake of simplicity we restrict attention to the cases $\kappa > \frac{1}{4 \pi}[\upzeta(3/2)]^{2/3}$ and $\kappa < \frac{1}{4 \pi}[\upzeta(3/2)]^{2/3}$.

\begin{theorem}[2-pdm and particle number variances]\label{thm:main2}
	Let $v$ satisfy the assumptions of Theorem~\ref{thm:norm-approximation}. We consider the limit $N \to \infty$, $\beta N^{2/3} \to \kappa \in (0,\infty)$. Then we have the following statements for the 2-pdm of the Gibbs state $G_{\beta,N}$ in \eqref{eq:interactingGibbsstate}.
    \begin{enumerate}[label=(\alph*)]
    \item Assume that $\kappa > \frac{1}{4 \pi}[\upzeta(3/2)]^{2/3}$ and that $p,q,r,s \in \Lambda_+^*$. Then we have
    \begin{align}
        \Tr[(a_0^* a_0)^2 G_{\beta,N}] &= \widetilde{N}^2_0(\beta,N) + N/(\beta \hat{v}(0)) + O(N^{5/3-1/12}), \nn \\
        \Tr [ a_0^* a^*_0 a_p a_{-p} G_{\beta,N} ] &= \widetilde{N}_0(\beta,N) \alpha_p + O(N^{5/3-1/96}), \nn \\
        \Tr [a_0^* a_0 a_p^* a_p  G_{\beta,N} ]  &=  \widetilde{N}_0(\beta,N) \gamma_p + O(N^{5/3-1/96}), \nn \\
        |\Tr[ a_p^* a_q^* a_r a_0 ]| &\lesssim N^{3/2}, \nn \\
        \Tr [a_p^* a^*_q a_{r} a_{s} G_{\beta,N}] &=\delta_{p,-q}\delta_{r,-s} \alpha_p \alpha_r  + (\delta_{p,r}\delta_{q,s} +\delta_{p,s}\delta_{q,r} ) \gamma_p\gamma_q + O(N^{4/3-1/96}), \nn \\
        \Tr[ \mathcal{N}_+^2 G_{\beta,N}] - (\Tr [ \mathcal{N}_+ G_{\beta,N} ])^2 &= \sum_{p \in \Lambda_+^*} ( \alpha_p^2 + \gamma_p^2 + \gamma_p ) + O(N^{4/3-1/96}). \label{eq:2-pdmBECPhase}
    \end{align}
    with $\gamma_p$ in \eqref{eq:gammap}, $\widetilde{N}_0 = N - \sum_{p \in \Lambda_+^*} \gamma_p$, $\alpha_p$ in \eqref{eq:alphap}, and $\mathcal{N}_+ = \sum_{p \in \Lambda_+^*} a_p^* a_p$.
    \item If $\kappa < \frac{1}{4 \pi}[\upzeta(3/2)]^{2/3}$ and $p,q,r,s \in \Lambda_+^*$ then we have
    \begin{align}
        \Tr[(a_0^* a_0)^2 G_{\beta,N}] &= 2 N_0^2(\beta,N) + N_0(\beta,N) + O(N^{-1/144}), \nn \\
        |\Tr [ a_0^* a^*_0 a_p a_{-p} G_{\beta,N} ]| &\lesssim N^{-1/144}, \nn \\
        \Tr [a_0^* a_0 a_p^* a_p  G_{\beta,N} ] &= \frac{N_0(\beta,N)}{\exp(\beta(p^2-\mu_0(\beta,N)))-1}  + O(N^{-1/144} ), \nn \\
        | \Tr[ a_p^* a_q^* a_r a_0 ]| &\lesssim N^{-1/144}, \label{eq:2-pdmNonCondensedPhase} \\
        \Tr [a_p^* a^*_q a_{r} a_{s} G_{\beta,N}] 
        &=  \frac{(\delta_{p,r}\delta_{q,s} +\delta_{p,s}\delta_{q,r} )}{(\exp(\beta(p^2 - \mu_0(\beta,N)))-1)(\exp(\beta(q^2 - \mu_0(\beta,N)))-1)} + O(N^{-1/144}), \nn \\
        \Tr[ \mathcal{N}_+^2 G_{\beta,N}] - (\Tr [ \mathcal{N}_+ G_{\beta,N} ])^2 &= \sum_{p \in \Lambda_+^*} \frac{1}{\exp(\beta(p^2 - \mu_0))-1}  \left( 1+ \frac{1}{\exp(\beta(p^2 - \mu_0))-1} \right) + O(N^{4/3-1/144}) \nn
    \end{align}
    with $\mu_0(\beta,N)$ and $N_0(\beta,N)$ in \eqref{eq:idealgase1pdmchempot} and \eqref{eq:crittemp}, respectively.
    \end{enumerate}
\end{theorem}

We have the following remarks concerning the above result.

\begin{remark}
\label{rem:2pdm}
    \begin{enumerate}[label=(\alph*)]
    \item The above result shows that the 2-pdm of $G_{\beta,N}$ can be approximated by that of $\Gamma_{\beta,N}$, as defined in Theorem~\ref{thm:norm-approximation}. An inspection of the 2-pdm of $\Gamma_{\beta,N}$ reveals that this state is not quasi-free if $N_0 \geq N^{2/3}$. A translation-invariant quasi free state is fully characterized by the three functions $\psi_p = \langle a_p \rangle$, $\alpha_p = \langle a_p a_{-p } \rangle$, and $\rho_p = \langle a_p^* a_p \rangle$. As argued in part (a) of Remark~\ref{rem:Theorem1}, the functions $\psi_p$ and $\alpha_p$ vanish in the state $\Gamma_{\beta,N}$. Hence, if it were quasi-free, its 2-pdm would solely be determined by $\Tr[ a_p^* a_p \Gamma_{\beta,N} ] = \gamma_p$, which is incorrect. 
    \item As mentioned at the end of Section~\ref{sec:DiscussionOfProof}, we require higher-order moment estimates to prove Theorem~\ref{thm:main2}. In contrast, the proofs of all other results are based on applications of Theorem~\ref{thm:correlation-intro}. In Section~\ref{sec:higherOrderCorrelationInequalities} we establish Theorem~\ref{thm:higher-moments}, a higher-order abstract correlation inequality, and use it to prove the desired moment bounds. The proof of Theorem~\ref{thm:main2} constitutes the second main technical novelty of our paper, the first being the proof of Theorem~\ref{thm:norm-approximation}. Finally, we emphasize that the 2-pdm cannot be accessed using the methods of \cite{DeuSeiYng-19,DeuSei-20,DeuSei-21}.
    \item For $\kappa > \frac{1}{4 \pi}[\upzeta(3/2)]^{2/3}$ the 2-pdm displays correlations that are induced by the interaction between the particles. This is reflected for example in the presence of $\alpha_p$ in the second, fifth and sixth line of \eqref{eq:2-pdmBECPhase}. The first term in the fifth line is often used as a signature for the presence of superfluidity in the system. It is of order $N^{4/3}$ in the condensed phase and, as the equation in the fifth line of \eqref{eq:2-pdmNonCondensedPhase} shows, it vanishes in the non-condensed phase corresponding to $\kappa < \frac{1}{4 \pi}[\upzeta(3/2)]^{2/3} $. The main contributions in \eqref{eq:2-pdmBECPhase} and \eqref{eq:2-pdmNonCondensedPhase} are given by the related expectations in the state $\Gamma_{\beta,N}$ in \eqref{eq:norm-approximation-intro}. Parts (a) and (d) of Remark~\ref{rem:1pdm} also apply to Theorem~\ref{thm:main2}.
    \item In combination, \eqref{eq:theorem21pdmb} and the first equation in \eqref{eq:2-pdmBECPhase} show that the variance of the number of particles in the BEC depends on the interaction between the particles and is of order $N^{5/3}$ if $\kappa > \frac{1}{4 \pi}[\upzeta(3/2)]^{2/3}$. This should be compared to the same quantity in the ideal gas, which is of order $N^2$. For $\kappa < \frac{1}{4 \pi}[\upzeta(3/2)]^{2/3}$ the variance of $a_0^* a_0$ is of order one with a constant that is determined by the ideal gas, see the first equation in \eqref{eq:2-pdmNonCondensedPhase}. As the last equations in \eqref{eq:2-pdmBECPhase} and \eqref{eq:2-pdmNonCondensedPhase} show, the variance of the number of thermally excited particles in always of order $N^{4/3}$. However, the constant multiplying this growth differs for $\kappa > \frac{1}{4 \pi}[\upzeta(3/2)]^{2/3}$ and $\kappa < \frac{1}{4 \pi}[\upzeta(3/2)]^{2/3} $. In the former case it depends on the interaction between the particles and in the latter case it does not. That is, also this quantity allows one to distinguish between the condensed and the non-condensed phases. 
    \end{enumerate}
\end{remark}

\subsection{Particle number distributions for the condensate}

Theorem~\ref{thm:norm-approximation} also allows us to access quantities that go strictly beyond one- and two-particle correlations. One example of particular interest for physicists is the distribution of the number of particles in the Bose--Einstein condensate. In Theorems \ref{thm:particleNumberDistributionBEC1} and~\ref{thm:particleNumberDistributionBEC2} below we present a detailed classification of its limiting distributions.

In the first statement we consider a quantity defined via coherent states.

\begin{theorem}[Particle number distribution BEC, Part I]
\label{thm:particleNumberDistributionBEC1}
    Let $v$ satisfy the assumptions of Theorem~\ref{thm:norm-approximation}. We consider the limit $N \to \infty$, $\beta N^{2/3} \to \kappa \in (0,\infty)$ and assume that $N_0(\beta,N) \gesssim N^{5/6 + \epsilon}$ holds with some fixed $0 < \epsilon \leq 1/6$. Then the probability distribution 
    \begin{equation}
        \zeta_G(z) = \Tr[ |z \rangle \langle z | G_{\beta,N} ]
        \label{eq:condensateDistributionGibbsStateContinuous}
    \end{equation}
     on $\mathbb{C}$ with respect to the measure $\de z = \de x \de y /\pi$, where $x = \Re \ z$ and $y = \Im \ z$, satisfies
    \begin{equation}
        \int_{\mathbb{C}} | \zeta_G(z) - g(|z|^2) | \de z \lesssim N^{-1/48}.
    \end{equation}
    Here $|z\rangle$ denotes the coherent state in \eqref{eq:coherentstate} and $g$ the Gaussian distribution
    \begin{equation}
        g(x) = \sqrt{ \frac{\beta \hat{v}(0)}{2 \pi N} } \exp\left( - \frac{\beta \hat{v}(0)}{2 N} \left( x - N_0(\beta,N) \right)^2 \right)
        \label{eq:Gaussian}
    \end{equation}
    with mean $N_0(\beta,N)$ and variance $N/(\beta \hat{v}(0))$. 
\end{theorem}

\begin{remark}
\label{rem:distributionBEC1}
    The particle number distribution related to the coherent state $|z \rangle$ is given by a Poisson distribution. Since $g$ is peaked around $\widetilde{N}_0 \sim N_0 \sim N$ if $\kappa > \frac{1}{4 \pi}[\upzeta(3/2)]^{2/3}$ the variance of the number of particles in $|z \rangle$ is of order $N$ in the relevant parameter regime. Accordingly, the variance of the number of particles in the condensate is described by $g$ in \eqref{eq:Gaussian} and is given by $N/(\beta \hat{v}(0)) \sim N^{5/3}$. This should be compared to the first line of \eqref{eq:2-pdmBECPhase} and to the result in Theorem~\ref{thm:particleNumberDistributionBEC2} below. We highlight that it is possible to approximate the expected number of particles in the condensate by $N_0$ instead of $\widetilde{N}_0$ in Theorem~\ref{thm:particleNumberDistributionBEC1}. The statement holds in the same way when $N_0$ is replaced by $\widetilde{N}_0$.
\end{remark}

In the following theorem we investigate limiting distributions for the number of particles in the condensate. Before we state our result we introduce some notation. Let us define the probability distribution
\begin{equation}
    g_0(x) = \frac{\exp\left( - \beta \left( \frac{\hat{v}(0)}{2N} x^2 - \mu x\right) \right)}{\int_0^{\infty} \exp\left( - \beta \left( \frac{\hat{v}(0)}{2N} x^2 - \mu x\right) \right) \de x},
    \label{eq:g0}
\end{equation}
where $\mu$ is chosen such that $\int_0^{\infty} x g(x) \de x = \widetilde{N}_0(\beta,N)$ holds with $\widetilde{N}_0$ below \eqref{eq:theorem21pdm}. The variance related to $g_0$ is denoted by $\mathbf{Var}(\beta,N)$. We also introduce the random variables $\mathbf{N}_0$ and $\widetilde{\mathbf{N}}_0$ by
\begin{equation}
    \mathbf{P}(\mathbf{N}_0 = n_0) = \Tr[ | n_0 \rangle \langle n_0 | G_{\beta,N} ] \quad \text{ and } \quad \widetilde{\mathbf{N}}_0 = \frac{\mathbf{N}_0 - \widetilde{N}_0(\beta,N)}{\sqrt{\mathbf{Var}(\beta,N)}},
    \label{eq:condensateDistributionGibbsStateDiscreteintro}
\end{equation}
where $| n_0 \rangle = 1/(\sqrt{n_0!}) (a^*_0)^{n_0} \Omega_0$ with the vaccuum vector $\Omega_0$ of the Fock space $\mathscr{F}_0$ in \eqref{eq:exponentialProperty} denotes the occupation number vector with $n_0 \in \mathbb{N}_0$ particles. Since the notations are similar, we highlight that it should not be confused with the coherent state in \eqref{eq:coherentstate}, which has a fluctuating particle number. The quantity $\mathbf{P}(\mathbf{N}_0 = n_0)$ equals the probability to find exactly $n_0$ particle with $p=0$ in the mean-field Bose gas described by the Gibbs state $G_{\beta,N}$. 

\begin{theorem}[Particle number distribution BEC, part II]
\label{thm:particleNumberDistributionBEC2}
    Let $v$ satisfy the assumptions of Theorem~\ref{thm:main1}. We consider the limit $N \to \infty$, $\beta N^{2/3} \to \kappa \in (0,\infty)$. Then we have
    \begin{equation}
        | N_0(\beta,N) - \widetilde{N}_0(\beta,N) | \lesssim  \frac{N_0(\beta,N)}{N^{1/3}}
        \label{eq:boundN0Theorem}
    \end{equation}
    with $\widetilde{N}_0$ below \eqref{eq:referenceState-intro1}, and the following statements hold.
    \begin{enumerate}[label=(\alph*)]
    \item If $N_0(\beta,N) \gg N^{5/6}$ we have 
    \begin{equation}
        \lim_{N \to \infty} \frac{\beta \hat{v}(0) \textbf{Var}(\beta,N)}{N} = 1
        \label{eq:varianceBoundaTheorem}
    \end{equation}
    and, as $N \to \infty$, the random variable $\widetilde{\mathbf{N}}_0$ converges in distribution to a standard normal random variable. 
    \item If $N_0(\beta,N) = t N^{5/6}$ with some fixed $t \in \mathbb{R}$ the parameter $\sigma = \mu \sqrt{\beta N/(2\hat{v}(0))}$ does not depend on $N$ and we have 
    \begin{equation}
        \lim_{N \to \infty} \frac{\beta \hat{v}(0) \textbf{Var}(\beta,N)}{2N} = B^2 = \frac{\int_{-\sigma}^{\infty} x^2 \exp(-x^2) \de x}{\int_{-\sigma}^{\infty} \exp(-x^2) \de x} - \left( \frac{\int_{-\sigma}^{\infty} x \exp(-x^2) \de x}{\int_{-\sigma}^{\infty} \exp(-x^2) \de x} \right)^2.
    \label{eq:varianceBoundbTheorem}
    \end{equation}
    Moreover, the distribution of $\widetilde{\mathbf{N}}_0$ converges, as $N \to \infty$, to
    \begin{equation}
            f_{\sigma,A,B}(x) = \frac{\exp\left(-\left( x B + A \right)^2 \right)}{\int_{\frac{-\sigma-A}{B}}^{\infty} \exp\left(-\left( x B + A \right)^2 \right) \de x} \quad \text{ with } \quad A = \frac{\int_{-\sigma}^{\infty} x \exp(-x^2) \de x }{\int_{-\sigma}^{\infty} \exp(-x^2) \de x}
            \label{eq:lemmaCharacteristicFunctiondaTheorem}
        \end{equation}
    and $B$ in \eqref{eq:varianceBoundbTheorem}.
    \item If $ 1 \ll N_0(\beta,N) \ll N^{5/6}$ we have
    \begin{equation}
        \lim_{N \to \infty} \frac{\textbf{Var}(\beta,N)}{[\widetilde{N}_0(\beta,N)]^2} = 1
        \label{eq:varianceBoundc_intro}
    \end{equation}
    and $\widetilde{\mathbf{N}}_0$ converges, as $N \to \infty$, in distribution to an exponential random variable with distribution
    \begin{equation}
        f(x) = \exp(-(1+x)) \quad \text{ with } \quad x \in [-1,\infty).
    \end{equation}
    \item If $N_0(\beta,N) = t$ with some $t > 0$ then $\beta \mu_0(\beta,N)$ does not depend on $N$ and $\mathbf{N}_0$ converges, as $N \to \infty$, in distribution to a geometric random variable with law
    \begin{equation}
        q(n) = \exp(\beta \mu_0(\beta,N)n)(1-\exp(\beta \mu_0(\beta,N))).
    \end{equation}
    \end{enumerate}
\end{theorem}   

\begin{remark}
\begin{enumerate}[label=(\alph*)]
    \item Theorem~\ref{thm:particleNumberDistributionBEC2} yields a complete description of the fluctuations of the number of particles in the BEC in our range of parameters. If $N_0 \gg N^{5/6}$ these fluctuations are given by a Gaussian. As shown already in \eqref{eq:2-pdmBECPhase} the related fluctuations live on the anomalously large scale $N^{5/6}$. If $N_0 \sim N^{5/6}$ the limiting distribution is more complicated, see $f_{\sigma,A,B}$ in \eqref{eq:lemmaCharacteristicFunctiondaTheorem}. Since for $N_0 \ll N^{5/6}$ we find an exponential distribution this reveals a second phase transition occurring in the system around the scaling $N_0(\beta,N) \sim N^{5/6}$. This is related to the fact that for $N_0 \ll N^{5/6}$, the interaction $\hat{v}(0)$ in $g_0$ in \eqref{eq:g0} becomes negligible.
    \item The above Theorem is another example of a statement that cannot be obtained with techniques that are only based on the use of coercivity. As already mentioned in point~(d) of Remark~\ref{rem:1pdm}, we know the free energy up to a remainder of order $N^{5/8}$, which only allows control over the number of particles with energy $e \sim 1$ up to the same accuracy. However, as point~(d) of Theorem~\ref{thm:particleNumberDistributionBEC2} shows, if $\Tr[a_0^* a_0 G_{\beta,N}] \sim 1$ we can compute the limiting distribution of the number of particles in the BEC.
    \item Theorem~\ref{thm:particleNumberDistributionBEC2} should be compared to \cite[Theorem~6]{BuffPul1983}, where the imperfect Bose gas has been considered. In contrast to this result, we see a dependence of the variance of $\mathbf{N}_0$ on the interaction, and we resolve the change of its distribution across the critical point.
    \end{enumerate}
    
\end{remark}

The two theorems above can be interpreted from a propabilistic viewpoint as follows. In \cite{DeuSeiYng-19,DeuSei-20,DeuSei-21} the BEC phase transition has been proved. To understand this transition, one needs to understand the leading order behavior of $\Tr[a_0^* a_0 G_{\beta,N}]$ as $N \to \infty$ as a function of $\beta N^{2/3}$. That is, one needs to prove a law of large numbers type statement for the random variable $\mathbf{N}_0$ defined in \eqref{eq:condensateDistributionGibbsStateDiscreteintro}. The statements in Theorems~\ref{thm:particleNumberDistributionBEC1} and~\ref{thm:particleNumberDistributionBEC2} are therefore central limit type theorems for an interacting quantum many-particle system as they concern the fluctuations of $\mathbf{N}_0$. From this perspective it is clear that their proof requires a very precise understanding of the Gibbs state. 

The above statements should also be compared to results for other (classical and quantum) statistical mechanics models, where fluctuations of observables can be studied, see e.g. \cite{BourgErdYau-2014,LeblSerf-2017,BekLebSerf-2018,Serfaty-2018,BorGuiKoz-2018,Gau-2021,SerfToAppear} for the classical case and \cite{AroKirkSchl-2013,BuSafSchl-2014,RaSchl-2019,KirkRadeSchl-2021,LewNamRou-21,
	FroKnoSchSoh-22,RaSei-2022,NamRad-23,DelLam-2024} for the quantum case.

\subsection{Free energy expansion}

In this section we discuss the asymptotic expansion of the free energy $F(\beta,N)$ as $N \to \infty$. It is a key ingredient for the proofs of our other main results but also of independent interest. Before we state our result, we introduce an effective free energy functional that describes the interacting BEC as well as the free energy related to the Bogoliubov Hamiltonian in \eqref{BogoliubovHamiltonianAndi}. The minimum of the condensate functional is a discrete version of the effective free energy that appeared recently in \cite{BocDeuSto-24} and later in \cite{CapDeu-23}. Our discrete version has the advantage of being the correct effective model also if order $1$ particles occupy the $p = 0$ momentum mode.

\subsubsection*{Effective condensate functional}

We define the free energy of a probability distribution $p$ on $\mathbb{N}_0$ by
\begin{equation}
	\mathcal{F}^{\mathrm{BEC}}(p) = \frac{\hat{v}(0)}{2N} \sum_{n =0}^{\infty} n^2 p(n) - \frac{1}{\beta} S(p), \quad \text{ where } \quad S(p) = - \sum_{n=0}^{\infty} p(n) \ln(p(n)) 
	\label{eq:condensatefunctional}
\end{equation}
denotes the classical entropy of $p$. The function $p$ should be interpreted as the distribution of the number of particles in a BEC. The free energy of a BEC with an expected number of $M$ particles is defined by
\begin{equation}
	F^{\mathrm{BEC}}(\beta,M) = \inf \left\{ \mathcal{F}^{\mathrm{BEC}}(p) \ \bigg| \ p \geq 0, \sum_{n = 0}^{\infty} p(n) = 1, \sum_{n = 0}^{\infty} n p(n) = M \right\} - \frac{\hat{v}(0) M^2}{2 N}.
	\label{eq:condensateFreeEnergy}
\end{equation}
We subtract the constant $\hat{v}(0) M^2/(2 N)$ because this allows us to state our main result in a more intuitive form. Because of the subtracted constant, $F^{\mathrm{BEC}}(\beta,M)$ should be interpreted as the free energy associated to the \textit{particle number fluctuations} in the BEC. 

The unique minimizer of the minimization problem in \eqref{eq:condensateFreeEnergy} is the Gibbs distribution
\begin{equation}
	g_{\beta,M}^{\mathrm{BEC}}(n) = \frac{\exp\left( -\beta\left( \frac{\hat{v}(0)}{2N} n^2 - \mu^{\mathrm{BEC}}(\beta,M) n \right) \right)}{\sum_{n=0}^{\infty} \exp\left( -\beta\left( \frac{\hat{v}(0)}{2N} n^2 - \mu^{\mathrm{BEC}}(\beta,M) n \right) \right)},
 \label{eq:GibbsdistributiomDiscrete}
\end{equation}
where the chemical potential $\mu^{\mathrm{BEC}}$ is chosen such that $\sum_{n \in \mathbb{N}_0} n g^{\mathrm{BEC}}_{\beta,M}(n) = M$ holds.

\subsubsection*{Bogoliubov free energy}

We define the free energy related to Bogoliubov Hamiltonian $\mathcal{H}^{\mathrm{Bog}}(z)$ in \eqref{BogoliubovHamiltonianAndi} by
\begin{align}
	F^{\mathrm{Bog}}(\beta,N) &= -\frac{1}{\beta} \ln\left( \tr_{\mathscr{F}_+} \exp\left( -\beta \mathcal{H}^{\mathrm{Bog}}(z) \right) \right) + \mu_0(\beta,N) \ \tr_{\mathscr{F}_+} [ \mathcal{N}_+ G^{\mathrm{Bog}} ] \nonumber \\
	&= \frac{1}{\beta} \sum_{p \in \Lambda^*_+} \ln\left( 1 - \exp\left( -\beta \epsilon(p) \right) \right) + \mu_0(\beta,N) \sum_{p \in \Lambda_+^*} \gamma_p,
	\label{eq:BogoliubovFreeEnergy}
\end{align}
where $\mathcal{N}_+ = \sum_{p \in \Lambda_+^*} a_p^* a_p$. We highlight that the arguments $\beta$ and $N$ of $F^{\mathrm{Bog}}$ do not indicate that the expected number of particles in the state $G^{\mathrm{Bog}}(z)$ in \eqref{eq:BogoliubovGibbsState} at inverse temperature $\beta$ equals $N$, but rather that the parameters $\mu_0$ and $N_0$ are evaluated at the point $(\beta,N)$. We also highlight that, although $G^{\mathrm{Bog}}(z)$ depends on $z \in \mathbb{C}$, $F^{\mathrm{Bog}}(\beta,N)$ does not.

\subsubsection*{Asymptotic expansion of the free energy}

We are now prepared to state our main results concerning the asymptotic expansion of the free energy $F(\beta,N)$. 

\begin{theorem}[Free energy expansion]\label{thm:main1}
	Let $v$ satisfy the assumptions of Theorem~\ref{thm:norm-approximation}. We consider the limit $N \to \infty$, $\beta N^{2/3} \to \kappa \in (0,\infty)$. The free energy $F(\beta,N)$ in \eqref{eq:freeenergy} satisfies
	\begin{equation}
		F(\beta,N) = F^{\mathrm{Bog}}(\beta,N) + \frac{\hat{v}(0) N}{2} + F^{\mathrm{BEC}}(\beta,N_0(\beta,N)) + O\left( N^{5/8} \right),
        \label{eq:ThmFreeEnergyBounds}
	\end{equation}
	with $F^{\mathrm{BEC}}$ in \eqref{eq:condensateFreeEnergy}, $F^{\mathrm{Bog}}$ in \eqref{eq:BogoliubovFreeEnergy}, and $N_0(\beta,N)$ in \eqref{eq:crittemp}. 
\end{theorem}

The terms on the right-hand side of \eqref{eq:ThmFreeEnergyBounds} are listed in descending order concerning their growth in $N$. The Bogoliubov free energy $F^{\mathrm{Bog}}(\beta,N)$ is of the order $\beta^{-1} N \sim N^{5/3}$ and satisfies $F^{\mathrm{Bog}}(\beta,N) = F_0^+(\beta,N) + \mathrm{const.} \ N^{2/3} + o(N^{2/3})$ with the free energy 
\begin{equation}
    F_0^+(\beta,N) = \frac{1}{\beta} \sum_{p \in \Lambda_+^*} \ln\left( 1- \exp\left( \beta(p^2 -\mu_0(\beta,N)) \right) \right)
    \label{eq:freeEnergyIdealGasCloud}
\end{equation}
of the thermally excited particles in the ideal gas. The correction of the order $N^{2/3}$ is related to the emergence of the Bogoliubov dispersion relation in \eqref{eq:gammap} in the condensed phase ($\kappa > \frac{1}{4 \pi}[\upzeta(3/2)]^{2/3}$). The term $\frac{\hat{v}(0) N}{2}$ in \eqref{eq:ThmFreeEnergyBounds} is a mean-field interaction between the particles. Finally, the free energy of the fluctuations of the number of particles in the BEC, $F^{\mathrm{BEC}}(\beta,N_0(\beta,N))$, yields two contributions if $\kappa > \frac{1}{4 \pi}[\upzeta(3/2)]^{2/3}$, one on the order $N^{2/3} \ln(N)$ and one on the order $N^{2/3}$. 
\begin{remark}
\begin{enumerate}[label=(\alph*)]
\item The fact that we use a discrete effective condensate theory allows us to provide a simple formula for the free energy of the BEC that is valid in all parameter regimes. This should be compared to the analysis in \cite{{BocDeuSto-24}}, where a continuous version of our condensate free energy, which is minimized by $g^{\mathrm{BEC}}$ in \eqref{eq:GibbsDistributionDiscrete}, is used, see \cite[(1.16), (1.17)]{BocDeuSto-24}. The minimum in \cite[(1.16)]{BocDeuSto-24} is needed because the continuous effective condensate theory in \cite[(1.17)]{BocDeuSto-24} fails to describe the free energy of the condensate if $N_0 \sim 1$. 
\item The discrete free energy functional in \eqref{eq:condensatefunctional} plays a crucial role in our proof of the first-order a priori estimates in Section~\ref{sec:FirstOrderEstimates}. There, we use the insight that the variance of the number of condensed particles and that of the total particle number agree to leading order. The former can be described by a c-number substitution, which motivates the appearance of $g^{\mathrm{BEC}}(z)$ in \eqref{eq:GibbsDistributionDiscrete}. However, when studying the free energy related to particle number fluctuations, we are naturally led to $\mathcal{F}^{\mathrm{BEC}}$ in \ref{eq:condensatefunctional}.
\item One can have the idea to try to approximate the Gibbs state $G_{\beta,N}$ by a reference state that is constructed with $g_{\beta,M}^{\mathrm{BEC}}(n)$ in \eqref{eq:GibbsdistributiomDiscrete} instead of $g^{\mathrm{BEC}}(z)$. However, we believe that this is not possible. The reason is that if we replace $|z \rangle \langle z |$ in the definition of the state $\Gamma$ below \eqref{eq:coherentstateEigenfunction} by $ | n_0 \rangle \langle n_0 |$ with $| n_0 \rangle$ in \eqref{eq:condensateDistributionGibbsStateDiscreteintro} then the terms proportional to $a_0^* a_0^* a_p a_{-p}$ and $a_p^* a_{-p}^* a_0 a_0$ in the energy in \eqref{eq:firstBogoliubovHamiltonian} vanish. 
\end{enumerate}
\end{remark}

The next statement provides us with a simplified formula for $F^{\mathrm{BEC}}(\beta,N_0(\beta,N))$ in two parameter regimes (condensed and non-condensed phase). A similar statement has been proved for a continuous version of our effective condensate theory in \cite[Proposition~1.2]{BocDeuSto-24}.

\begin{proposition}[Free energy of the effective condensate theory]
\label{prof:freeEnergyEffectiveCondensateTheoryMain}
    We consider the limit $N \to \infty$, $\beta N^{2/3} \to \kappa \in (0,\infty)$. The following statements hold for given $\epsilon > 0$:
	\begin{enumerate}[label=(\alph*)]
		\item Assume that $M \gtrsim N^{5/6 + \epsilon}$. Then we have
		\begin{equation}
			F^{\mathrm{BEC}}(\beta,M) = \frac{1}{2 \beta} \ln \left( \frac{ \hat{v}(0) \beta  }{2 \pi N} \right) + O\left( N^{1/3} \right).
			\label{eq:FreeEnergyBECInteractingLimit}
		\end{equation}
		\item Assume that $M \lesssim N^{5/6 - \epsilon}$. Then we have
		\begin{equation}
			F^{\mathrm{BEC}}(\beta,M) = F_0^{\mathrm{BEC}}(\beta,N) + O\left( N^{2/3 - 2 \epsilon} \right)
			\label{eq:FreeEnergyBECNonInteractingLimit}
		\end{equation}
		with the free energy
        \begin{equation}
            F_0^{\mathrm{BEC}}(\beta,N) = \frac{1}{\beta} \ln\left( 1 - \exp\left(\beta \mu_0(\beta,N) \right) \right)
            \label{eq:freeEnergyIdealGasBEC}
        \end{equation}
        of the condensed particles in the ideal gas. In particular, $F^{\mathrm{BEC}}(\beta,M)$ is independent of $\hat{v}(0)$ at the given level of accuracy. 
	\end{enumerate} 
\end{proposition}

A direct consequence of the above proposition is the following corollary.

\begin{corollary}[Free energy expansion]
    \label{cor:main1}
    Let $v$ satisfy the assumptions of Theorem~\ref{thm:main1}. We consider the limit $N \to \infty$, $\beta N^{2/3} \to \kappa \in (0,\infty)$. If $\kappa \in (1,\infty)$ then
    \begin{equation}
        F(\beta,N) = F^{\mathrm{Bog}}(\beta,N) + \frac{\hat{v}(0) N}{2} + \frac{1}{2 \beta} \ln \left( \frac{ \hat{v}(0) \beta  }{2 \pi N} \right) + O\left( N^{5/8} \right),
        \label{eq:CorollaryBEC}
    \end{equation}
    and if $\kappa \in (0,1)$ we have
    \begin{equation}
        F(\beta,N) = F_0(\beta,N) + \frac{\hat{v}(0) N}{2} + O\left( N^{5/8} \right)
        \label{eq:CorollaryNoBEC}
    \end{equation}
    with $F_0(\beta,N) = F_0^{\mathrm{BEC}}(\beta,N) + F_0^+(\beta,N)$. The free energies $F_0^{\mathrm{BEC}}$ and $F_0^{+}$ are defined in \eqref{eq:freeEnergyIdealGasBEC} and \eqref{eq:freeEnergyIdealGasCloud}, respectively.
\end{corollary}

The interpretation of the above result is the following: in the condensed phase ($\kappa > \frac{1}{4 \pi}[\upzeta(3/2)]^{2/3}$) we have BEC and because of this Bogoliubov modes form. The third term on the right-hand side of \eqref{eq:CorollaryBEC} is the free energy related to the fluctuations of the number of particles in the BEC. It is interesting to note that it depends on the interaction between the particles but not on $N_0$. This is related to the fact that the variance of the number of particles in the condensate does not depend on $N_0$ if $\kappa > \frac{1}{4 \pi}[\upzeta(3/2)]^{2/3}$, see the first line of \eqref{eq:2-pdmBECPhase}.

In the non-condensed phase ($\kappa < \frac{1}{4 \pi}[\upzeta(3/2)]^{2/3}$) Bogoliubov modes do not form and the free energy related to the particle number fluctuations in the BEC is absent. While we assume to be either strictly above or strictly below the critical point of the BEC phase transition in Corollary~\ref{cor:main1}, Theorem~\ref{thm:main1} also applies to the region around the critical point ($\kappa = 1$).

\subsection{Organization of the article}
\label{sec:organizationOfArticle}
For the convenience of the reader we describe here briefly the organization of our article before turning to the proof in the remaining sections.

We recall that the starting point of our proof is the Gibbs variational principle, which states that the Gibbs state is the unique minimizer of the free energy functional in \eqref{eq:GibbsFreeEnergyFunctional}. We will prove bounds for the free energy in \eqref{eq:freeenergy} and then infer information on the Gibbs state. In particular, we will prove the free energy expansion in Theorem \ref{thm:main1} by justifying the upper and lower bounds separately.  

In Section~\ref{sec:upperBoundFreeEnergy} we use the state $\Gamma_{\beta,N}$ to prove a sharp upper bound for $F(\beta,N)$ that is captured in Proposition~\ref{prop:upperbound}. The analysis here is a simplified version of that in \cite{BocDeuSto-24,CapDeu-23}, where an upper bound for the free energy in the more challenging Gross--Pitaevskii regime has been considered. Since several ingredients of the proof of the upper bound will later be used in other parts of our proof we provide all details.

Proving the sharp lower bound for the free energy requires second order correlation inequalities to justify Bogoliubov theory. As a preparation for these bounds, we need first-order bounds for the Gibbs states under suitable perturbations on the optimal scale. The general idea is to use a Griffiths argument and the coercivity of the relevant energy functional. The detailed  derivation of the first-order bounds is technically demanding and will occupy Sections ~\ref{sec:condensate Fraction},  ~\ref{sec:FirstOrderEstimates}, and~\ref{sec:boundsChemicalPotential}.


In Section~\ref{sec:condensate Fraction}, we derive a rough bound for the expected number of particles in the condensate. This follows from a rough lower bound for the free energy that is obtain via an application of the Onsager-type estimate in Lemma~\ref{lem:OnsagersLemma}. We also use \cite[Theorem~6.1]{LewNamRou-21}, see Lemma~\ref{lemcoercivityrelentr} below, to quantify the coercivity of the relative entropy with respect to the trace norm of the 1-pdms of the states under consideration. A refined estimate for the 1-pdm will be provided later in Section~\ref{sec:1pdm}.

Section~\ref{sec:FirstOrderEstimates} is devoted to proving first-order a priori estimates for a family of perturbed Gibbs states using a Griffiths argument. To obtain these bounds, we find it convenient to work with the grand potential instead of the free energy. Specifically, we introduce a chemical potential $\mu$
and minimize over all states on the Fock space rather than only those with an expected number of $N$ particles. This unconstrained minimization introduces additional mathematical difficulties. Part of these difficulties is resolved by introducing a new nonlinear equation in \eqref{eq:GrantCanonicalEffectiveIddealGasChemPot}, which a certain effective chemical potential solves. In particular, this allows us to approximately compute the expected number of condensed particles as a function of $\beta$ and $\mu$. This should be compared to \cite{DeuSeiYng-19,DeuSei-20,DeuSei-21}, where this quantity has been computed as a function of $\beta$ and $N$. First, we establish an estimate for the grand potential of the perturbed systems, accurate up to a remainder of order $N^{2/3}$. Achieving this accuracy requires taking into account the contribution of the effective condensate free energy in \eqref{eq:condensateFreeEnergy} at the scale $N^{2/3} \ln(N)$. We recall that this term contributes on two scales: $N^{2/3} \ln(N)$ and $N^{2/3}$. This step relies on an argument showing that, in the condensed phase, the fluctuations of the number of particles in the BEC coincide with those of the total number of particles. Finally, we use the concavity and monotonicity properties of the grand potential to derive estimates for its derivatives, which encode the relevant information. 

In Section~\ref{sec:boundsChemicalPotential}, by a first order Griffith argument that uses bounds for the free energy instead of the grand potential, we establish also bounds for the interacting chemical potential $\mu_{\beta,N}$ appearing in the definition of the Gibbs state $G_{\beta,N}$. The bounds for $\mu_{\beta,N}$ are in terms of $\hat{v}(0)$ and the chemical potential $\mu_0(\beta,N)$ of the ideal gas. From this we learn that the bounds we obtained in Section~\ref{sec:FirstOrderEstimates} also apply to $G_{\beta,N}$.

In Section~\ref{sec:correlationInequalities} we prove several correlation inequalities for $G_{\beta,N}$ on the optimal scale using our first new abstract correlation inequality in Theorem~\ref{thm:correlation-intro}. We first state and prove an infinite-dimensional version of Stahl's theorem, see Theorem~\ref{sec:StahlsTheorem}, which is then used to prove Theorem~\ref{thm:correlation-intro}.  Afterwards, we combine Theorem~\ref{thm:correlation-intro}, the first order estimates in Section~\ref{sec:FirstOrderEstimates}, and the bound for the chemical potential in Section~\ref{sec:boundsChemicalPotential} to derive second moment estimates for $G_{\beta,N}$. These bounds are captured in Theorem~\ref{thm:secondOrderEstimates}. Unlike the first-order estimates, the second-order bounds in this section cannot be expected to holds for approximate minimizers of the Gibbs free energy functional. That is, this part of the analysis goes strictly beyond the use of coercivity.

In Section~\ref{sec:lowerBound}, we apply our second-order correlation estimates for the Gibbs state $G_{\beta,N}$ to derive a sharp lower bound for the free energy $F(\beta,N)$,  thus completing the proof of Theorem \ref{thm:main1}. Two other ingredients of this proof are a c-number substitution in the spirit of \cite{LieSeiYng-05} and \cite[Lemma~1]{DeuSei-20}, which provides a bound for the entropy in the context of the c-number substitution. 

With the free energy bounds at hand, we conclude in Section~\ref{sec:Gibbs-state-III} the proofs of Theorem~\ref{thm:norm-approximation} (Trace norm approximation of the Gibbs state), Theorem~\ref{thm:1-pdm} (Trace norm and pointwise bounds for the 1-pdm), Theorem~\ref{thm:particleNumberDistributionBEC1} (Coherent state distribution of the BEC), and Theorem~\eqref{thm:particleNumberDistributionBEC2} (Limiting distributions for the BEC). All statements for the 1-pdm and the condensate distributions, except for the trace norm bound for the 1-pdm, are based on applications of Theorem~\ref{thm:norm-approximation}. The trace norm approximation for the 1-pdm is proved with a Griffith argument. The reason for this is that the state $\Gamma_{\beta,N}$ is not quasi-free if $N_0(\beta,N) \geq N^{2/3}$, which prevents us from applying techniques similar to those used in Section~\ref{sec:condensate Fraction}. For a more detailed discussion of this issue we refer to Section~\ref{sec:1pdm}. 

As explained at the end of Section~\ref{sec:DiscussionOfProof}, the second-moment estimates proved in Section~\ref{sec:correlationInequalities} are insufficient to establish Theorem~\ref{thm:main2} for the 2-pdm. Therefore, in Section~\ref{sec:higherOrderCorrelationInequalities}, we discuss higher-order moment bounds for the Gibbs state. These results, presented in Theorem~\ref{thm:higherOrderEstimates}, rely on the first-order estimates in Section~\ref{sec:FirstOrderEstimates} and an application of our second new abstract correlation inequality in Theorem~\ref{thm:higher-moments}. We note that the bounds in Theorem~\ref{thm:higherOrderEstimates} supersede some of those in Theorem~\ref{thm:secondOrderEstimates}. We provide both statements because the more complex results in Theorem~\ref{thm:higherOrderEstimates} are only required for proving bounds for the 2-pdm. This allows us to avoid unnecessarily complicating the proofs of the other main results.

In Section~\ref{sec:Gibbs-state-V} we apply the trace norm approximation of the Gibbs state $G_{\beta,N}$ in Theorem~\ref{thm:norm-approximation} and the higher-order moment bounds in Theorem~\ref{thm:higherOrderEstimates} to prove Theorem~\ref{thm:main2} for the 2-pdm. 

A considerable number of technical lemmas related to the various effective models appearing in our statements are presented in an appendix. Several bounds in the appendix are rather tedious because we require their uniformity across the critical point. The motivation for proving them in the appendix is to avoid interrupting the main flow of our arguments.

\vspace{0.5cm}

    \textbf{Acknowledgments.} It is a pleasure for A. D. to thank Robert Seiringer for inspiring discussions. P. T. N. would like to thank Mathieu Lewin and Nicolas Rougerie for various helpful discussions on Gibbs states, while M. N. would like to thank Błażej Ruba for useful remarks concerning the manuscript. A. D. gratefully acknowledges funding from the Swiss National Science Foundation (SNSF) through the Ambizione grant PZ00P2 185851. P. T. N. was partially supported by the European Research Council via the ERC Consolidator Grant RAMBAS (Project No. 10104424). The work of M. N. was supported by the National Science Centre (NCN) grant Sonata Bis 13 number 2023/50/E/ST1/00439.

\vspace{0.5cm}
\section{Sharp upper bound for the free energy}
\label{sec:upperBoundFreeEnergy}

The goal of this section is to prove the upper bound for the free energy in Theorem \ref{thm:main1}. The precise statement is captured in the following proposition.

\begin{proposition}
	\label{prop:upperbound}
	Let $v$ satisfy the assumptions of Theorem~\ref{thm:norm-approximation}. We consider the limit $N \to \infty$, $\beta N^{2/3} \to \kappa \in (0,\infty)$. The free energy $F(\beta,N)$ in \eqref{eq:freeenergy} satisfies the upper bound
	\begin{equation}
		F(\beta,N) \leq F^{\mathrm{Bog}}(\beta,N) + F^{\mathrm{BEC}}(\beta,N_0(\beta,N)) + \frac{\hat{v}(0)N}{2} + C_{\epsilon} N^{1/3 + \epsilon}
    \label{eq:upperBoundFreeEnergy}
	\end{equation}
	for any fixed $\epsilon > 0$.
\end{proposition}

To prove the above proposition we apply a trial state argument with the state $\Gamma_{\beta,N}$, which has been defined in Theorem~\ref{thm:norm-approximation}.
\subsection{Definition of the trial state}
\label{sec:trialstate}
We first consider the case $N_0 \geq N^{2/3}$ and recall \eqref{eq:exponentialProperty}. The case $N_0 < N^{2/3}$ will be discussed at the end of Section~\ref{sec:Finalupperbound}. As trial state we choose $\Gamma_{\beta,N}$ in \eqref{eq:referenceState-intro1}. The coherent state describes a BEC with an expected number of $|z|^2$ particles in the constant function $z/|z| \in L^2(\Lambda)$. In contrast, the thermally excited particles are described by $G^{\mathrm{Bog}}(z)$, which acts on $\mathscr{F}_+$. For the sake of a lighter notation we will denote the function $g^{\mathrm{BEC}}$ defined in \eqref{eq:GibbsDistributionDiscrete} in this section by $g$. It is the unique minimizer of the free energy functional 
\begin{equation}
	\mathcal{F}^{\mathrm{BEC}}_{\mathrm{c}}(\zeta) = \frac{\hat{v}(0)}{2N}  \int_{\mathbb{C}} |z|^4 \zeta(z) \de z  - \frac{1}{\beta} S(\zeta), \quad \text{ where } \quad S(\zeta) = - \int_{\mathbb{C}} \zeta(z) \ln(\zeta(z)) \de z
	\label{eq:condensatefunctionalcontinuous}
\end{equation}
denotes the classical entropy of $\zeta$, in the set
\begin{equation}\label{eq:MContinuousBECTheory}
	\mathcal{M}_{\mathrm{c}}(\widetilde{N}_0) = \left\{ \zeta \in L^1(\mathbb{C}) \ \Big| \ \zeta \geq 0, \int_{\mathbb{C}} \zeta(z) \de z = 1, \int_{\mathbb{C}} |z|^2 \zeta(z) \de z = \widetilde{N}_0  \right\}.
\end{equation}

The free energy of $g$ minus $\hat{v}(0) \widetilde{N}_0^2/(2N)$ is denoted by $F^{\mathrm{BEC}}_{\mathrm{c}}(\beta,\widetilde{N}_0)$. The free energy functional $\mathcal{F}_{\mathrm{c}}^{\mathrm{BEC}}$ is a continuous version of the discrete free energy functional $\mathcal{F}^{\mathrm{BEC}}$ in \eqref{eq:condensatefunctional}. It has been introduced recently in \cite{BocDeuSto-24} and later also appeared in \cite{CapDeu-23}. 

We choose the parameter $\widetilde{N}_0 \geq 0$ such that the expected number of particles in our trial state equals $N$, i.e.
\begin{equation}
	N = \tr[\mathcal{N} \Gamma_{\beta,N}] = \int_{\mathbb{C}} |z|^2 g(z) \de z + \int_{\mathbb{C}} \tr_{\mathscr{F}_+} [ \mathcal{N}_+ G^{\mathrm{Bog}}(z) ] g(z) \de z = \widetilde{N}_0(\beta,N) + \int_{\mathbb{C}} \tr_{\mathscr{F}_+} [ \mathcal{N}_+ G^{\mathrm{Bog}}(z) ] g(z) \de z.
	\label{eq:defN0}
\end{equation} 
As we show in Remark~\ref{rem:particleNumberTrialState} below, this is always possible. In Lemma~\ref{lem:BoundN0} below we show that \eqref{eq:defN0} implies the bound $| \widetilde{N}_0 - N_0 | \lesssim (1+ \beta^{-1}) \lesssim N^{2/3}$ with $N_0(\beta,N)$ in \eqref{eq:crittemp}. 
\subsection{Properties of the trial state}
\label{sec:preplemmas}

In this section we collect three lemmas that are needed in the proof of the upper bound for the free energy of our trial state $\Gamma_{\beta,N}$. We state them here to not interrupt the main line of the argument later.

The first lemma provides us with a Bogoliubov transformation that diagonalizes the Bogoliubov Hamiltonian $\mathcal{H}^{\mathrm{Bog}} (z)$. Before we state it, we introduce some notation. For $p \in \Lambda^*$ and $z \in \mathbb{C}$, we define the function 
\begin{equation}
	\varphi_{p,z}(x) = (z/|z|) e^{\mathrm{i}px} \quad \text{ and denote } \quad a_{p,z} = a(\varphi_{p,z}).
    \label{eq:planeWavesWithPhase}
\end{equation}
The operators $a_{p,z}$ and $a_{p,z}^*$ satisfy the canonical commutation relations in \eqref{eq:CCR}. We also define the Bogoliubov transformation $\mathcal{U}_z$ on $\mathscr{F}_+$ by $\mathcal{U}_z \Omega = \Omega$ with the vacuum vector $\Omega \in \mathscr{F}_+$ and
\begin{equation}
	\mathcal{U}_z^* a_{p,z} \mathcal{U}_z = u_{p} a_{p,z} + v_{p} a_{-p,z}^*, \quad \quad p \in \Lambda^*_+.
	\label{eq:Bogtrafo}
\end{equation}
Here the functions $u_{p}$ and $v_{p}$ are defined by
\begin{align}
	v_{p} &= \frac{1}{2} \left( \frac{p^2-\mu_0}{p^2-\mu_0 + 2 \hat{v}(p) N_0/N} \right)^{1/4} - \frac{1}{2} \left( \frac{p^2-\mu_0}{p^2-\mu_0 + 2 \hat{v}(p) N_0/N} \right)^{-1/4} \qquad \text{and} \qquad   u_p= \sqrt{1+v^2_p}. 
	\label{eq:coefficientsBogtrafoAndi}
\end{align}
The chemical potential $\mu_0(\beta,N)$ and the number of condensed particles $N_0(\beta,N)$ in the ideal gas are defined in \eqref{eq:idealgase1pdmchempot} and \eqref{eq:crittemp}, respectively. We are now prepared to state our first lemma, whose proof is a standard computation based on \eqref{eq:planeWavesWithPhase}--\eqref{eq:coefficientsBogtrafoAndi}  (we refer to \cite{NamNapSol-16} for a general theory of bosonic quadratic Hamiltonians).

\begin{lemma}
	\label{lem:Diaggamma}
	The Bogoliubov Hamiltonian $\mathcal{H}^{\mathrm{Bog}} (z)$ in \eqref{BogoliubovHamiltonianAndi} satisfies
	\begin{equation}
		\mathcal{U}_z \mathcal{H}^{\mathrm{Bog}} (z) \mathcal{U}_z^* = E_0 + \sum_{p \in \Lambda^*_+} \epsilon(p) a_{p}^* a_p	
		\label{eq:diagonalform1}
	\end{equation}
	with the Bogoliubov dispersion relation $\epsilon(p)$ in \eqref{eq:BogoliubovDispersion} and the ground state energy
	\begin{equation}
		E_0 = -\frac{1}{2} \sum_{p\in \Lambda^*_+} \left[ p^2 - \mu_0(\beta,N) + \hat{v}(p) N_0(\beta,N)/N - \epsilon(p) \right]. 
		\label{eq:epsilonz}
	\end{equation}
\end{lemma}

In the next lemma we compute the 1-pdm and the pairing function of the state $G(z)$.	
	
\begin{lemma}\label{lem:1pdmAndPairingFunction}
	Let $G^{\mathrm{Bog}}(z)$ be the state in \eqref{eq:referenceState-intro1}. Its 1-pdm and pairing function are given for $p,q \in \Lambda^*_+$ by
	\begin{equation}
		\tr[ a_q^* a_p G^{\mathrm{Bog}}(z) ] = \delta_{p,q} \underbrace{\left\{  \left( u_{p}^2 + v_{p}^2 \right) \gamma^{\mathrm{Bog}}_{p} + v_{p}^2 \right\}}_{\gamma_{p}} \quad \text{ and } \quad \tr[ a_p a_q G^{\mathrm{Bog}}(z) ] = \delta_{p,-q} (z/|z|)^{-2} \underbrace{ u_{p} v_{p} \left\{ 2  \gamma^{\mathrm{Bog}}_{p} + 1 \right\}}_{\alpha_{p}},
		\label{eq:1pdm}
	\end{equation}
	respectively, with the functions $u_{p}$, $v_{p}$ in \eqref{eq:coefficientsBogtrafoAndi}. Moreover,
	\begin{equation}
	\gamma^{\mathrm{Bog}}_{p} = \frac{1}{\exp(\beta \epsilon(p))-1}.
    \label{eq:gammapBog}
	\end{equation}
with $\epsilon(p)$ in \eqref{eq:BogoliubovDispersion}.
\end{lemma}

The proof of the above lemma is a straightforward computation that uses \eqref{eq:Bogtrafo} and Wicks rule for the diagonalized version of the state $G^{\mathrm{Bog}}(z)$. Next we state and prove a bound for $\widetilde{N}_0(\beta,N)$ showing that it equals $N_0(\beta,N)$ to leading order as $N \to \infty$.

\begin{lemma}
	\label{lem:BoundN0}
	Assume that $\widetilde{N}_0(\beta,N)$ is defined as in \eqref{eq:defN0} and let $N_0(\beta,N)$ be given as in \eqref{eq:idealgase1pdmchempot}. Then we have
	\begin{equation}
		| \widetilde{N}_0(\beta,N) - N_0(\beta,N) | \lesssim \left( 1 + \frac{1}{\beta} \right) \frac{N_0^2(\beta,N)}{N^2} + \frac{N_0(\beta,N)}{\beta N}.
	\end{equation}
\end{lemma}
\begin{proof}
	To prove a bound for $\widetilde{N}_0$, we need to consider the quantity
	\begin{equation}
		\int_{\mathbb{C}} \tr_{\mathscr{F}_+} [\mathcal{N}_+ G^{\mathrm{Bog}}(z)] \zeta(z) \de z = \sum_{p \in \Lambda^*_+} \gamma_{p}  = \sum_{p \in \Lambda^*_+} \left[ \gamma^{\mathrm{Bog}}_{p}  + (1+2\gamma_{p}^{\mathrm{Bog}}) v_{p}^2  \right].
		\label{eq:preplemmas31}
	\end{equation} 	
	To obtain it, we applied Lemma~\ref{lem:1pdmAndPairingFunction} and used $u_p^2 - v_p^2 = 1$. Before we analyze these terms, we compute
	\begin{equation}
		v_{p}^2 = \frac{1}{4} \left( \frac{p^2-\mu_0}{p^2-\mu_0 + 2 \hat{v}(p) N_0/N} \right)^{1/2} + \frac{1}{4} \left( \frac{p^2-\mu_0}{p^2-\mu_0 + 2 \hat{v}(p) N_0/N} \right)^{-1/2} - \frac{1}{2}.
	\end{equation}
	This, the bound $0 \leq (1+x)^{-1/2} + (1+x)^{1/2} - 2 \leq x^2/4$ for $x \geq 0$, and $\mu_0<0$ imply
	\begin{equation}
		0 \leq v_{p}^2 \leq \frac{\hat{v}^2(p) N_0^2}{4 N^2 p^4}.
		\label{eq:boundvp}
	\end{equation}
	We also have
	\begin{equation}
		\gamma_{p}^{\mathrm{Bog}} \leq \frac{1}{\exp(\beta(p^2-\mu_0))-1} \leq \frac{1}{\beta p^2},
		\label{eq:pwboundgammaB}
	\end{equation}
	which follows from $(\exp(x)-1)^{-1} \leq 1/x$ for $x \geq 0$ and $\epsilon(p) \geq p^2-\mu_0 \geq p^2$. We apply \eqref{eq:boundvp} and \eqref{eq:pwboundgammaB} to see that
	\begin{equation}
		0 \leq \sum_{p \in \Lambda^*_+} (1+2\gamma_{p}^{\mathrm{Bog}}) v_{p}^2  \lesssim \Vert \hat{v} \Vert_{\infty}^2 \left( 1 + \frac{1}{\beta} \right) \frac{N^2_0}{N^2} \sum_{p \in \Lambda^*_+} \frac{1}{p^4} \lesssim \left( 1 + \frac{1}{\beta} \right) \frac{N^2_0}{N^2}
		\label{eq:preplemma32}
	\end{equation}
	holds. It remains to give a bound for the first term on the r.h.s. of \eqref{eq:preplemmas31}.
	
	A first order Taylor approximation shows
	\begin{equation}
		0 \leq \sum_{p \in \Lambda^*_+} \left( \frac{1}{\exp(\beta(p^2-\mu_0))-1} - \gamma_{p}^{\mathrm{Bog}} \right) \leq \sum_{p \in \Lambda^*_+} \int_0^1 \frac{\beta (\epsilon(p) - (p^2 - \mu_0))}{4 \sinh^2(\beta( p^2-\mu_0 + t (\epsilon(p) - (p^2 - \mu_0)) )/2)} \de t.
		\label{eq:preplemma32c}
	\end{equation}
	We have $\sinh(x) \geq x$ for $x \geq 0$. That is, the r.h.s. of the above equation is bounded by
	\begin{equation}
		\sum_{p \in \Lambda^*_+} \frac{(\epsilon(p) - (p^2 - \mu_0))}{\beta(p^2 - \mu_0)^2} = \sum_{p \in \Lambda^*_+} \frac{\sqrt{1+2 \hat{v}(p) (N_0/N) /(p^2-\mu_0)}-1}{\beta(p^2 - \mu_0)} \leq  \frac{\Vert \hat{v} \Vert_{\infty} N_0}{\beta N} \sum_{p \in \Lambda^*_+} \frac{1}{p^4} \lesssim \frac{N_0}{\beta N}. 
		\label{eq:preplemma32d}
	\end{equation}
	To obtain the inequality in the above equation, we used $\sqrt{1+x} - 1 \leq x/2$ for $x \geq 0$ and $\mu_0<0$. 
	In combination, \eqref{eq:defN0}, \eqref{eq:preplemmas31}, \eqref{eq:preplemma32}, \eqref{eq:preplemma32c} and \eqref{eq:preplemma32d} show 
	\begin{equation}
		| \widetilde{N}_0 - N_0 | \lesssim \left( 1 + \frac{1}{\beta} \right) \frac{N_0^2}{N^2} + \frac{N_0}{\beta N},
		\label{eq:preplemma33b}
	\end{equation}
	which proves the claim.
	\end{proof}
	\begin{remark}
    \label{rem:particleNumberTrialState}
		The bounds in the proof of Lemma~\ref{lem:BoundN0} also show that there always exists $\widetilde{N}_0 > 0$ such that \eqref{eq:defN0} is satisfied. To see this one also needs to use that $\beta \gtrsim \beta_{\mathrm{c}}$ with $\beta_{\mathrm{c}}$ in \eqref{eq:crittemp} implies $N_0 \gtrsim 1$.
	\end{remark}
\subsection{Upper bound for the free energy of the trial state}
\label{sec:upperboundfreeenergytrialstate}

In this section we prove an upper bound for the free energy of $\Gamma_{\beta,N}$. We start with a bound for the energy.

\subsubsection{Bound for the energy}

To compute the energy of $\Gamma$, we write the Hamiltonian as

\begin{align}
    \mathcal{H}_N =& \sum_{p \in \Lambda^*_+} (p^2 - \mu_0) a_p^* a_p + \frac{1}{2N} \sum_{p \in \Lambda^*_+} \hat{v}(p) \left\{ 2 a_p^* a_0^* a_p a_0 + a_0^* a_0^* a_p a_{-p} + a_p^* a_{-p}^* a_0 a_0 \right\} + \mu_0 \sum_{p \in \Lambda^*_+} a_p^* a_p \nonumber \\
	&+ \frac{1}{N} \sum_{p,k,p+k \in \Lambda_+^*} \hat{v}(p) \left\{ a^*_{k+p} a^*_{-p} a_k a_0 + h.c.   \right\} + \frac{\hat{v}(0)}{2N} \sum_{u,v \in \Lambda^*} a_u^* a_v^* a_u a_v \nonumber \\
	&+ \frac{1}{2N} \sum_{u,v,p,u+p,v-p \in \Lambda^*_+} \hat{v}(p) a^*_{u+p} a^*_{v-p} a_u a_v,
	\label{eq:decompH}
\end{align}
where $\mu_0 = \mu_0(\beta,N)$. A brief computation that uses $\int_{\mathbb{C}} |z|^2 g(z) \de z = \widetilde{N}_0(\beta,N)$ shows that the expectation of the first two terms on the right-hand side in our trial state $\Gamma$ equals
\begin{equation}
	\int_{\mathbb{C}} \tr_{\mathscr{F}_+}[ \mathcal{H}^{\mathrm{Bog}} (z)G^{\mathrm{Bog}}(z) ] g(z) \de z + \frac{\widetilde{N_0}-N_0}{N} \sum_{p \in \Lambda^*_+} \hat{v}(p) \left( \gamma_p + \alpha_p \right) 
	\label{eq:Bogtermsinenergy}
\end{equation}
with the Bogoliubov Hamiltonian $\mathcal{H}^{\mathrm{Bog}} (z)$ in \eqref{BogoliubovHamiltonianAndi}. The functions $\gamma_p$ and $\alpha_p$ are defined in \eqref{eq:1pdm}. The first term will later be combined with a term coming from the entropy to give the expectation of the free energy of the Bogoliubov Hamiltonian $\mathcal{H}^{\mathrm{Bog}} (z)$ with respect to the probability measure $g(z) \de z$. We use \eqref{eq:1pdm}, \eqref{eq:boundvp}, \eqref{eq:pwboundgammaB}, and $u_p^2-v_p^2=1$ to see that $\gamma_p$ satisfies the bound 
\begin{equation}
	0 \leq \gamma_p \leq \frac{\hat{v}^2(p) N_0^2}{N^2 p^4} \left( 1 + \frac{2}{\beta p^2} \right) + \frac{1}{\beta p^2} \lesssim \left( 1 + \frac{1}{\beta} \right) \frac{1}{p^2}.
	\label{eq:boundgammap}
\end{equation}
To obtain a bound for $\alpha_p$, we first note that
\begin{equation}
	u_{p} v_{p} = \frac{1}{4} \left( \frac{p^2-\mu_0}{p^2-\mu_0 + 2 \hat{v}(p) N_0/N} \right)^{1/2} - \frac{1}{4} \left( \frac{p^2-\mu_0}{p^2-\mu_0 + 2 \hat{v}(p) N_0/N} \right)^{-1/2}.
\end{equation}
The inequality $0 \leq (1+x)^{1/2} - (1+x)^{-1/2} \leq x$ for $x \geq 0$ therefore shows
\begin{equation}
	0 \geq u_{p} v_{p} \geq -\frac{\hat{v}(p) N_0}{2 N p^2}, 
\end{equation}
which, in combination with \eqref{eq:1pdm}, implies the bound
\begin{equation}
	| \alpha_{p} | \leq \frac{\hat{v}(p) N_0}{2 N p^2} \left( 2 \gamma_{p}^{\mathrm{Bog}} + 1 \right) \lesssim \frac{\hat{v}(p) N_0}{N p^2} \left( 1 + \frac{1}{\beta } \right). 
	\label{eq:boundalphap}
\end{equation}
We use Lemma~\ref{lem:BoundN0}, \eqref{eq:boundgammap}, and \eqref{eq:boundalphap} to see that the second term   of \eqref{eq:Bogtermsinenergy} is bounded from above by a constant times
\begin{equation}
	\left( 1 + \frac{1}{\beta} \right)^2 \frac{N_0}{N^2} \sum_{p \in \Lambda_+^*} \frac{\hat{v}(p)}{p^2}.
\end{equation}
Our assumption $\hat{v} \in L^1(\Lambda^*)$ guarantees that the sum on the right-hand side is finite. Hence,
\begin{align}
	&\tr\left[ \left( \sum_{p \in \Lambda^*_+} (p^2 - \mu_0) a_p^* a_p + \frac{1}{2N} \sum_{p \in \Lambda^*_+} \hat{v}(p) \left\{ 2 a_p^* a_0^* a_p a_0 + a_0^* a_0^* a_p a_{-p} + a_p^* a_{-p}^* \right\} \right) \Gamma_{\beta,N} \right] \nonumber \\
	&\hspace{6cm}\leq \int_{\mathbb{C}} \tr_{\mathscr{F}_+}[ \mathcal{H}^{\mathrm{Bog}}(z) G^{\mathrm{Bog}}(z) ] g(z) \de z + C \left( 1 + \frac{1}{\beta} \right)^2 \frac{N_0}{N^2} \label{eq:finalupperenergyBogterm}
\end{align}
holds.

The expectation of the third term on the right-hand side of \eqref{eq:decompH} equals
\begin{equation}
	\mu_0 \sum_{p \in \Lambda^*_+} \int_{\mathbb{C}} \tr_{\mathscr{F}_>}[ a_p^* a_p G^{\mathrm{Bog}}(z) ] \zeta(z) \de z = \mu_0 \sum_{p \in \Lambda_+} \gamma_p.
	\label{eq:finalupperenergyBogtermb}
\end{equation}
Using that $G(z)$ is a quasi free state, for which the expectation of an odd number of creation and annihilation operators vanishes, we check that the expectation of the fourth term on the r.h.s. of \eqref{eq:decompH} in the state $\Gamma$ vanishes. It remains to compute the expectation of the fifth and the sixths term. We start with the computation of the expectation of the sixths term.

It reads
\begin{equation}
	\frac{1}{2N} \sum_{u,v,p,u+p,v-p \in \Lambda^*_+} \hat{v}(p) \int_{\mathbb{C}} \tr_{\mathscr{F}_+}[ a^*_{u+p} a^*_{v-p} a_u a_v G^{\mathrm{Bog}}(z) ] g(z) \de z.
	\label{eq:upperboundA1}
\end{equation}
Using Wick's theorem, Lemma~\ref{lem:Diaggamma} and $p \neq 0$, we compute
\begin{equation}
	\tr_{\mathscr{F}_+}[ a^*_{u+p} a^*_{v-p} a_u a_v G^{\mathrm{Bog}}(z) ] = \delta_{u+p,v} \gamma_{u} \gamma_{v} + \delta_{u,-v} \alpha_{u+p} \alpha_{u}
	\label{eq:expectation}
\end{equation}
with $\gamma_{p}$ and $\alpha_{p}$ in \eqref{eq:1pdm}. We insert the first term on the right-hand side into \eqref{eq:upperboundA1}, use Young's inequality, \eqref{eq:boundgammap}, \eqref{eq:boundalphap}, and find 
\begin{equation}
	\frac{1}{2N} \sum_{u,v,u-v \in \Lambda^*_+} \hat{v}(u-v)  \gamma_{u} \gamma_{v} + \frac{1}{2N} \sum_{u,p,u+p \in \Lambda^*_+} \hat{v}(p)  \alpha_{u+p} \alpha_{u} \lesssim  \frac{\Vert \hat{v} \Vert_1}{N} \left( 1 + \frac{1}{\beta} \right)^2.
	\label{eq:upperboundA2}
\end{equation}
It remains to consider the expectation of the fifth term on the right-hand side of \eqref{eq:decompH}.

A short computation shows
\begin{align}
	\frac{\hat{v}(0)}{2N} \sum_{u,v \in \Lambda^*} \tr[ a_u^* a_v^* a_u a_v \Gamma_{\beta,N} ] =& \frac{\hat{v}(0)}{2N} \int_{\mathbb{C}} |z|^4 g(z) \de z + \frac{\hat{v}(0) \widetilde{N}_0}{N} \sum_{p \in \Lambda^*_+} \gamma_p \nonumber \\
	&+ \frac{\hat{v}(0)}{2N} \int_{\mathbb{C}} \left( \sum_{u,v \in \Lambda^*_+} \tr_{\mathscr{F}_+}[ a_u^* a_v^* a_u a_v G^{\mathrm{Bog}}(z) ] \right) \zeta(z) \de z.
	\label{eq:upperboundA4a}
\end{align}
An application of Wick's theorem allows us to see that 
\begin{align}
	\frac{\hat{v}(0)}{2N} \int_{\mathbb{C}} \left( \sum_{u,v \in \Lambda^*_+} \tr_{\mathscr{F}_+}[ a_u^* a_v^* a_u a_v G^{\mathrm{Bog}}(z) ] \right) \zeta(z) \de z & = \frac{\hat{v}(0)}{2N} \left[ \left( \sum_{u \in \Lambda^*_+} \gamma_{u} \right)^2 + \sum_{u \in \Lambda^*_+} \gamma_{u}^2 +\sum_{u \in \Lambda^*_+} |\alpha _{u}|^2\right] \nn \\
 & \leq \frac{\hat{v}(0)}{2N} \left( \sum_{u \in \Lambda^*_+} \gamma_{u} \right)^2+\frac{C}{N} \left( 1 + \frac{1}{\beta} \right)^2,
	\label{eq:upperboundA4c}
\end{align}
where we used \eqref{eq:boundgammap} and \eqref{eq:boundalphap} in the last step. In combination, \eqref{eq:finalupperenergyBogterm}, \eqref{eq:finalupperenergyBogtermb}, \eqref{eq:upperboundA2}, \eqref{eq:upperboundA4a}, \eqref{eq:upperboundA4c}, and 
\begin{equation}
	\widetilde{N}_0 + \sum_{p \in \Lambda_+^*} \gamma_p = N,
\end{equation}
which follows from \eqref{eq:defN0}, prove the bound
\begin{align}
	\tr[\mathcal{H}_N \Gamma_{\beta,N}] \leq& \int_{\mathbb{C}} \tr_{\mathscr{F}_+}[ \mathcal{H}^{\mathrm{Bog}}(z) G^{\mathrm{Bog}}(z) ] g(z) \de z + \mu_0(\beta,N) \sum_{p \in \Lambda_+^*} \gamma_p + \frac{\hat{v}(0) N}{2} \nonumber \\
	&+ \frac{\hat{v}(0)}{2N} \left( \int_{\mathbb{C}} |z|^4 g(z) \de z - \left( \int_{\mathbb{C}} |z|^2 g(z) \de z \right)^2 \right) 
	+ \frac{C}{N} \left( 1+ \frac{1}{\beta} \right)^2.
	\label{eq:finalupper4}
\end{align}
We highlight that the error term on the right-hand side is bounded by a constant times $N^{1/3}$.
\subsubsection{Bound for the entropy}
In this section we derive an upper bound for the entropy of $\Gamma$. To that end, we need the following lemma, which provides us with a Berezin--Lieb inequality in the spirit of \cite{Berezin-72,Lieb-73}. Its proof can be found in \cite[Lemma~2.4]{BocDeuSto-24}.

\begin{lemma}
	\label{lem:BerezinLieb}
	Let $\{ G(z) \}_{z \in \mathbb{C}}$ be a family of states on a (separable complex) Hilbert space with eigenvalues $g_{\alpha}(z)$ and eigenvectors $v_{\alpha}(z)$, $\alpha \in \mathbb{N}$, and let $p: \mathbb{C} \to \mathbb{R}$ be a probability distribution. We assume that the functions $z \mapsto g_{\alpha}(z)$ are measurable, that the functions $z \mapsto v_{\alpha}(z)$ are weakly measurable, and that $p$ satisfies $\int_{\mathbb{C}} | p(z) \ln(p(z)) | \de z < +\infty$. Then the entropy of the state
	\begin{equation}
		\Gamma = \int_{\mathbb{C}} |z \rangle \langle z | \otimes G(z) p(z) \de z,
		\label{eq:entropy2}
	\end{equation} 
	where the integral is understood in the sense of Lebesgue with respect to the weak operator topology, satisfies the lower bound
	\begin{equation}
		S(\Gamma) \geq \int_{\mathbb{C}} S(G(z)) p(z) \de z + S(p)
		\label{eq:entropy3}
	\end{equation}
	with $S(p)$ in \eqref{eq:condensatefunctionalcontinuous}.
\end{lemma}

An application of the above lemma shows that the entropy of $\Gamma_{\beta,N}$ satisfies the lower bound
\begin{equation}
	S(\Gamma_{\beta,N}) \geq \int_{\mathbb{C}} S(G^{\mathrm{Bog}}(z)) g(z) \de z + S(g),
	\label{eq:entropy7}
\end{equation}
which is the final result of this section.
\subsubsection{Final upper bound for the free energy}
\label{sec:Finalupperbound}
When we combine \eqref{eq:finalupper4}, \eqref{eq:entropy7} and Lemma~\ref{lem:Diaggamma}, we find the bound 
\begin{equation}
	\mathcal{F}(\Gamma_{\beta,N}) \leq F^{\mathrm{Bog}}(\beta,N) + E_0 + F^{\mathrm{BEC}}_{\mathrm{c}}(\beta,\widetilde{N}_0) + \frac{\hat{v}(0) N}{2} + \frac{C}{N} \left( 1 + \frac{1}{\beta} \right)^2
	\label{equpperboundfreeenergy1}
\end{equation}
with $F^{\mathrm{Bog}}$ in \eqref{eq:BogoliubovFreeEnergy}, $E_0$ in \eqref{eq:epsilonz}, and $F_{\mathrm{c}}^{\mathrm{BEC}}$ defined below \eqref{eq:MContinuousBECTheory}. Since $E_0$ is negative it can be dropped for an upper bound. It therefore only remains to replace $F_{\mathrm{c}}^{\mathrm{BEC}}(\beta,\widetilde{N}_0)$ by $F^{\mathrm{BEC}}(\beta,N_0)$ with $F^{\mathrm{BEC}}$ in \eqref{eq:condensateFreeEnergy} and $N_0$ in \eqref{eq:crittemp}. In this analysis we will focus on the case $N_0 \gtrsim N^{2/3}$. The parameter regime, where $N_0 \lesssim N^{2/3}$ holds will be treated with a separate argument later.

In the first step we replace $F_{\mathrm{c}}^{\mathrm{BEC}}(\beta,\widetilde{N}_0)$ by $F_{\mathrm{c}}^{\mathrm{BEC}}(\beta,N_0)$. Let us denote by $\widetilde{\mu}$ and $\mu$ the chemical potentials related to $\widetilde{N}_0$ and $N_0$, respectively. We have
\begin{align}
	F_{\mathrm{c}}^{\mathrm{BEC}}(\beta,\widetilde{N}_0) &= -\frac{1}{\beta} \ln\left( \int_{\mathbb{C}} \exp\left(-\beta\left( \frac{\hat{v}(0)}{2N} |z|^4 - \widetilde{\mu} |z|^2 \right) \right) \de z \right) + \widetilde{\mu} \widetilde{N}_0 - \frac{\hat{v}(0) \widetilde{N}_0^2}{2N} \nonumber \\
	&\leq -\frac{1}{\beta} \ln\left( \int_{\mathbb{C}} \exp\left(-\beta\left( \frac{\hat{v}(0)}{2N} |z|^4 - \mu |z|^2 \right) \right) \de z \right) + (\mu - \widetilde{\mu}) N_0 + \widetilde{\mu} \widetilde{N}_0 - \frac{\hat{v}(0) \widetilde{N}_0^2}{2N}.
	\label{eq:UpperBoundAndi1}
\end{align}
To obtain this bound, we used that the first term after the smaller or equal sign is concave in $\mu$, and that its first derivative with respect to $\mu$ equals $N_0$. We distinguish two case and first assume that $N_0 \geq N^{5/6 + \epsilon}$ for some fixed $0 < \epsilon < 1/6$. Applications of Lemma~\ref{lem:BoundN0} and part~(a) of Lemma~\ref{lem:ChemPotBECCont} in Appendix~\ref{app:effcondensate} show that 
\begin{equation}
	(\mu - \widetilde{\mu}) N_0 + \widetilde{\mu} \widetilde{N}_0 - \frac{\hat{v}(0) \widetilde{N}_0^2}{2N} \leq \mu N_0 - \frac{\hat{v}(0) N_0}{2N} + C \left( N^{1/3} + \exp(- c N^{\epsilon} ) \right).
	\label{eq:UpperBoundAndi2}
\end{equation}
If $N^{2/3} \leq N_0 < N^{5/6 + \epsilon}$ we apply Lemma~\ref{lem:BoundN0} and part~(c) of Lemma~\ref{lem:ChemPotBECCont}, which gives 
\begin{equation}
	(\mu - \widetilde{\mu}) N_0 + \widetilde{\mu} \widetilde{N}_0 - \frac{\hat{v}(0) \widetilde{N}_0^2}{2N} \leq \mu N_0 - \frac{\hat{v}(0) N_0}{2N} + C N^{1/3+2 \epsilon}.
	\label{eq:UpperBoundAndi3}
\end{equation}
In combination, these considerations imply
\begin{equation}
	F_{\mathrm{c}}^{\mathrm{BEC}}(\beta,\widetilde{N}_0) \leq F_{\mathrm{c}}^{\mathrm{BEC}}(\beta,N_0) + C \left( N^{1/3+2 \epsilon} + \exp(- c N^{\epsilon} ) \right).
	\label{eq:UpperBoundAndi4}
\end{equation}
Moreover, an application of Lemma~\ref{lem:comparisonContinuousDiscreteCondensateFreeEnergy} in Appendix~\ref{app:effcondensate} shows
\begin{equation}
	F_{\mathrm{c}}^{\mathrm{BEC}}(\beta,N_0) \leq F^{\mathrm{BEC}}(\beta,N_0) + C N^{1/3}.
	\label{eq:UpperBoundAndi9}
\end{equation}

To obtain the final bound for the free energy of our trial $\Gamma_{\beta,N}$, we collect \eqref{equpperboundfreeenergy1}, \eqref{eq:UpperBoundAndi4}, and \eqref{eq:UpperBoundAndi9}, which gives 
\begin{equation}
	\mathcal{F}(\Gamma_{\beta,N}) \leq F^{\mathrm{Bog}}(\beta,N) + F^{\mathrm{BEC}}(\beta,N_0) + \frac{\hat{v}(0) N}{2} + C \left( N^{1/3+2 \epsilon} + \exp(- c N^{\epsilon} ) \right)
	\label{eq:UpperBoundAndi10}
\end{equation}
for $0 < \epsilon < 1/6$. Our bound has been derived under the assumption $N_0(\beta,N) \geq N^{2/3}$. It therefore remains to consider the parameter regime, where $N_0(\beta,N) < N^{2/3}$ holds.

In this case we choose the Gibbs state $G_{\beta,N}^{\mathrm{id}}$ of the ideal gas in \eqref{eq:GibbsStateIdealGas} as trial state. A straightforward computation shows 
\begin{equation}
	\mathcal{F}(G_{\beta,N}^{\mathrm{id}}) \leq \frac{1}{\beta} \sum_{p \in \Lambda^*} \ln \left( 1 - \exp\left( -\beta(p^2 - \mu_0(\beta,N)) \right) \right) + \mu_0(\beta,N) N + \frac{\hat{v}(0)N}{2} + C N^{1/3}
	\label{eq:UpperBoundAndi11}
\end{equation}
with $\mu_0$ in \eqref{eq:idealgase1pdmchempot}. The above expressions need to be compared to the ones appearing in Theorem~\ref{thm:main1}. Let $p$ be a probability distribution on $\mathbb{N}_0$ with $\sum_n n p(n) = N_0(\beta,N)$. We have
\begin{align}
	\mathcal{F}^{\mathrm{BEC}}(p) &= \sum_{n=0}^{\infty} \frac{\hat{v}(0)}{2N} n^2 - \frac{1}{\beta} S(p) \geq - \mu_0(\beta,N) \sum_{n=0}^{\infty} n p(n) - \frac{1}{\beta} S(p) + \mu_0(\beta,N) N_0(\beta,N) \nonumber \\
	&\geq -\frac{1}{\beta} \ln\left( \sum_{n=0}^{\infty} \exp(\beta \mu_0 n) \right) + \mu_0 N_0  = \frac{1}{\beta} \ln\left( 1 - \exp(\beta \mu_0) \right) + \mu_0 N_0,
	\label{eq:UpperBoundAndi12}
\end{align}
and hence
\begin{equation}
	\frac{1}{\beta} \ln\left( 1 - \exp(\beta \mu_0) \right) + \mu_0 N_0 \leq F^{\mathrm{BEC}}(\beta, N_0) + C N^{1/3}.
	\label{eq:UpperBoundAndi13}
\end{equation}
To obtain the final bound we also used $\hat{v}(0) N_0^2/N \lesssim N^{1/3}$. The bound in \eqref{eq:UpperBoundAndi13} will be used for the term with $p=0$ in the sum on the right-hand side of \eqref{eq:UpperBoundAndi11}. Using that $\epsilon(p) \geq p^2-\mu_0$ with $\epsilon(p)$ in \eqref{eq:BogoliubovDispersion}, we bound the remaining part of the sum as follows:
\begin{equation}
	\sum_{p \in \Lambda^*_+} \ln \left( 1 - \exp\left( -\beta(p^2 - \mu_0) \right) \right) \leq \sum_{p \in \Lambda^*_+} \ln \left( 1 - \exp\left( -\beta \epsilon(p) \right) \right).
	\label{eq:UpperBoundAndi14}
\end{equation}
It remains to replace $\mu_0(N-N_0)$ by $\mu_0 \sum_{p \in \Lambda_+^*} \gamma_p$ with $\gamma_p$ in \eqref{eq:1pdm}. This can be done with Lemma~\ref{lem:BoundN0} and the bound $-\mu_0 \lesssim 1/(\beta N_0)$, and we find
\begin{equation}
	\mu_0(N-N_0) \leq \mu_0 \sum_{p \in \Lambda_+^*} \gamma_p + C N^{1/3}.
	\label{eq:UpperBoundAndi15}
\end{equation} 
In combination, these consideration show
\begin{equation}
	\mathcal{F}(G_{\beta,N}^{\mathrm{id}}) \leq F^{\mathrm{Bog}}(\beta,N) + F^{\mathrm{BEC}}(\beta, N_0) + \frac{\hat{v}(0) N}{2} + CN^{1/3},
	\label{eq:UpperBoundAndi16}
\end{equation}
which holds under the assumption $N_0 < N^{2/3}$. Together with the Gibbs variational principle in \eqref{eq:freeenergy} and \eqref{eq:UpperBoundAndi10}, \eqref{eq:UpperBoundAndi16} proves Proposition~\ref{prop:upperbound}.

\section{A rough bound for the expected number of particles in the condensate}
\label{sec:condensate Fraction}
In this section we provide a rough bound for the expected number of particles in the condensate, which will be used in Section~\ref{sec:lowerBound}. 

\begin{proposition}
	\label{prop:roughapriori}
    Let $v$ satisfy the assumptions of Theorem~\ref{thm:norm-approximation}. We consider the limit $N \to \infty$, $\beta N^{2/3} \to \kappa \in (0,\infty)$. Then we have
    \begin{equation}
        |\Tr[a_0^* a_0 G_{\beta,N}] - N_0(\beta,N) | \lesssim N^{2/3} \ln(N)
        \label{eq:aPrioriCondensateFraction1}
    \end{equation}
    with $G_{\beta,N}$ in \eqref{eq:interactingGibbsstate} and $N_0(\beta,N)$ in \eqref{eq:crittemp}.
\end{proposition}

Before we give the proof of the above proposition, we discuss three lemmas to not interrupt the main line of the argument later. 

\subsection{Preparations}

The first lemma is well-known and will be used to obtain a lower bound for the interaction term in the Hamiltonian. 

\begin{lemma}\label{lem:OnsagersLemma}
	Let $v \in L^1(\Lambda)$ be a periodic function with summable Fourier coefficients $\hat{v} \geq 0$ and denote the second term in \eqref{eq:Hamiltonian} by $\mathcal{V}_{\eta}$. Then we have 
	\begin{equation}
		\mathcal{V}_{\eta} \geq \frac{\hat{v}(0) \mathcal{N}^2}{2 N} - \frac{v(0) \mathcal{N}}{2 N}.
		\label{eq:OnsagersInequality}
	\end{equation}
\end{lemma}
\begin{proof}
	Since $\hat{v} \geq 0$ we know that
	\begin{equation}
		\sum_{1 \leq i < j \leq n} v(x_i - x_j) \geq \frac{\hat{v}(0) n^2}{2} - \frac{ v(0) n }{2}
		\label{eq:OnsagersInequality1stQ}
	\end{equation}
	holds for any $n \geq 2$, see e.g. \cite[Eq.~(8)]{Seiringer-11}. Eq.~\eqref{eq:OnsagersInequality} is a direct consequence of \eqref{eq:OnsagersInequality1stQ}.
\end{proof}

Let $\Gamma$ and $\Gamma'$ be two states on the bosonic Fock space $\mathscr{F}$. The relative entropy of $\Gamma$ with respect to $\Gamma'$ is defined by
\begin{equation}
	S(\Gamma,\Gamma') = \tr\left[ \Gamma \left( \ln(\Gamma) - \ln(\Gamma') \right) \right] \geq 0.
	\label{eq:relativeentropy}
\end{equation}
We recall the unitary equivalence $\mathscr{F} \cong \mathscr{F}_0 \otimes \mathscr{F}_+$, which has been discussed in \eqref{eq:exponentialProperty}. For a state $\Gamma$ on $\mathscr{F}_0 \otimes \mathscr{F}_+$ we denote by $\Gamma_0$ and $\Gamma_+$ the restriction (via a partial trace) to the first and the second tensor factor, respectively. The relative entropy satisfies the following lower bound with respect to restriction of states.
\begin{lemma}
	\label{restrictionrelentr}
	Let $\Gamma'$ be a translation-invariant quasi-free state on $\mathscr{F}_0 \otimes \mathscr{F}_+$. For any state $\Gamma$ on $\mathscr{F}_0 \otimes \mathscr{F}_+$, we have
	\begin{equation}
		S(\Gamma,\Gamma') \geq S(\Gamma_0,\Gamma'_0) + S(\Gamma_+,\Gamma'_+).
	\end{equation}
\end{lemma}
\begin{proof}
	The proof is a direct consequence of \cite[Lemma~4]{Seiringer-08}. 
\end{proof}

The next lemma quantifies the coercivity of the relative entropy with respect to the 1-pdms of the states under consideration. The statement and its proof can be found in \cite{LewNamRou-21}, see Eq.~(6.3).

\begin{lemma}
	\label{lemcoercivityrelentr}
	Let $h>0$ be a positive operator on a (separable complex) Hilbert space $\mathfrak{h}$ that satisfies $\mathrm{tr}[h^{-2}] < + \infty$. Consider the quasi free state
	\begin{equation}
		\Gamma' = \frac{\exp(-\de \Gamma(h))}{\tr[ \exp(-\de \Gamma(h)) ]}
	\end{equation}
	with 1-pdm $\gamma'$ on the bosonic Fock space $\mathscr{F}(\mathfrak{h})$. Then for any state $\Gamma$ on $\mathscr{F}(\mathfrak{h})$ with 1-pdm $\gamma$, we have
	\begin{equation}
		\Vert \gamma - \gamma' \Vert_1 \leq 2 \sqrt{2} \sqrt{\mathrm{tr}[h^{-2}]} \sqrt{ S(\Gamma,\Gamma')} + 2 \Vert h^{-1} \Vert_{\infty} \ S(\Gamma,\Gamma'),
	\end{equation} 
    where $\Vert \cdot \Vert_{\infty}$ denotes the operator norm.
\end{lemma}

We are now prepared to give the proof of Proposition~\ref{prop:roughapriori}.
\subsection{Proof of Proposition~\ref{prop:roughapriori}}
In the first step we provide a bound for the free energy in \eqref{eq:freeenergy}, which reads
\begin{equation}
    F(\beta,N) \leq F^+_0(\beta,N) + \frac{\hat{v}(0) N}{2} + C N^{1/3}
    \label{eq:easyUpperBound}
\end{equation}
with $F_0^+$ in \eqref{eq:freeEnergyIdealGasCloud}. This statement follows from a trial state argument with the state 
\begin{align}
    \Gamma^{\mathrm{trial}} &= | \sqrt{N_0(\beta,N)} \rangle \langle \sqrt{N_0(\beta,N)} | \otimes G_{+}^{\mathrm{id}}(\beta,N) \quad \text{ with } \nonumber \\ 
    G_{+}^{\mathrm{id}}(\beta,N)  &=   \frac{\exp\left( -\beta (\de \Upsilon(-Q \Delta) - \mu_0(\beta,N) \mathcal{N}_+) \right)}{\tr_{\mathscr{F}_+} \exp\left( -\beta (\de \Upsilon(-Q \Delta) - \mu_0(\beta,N) \mathcal{N}_+) \right)},
    \label{eq:aPrioriCondensateFraction2}
\end{align}
the coherent state in $|z \rangle$ in \eqref{eq:coherentstate}, $Q = \mathds{1}(-\Delta \neq 0)$, $\mu_0$ in \eqref{eq:idealgase1pdmchempot}, $N_0$ in \eqref{eq:crittemp}, and $\cN_+=\sum_{q\in \Lambda_+^*} a_q^* a_q$. The proof of \eqref{eq:easyUpperBound} is a strongly simplified version of that of Proposition~\ref{prop:upperbound}, and therefore left to the reader. Next, we prove a corresponding lower bound.

An application of Lemma~\ref{lem:OnsagersLemma} shows
\begin{equation}
    \mathcal{F}(G_{\beta,N}) \geq F_0(\beta,N) + \frac{1}{\beta} S(G_{\beta,N},G^{\mathrm{id}}_{\beta,N}) + \frac{\hat{v}(0) \Tr[ \mathcal{N}^2 G_{\beta,N} ]}{2 N} - \frac{v(0)}{2}
    \label{eq:aPrioriCondensateFraction3}
\end{equation}
with $G^{\mathrm{id}}_{\beta,N}$ in \eqref{eq:GibbsStateIdealGas} and the relative entropy $S(G_{\beta,N},G^{\mathrm{id}}_{\beta,N})$ in \eqref{eq:relativeentropy}. Here we also used the identity 
\begin{equation}
    S(G_{\beta,N},G^{\mathrm{id}}_{\beta,N}) = \beta \Tr[\de \Upsilon(-\Delta) ( G_{\beta,N} - G^{\mathrm{id}}_{\beta,N} )] -  S(G_{\beta,N}) + S(G^{\mathrm{id}}_{\beta,N}).
    \label{eq:aPrioriCondensateFraction4}
\end{equation}
When we additionally use $\Tr[\mathcal{N}^2 G_{\beta,N}] \geq (\Tr[\mathcal{N} G_{\beta,N}] )^2$ and $F_0(\beta,N) \geq F_0^+(\beta,N) - C N^{2/3} \ln(N)$, we find
\begin{equation}
    F(\beta,N) \geq F^+_0(\beta,N) + \frac{\hat{v}(0) N}{2} + \frac{1}{\beta}S(G_{\beta,N},G^{\mathrm{id}}_{\beta,N})  - C N^{2/3} \ln(N).
    \label{eq:aPrioriCondensateFraction5}
\end{equation}

In combination, \eqref{eq:easyUpperBound} and \eqref{eq:aPrioriCondensateFraction5} imply the bound
\begin{equation}
    S(G_{\beta,N},G^{\mathrm{id}}_{\beta,N}) \lesssim \ln(N)
    \label{eq:aPrioriCondensateFraction6}
\end{equation}
for the relative entropy.

Our proof of the bound in \eqref{eq:aPrioriCondensateFraction1} uses that the 1-pdm a translation-invariant state is diagonal in Fourier space. We hightlight that both, $G_{\beta,N}$ and $G_{\beta,N}^{\mathrm{id}}$ are translation-invariant. In the following we denote the 1-pdms of $G_{\beta,N}$ and $G_{\beta,N}^{\mathrm{id}}$ by $\gamma_{\beta,N}$ and $\gamma_{\beta,N}^{\mathrm{id}}$, respectively. We first derive a bound for $\Vert Q(\gamma_{\beta,N} - \gamma_{\beta,N}^{\mathrm{id}}) \Vert_1$. From Lemma~\ref{restrictionrelentr} we know that $S(G_{\beta,N},G_{\beta,N}^{\mathrm{id}}) \geq S(G_{\beta,N,+},G_{\beta,N,+}^{\mathrm{id}})$. Moreover, an application of Lemma~\ref{lemcoercivityrelentr} applied to $S(G_{\beta,N,+},G_{\beta,N,+}^{\mathrm{id}})$ with the choice $h = - \beta \Delta Q$ proves 
	\begin{equation}
		\Vert Q ( \gamma_{\beta,N} - \gamma_{\beta,N}^{\mathrm{id}} ) \Vert_1 \lesssim N^{2/3} \ln(N).
		\label{eq:roughapriori5}
	\end{equation}
We note that the bound in \eqref{eq:roughapriori5} would be worse if we had not considered $Q \gamma_{\beta,N}$ and $Q \gamma_{\beta,N}^{\mathrm{id}} $ instead of $\gamma_{\beta,N}$ and $\gamma_{\beta,N}^{\mathrm{id}} $. This is because the largest eigenvalue of $\gamma_{\beta,N}^{\mathrm{id}}$ equals $N_0$ and is potentially of order $N$ while that of $Q \gamma_{\beta,N}^{\mathrm{id}}$ is bounded by a constant times $\beta^{-1} \lesssim N^{2/3}$. Using \eqref{eq:roughapriori5} and 
\begin{equation}
    \Tr[ a_0^* a_0 G_{\beta,N} ] + \trs[ Q \gamma_{\beta,N} ] = N = N_0(\beta,N) + \trs[Q \gamma_{\beta,N}^{\mathrm{id}}],
\end{equation}
where $\trs$ denotes the trace over $L^2(\Lambda)$, we infer that 
\begin{equation}
    | \Tr[ a_0^* a_0 G_{\beta,N} ] - N_0(\beta,N) | \lesssim N^{2/3} \ln(N),
\end{equation}
which proves the claim of Proposition~\ref{prop:roughapriori}. 
\section{The Gibbs state part I: first order a priori estimates}
\label{sec:FirstOrderEstimates}
In this section, we prove a priori estimates for a generalized version of the Gibbs state in \eqref{eq:interactingGibbsstate} using a first-order Griffith argument (i.e. a Hellmann–Feynman type argument). These estimates serve as ingredients for proving the correlation inequalities in Section~\ref{sec:correlationInequalities} and the sharp lower bound for the free energy in Section~\ref{sec:lowerBound}. To prove these bounds we find it more convenient to work with the grand potential than with the free energy. We also replace the Laplacian by a more general one-particle operator, which allows us to add perturbations. We begin our discussion with the mathematical set-up used in this section.

Let $h$ be a self-adjoint operator on $L^2(\Lambda)$ with the following three properties:
\begin{equation}
\label{eq:generalizedOneParticleHamiltonian}
h \gesim -\Delta,  \quad \|h+\Delta\|_{\infty} \lesim 1, \quad \text{ and } \quad h \varphi_0=0,
\end{equation}
where $\Vert \cdot \Vert_{\infty}$ denotes the operator norm and $\varphi_0(x) = 1$ for all $x \in \Lambda$. The second quantized Hamiltonian of our system reads
\begin{equation}
    \mathcal{H}_{h,\eta} = \de \Upsilon(h) + \frac{1}{2 \eta} \sum_{p,u,v \in \Lambda^*} \hat{v}(p) a_{u+p}^* a_{v-p}^* a_u a_v
    \label{eq:perturbedFockSpaceHamiltonian}
\end{equation}
Here we replaced the factor $1/N$ in front of the interaction potential by $1/\eta$ with some $\eta > 0$. The reason for this is twofold: we want to reserve the letter $N$ for the expected number of particles in the system and it is mathematically convenient for us to allow the coupling constant to differ from $1/N$. Later we will choose $\eta \sim N$. We also introduce the Gibbs state
\begin{equation}
    G_{h,\eta}(\beta,\mu) = \frac{\exp(-\beta( \mathcal{H}_{h,\eta} - \mu \mathcal{N} ))}{\Tr[\exp(-\beta( \mathcal{H}_{h,\eta} - \mu \mathcal{N} ))]}
    \label{eq:GeneralGibbsState}
\end{equation}
with $\mu \in \mathbb{R}$. By $G_{\eta}(\beta,\mu)$ we denote the Gibbs state in the special case $h = -\Delta$. 

We highlight that, in contrast to $-\Delta$, the operator $h$ is, in general, not assumed to commute with translations in position space. This is important as it allows us to compute the expectation of the non translation-invariant operator $B_p$ defined in part~(a) of Theorem~\ref{thm:firstOrderAPriori} below in the perturbed Gibbs state $G_{h,\eta}(\beta,\mu)$. The expectation of this term allows us to obtain a bound on the exchange term related to the Gibbs state $G_{\beta,N}$, which is a crucial ingredient for the proof of the lower bound for the free energy, see \eqref{eq:AA-G} and \eqref{eq:lowerBoundFE5} in Section~\ref{sec:lowerBound}.

The following theorem is the main result of this section.

\begin{theorem}[First order a priori estimates]
\label{thm:firstOrderAPriori}
Let $v$ satisfy the assumptions of Theorem~\ref{thm:norm-approximation}. We consider the limit $\eta \to \infty$, $\beta\eta^{2/3} \to \kappa \in (0,\infty)$. The chemical potential $\mu$, which may depend on $\eta$, is assumed to satisfy $-\eta^{2/3} \lesssim \mu \lesssim 1$. Let $\widetilde \mu < 0$ be the unique solution to the equation
    \begin{equation}
        \sum_{p \in \Lambda^*} \frac{1}{e^{\beta(p^2 - \widetilde{\mu})}-1} = \frac{(\mu - \widetilde{\mu})\eta}{\hat{v}(0)}
        \label{eq:GrantCanonicalEffectiveIddealGasChemPot}
    \end{equation}
and define
    \begin{equation}
        M(\beta,\widetilde{\mu}) = \sum_{p \in \Lambda^*} \frac{1}{e^{\beta(p^2 - \widetilde{\mu})}-1}.
        \label{eq:particleNumber}
    \end{equation}
Moreover, let $\mu_0(\beta,M)$ and $N_0(\beta,M)$ be defined as in \eqref{eq:idealgase1pdmchempot} and \eqref{eq:crittemp}, respectively (note that $\mu_0(\beta,M) = \widetilde{\mu}$). Then the following holds: 
\begin{enumerate}[label=(\alph*)]
    \item For $p \in \Lambda_+^*$ we have 
    \begin{equation} \label{eq:unperturbedFirstOrderBounds}
        \tr[ a_p^* a_p G_{h,\eta}(\beta,\mu) ] \lesssim \frac{1}{\beta p^2} \quad \text{ and } \quad |\tr[ B_p G_{h,\eta}(\beta,\mu) | \lesssim \eta^{2/3}, 
    \end{equation}
    where $B_p = \de \Upsilon( Q \cos(p \cdot x) Q )$ with $Q$ in \eqref{eq:excitationFockSpace}. Moreover, $\cos(p \cdot x)$ denotes the multiplication operator with this function in position space.  
    \item We have
    \begin{align}
        \left| \Tr [\cN_+ G_{h,\eta}(\beta,\mu)] - \Tr [\cN_+ G_{\eta}(\beta,\mu) ] \right| &\lesssim \eta^{2/3} \quad \text{ and } \nonumber \\
        \left| \Tr [ \cN_+ G_{\eta}(\beta,\mu) ] - (M-N_0(\beta,M)) \right| &\lesssim \eta^{2/3}. \label{eq:perturbedFirstOrderBounds-b}
    \end{align}
    \item We have
    \begin{equation}
    \Tr [\cN^2 G_{h,\eta}(\beta,\mu) ]  \lesim  \eta^2.
    \label{eq:particleNumberBoundsPerturbedState}
    \end{equation}
    \item Let $w$ be a translation-invariant operator on $L^2(\Lambda)$, which satisfies
    \begin{equation}
        - c \Delta - (|\mu_0(\beta,M)|/4) \ |\varphi_p \rangle \langle \varphi_p | \leq w \leq - \frac{1}{c} \Delta + (|\mu_0(\beta,M)|/4) \ |\varphi_p \rangle \langle \varphi_p |
        \label{eq:definitionOfW}
    \end{equation}
    with some $c>0$. If $N_0(\beta,M) \lesssim N^{2/3}$ then
    \begin{equation} \label{eq:firstorderapapbound}
        \tr[ a_p^* a_p G_{w,\eta}(\beta,\mu) ] \lesssim \frac{1}{\beta (p^2-\mu_0(\beta,M))}
    \end{equation}
    holds for all $p \in \Lambda^*$.
\end{enumerate}
\end{theorem}

\begin{remark}
\label{rem:interactingChemicalPotential}
    The assumption $-\eta^{2/3} \lesssim \mu \lesssim 1$ will be justified in Section~\ref{sec:boundsChemicalPotential}, where we show that $\mu_{\beta,N}$ satisfies this bound with $\eta$ replaced by $N$. Accordingly, the bounds in Theorem~\ref{thm:firstOrderAPriori} hold for the choices $h = - \Delta$, $\eta = N$, and $\mu = \mu_{\beta,N}$. In this case we have $G_{h,\eta}(\beta,\mu) = G_{\beta,N}$.
\end{remark}
\begin{remark}
\label{rem:effectiveParticleNumbers}
    In \eqref{eq:GrantCanonicalEffectiveIddealGasChemPot} and \eqref{eq:particleNumber} we define an effective chemical potential and an effective particle number. It can be shown that the expected number of particles in the system is, to leading order as $\eta \to \infty$, given by $M$. It can also be shown that the expected number of particle in the condensate is, in the same limit and to leading order, given by $1/(\exp(\beta \widetilde{\mu})-1)$. Since both statements are not needed in our analysis, we do not prove them. Their proves can, however, be found in \cite{DeuNapNam-25}, see Proposition~1.
\end{remark}

The proof of the above theorem will be carried our in two steps. In the first step we derive in Section~\ref{sec:boundsPerturbedGrandPotential} bounds for the grand potential
\begin{equation}
    \Phi_{h,\eta}(\beta,\mu) = -\frac{1}{\beta} \ln\left( \tr[\exp(-\beta(\mathcal{H}_{h,\eta} - \mu \mathcal{N}))] \right).
    \label{eq:perturbedGrandPotential}
\end{equation}
In the second step we use these bounds in Sections~\ref{sec:proofFirstOrderApriori1} to prove Theorem~\ref{thm:firstOrderAPriori} with a Griffith argument.
\subsection{Bounds for the perturbed grand potential}
\label{sec:boundsPerturbedGrandPotential}
The bounds for the grand potential are captured in the following proposition.

\begin{proposition}
\label{prop:boundsPerturbedGrandPotential}
    Let $v$ satisfy the assumptions of Theorem~\ref{thm:norm-approximation}. We consider the limit $\eta \to \infty$, $\beta \eta^{2/3} \to \kappa  \in (0,\infty)$. The chemical potential $\mu$, which may depend on $\eta$, is assumed to satisfy $-\eta^{2/3} \lesssim \mu \lesssim 1$. Let $M$, $\mu_0(\beta,M)$, and $N_0(\beta,M)$ be defined as in \eqref{eq:particleNumber}, \eqref{eq:idealgase1pdmchempot}, and \eqref{eq:crittemp}, respectively. Then we have
    \begin{equation}
    \label{eq:boundsGrandPotential}
        \Phi_{h}(\beta,\mu) = \Phi_{h}^{\mathrm{id}}(\beta,\mu_0(\beta,M)) - \frac{(\mu - \mu_0(\beta,M))^2 \eta}{2 \hat{v}(0)} + \Theta^{\mathrm{BEC}}(\beta,N_0(\beta,M)) + O(\eta^{2/3}).
    \end{equation}
    Here 
    \begin{equation}
        \Phi_{h}^{\mathrm{id}}(\beta,\mu_0(\beta,M)) = \beta^{-1} \trs_+ \left[ \ln\left( 1 - \exp(-\beta Q (h - \mu_0(\beta,M))) \right) \right],
        \label{eq:PhiId}
    \end{equation}
    where $Q$ has been defined in \eqref{eq:excitationFockSpace}, $\trs_+[A] = \trs[Q A Q]$, $A \in \mathcal{B}(L^2(\Lambda))$, and
    \begin{equation}
    \label{eq:ThetaBEC}
        \Theta^{\mathrm{BEC}}(\beta,N_0) = - \begin{cases} \frac{5}{6 \beta} \ln(\eta) & \text{ if } N_0 \geq \eta^{5/6}, \\ \frac{1}{\beta} \ln(N_0) & \text{ if } N_0 < \eta^{5/6}. \end{cases}
    \end{equation}

    Moreover, if $w$ is given as in part~(d) of Theorem~\ref{thm:firstOrderAPriori} and $N_0(\beta,M) \lesssim \eta^{2/3}$ then we have
    \begin{equation}
        \Phi_w(\beta,\mu) = \widetilde{\Phi}_w^{\mathrm{id}}(\beta,\mu_0(\beta,M)) - \frac{(\mu - \mu_0(\beta,M))^2 \eta}{2 \hat{v}(0)} + O(\eta^{2/3}),
        \label{eq:newBoundAboveCriticalPoint}
    \end{equation}
    where $\widetilde{\Phi}_w^{\mathrm{id}}$ equals $\Phi_{h}^{\mathrm{id}}$ in \eqref{eq:PhiId} with $h$ and $\trs_+[\cdot]$ are replaced by $w$ and $\trs[\cdot]$, respectively.
\end{proposition}
The proof of the above proposition will also be given in two steps. In the first step we prove an upper bound for $\Phi_{h}(\beta,\mu)$ and $\Phi_{w}(\beta,\mu)$, in the second step a lower bound. 
\subsubsection{Upper bound for the perturbed grand potential}
\label{sec:upperBoundGrandPotential}
As in the proof of the sharp upper bounds for the free energy in Section~\ref{sec:upperBoundFreeEnergy} we apply a trial state argument, and we start with the upper bound in \eqref{eq:boundsGrandPotential}. The main difference between the analysis in these two sections is that the one-particle operator $h$, in contrast to the Laplacian, does not commute with translations. 

We start our analysis with the definition of the trial state. It reads
\begin{equation}
\label{eq:trialStateFirstOrderEstimates}
    \Gamma^{\mathrm{trial}} = \int_{\mathbb{C}} | z \rangle \langle z | \ g(z) \de z \otimes G^{\mathrm{id}}_{h,+}(\beta,M), \quad G_{h,+}^{\mathrm{id}}(\beta,M) =   \frac{\exp\left( -\beta (\de \Upsilon(Q h) - \mu_0(\beta,M) \mathcal{N}_+) \right)}{\tr_{\mathscr{F}_+} \exp\left( -\beta (\de \Upsilon(Q h) - \mu_0(\beta,M) \mathcal{N}_+) \right)} 
\end{equation}
with the coherent state $| z \rangle$ in \eqref{eq:coherentstate}, the projection $Q$ in \eqref{eq:excitationFockSpace}, the chemical potential $\mu_0$ in \eqref{eq:idealgase1pdmchempot}, and where the chemical potential $\mu$ in the definition of the Gibbs distribution $g$ in \eqref{eq:GibbsDistributionDiscrete} (as in Section~\ref{sec:upperBoundFreeEnergy} we omit the superscript $\mathrm{BEC}$) is chosen such that $\int |z|^2 g(z) \de z = N_0(\beta,M)$ holds. Since $M$ is defined via \eqref{eq:particleNumber} the expected number of particles in the trial state will, in general, not equal $M$. However, it is given by $M$ to leading order, see Remark~\ref{rem:effectiveParticleNumbers} above. An application of Lemma~\ref{lem:BerezinLieb} shows that the entropy of our trial state satisfies 
\begin{equation}
    S(\Gamma^{\mathrm{trial}}) \geq S(G_{h,+}^{\mathrm{id}}(\beta,M)) + S(g). 
    \label{eq:UpperBoundPerturbedGrandPotential1}
\end{equation}
It remains to compute the energy of $\Gamma^{\mathrm{trial}}$.

We recall the decomposition of the Hamiltonian $\mathcal{H}_N$ in \eqref{eq:decompH}, introduce for $p,q \in \Lambda_+^*$ the notation $\gamma_h^{\mathrm{id}}(p,q) = \tr [ a_q^* a_p G_+^{\mathrm{id}}(\beta,M) ]$, and abbreviate $\mu_0(\beta,M) = \mu_0$, $N_0(\beta,M) = N_0$. A short computation that uses the Wick theorem shows 
\begin{align}
	&\tr[(\mathcal{H}_{h,\eta}-\mu \mathcal{N} ) \Gamma^{\mathrm{trial}}] = \tr_{\mathscr{F}_+}[ \de \Upsilon(Q(h - \mu_0)) G^{\mathrm{id}}_{h,+}(\beta,M)] + \mu_0 \trs_+[\gamma_h^{\mathrm{id}}] - \mu (N_0+\trs_+[\gamma_h^{\mathrm{id}}] )   \nonumber \\
	&\hspace{+0.3cm}+ \frac{N_0}{\eta} \sum_{p\in \Lambda_+^*} \hat{v}(p) \gamma_h^{\mathrm{id}}(p,p) + \frac{\hat{v}(0)}{2\eta} \Bigg\{ \left( N_0 + \trs_+[\gamma_h^{\mathrm{id}}] \right)^2 + \int_{\mathbb{C}} |z|^4 g(z) \de z - N_0^2 + \sum_{u,v \in \Lambda_+^*} | \gamma_h^{\mathrm{id}}(u,v) |^2 \Bigg\} \nonumber 	  \\
	&\hspace{+0.3cm}+ \frac{1}{2\eta} \sum_{p,u,v,u+p,v-p \in \Lambda_+^*} \hat{v}(p) \left\{ \gamma_h^{\mathrm{id}}(u+p,u) \gamma_h^{\mathrm{id}}(v-p,v) + \gamma_h^{\mathrm{id}}(u+p,v) \gamma_h^{\mathrm{id}}(v-p,u) \right\}, \label{eq:firstOrderUpperBound1}
\end{align}
where $\trs_+$ denotes the trace over $Q L^2(\Lambda)$. Our assumption $h \gtrsim - \Delta$ implies that the $k$-th eigenvalue of $h$ is bounded from below by the $k$-th eigenvalue of $-\Delta$ (with a constant that does not depend on $k$). In combination with $(e^x-1)^{-1} \leq x$ for $x>0$ and $\mu_0 < 0$, this, in particular, implies $\gamma_h^{\mathrm{id}} \leq \beta^{-1}$. Using these two bounds, $N_0(\beta,M) \leq M$, and the fact that the Hilbert--Schmidt norm of a self-adjoint operator can be expressed in terms of its eigenvalues, we check that
\begin{equation}
	\frac{N_0}{\eta} \sum_{p\in \Lambda_+^*} \hat{v}(p) \gamma_h^{\mathrm{id}}(p,p) \lesssim \frac{M}{\beta \eta} \quad \text{ as well as } \quad \sum_{u,v \in \Lambda_+^*} | \gamma_h^{\mathrm{id}}(u,v) |^2 = \Vert \gamma_h^{\mathrm{id}} \Vert_{2}^2 \leq \frac{1}{\beta^2} \sum_{p \in \Lambda_+^*} \frac{1}{|p|^4}.
	\label{eq:firstOrderUpperBound2}
\end{equation}  
Here $\Vert \cdot \Vert_2$ denotes the Hilbert--Schmidt norm. The last term on the right-hand side of \eqref{eq:firstOrderUpperBound1} is bounded from above by
\begin{equation}
	\frac{1}{4\eta} \sum_{p,u,v,u+p,v-p \in \Lambda_+^*} \hat{v}(p) \left\{|\gamma_h^{\mathrm{id}}(u+p,v)|^2 + |\gamma_h^{\mathrm{id}}(v-p,u) |^2 \right\} = \frac{\Vert \gamma_h^{\mathrm{id}} \Vert_2^2}{2 \eta} \sum_{p \in \Lambda_+^*} \hat{v}(p) \lesssim \frac{1}{\beta^2 \eta}.
	\label{eq:firstOrderUpperBound3}
\end{equation}

In the next step we derive a bound for the first term in the last line of \eqref{eq:firstOrderUpperBound1}. To that end, we need the following lemma.
\begin{lemma} \label{lem:perturbedtrace}
	Let $h_0, X$ be self-adjoint operators on $Q L^2(\Lambda)$ with $Q$ in \eqref{eq:excitationFockSpace} and assume that $X$ is bounded. We also assume that the operators $h_0$ and $h = h_0 + X$ satisfy
	\begin{align} 
		\label{eq:trace-diff-exp-h-h0-0}
		h, h_0 \gtrsim -\Delta \quad \text{ on } \quad Q L^2(\Lambda).
	\end{align}
	For any $\beta > 0$ and any $\mu < 0 $ we then have
	\begin{align} \label{eq:trace-diff-exp-h-h0-2}
		\trs_+  \left| \frac{1}{e^{\beta (h-\mu)} -1} - \frac{1}{e^{\beta (h_0-\mu )}-1}  \right|   \lesssim  \beta^{-1}.
	\end{align}
\end{lemma}
\begin{proof}
	An application of the Cauchy-Schwarz inequality shows
	\begin{align} 
		\trs_{+}  \left| \frac{1}{e^{\beta (h-\mu)} -1} - \frac{1}{e^{\beta (h_0-\mu)}-1}  \right|   &\le   \|h_0^{-1}\|_{2} \left\| (h_0 -\mu) \left( \frac{1}{e^{\beta (h-\mu)} -1} - \frac{1}{e^{\beta (h_0-\mu)}-1} \right)\right\|_{2}  \nonumber \\
        &\lesssim \left\|(h_0-\mu) \left( \frac{1}{e^{\beta (h\mu)} -1} - \frac{1}{e^{\beta (h_0-\mu)}-1} \right)\right\|_{2}. \label{eq:trace-diff-exp-h-h0-3}
	\end{align}
	The term on the right-hand side can be estimated with Klein's inequality, see e.g. \cite[Prop. 3.16]{OhyPer-93}, which we recall now. 
	
	Let $A,B$ be two self adjoint operators with spectra $\sigma(A), \sigma(B)$, let $\{ f_k   \},\{g_k\}$ be two families of functions, where $f_k : \sigma(A) \to \mathbb{C}$, $g_k : \sigma(B) \to \mathbb{C}$, and assume that 
	\begin{equation}
		\sum_k f_k(a) g_k(b) \ge 0 ,\quad \forall a \in \sigma(A), b \in \sigma(B).
	\end{equation}
	Then Klein's inequality states that
	\begin{equation}
		\trs \left( \sum_k f_k(A) g_k(B) \right) \ge 0.
	\end{equation}
		
	Let us introduce the notation $f(x) = (\exp(x)-1)^{-1}$ for $x>0$ and note that $-f'(x) = (4 \sinh^2(x/2))^{-1} \leq 1/x^2$. We claim that 
	\begin{align}\label{eq:Klein-pointwise}
		y^2 (f(x) - f(y))^2 \leq 4 (x-y)^2 \left[ x^{-2} + y^{-2} \right]
	\end{align}
	holds for any $x, y > 0$. Indeed, if $y\le x$, and for some $\xi \in \{ tx + (1-t)y \ | \ t \in [0,1] \}$ we have
	\begin{align*}
		y^2 (f(x) - f(y))^2 &= y^2 |f'(\xi)|^2 (x-y)^2 \leq (x-y)^2 y^{-2}.
	\end{align*}
	Moreover, if $x \leq y$ we know that
	\begin{align*}
		y^2 (f(x) - f(y))^2 \le 2 x^2 (f(x) - f(y))^2 + 2(x-y)^2(f(x) - f(y))^2.
	\end{align*}
	The first term on the right-hand side is bounded by $2 (x-y)^2 x^{-2}$. For the second term we have
	\begin{equation}
		(x-y)^2(f(x) - f(y))^2 \leq (x-y)^2  f^2(x)  \leq (x-y)^2 x^{-2}.
	\end{equation}
	In the last step we used $f(x) \leq 1/x$. 
	
	In combination, \eqref{eq:Klein-pointwise}, Klein's inequality, and Hölder's inequalty for traces, see e.g. \cite[Proposition~5]{RS1975}, imply
	\begin{align} \label{eq:trace-diff-exp-h-h0-4}
		\trs_+ \left[ \beta^2 (h_0-\mu)^2 ( f(\beta (h-\mu))-f(\beta (h_0-\mu)))^2 \right] & \leq 4 \trs_+ \left[ (\beta h- \beta h_0)^2 \left( (\beta h)^{-2} + (\beta h_0)^{-2} \right) \right] \nonumber\\
		&\le 4 \Vert h-h_0 \Vert_{\infty} \trs_+  \left[ h^{-2}+h_0^{-2}  \right] \lesssim  \trs_+ \left[ (-\Delta)^{-2} \right] < +\infty.
	\end{align}
	Here we used $\Vert X \Vert_{\infty} < + \infty$,  \eqref{eq:trace-diff-exp-h-h0-0}, and the bound $\Tr [A^{-2} ] \le \Tr [ B^{-2} ]$ for two operators that satisfy $A \ge B > 0$. When we put \eqref{eq:trace-diff-exp-h-h0-3} and \eqref{eq:trace-diff-exp-h-h0-4} together, we find \eqref{eq:trace-diff-exp-h-h0-2} and the lemma is proved.
\end{proof}

To obtain a bound for the first term in the last line of \eqref{eq:firstOrderUpperBound1}, we write it as
\begin{align}
	\frac{1}{2\eta} \sum_{p \in \Lambda_+^*} \hat{v}(p) \left[ \sum_{u,u+p \in \Lambda_+^*} \langle \varphi_{u+p}, \gamma_h^{\mathrm{id}} \varphi_u \rangle \right] \left[ \sum_{v,v-p \in \Lambda_+^*} \langle \varphi_{v-p}, \gamma_h^{\mathrm{id}} \varphi_v \rangle \right].
    \label{eq:firstOrderUpperBound4}
	\end{align}
In the following, we denote by $\gamma^{\mathrm{id}}$ the 1-pdm of $G_{h,+}^{\mathrm{id}}(\beta,M)$ in the special case $h = -\Delta$. The translation-invariance of $\gamma^{\mathrm{id}}$, the fact that $p \in \Lambda_+^*$, the Cauchy--Schwarz inequality, and Lemma~\ref{lem:perturbedtrace} can be used to show that
\begin{equation}
	\left| \sum_{u,u+p \in \Lambda_+^*} \langle \varphi_{u+p}, \gamma_h^{\mathrm{id}} \varphi_u \rangle \right| = \left| \sum_{u,u+p \in \Lambda_+^*} \langle \varphi_{u+p}, ( \gamma_h^{\mathrm{id}} - \gamma^{\mathrm{id}} ) \varphi_u \rangle \right| \leq \trs_+ \left| \gamma_h^{\mathrm{id}} - \gamma^{\mathrm{id}} \right| \lesssim \beta^{-1},
    \label{eq:firstOrderUpperBound5}
\end{equation}
and hence
\begin{equation}
	\frac{1}{2 \eta} \left| \sum_{p,u,v,u+p,v-p \in \Lambda_+^*} \hat{v}(p) \gamma_h^{\mathrm{id}}(u+p,u) \gamma_h^{\mathrm{id}}(v-p,v) \right| \lesssim \frac{1}{\beta^2 \eta}.
	\label{eq:BoundDifficultExchangeTerm}
\end{equation}

Next, we consider the second and the third term on the right-hand side of \eqref{eq:firstOrderUpperBound1}. Using $\mu_0 < 0$ and Lemma~\ref{lem:perturbedtrace}, we check that
\begin{equation}
    \mu_0 \trs_+[\gamma_h^{\mathrm{id}}] - \mu (N_0 + \trs_+[\gamma_h^{\mathrm{id}}]) \leq ( \mu_0 - \mu ) (N_0 + \trs_+[\gamma_h^{\mathrm{id}}]) \leq ( \mu_0 - \mu ) M + \frac{C | \mu_0 - \mu |}{\beta}.
    \label{eq:firstOrderUpperBound6}
\end{equation}
Another application of Lemma~\ref{lem:perturbedtrace} shows that the second term in the second line of \eqref{eq:firstOrderUpperBound1} satisfies
\begin{equation}
    \frac{\hat{v}(0)}{2\eta} \left( N_0 + \trs_+[\gamma_h^{\mathrm{id}}] \right)^2 \leq \frac{\hat{v}(0)}{2\eta} M^2 + \frac{C M}{\beta \eta}.
    \label{eq:firstOrderUpperBound6a}
\end{equation}

From Lemma~\ref{lem:effectiveChemicalPotentialAPriori} we know that $M \lesssim \eta$ and $|\mu - \mu_0| \lesssim 1$. When we put the above considerations together, use these two bounds and $\beta \gtrsim \eta^{-2/3}$, we find  
\begin{align}
	\tr[(\mathcal{H}_{h,\eta}-\mu \mathcal{N} ) \Gamma^{\mathrm{trial}}] - \frac{1}{\beta} S(\Gamma^{\mathrm{trial}}) \leq& \Phi_{h}^{\mathrm{id}}(\beta,\mu_0(\beta,M))  + \frac{\hat{v}(0)}{2\eta} M^2 - (\mu-\mu_0(\beta,M)) M \nonumber \\
	&+ F_{\mathrm{c}}^{\mathrm{BEC}}(\beta,N_0(\beta,M)) + C \eta^{2/3} \label{eq:firstOrderUpperBound6b}
\end{align}
with $F_{\mathrm{c}}^{\mathrm{BEC}}$ defined below \eqref{eq:GibbsDistributionDiscrete}. Our choice for $M$, see \eqref{eq:particleNumber}, guarantees that $M = (\mu - \mu_0(\beta,M))\eta/\hat{v}(0)$ holds, and hence
\begin{equation}
    \frac{\hat{v}(0)}{2\eta} M^2 - (\mu-\mu_0(\beta,M)) M = - \frac{(\mu - \mu_0(\beta,M))^2 \eta}{2 \hat{v}(0)}.
    \label{eq:firstOrderUpperBound7}
\end{equation}
From Lemma~\ref{prop:FreeEnergyBEC} in Appendix~\ref{app:effcondensate} we know that $| F_{\mathrm{c}}^{\mathrm{BEC}}(\beta,N_0(\beta,M)) - \Theta^{\mathrm{BEC}}(\beta,N_0(\beta,M)) | \lesssim \eta^{2/3}$. In combination, this bound, \eqref{eq:firstOrderUpperBound6b}, and \eqref{eq:firstOrderUpperBound7} imply the final upper bound
\begin{equation}
    \Phi_{h}(\beta,\mu) \leq \Phi_{h}^{\mathrm{id}}(\beta,\mu_0(\beta,M)) - \frac{(\mu - \mu_0(\beta,M))^2 \eta}{2 \hat{v}(0)} + \Theta^{\mathrm{BEC}}(\beta,N_0(\beta,M)) + C \eta^{2/3}
    \label{eq:firstOrderUpperBound8}
\end{equation}
for the perturbed grand potential. It remains to prove the upper bound in \eqref{eq:newUpperBoundAboveCriticalPoint}.

Since $w$ is translation-invariant we can argue similarly as in the proof for the upper bound in the non-condensed phase at the end of Section~\ref{sec:Finalupperbound}. As trial state we choose the Gibbs state of the ideal gas in \eqref{eq:GibbsStateIdealGas} with $p^2$ replaced by $w(p)$ and $\mu_0(\beta,N)$ replaced by $\mu_0(\beta,M)$, which we denote by $G^{\mathrm{id}}_{w,\beta}$. We claim that
\begin{equation}
    \Tr[ (\mathcal{H}_{w,\eta}- \mu \mathcal{N} ) G^{\mathrm{id}}_{w,\eta} ] - \frac{1}{\beta} S(G^{\mathrm{id}}_{w,\eta}) \leq \widetilde{\Phi}_w^{\mathrm{id}}(\beta,\mu_0(\beta,M)) - \frac{(\mu - \mu_0(\beta,M))^2 \eta}{2 \hat{v}(0)} + O(\eta^{2/3})
    \label{eq:newUpperBoundAboveCriticalPoint}
\end{equation}
holds. The proof of this bound follows from \eqref{eq:firstOrderUpperBound1}, $N_0(\beta,M) \lesssim \eta^2$, which implies $\mu_0(\beta,M) \lesssim - 1$, and strongly simplified versions of arguments that have been used to prove \eqref{eq:firstOrderUpperBound8}. It is therefore left to the reader.
\subsubsection{Lower bound for the perturbed grand potential}
\label{sec:lowerBoundGrandPotential}
In this section we prove a lower bound for the perturbed grand potential that agrees with the upper bound in \eqref{eq:firstOrderUpperBound8} up to a remainder of the order $\eta^{2/3}$. The main idea of our proof is that if $N_0 \gesssim N^{5/6}$, the variances of $\mathcal{N}_0$ and $\mathcal{N}$ in the Gibbs state $G_{\beta,N}$ are both of order $N^{5/3}$. In contrast, the variance of $\mathcal{N}_+$ is only of the order $N^{4/3}$. Because of this we can extract the contribution of the condensate to the grand potential without distinguishing between the fluctuations of the total particle number and that of the condensed particles.

We apply Lemma~\ref{lem:OnsagersLemma} to obtain a lower bound for the interaction term in our Hamiltonian and find
\begin{equation}
    \mathcal{H}_{h,\eta} \geq \de \Upsilon(h) + \frac{\hat{v}(0) \mathcal{N}^2}{2 \eta} - \frac{v(0) \mathcal{N}}{\eta}. 
    \label{eq:firstOrderLowerBound2}
\end{equation}
Since the right-hand side of \eqref{eq:firstOrderLowerBound2} commutes with $\mathcal{N}$, we can restrict the minimization of the free energy functional related to it to states that also commute with $\mathcal{N}$. In the following, we denote by $\mathcal{S}$ the set of all states on $\mathscr{F}$ and by $\mathcal{S}^{\mathrm{block}} \subset \mathcal{S}$ the subset of all states that commute with $\mathcal{N}$. Let $P_n$ denote the projection onto the $n$-particle sector of the bosonic Fock space. All states in $\mathcal{S}^{\mathrm{block}}$ are block diagonal with respect to the family $\{ P_n \}_{n \in \mathbb{N}_0}$. Given a state $\Gamma \in \mathcal{S}^{\mathrm{block}}$ we write $\Gamma = \sum_{n=0}^{\infty} c_n \Gamma_n$ with $c_n = \tr[P_n \Gamma]$ and $\Gamma_n = P_n \Gamma/\tr[P_n \Gamma]$. For the entropy of $\Gamma$ this implies the identity
\begin{equation}
    S(\Gamma) = \sum_{n} c_n S(\Gamma_n) + S(c) \quad \text{ with the classical entropy } \quad S(c) = - \sum_n c_n \ln(c_n)
    \label{eq:firstOrderLowerBound3}
\end{equation}
of the sequence $\{c_n \}_{n \in \mathbb{N}_0}$. We use this identity for the entropy and the lower bound for $\mathcal{H}_{h,\eta}$ in \eqref{eq:firstOrderLowerBound2} to show that
\begin{align}
    \inf_{\Gamma \in \mathcal{S}} \left\{ \tr[( \mathcal{H}_{h,\eta} -  \mu  \mathcal{N}) \Gamma] - \frac{1}{\beta} S(\Gamma) \right\} \geq \inf_{\Gamma \in \mathcal{S}^{\mathrm{block}}} \Bigg\{ &\sum_n c_n \left( \tr[ \de \Upsilon(h(\lambda))) \Gamma_n ] - \frac{1}{\beta} S(\Gamma_n) + \frac{\hat{v}(0)n^2}{2\eta} - \left(  \mu  + \frac{v(0)}{2\eta} \right) n \right)  \nonumber \\
    &- \frac{1}{\beta} S(c) \Bigg\}
    \label{eq:firstOrderLowerBound4}
\end{align}
holds. In the following we abbreviate $\mu_0 = \mu_0(\beta,M)$ and $N_0 = N_0(\beta,M)$. 

From the Gibbs variational principle we know that
\begin{equation}
    \tr[ \de \Upsilon(h) \Gamma_n ] - \frac{1}{\beta} S(\Gamma_n) \geq F^{\mathrm{can}}(\beta,n),
    \label{eq:firstOrderLowerBound5}
\end{equation}
where $F^{\mathrm{can}}(\beta,n)$ denotes the canonical free energy related to the one-particle Hamiltonian $h$ at inverse temperature $\beta$ and particle number $n$. It is defined as the infimum of the left-hand side of \eqref{eq:firstOrderLowerBound5} over $\mathcal{S}_n^{\mathrm{can}}$, the set of states on the $n$-particle Hilbert space $P_n \mathscr{F}$. The map $n \mapsto F^{\mathrm{can}}(\beta,n)$ is convex, see e.g. \cite[Proposition~A.1]{DeuSeiYng-19}, and hence
\begin{equation}
    \sum_n c_n F^{\mathrm{can}}(\beta,n) \geq F^{\mathrm{can}}(\beta,\overline{N}),
    \label{eq:firstOrderLowerBound6}
\end{equation}
where $\overline{N} = \sum_n n c_n$. 

We also denote by $\mathscr{F}_+^{\leq n}$ the excitation Fock space with particle number cutoff $n$ over the one-particle Hilbert space $QL^2(\Lambda)$. That is, 
\begin{equation}
    \mathscr{F}_+^{\leq n} = P_{+,\leq n} \mathscr{F}(QL^2(\Lambda)), \quad \text{ where } \quad P_{+,\leq n} = \sum_{j=0}^n P_j
\end{equation}
denotes the projection onto the Fock space sectors of $\mathscr{F}(QL^2(\Lambda))$ with particle number smaller than or equal to $n$. By $S_{\leq n}^+$ we denote the set of states on $\mathscr{F}_+^{\leq n}$, and $U_n$ is the unitary map from $P_n \mathscr{F}$ to $\mathscr{F}_+^{\leq n}$ defined by
\begin{equation}
    U_n ( \varphi_0^{\otimes n} \otimes_{\mathrm{sym}} \psi_0 + \varphi_0^{\otimes n-1} \otimes_{\mathrm{sym}} \psi_{1} + ... + \psi_n ) = (\psi_0, \psi_1, ..., \psi_n),
    \label{eq:firstOrderLowerBound7}
\end{equation}
where $\psi_j$, $j \in \{ 0,1,...,n\}$ denotes a symmetrized $j$-particle function and $\otimes_{\mathrm{sym}}$ the symmetric tensor product. The map $U_n$ has been introduced in \cite[Eq.~(2.14)]{LewNamSerSol-15}. Using the unitarity of $U_n$ and the unitary invariance of the entropy, we check that
\begin{align}
    F^{\mathrm{can}}(\beta,\overline{N}) &=  \inf_{\Gamma \in \mathcal{S}_{\overline{N}}^{\mathrm{c}}} \left\{ \tr[ \de \Upsilon(h) \Gamma] - \frac{1}{\beta} S(\Gamma) \right\}  = \inf_{\Gamma \in S^+_{\leq \overline{N}}} \left\{ \tr[ \de \Upsilon(Q h) \Gamma] - \frac{1}{\beta} S(\Gamma) \right\} \nonumber \\
    & \geq \inf_{\Gamma \in S^+_{\leq \overline{N}}} \left\{ \tr[ \de \Upsilon(Q( h - \mu_0 ) ) \Gamma] - \frac{1}{\beta} S(\Gamma) \right\} + \mu_0 \overline{N} \nonumber \\
    &\geq \frac{1}{\beta} \tr_+ \left[ \ln(1-\exp(-\beta Q (h-\mu_0))) \right] + \mu_0 \overline{N}  \label{eq:firstOrderLowerBound8}
\end{align}
holds. To come to the second line we also used $\mu_0 < 0$. In combination, \eqref{eq:firstOrderLowerBound5}, \eqref{eq:firstOrderLowerBound6}, and \eqref{eq:firstOrderLowerBound8} show
\begin{equation}
    \sum_n c_n \left( \tr[ \de \Upsilon(-\Delta) \Gamma_n ] - \frac{1}{\beta} S(\Gamma_n) \right) \geq \frac{1}{\beta} \trs_+ \left[ \ln(1-\exp(-\beta Q (h-\mu_0))) \right] + \mu_0 \sum_n n c_n.
    \label{eq:firstOrderLowerBound9}
\end{equation}

Let us have a closer look at the remaining terms, which read
\begin{equation}
    \sum_n c_n \left( \frac{\hat{v}(0)n^2}{2 \eta} - \left( \mu-\mu_0 + \frac{v(0)}{2\eta} \right) n \right) - \frac{1}{\beta} S(c) \geq - \frac{1}{\beta} \ln\left( \sum_n \exp\left(-\beta \left( \frac{\hat{v}(0)n^2}{2N} - \left( \mu-\mu_0 + \frac{v(0)}{2N} \right) n \right) \right) \right).
    \label{eq:firstOrderLowerBound11}
\end{equation}  
From Lemma~\ref{lem:effectiveChemicalPotentialAPriori} we know that there exists a constant $c>0$ such that $c^{-1} \leq \mu - \mu_0 \leq c$. Using this and Corollary~\ref{cor:ComparisonCondensateGrandPotentials}, we see that 
\begin{align}
    &- \frac{1}{\beta} \ln\left( \sum_n \exp\left(-\beta \left( \frac{\hat{v}(0)n^2}{2\eta} - \left( \mu-\mu_0 + \frac{v(0)}{2\eta} \right) n \right) \right) \right) \nonumber \\
    &\hspace{4cm} \geq - \frac{1}{\beta} \ln\left( \int_{\mathbb{C}} \exp\left(-\beta \left( \frac{\hat{v}(0)|z|^4}{2\eta} - \left( \mu-\mu_0 + \frac{v(0)}{2\eta} \right) |z|^2 \right) \right) \de z \right) - C.
    \label{eq:firstOrderLowerBound12}
\end{align}
Let $\widetilde{N}_0$ be the expected particle number related to the Gibbs distribution $g$ in \eqref{eq:GibbsDistributionDiscrete} at inverse temperature $\beta$ and chemical potential $\widetilde{\mu} = \mu - \mu_0 + v(0)/(2\eta)$. Since $c^{-1} \leq \widetilde{\mu} \leq c$, we can follow the proof of \eqref{eq:LemChemPotBEC1} in Lemma~\ref{lem:ChemPotBECCont} to obtain $\widetilde{N}_0 = \widetilde{\mu} \eta/\hat{v}(0) + O(\exp(-c \eta^{1/6}))$. Using this, Proposition~\ref{prop:FreeEnergyBEC}, and $|\mu-\mu_0| \lesssim 1$, we get the lower bound
\begin{equation}
    - \frac{1}{\beta} \ln\left( \int_{\mathbb{C}} \exp\left(-\beta \left( \frac{\hat{v}(0)|z|^4}{2\eta} - \left( \mu-\mu_0 + \frac{v(0)}{2\eta} \right) |z|^2 \right) \right) \de z \right) \geq -\frac{5}{6 \beta} \ln(\eta) - \frac{(\mu - \mu_0)^2 \eta}{2 \hat{v}(0)} - C. \label{eq:firstOrderLowerBound13}
\end{equation}

In combination, \eqref{eq:firstOrderLowerBound4}, \eqref{eq:firstOrderLowerBound9}, and \eqref{eq:firstOrderLowerBound13} imply the lower bound
\begin{equation}
    \Phi_{h}(\beta,\mu) \geq \Phi_{h}^{\mathrm{id}}(\beta,\mu_0(\beta,M)) - \frac{(\mu - \mu_0(\beta,M))^2 \eta}{2 \hat{v}(0)} - \frac{5}{6 \beta} \ln(\eta) - C.
    \label{eq:firstOrderLowerBound14}
\end{equation}
We use \eqref{eq:firstOrderLowerBound14} if $N_0(\beta,M) \geq \eta^{5/6}$. If $N_0(\beta,M) < \eta^{5/6}$ we use another bound that we derive now.

Applications of \eqref{eq:firstOrderLowerBound2} and the Gibbs variational principle show that
\begin{equation}
    \tr[\mathcal{H}_{h,\eta} \Gamma] - \frac{1}{\beta} S(\Gamma) \geq \frac{1}{\beta} \trs \left[ \ln(1-\exp(-\beta (h-\mu_0))) \right] - \frac{\widetilde{\mu}^2 \eta}{2 \hat{v}(0)} + \frac{\hat{v}(0)}{2 \eta} \Tr\left[ \left( \mathcal{N} - \frac{\widetilde{\mu} \eta}{\hat{v}(0)} \right)^2 \Gamma \right]
    \label{eq:firstOrderLowerBound15}
\end{equation}
holds for any state $\Gamma \in \mathcal{S}$. The second term on the right-hand side is bounded from below by $-(\mu-\mu_0)^2/(2\hat{v}(0)) - C$ and the third term can be dropped for a lower bound. Moreover, using $h \varphi_0 = 0$, $-\mu_0 = \beta^{-1} \ln(1+N_0^{-1})$, and $N_0 \gesssim 1$, we see that 
\begin{equation}
    \trs \left[ \ln(1-\exp(-\beta (h-\mu_0))) \right] \geq \trs_+ \left[ \ln(1-\exp(-\beta Q (h-\mu_0))) \right] - \frac{\ln(N_0)}{\beta} - \frac{C}{\beta}.
    \label{eq:firstOrderLowerBound16}
\end{equation}
In combination, these considerations prove the lower bound
\begin{equation}
    \Phi_{h}(\beta,\mu) \geq \Phi_{h}^{\mathrm{id}}(\beta,\mu_0(\beta,M)) - \frac{(\mu - \mu_0(\beta,M))^2 \eta}{2 \hat{v}(0)} - \frac{\ln(N_0(\beta,M))}{\beta} - C \left( 1 + \frac{1}{\beta} \right). 
    \label{eq:firstOrderLowerBound17}
\end{equation}
When we put \eqref{eq:firstOrderUpperBound8}, \eqref{eq:firstOrderLowerBound14}, and \eqref{eq:firstOrderLowerBound17} together, we obtain a proof of \eqref{eq:boundsGrandPotential}. It remains to prove \eqref{eq:newUpperBoundAboveCriticalPoint}.

In this case we use \eqref{eq:firstOrderLowerBound2}, $\tr[\mathcal{N}^2 \Gamma] \geq (\tr[\mathcal{N} \Gamma])^2$ for a state $\Gamma$ on $\mathscr{F}$, and the Gibbs variational principle to estimate
\begin{align}
    \tr[(\mathcal{H}_{w,\eta} - \mu \mathcal{N})\Gamma] - \frac{1}{\beta} S(\Gamma) \geq& \sum_{p \in \Lambda^*} (w(p) - \mu_0(\beta,M)) \tr[ a_p^* a_p \Gamma ] - \frac{1}{\beta} S(\Gamma) \nn \\
    &+ \frac{\hat{v}}{2 \eta} (\tr[\mathcal{N} \Gamma])^2 + \left(\mu_0(\beta,M) - \mu - \frac{v(0)}{2 \eta} \right) \tr[\mathcal{N} \Gamma] \nn \\
    \geq& \widetilde{\Phi}_w^{\mathrm{id}}(\beta,\mu_0(\beta,M)) - \frac{\left(\mu - \mu_0(\beta,M) - \frac{v(0)}{2 \eta} \right)^2 \eta}{2 \hat{v}(0)}
\end{align}
To come to the last line we completed a square. From Lemma~\ref{lem:effectiveChemicalPotentialAPriori} we know that $| \mu_0(\beta,M) - \mu| \lesssim 1$, and hence the second term on the right-hand side is bounded from below by $(\mu - \mu_0(\beta,M))^2 \eta/(2 \hat{v}(0))$ minus a constant. Putting our lower bound and \eqref{eq:newUpperBoundAboveCriticalPoint} together, we obtain \eqref{eq:newBoundAboveCriticalPoint} and Proposition~\ref{prop:boundsPerturbedGrandPotential} is proved.
\subsection{Proof of Theorem~\ref{thm:firstOrderAPriori}}
\label{sec:proofFirstOrderApriori1}
We are now prepared to give the proof of Theorem~\ref{thm:firstOrderAPriori}. Let $B$ be a bounded self-adjoint operator on $L^2(\Lambda)$ and assume that $\varphi_0(x)=1$ is in the kernel of $B$. For $t \in \mathbb{R}$ we define the one-particle Hamiltonian
\begin{equation}
    h_t = h + t B.
    \label{eq:h-lambda-general}
\end{equation}
As long as $|t|$ is sufficiently small (depending on $h$ and $B$), we have
\begin{align} \label{eq:h-lambda-requirement}
h_t \gesim -\Delta  \quad \text{ and } \quad \|h_t+\Delta\|_{\infty} \lesim 1.
\end{align}
The grand potentials $\Phi_{h_t}(\beta,\mu)$ and $\Phi_{h_t}^{\mathrm{id}}(\beta,\mu_0(\beta,M))$ are concave in $t$, and differentiation with respect to $t$ yields
\begin{equation}
\frac{\partial \Phi_{h_t}(\beta,\mu) }{\partial t}  = \Tr [ A G_{h_t,\eta}(\beta,\mu) ] \quad \text{ and } \quad \frac{\partial \Phi_{h_t}^{\mathrm{id}}(\beta,\mu_0(\beta,M)) }{\partial t} = \trs_+ \left[ B \frac{1}{e^{\beta h_t} -1} \right].
\label{eq:proofFirstOrderBounds1}
\end{equation}
If the operator $A$ is positive then $\Phi_{h_t}(\beta,\mu)$ and $\Phi_{h_t}^{\mathrm{id}}(\beta,\mu_0(\beta,M))$ are additionally monotone increasing in $t$.

\bigskip \noindent
\textbf{Proof of the first bound in \eqref{eq:unperturbedFirstOrderBounds}:} We choose $B=p^2 |\varphi_p \rangle \langle \varphi_p|$ with $\varphi_p(x) = e^{\mathrm{i} p \cdot x}$, $p \in \Lambda_+^*$ and some fixed $t > 0$. Using the concavity and monotonicity of the grand potentials and Proposition~\ref{prop:boundsPerturbedGrandPotential}, we find 
\begin{align}
    p^2 \tr[a_p^* a_p G_{h,\eta}(\beta,\mu)] &\leq \frac{\Phi_{h_{t}}(\beta,\mu)- \Phi_{h}(\beta,\mu)}{t}  \nonumber \\
    &\leq \frac{\Phi_{h_{t}}^{\mathrm{id}}(\beta,\mu_0(\beta,M))- \Phi_{h}^{\mathrm{id}}(\beta,\mu_0(\beta,M))}{t} + \frac{C \eta^{2/3}}{t}.
    \label{eq:proofFirstOrderBounds2}
\end{align}
Since $\Phi_{h_{t}}^{\mathrm{id}}(\beta,\mu_0(\beta,M))$ is concave and monotone increasing in $t$, an application of \eqref{eq:proofFirstOrderBounds1} shows
\begin{align}
    &\frac{\Phi_{h_{t}}^{\mathrm{id}}(\beta,\mu_0(\beta,M))- \Phi_{h}^{\mathrm{id}}(\beta,\mu_0(\beta,M))}{t} \leq \frac{\partial \Phi_{h_{s}}^{\mathrm{id}}(\beta,\mu_0(\beta,M))}{\partial s} \Bigg|_{s = t} \nonumber \\
    &\hspace{3.5cm}= \langle \varphi_p, \frac{p^2}{e^{\beta (h_t-\mu_0(\beta,M))} -1} \varphi_p \rangle \leq p^2 \langle \varphi_p, (\beta (h_t-\mu_0))^{-1} \varphi_p \rangle.
    \label{eq:proofFirstOrderBounds3}
\end{align}
To obtain the last bound we additionally used $(\exp(x)-1)^{-1} \leq x^{-1}$ for $x>0$. The function $x \mapsto -1/x$ is operator monotone, see e.g. \cite[Proposition~V.1.6]{Bhatia1997}. We use this, $h_t \gesssim - \Delta$, and $\mu_0 < 0$ to see that the right-hand side of \eqref{eq:proofFirstOrderBounds3} is bounded from a above by a constant times $\beta^{-1}$. When we put the above considerations together, we obtain a proof of the first bound in \eqref{eq:unperturbedFirstOrderBounds}. 

\bigskip \noindent
\textbf{Proof of the second bound in \eqref{eq:unperturbedFirstOrderBounds}:} We choose $B= Q \cos(p \cdot x) Q$ with $Q$ in \eqref{eq:excitationFockSpace}, $p \in \Lambda_+^*$ and some fixed $t > 0$ in such a way that $h \pm t Q \cos(p \cdot x) Q \gesssim - \Delta$ holds. We also recall the notation $B_p = \de \Upsilon(Q \cos(p \cdot x) Q)$. Since the operator $B$ is not positive the grand potentials cannot be expected to be monotone increasing in $t$. We therefore use the bound
\begin{align}
    | \tr[B_p  G_{h,\eta}(\beta,\mu) ]| \leq& t^{-1} \left( | \Phi_{h}(\beta,\mu) - \Phi_{h_{-t}}(\beta,\mu) | + | \Phi_{h_{t}}(\beta,\mu) - \Phi_{h}(\beta,\mu) | \right) \nonumber \\
    \leq& t^{-1} \left( | \Phi_{h}^{\mathrm{id}}(\beta,\mu_0(\beta,M)) - \Phi_{h_{-t}}^{\mathrm{id}}(\beta,\mu_0(\beta,M)) | + | \Phi_{h_{t}}^{\mathrm{id}}(\beta,\mu_0(\beta,M)) - \Phi_{h}^{\mathrm{id}}(\beta,\mu_0(\beta,M)) | \right) \nonumber \\
    &+ C \eta^{2/3}/t, 
    \label{eq:proofFirstOrderBounds4}
\end{align}
where the second inequality follows from Proposition~\ref{prop:boundsPerturbedGrandPotential}. We also have
\begin{align}
    &| \Phi_{h_{t}}^{\mathrm{id}}(\beta,\mu_0(\beta,M)) - \Phi_{h}^{\mathrm{id}}(\beta,\mu_0(\beta,M)) | \nonumber \\
    &\hspace{3cm} \leq \int_0^t \left| \trs_+ \left[ Q \cos(p \cdot x) Q \frac{1}{\exp(\beta(h + s Q \cos(p \cdot x) Q) - \mu_0(\beta,M)) - 1} \right] \right| \de s
    \label{eq:proofFirstOrderBounds5}
\end{align}
and 
\begin{equation}
    \trs_+ \left[ Q \cos(p \cdot x) Q \frac{1}{\exp(\beta( -\Delta - \mu_0(\beta,M)) - 1} \right] = 0. 
    \label{eq:proofFirstOrderBounds6}
\end{equation}
In combination, \eqref{eq:proofFirstOrderBounds5}, \eqref{eq:proofFirstOrderBounds6}, $\Vert Q \cos(p \cdot x) Q \Vert_{\infty} < +\infty$, and Lemma~\ref{lem:perturbedtrace} imply
\begin{align}
    | \Phi_{h_{t}}^{\mathrm{id}}(\beta,\mu_0(\beta,M)) - \Phi_{h}^{\mathrm{id}}(\beta,\mu_0(\beta,M)) | \lesssim t \beta^{-1}.
    \label{eq:proofFirstOrderBounds7}
\end{align}
The first term on the right-hand side of \eqref{eq:proofFirstOrderBounds4} can be bounded with the same argument by the right-hand side of \eqref{eq:proofFirstOrderBounds7}, too. Putting the above considerations together, we find
\begin{equation}
    | \tr[B_p  G_{h,\eta}(\beta,\mu) ]| \lesssim \eta^{2/3},
    \label{eq:proofFirstOrderBounds8}
\end{equation}
which proves the second bound in \eqref{eq:unperturbedFirstOrderBounds}.

\bigskip \noindent
\textbf{Proof of \eqref{eq:perturbedFirstOrderBounds-b}:} We choose $B= Q$ with $Q$ in \eqref{eq:excitationFockSpace} and some fixed $t > 0$ in such a way that $h \pm t Q \gesssim - \Delta$ holds. We also recall the identity $\de \Upsilon(Q) = \mathcal{N}_+$. The argument in \eqref{eq:proofFirstOrderBounds2} and \eqref{eq:proofFirstOrderBounds3} and an application of Lemma~\ref{lem:perturbedtrace} show
\begin{align}
    \tr[ \mathcal{N}_+ G_{h,\eta}(\beta,\mu) ] &\leq \trs_+\left[ \frac{1}{\exp(\beta Q (h_t-\mu_0(\beta,M)) - 1} \right] + \frac{C \eta^{2/3}}{t} \nonumber \\
    &\leq \trs_+\left[ \frac{1}{\exp(\beta Q( -\Delta - \mu_0(\beta,M)) - 1} \right] + C \eta^{2/3} (1+t^{-1}).
    \label{eq:proofFirstOrderBounds9}
\end{align}
When we replace $t$ by $-t$ a similar argument also shows
\begin{equation}
    \tr[ \mathcal{N}_+ G_{h,\eta}(\beta,\mu) ] \geq \trs_+\left[ \frac{1}{\exp(\beta Q( -\Delta - \mu_0(\beta,M)) - 1} \right] - C \eta^{2/3} (1+t^{-1}).
    \label{eq:proofFirstOrderBounds10}
\end{equation}
With a similar argument we obtain the same bounds with $G_{h,\eta}(\beta,\mu)$ replaced by $G_{\eta}(\beta,\mu)$. This proves the first bound in \eqref{eq:perturbedFirstOrderBounds-b}. The second bound is a consequence of the bounds for $G_{\eta}(\beta,\mu)$. 

\bigskip \noindent
\textbf{Proof of \eqref{eq:particleNumberBoundsPerturbedState}:} To prove this bound we perturb the interaction potential. More precisely, we replace $\hat{v}(0)$ by $\hat{v}(0) + t$ with some fixed $t > 0$. In the following, we denote the grand potential and the Gibbs state related to the Hamiltonian with $v$ shifted in this way by $\Phi_{h,t}(\beta,\mu)$ and $G_{h,\eta,t}(\beta,\mu)$, respectively. The function $\Phi_{h,t}(\beta,\mu)$ is monotone increasing and concave in $t$ and we have
\begin{equation}
\frac{\partial \Phi_{h,t}(\beta,\mu)}{\partial t} = \eta^{-1} \Tr [ \mathcal{N}^2 G_{h,\eta,t}(\beta,\mu) ].
\label{eq:proofFirstOrderBounds11}
\end{equation}
In combination with an application of Proposition~\ref{prop:boundsPerturbedGrandPotential}, this implies
\begin{align}
\eta^{-1} \Tr [ \mathcal{N}^2 G_{h,\eta,t}(\beta,\mu) ] \leq& \frac{\Phi_{h,t}(\beta,\mu)- \Phi_{h}(\beta,\mu)}{t}  \nonumber \\
\leq& \frac{1}{t} \bigg[ \Phi^{\mathrm{id}}_h(\beta,\mu_0(\beta,M_t)) - \Phi^{\mathrm{id}}_h(\beta,\mu_0(\beta,M)) - \frac{(\mu - \mu_0(\beta,M_t))^2 \eta}{2( \hat{v}(0) + t) } +  \frac{(\mu - \mu_0(\beta,M))^2 \eta}{2\hat{v}(0) }  \nonumber \\
&\hspace{0.35cm}+ \Theta^{\mathrm{BEC}}(\beta,N_0(\beta,M_t)) - \Theta^{\mathrm{BEC}}(\beta,N_0(\beta,M)) + C \eta^{2/3} \bigg],
\label{eq:proofFirstOrderBounds12}
\end{align}
where $M_t$ is defined as in \eqref{eq:particleNumber} with $\hat{v}(0)$ replaced by $\hat{v}(0) + t$ in \eqref{eq:GrantCanonicalEffectiveIddealGasChemPot}. From Lemma~\ref{lem:effectiveChemicalPotentialPerturbation2} in Appendix~\ref{app:effectiveChemicalPotential} we know that $0 < \mu_0(\beta,M) - \mu_0(\beta,M_t) \lesssim t$. With this bound and part~(c) of Lemma~\ref{lem:effectiveChemicalPotentialAPriori} we check that the difference of the two grand potentials on the right-hand side of \eqref{eq:proofFirstOrderBounds12} is bounded from above by
\begin{equation}
\sum_{p \in \Lambda^*} \frac{\beta(\mu_0(\beta,M) - \mu_0(\beta,M_t))}{\exp(\beta(p^2 - \mu_0(\beta,M))) - 1} \lesssim t \eta^{1/3}.
\label{eq:proofFirstOrderBounds13}
\end{equation}
Moreover, an application of part~(b) of Lemma~\ref{lem:effectiveChemicalPotentialAPriori} allows us to see that the sum of the third and the fourth term on the right-hand side of \eqref{eq:proofFirstOrderBounds12} is bounded from above by a constant times $\eta/t$. Finally, the terms in the third line of \eqref{eq:proofFirstOrderBounds12} are bounded by a constant times $\eta^{2/3} \ln(\eta)/t$. When we put these considerations together and use that $t>0$ has been chosen independently of $\eta$, we find
\begin{equation}
\eta^{-1} \Tr [ \mathcal{N}^2 G_{h,\eta,t}(\beta,\mu) ] \lesssim \eta.
\label{eq:proofFirstOrderBounds14}
\end{equation}
This proves \eqref{eq:particleNumberBoundsPerturbedState}.

\bigskip \noindent
\textbf{Proof of \eqref{eq:firstorderapapbound}:} We choose $B=(p^2-\mu_0(\beta,M)) |\varphi_p \rangle \langle \varphi_p|$ with $p \in \Lambda_+^*$ and some fixed $t > 0$. The proof is almost literally the same as that of the first bound in \eqref{eq:unperturbedFirstOrderBounds}, and therefore left to the reader. The proof of Theorem~\ref{thm:firstOrderAPriori} is complete.

\section{A bound for the chemical potential}
\label{sec:boundsChemicalPotential}
In this section we provide a bound for the chemical potential $\mu_{\beta,N}$. The results in the previous section are proved for Gibbs states with a general chemical potential that satisfies certain bounds. From Lemma~\ref{lem:roughBoundChemicalPotential} below we know that these properties are shared by the Gibbs state $G_{\beta,N}$. More details on this can be found in Remark~\ref{rem:interactingChemicalPotential}.

\begin{lemma}
	\label{lem:roughBoundChemicalPotential}
	Let $v$ satisfy the assumptions of Theorem~\ref{thm:norm-approximation}. We consider the limit $N \to \infty$, $\beta N^{2/3} \to \kappa \in (0,\infty)$. The chemical potential satisfies the bound
	\begin{align}
		| \mu_{\beta,N} - \hat{v}(0) - \mu_0(\beta,N) | \lesssim& N^{1/3} \sqrt{ \frac{1}{N} + \frac{1}{\beta N_0^2(\beta,N) + \beta^{-1} ( \max\{ 1, 1/(\beta N_0(\beta,N)) \} )^{-1/2}}  } \nn \\
        &+ \frac{1}{\beta} \min \left\{ \frac{1}{N_0(\beta,N)}, \frac{N_0(\beta,N)}{N^{4/3}}  \right\}.
		\label{eq:boundChemicalPotential}
	\end{align}
\end{lemma}
\begin{remark}
    \label{rem:roughBoundChemicalPotential}
    The right-hand side of \eqref{eq:boundChemicalPotential} scales as $N^{-1/6}$ if $N_0 \sim N$ (condensed phase) and as $N^{1/6}$ if $N_0 \sim 1$ (non-condensed phase). In the parameter regime around the critical point, where $1 \ll N_0 \ll N$, it interpolates between these two behaviors. We learn from \eqref{eq:boundChemicalPotential} and the behavior of $\mu_0$ described below \eqref{eq:crittemp} that $\mu_{\beta,N}$ satisfies the bound
    \begin{equation}
        -N^{2/3} \lesssim \mu_{\beta,N} \lesssim 1.
    \end{equation}
\end{remark}
The remainder of this section is devoted to the proof of Lemma~\ref{lem:roughBoundChemicalPotential}. We fix $\eta > 0$, recall the definition of $\mathcal{H}_{\eta}$ in \eqref{eq:FockSpaceHamiltonian1}, choose $\beta = \kappa \beta_{\mathrm{c}}(\eta)$ with $\kappa \in (0,\infty)$, and introduce the notation
\begin{equation}
    F_{\eta}(\beta,N) = -\frac{1}{\beta} \ln\left( \tr[\exp(-\beta( \mathcal{H}_{\eta} - \mu_{\beta,N,\eta}\mathcal{N}))] \right) + \mu_{\beta,N,\eta} N.
    \label{eq:etaFreeEneregy}
\end{equation}
Here the chemical potential $\mu_{\beta,N,\eta}$ is chosen such that the expected number of particles in the Gibbs state $G_{\beta,N,\eta}$ related to $F_{\eta}(\beta,N)$ equals $N$. By definition, we have $F(\beta,N) = F_{N}(\beta,N)$ and $\mu_{\beta,N} = \mu_{\beta,N,N}$. Later we will choose $\eta = N$. During our analysis we therefore assume that there exists a constant $0 < c < 1$ such that $ c N \leq \eta \leq c^{-1} N $ holds.

The chemical potential equals a first derivative of the free energy:
\begin{equation}
	\frac{\partial F_{\eta}(\beta,N)}{\partial N} = \mu_{\beta,N,\eta}.
	\label{eq:boundChemicalPotential1}
\end{equation}
Moreover, differentiation of both sides of the equation $\tr[\mathcal{N} G_{\beta,N,\eta}] = N$ with respect to $N$ yields
\begin{equation}
	\frac{\partial \mu_{\beta,N,\eta}}{\partial N} = \frac{1}{\beta ( \tr[\mathcal{N}^2 G_{\beta,N,\eta}] - (\tr[\mathcal{N} G_{\beta,N,\eta}])^2 )} > 0,
	\label{eq:boundChemicalPotential2}
\end{equation}
and we conclude that the map $N \mapsto F_{\eta}(\beta,N)$ is convex. Using this convexity and \eqref{eq:boundChemicalPotential1}, we see that
\begin{equation}
	\frac{F_{\eta}(\beta,N) - F_{\eta}(\beta,N-M)}{M} \leq \mu_{\beta,N,\eta} \leq \frac{F_{\eta}(\beta,N+M) - F_{\eta}(N)}{M}
	\label{eq:boundChemicalPotential3}
\end{equation}
holds for any $N,M > 0$ with $N-M \geq 0$. To make use of the above inequality, we now derive upper and lower bounds for $F_{\eta}(\beta,N)$. 

\subsubsection*{Free energy bounds}

With the techniques in Sections~\eqref{sec:upperBoundGrandPotential} and \eqref{sec:lowerBoundGrandPotential} it is not difficult to check that the free energy satisfies 
\begin{equation}
	F_{\eta}(\beta,N) = F^+_0(\beta,N) + \frac{\hat{v}(0) N^2}{2 \eta} + \Theta^{\mathrm{BEC}}(\beta,N_0(\beta,N)) + O(N^{2/3})
	\label{eq:boundChemicalPotential5}
\end{equation}
with $F^+_0$ in \eqref{eq:freeEnergyIdealGasCloud} and $\Theta^{\mathrm{BEC}}$ in \eqref{eq:ThetaBEC}. It should be highlighted that we cannot simply use \eqref{eq:boundsGrandPotential} and the relation $F(\beta,N) = \Phi(\beta,\mu_{\beta,N}) + \mu_{\beta,N} N$ between the free energy and the grand potential to obtain the asymptotic expansion in \eqref{eq:boundChemicalPotential5}. This is because the asymptotic expansion of $\Phi(\beta,\mu_{\beta,N}) + \mu_{\beta,N} N$ depends on $\mu_{\beta,N}$ and we have no a-priori bounds for it. 

Let us briefly comment on some changes that the proof of \eqref{eq:boundChemicalPotential5} requires with respect to the analysis in Sections~\eqref{sec:upperBoundGrandPotential} and \eqref{sec:lowerBoundGrandPotential}. As trial state for the upper bound we choose the state in \eqref{eq:trialStateFirstOrderEstimates} with $h = -\Delta$ and $M=N$. With these choices its expected number of particles equals $N$. The computation of the free energy of this trial state is a simplified version of the analysis in Section~\eqref{sec:upperBoundGrandPotential} because $-\Delta$ is diagonal in momentum space. In the lower bound we have to use the constraint $\Tr[\mathcal{N} G_{\beta,N}] = N$ because
$\mu = 0$. More precisely, when we minimize the effective free energy function in \eqref{eq:firstOrderLowerBound11} we have to use that $\sum_n n c_n = N$ holds.  The rest of the lower bound goes through without additional complications.

\subsubsection*{Bounds on the chemical potential} 

In the following we assume $M \ll N$. This, in particular, implies $N_0(\beta,N+M) \simeq N_0(\beta,N)$ and $\mu_0(\beta,N+M) \simeq \mu_0(\beta,N)$. When we combine \eqref{eq:boundChemicalPotential3} and \eqref{eq:boundChemicalPotential5}, we find 
\begin{align}
    &\frac{F^+_0(\beta,N) - F^+_0(\beta,N-M) + \frac{\hat{v}(0)(2N M - M^2)}{2\eta} + \Theta^{\mathrm{BEC}}(\beta,N_0(\beta,N)) - \Theta^{\mathrm{BEC}}(\beta,N_0(\beta,N-M)) - C N^{2/3} }{M} \leq \mu_{\beta,N,\eta} \nonumber \\
    &\leq \frac{F^+_0(\beta,N+M) - F^+_0(\beta,N) + \frac{\hat{v}(0)(2N M + M^2)}{2\eta} + \Theta^{\mathrm{BEC}}(\beta,N_0(\beta,N+M)) - \Theta^{\mathrm{BEC}}(\beta,N_0(\beta,N)) + C N^{2/3} }{M}.
    \label{eq:boundChemicalPotential10}
\end{align}
To obtain a bound for $F^+_0(\beta,N) - F^+_0(\beta,N-M)$ we write $F_0^+ = F_0 - F_0^{\mathrm{BEC}}$ with $F_0^{\mathrm{BEC}}$ in \eqref{eq:freeEnergyIdealGasBEC} and $F_0$ below \eqref{eq:CorollaryNoBEC}, and first consider the term $F_0(\beta,N) - F_0(\beta,N-M)$. Eqs.~\eqref{eq:etaFreeEneregy}, \eqref{eq:boundChemicalPotential1} and a first order Taylor approximation allow us to write
\begin{equation}
    \pm \frac{F_0(\beta,N \pm M) \mp F_0(\beta,N)}{M} = \mu_0(\beta,N) \pm \frac{M}{2\beta \textbf{Var}_{G^{\mathrm{id}}_{\beta,\xi}}(\mathcal{N}) }, 
    \label{eq:boundChemicalPotential11}
\end{equation}
where $\textbf{Var}_{G^{\mathrm{id}}_{\beta,\xi}}(\mathcal{N})$ denotes the variance of the operator $\mathcal{N}$ in the state Gibbs state $G^{\mathrm{id}}_{\beta,\xi}$ of the ideal gas in \eqref{eq:GibbsStateIdealGas} and $\xi \in [N,N \pm M]$. Let us introduce the grand potential 
\begin{equation}
    \Phi_0(\beta,\mu_0) = -\frac{1}{\beta} \ln\left( \tr \exp(-\beta \de \Upsilon(-\Delta - \mu_0)) \right) = \frac{1}{\beta} \sum_{p \in \Lambda^*} \ln\left( 1 - \exp(-\beta(p^2 - \mu_0)) \right)
    \label{eq:grandPotentialIdealGas}
\end{equation}
of the ideal gas. The above variance can be written in terms of the grand potential as
\begin{equation}
    \textbf{Var}_{G^{\mathrm{id}}_{\beta,\xi}}(\mathcal{N}) = -\frac{1}{\beta} \frac{\partial^2 \Phi_0}{\partial \mu^2_0}(\beta,\mu_0(\beta,\xi)) = \sum_{p \in \Lambda^*} \frac{1}{ 4 \sinh^2 \left( \frac{\beta(p^2 - \mu_0(\beta,\xi))}{2} \right) }.  
    \label{eq:boundChemicalPotential12}
\end{equation}
We need a lower bound for the term on the right-hand side of the above equation that is uniform in $\xi \in [N-M,N+M]$. Let us choose $c > 0$ such that $c/\beta \geq -2 \mu_0$ holds. This is always possible because $-\mu_0=\ln(1+N_0^{-1})$ and $N_0 \gesssim 1$, which is a consequence of our assumption on $\beta$. With this choice of $c$ we find
\begin{equation}
    \sum_{p \in \Lambda^*} \frac{1}{ 4 \sinh^2 \left( \frac{\beta(p^2 - \mu_0)}{2} \right) } \gesssim \left( \frac{1}{\beta \mu_0} \right)^2 + \sum_{p \in \Lambda_+^*, - \mu_0 < p^2 \leq c/\beta} \left( \frac{1}{\beta p^2} \right)^2  \gesssim \left( \frac{1}{\beta \mu_0} \right)^2 + \beta^{-2} (\max\{1,-\mu_0\} )^{-1/2}. \label{eq:boundChemicalPotential14} 
\end{equation}
We use $-\mu_0 \sim (\beta N_0)^{-1}$ and $N_0(\beta,N+M) \simeq N_0(\beta,N)$ to see that the right-hand side is bounded from below by a constant times $N_0(\beta,N)^2 + \beta^{-2} ( \max\{ 1, 1/(\beta N_0(\beta,N)) \} )^{-1/2}$. That is, we have
\begin{equation}
    \left| \pm \frac{F_0(\beta,N \pm M) \mp F_0(\beta,N)}{M} - \mu_0(\beta,N) \right| \lesssim \frac{M \left( N_0(\beta,N)^2 + \beta^{-2} ( \max\{ 1, 1/(\beta N_0(\beta,N)) \} )^{-1/2} \right)}{\beta}.
    \label{eq:boundChemicalPotential14A1}
\end{equation}

Next, we have a closer look at $F_0^{\mathrm{BEC}}(\beta,N)$. A short computation that uses the identity $\mu_0 = -\beta^{-1} \ln(1+N_0^{-1})$ shows
\begin{equation}
    F_0^{\mathrm{BEC}}(\beta,N) = \frac{1}{\beta} \ln\left( 1-\exp(\beta \mu_0(\beta,N)) \right) + \mu_0(\beta,N) N_0(\beta,N) = \frac{-\ln(N_0(\beta,N))}{\beta} + O(N^{2/3}).
    \label{eq:boundChemicalPotential14A2}
\end{equation}
With a first order Taylor expansion we also get
\begin{equation}
    N_0(\beta,N \pm M) = N_0(\beta,N) \pm \frac{\beta M \frac{\partial \mu_0}{\partial N} (\beta,\xi)}{4 \sinh^2\left( \frac{-\beta \mu_0(\beta,\xi)}{2} \right)}
    \label{eq:boundChemicalPotential14A3}
\end{equation}
with some $\xi \in \{ N \pm tM \ | \ t \in [0,1] \}$. In combination with Lemma~\ref{lem:BoundDerivativeCHemPotWrtN} this shows
\begin{equation}
    | N_0(\beta,N \pm M) - N_0(\beta,N) | \lesssim M \min\{ 1, N^2_0(\beta,N)/N^{4/3} \},
    \label{eq:boundChemicalPotential14A4}
\end{equation}
and hence
\begin{equation}
    | F_0^{\mathrm{BEC}}(\beta,N \pm M) - F_0^{\mathrm{BEC}}(\beta,N) | \lesssim \frac{M \min\{ 1, N^2_0(\beta,N)/N^{4/3} \}}{\beta N_0(\beta,N)} + N^{2/3} 
    \label{eq:boundChemicalPotential14A4a}
\end{equation}
Similarly, we check that $| \Theta^{\mathrm{BEC}}(\beta,N \pm M) - \Theta^{\mathrm{BEC}}(\beta,N) | $ is bounded by the right-hand side of \eqref{eq:boundChemicalPotential14A4a}, too.

In combination, \eqref{eq:boundChemicalPotential10}--\eqref{eq:boundChemicalPotential14A4a} imply the bound
\begin{align}
    | \mu_{\beta,N,\eta} - \hat{v}(0) (N/\eta) - \mu_0(\beta,N) | \lesssim& \frac{M}{N} + \frac{N^{2/3} }{M} + \frac{M}{\beta N_0^2(\beta,N) + \beta^{-1} ( \max\{ 1, 1/(\beta N_0(\beta,N)) \} )^{-1/2} } \nn \\
    &+ \frac{1}{\beta} \min \left\{ \frac{1}{N_0(\beta,N)}, \frac{N_0(\beta,N)}{N^{4/3}}  \right\}.
    \label{eq:boundChemicalPotential15} 
\end{align}
The optimal choice for $M$ reads
\begin{equation}
    M = \frac{N^{1/3} }{ A^{1/2} }  \quad \text{ with } \quad A = \frac{1}{N} + \frac{1}{\beta N_0^2(\beta,N) + \beta^{-1} ( \max\{ 1, 1/(\beta N_0(\beta,N)) \} )^{-1/2} }
\end{equation}
and satisfies the bound $M \lesssim N^{5/6} \ll N$. The error term on the right-hand side of \eqref{eq:boundChemicalPotential15} is bounded by a constant times
\begin{equation}
    N^{1/3} \sqrt{ \frac{1}{N} + \frac{1}{\beta N_0^2(\beta,N) + \beta^{-1} ( \max\{ 1, 1/(\beta N_0(\beta,N)) \} )^{-1/2}}  }. 
\end{equation}
With the choice $\eta = N$ this proves the claim of Lemma~\ref{lem:roughBoundChemicalPotential}.

\section{The Gibbs state part II: second order correlation inequalities} \label{sec:correlationInequalities}
In this section we prove several correlation inequalities for the Gibbs state $G_{\beta,N}$ in \eqref{eq:interactingGibbsstate}. The proof of these inequalities is based on the first order estimates in Theorem \ref{thm:firstOrderAPriori} and a new abstract correlation inequality. A crucial ingredient for the proof of the abstract correlation inequality is an infinite-dimensional version of Stahl's theorem, see \cite{Stahl2013}, which may be of separate interest. 
\subsection{An infinite-dimensional version of Stahl's theorem}
\label{sec:StahlsTheorem}
The goal of this section is to prove the following theorem.

\begin{theorem}[Stahl's theorem in infinite dimensions]
\label{prop:stahlInfinite}
    Let $A$ and $B$ be two self-adjoint operators on a separable complex Hilbert space. We assume that $A + t B$ is self-adjoint on the domain of $A$ and that $\exp(-(A +t B))$ is trace-class for $t \in [-1,1]$. We also define the function
    \begin{equation}
        Z(t) = \Tr[ \exp(-(A+tB)) ].
        \label{eq:Stahl2}
    \end{equation}
    Then there exists a nonnegative Borel measure $\mu$ on $\mathbb{R}$ with
    \begin{equation}
        \int_{-\infty}^{\infty} \cosh(s) \de \mu(s) < + \infty
        \label{eq:Stahl3}
    \end{equation}
    such that
    \begin{equation}
        Z(t) = \int_{-\infty}^\infty e^{-ts} \de \mu(s)
        \label{eq:Stahl4}
    \end{equation}
    holds for all $t \in [-1,1]$. 
\end{theorem}
    
\begin{remark}
    Although it is not directly relevant for our paper, we also mention the following alternative and equally interesting version of the above theorem: assume $A+t B$, $t \in [0,1]$ is self-adjoint on the domain of $A$, $B \geq 0$, and $\exp(-A)$ is trace class. Then there exists a positive finite Borel measure $\mu$ on $[0,\infty)$ such that
    \begin{equation}
           Z(t) = \int_0^{\infty} e^{-ts} \de \mu(s)
           \label{eq:stahl17}
    \end{equation}
    holds for all $t \geq 0$. One can prove this statement along the same lines as  Theorem~\ref{prop:stahlInfinite}. There is, however, a shorter proof that uses well-known results for completely monotone functions and that we briefly explain now. A function $f$ on $(0,\infty)$ is called completely monotone if it is of class $\mathcal{C}^{\infty}$ and $(-1)^n f^{(n)}(x) \geq 0 $ holds for all $n \in \mathbb{N}_0$ and all $x > 0$. From Theorem~\ref{thm:Stahl} below we know that the functions $Z_n(t)$ defined in \eqref{eq:stahl6} are completely monotone. The set of completely monotone functions is closed under pointwise convergence, see \cite[Corollary~1.6]{SchiSongVond2012}, and hence $\lim_{n \to \infty} Z_n(t) = Z(t)$ for all $t > 0$ (proof see below) implies that $Z$ is completely monotone. Eq.~\eqref{eq:stahl17} now follows from Bernstein's Theorem, see \cite[Theorem~1.4]{SchiSongVond2012}, which states that every completely monotone function is the Laplace transform of a positive finite Borel measure on $[0,\infty)$. 
\end{remark}

The proof of Theorem~\ref{prop:stahlInfinite} is based on the following version of Stahl's theorem, see \cite{Stahl2013,Eremenko2015}, which was formerly known as the Bessis--Moussa--Villani (BMV) conjecture. 

\begin{theorem}[Stahl's Theorem]\label{thm:Stahl}
    Let $A$ and $B$ be two hermitian $n \times n$ matrices and denote by $b_1$ and $b_n$ the smallest and the largest eigenvalues of $B$, respectively. Then there exists a nonnegative Borel measure $\mu$ on $[b_1,b_n]$ such that
    \begin{equation}
        \tr[ \exp(-(A+tB)) ] = \int_{b_1}^{b_n} e^{-ts} \de \mu(s).
        \label{eq:stahl5}
    \end{equation}
\end{theorem}

With the above theorem at hand we are prepared to give the proof of Theorem~\ref{prop:stahlInfinite}.
\begin{proof}[Proof of Theorem~\ref{prop:stahlInfinite}]
    Our assumptions guarantee that $A$ is bounded from below and has only discrete spectrum, and hence the projections $P_n = \mathds{1}(A \leq n)$, $n \in \mathbb{N}$ have finite rank. This allows us to define the functions
    \begin{equation}
        Z_n(t) = \Tr[P_n\exp(-P_n(A+tB)P_n)].
        \label{eq:stahl6}
    \end{equation}
    Theorem~\ref{thm:Stahl} implies that there exists a sequence $\{ \mu_n \}_{n=1}^{\infty}$ of nonnegative Borel measures with compact support on $\mathbb{R}$ such that
    \begin{equation}
        Z_n(t) = \int_{-\infty}^{\infty} e^{-ts} \de \mu_n(s)
        \label{eq:stahl7}
    \end{equation}
    holds.

    Let us prove the pointwise convergence of $Z_n$ to $Z$. We write
    \begin{equation}
        Z_n(t) = \sum_{\alpha = 1}^{\infty} \mathds{1}(\alpha \leq n) \langle \psi_{\alpha}, \exp(-P_n(A+tB)P_n) \psi_{\alpha} \rangle,
        \label{eq:stahl8a}
    \end{equation}
    where we choose the vectors $\psi_{\alpha}$ such for every $n \in \mathbb{N}$, the span of $\{ \psi_{\alpha} \}_{\alpha=1}^n$ equals the range of $P_n$. From the min-max principle, see e.g. \cite[Theorem~12.1]{LiebLoss2010}, we know that the $k$-th eigenvalue of $P_n (A+tB) P_n$ is bounded from below by the $k$-th eigenvalue of $A+tB$ provided $k \leq n$. Using this, we check that the summand in \eqref{eq:stahl8a} is, for fixed $\alpha$, monotone increasing in $n$. We also claim that $\exp(-P_n(A+tB)P_n)$ converges to $\exp(-(A+tB))$ in the strong operator topology. To prove this claim we will show that $P_n(A+tB)P_n$ converges to $A+tB$ in strong resolvent sense. 
    
    Since strong resolvent convergence is implied by weak resolvent convergence, see \cite[p.~284]{RS1980}, it suffices to show that for fixed $z \in \mathbb{C}$ with $\Im \ z \neq 0$ and $\psi$ in the Hilbert space we have
    \begin{equation}
        \lim_{n \to \infty} \left| \langle \psi,  \left( \frac{1}{P_n(A+tB)P_n - z} - \frac{1}{A+tB - z} \right) \psi \rangle \right| = 0. 
        \label{eq:stahl8}
    \end{equation}
    Using the resolvent identity, we write the operator in the above equation as 
    \begin{align}
        \frac{1}{P_n(A+tB)P_n - z} - \frac{1}{A+tB - z}  &= \frac{1}{P_n(A+tB)P_n - z} \left( A+tB - P_n (A+tB) P_n \right) \frac{1}{A+tB - z} \nonumber \\
        &= \frac{1}{P_n(A+tB)P_n - z} \left( (A+tB)Q_n + Q_n (A+tB) P_n \right) \frac{1}{A+tB - z},
        \label{eq:stahl9}
    \end{align} 
    where $Q_n = 1-P_n$. We also have the bounds
    \begin{align}
        &\left| \langle \psi, \frac{1}{P_n(A+tB)P_n - z} (A+tB)Q_n \frac{1}{A+tB - z} \psi \rangle \right| \leq \frac{\Vert \psi \Vert}{|\Im z|} \left( \left\Vert A Q_n \frac{1}{A+tB - z} \psi \right\Vert + |t| \ \left\Vert B Q_n \frac{1}{A+tB - z} \psi \right\Vert \right) \nonumber \\
        &\hspace{4cm}\leq \frac{\Vert \psi \Vert}{|\Im z|} \left( \left( 1 + |t| \ \left\Vert B \frac{1}{A-z} \right\Vert \right) \left\Vert Q_n A \frac{1}{A+tB - z} \psi \right\Vert + |tz| \left\Vert Q_n \frac{1}{A+tB - z} \psi \right\Vert \right) \label{eq:stahl10_0}
    \end{align}
    and
    \begin{align}
        &\left| \langle \psi, \frac{1}{P_n(A+tB)P_n - z} Q_n(A+tB)P_n \frac{1}{A+tB - z} \psi \rangle \right| \leq \left\Vert Q_n \frac{1}{P_n(A+tB)P_n-z} \psi \right\Vert \ \left\Vert (A+tB) P_n \frac{1}{A+tB-z} \psi \right\Vert \nonumber \\
        &\hspace{4cm}\leq \frac{\Vert Q_n \psi \Vert }{|\Im \ z|} \left( \left(1+|t| \ \left\Vert B \frac{1}{A-z} \right\Vert \right) \left\Vert A \frac{1}{A+tB-z} \psi \right\Vert + |tz| \ \left\Vert \frac{1}{A+tB-z} \psi \right\Vert \right). \label{eq:stahl10}
    \end{align}
    To obtain the result we used $[A,P_n]=0$. Note that the right-hand sides in both equations converge to zero as $n \to \infty$ because the operators $A+tB$, $t \in [-1,1]$ share a common domain and $Q_n \to 0$ in the strong operator topology. We conclude that $P_n(A+tB)P_n$ converges to $A+tB$ in strong resolvent sense. This, in particular, implies the strong convergence of $\exp(-P_n(A+tB)P_n)$ to $\exp(-(A+tB))$. In combination with the considerations below \eqref{eq:stahl8} and an application of the monotone convergence theorem, we conclude that $Z_n$ converges pointwise to $Z(t)$ for all $t \in [-1,1]$.

    Next, we investigate the sequence of measures $\{ \mu_n \}_{n=1}^{\infty}$. From 
    \eqref{eq:stahl7} we know that
    \begin{equation}
        \int_{-\infty}^{\infty} \cosh(s) \de \mu_n(s) \leq \max\{ Z(-1),Z(1) \}.
        \label{eq:stahl13}
    \end{equation}
    The bound \eqref{eq:stahl13} implies that the sequence of measures  $\{\mu_n\}$ is tight, and hence by Prokhorov's theorem we conclude that there exist a subsequence $\{ \mu_{n_k} \}_{k=1}^{\infty}$ and a nonnegative Borel measure $\mu$ on $\mathbb{R}$ such that
    \begin{align}
        \lim_{k \to \infty} \int_{-\infty}^{\infty} f(s) \de \mu_{n_k}(s) &= \int_{-\infty}^{\infty} f(s) \de \mu(s) \quad \text{ and } \nonumber \\
        \lim_{k \to \infty} \int_{-\infty}^{\infty} f(s) \cosh(s) \de \mu_{n_k}(s) &= \int_{-\infty}^{\infty} f(s) \cosh(s) \de \mu(s)
        \label{eq:stahl14}
    \end{align}
    hold for all $f \in \mathcal{C}_0(\mathbb{R})$. Here $\mathcal{C}_0(\mathbb{R})$ denotes the set of all real-valued continuous functions on $\mathbb{R}$ that vanish at infinity. Moreover, by Fatou's lemma for sequences of measures, see e.g. \cite[Theorem~30.2]{Bauer2001},
    \begin{equation}
        \max\{ Z(-1),Z(1) \} \geq \limsup_{k \to \infty} \int_{-\infty}^{\infty} \cosh(s) \de \mu_{n_k} \geq \int_{-\infty}^{\infty} \cosh(s) \de \mu(s).
        \label{eq:stahl15}
    \end{equation}

    It remains to prove 
    \begin{equation}
        \lim_{k \to \infty} \int_{-\infty}^{\infty} e^{-ts} \de \mu_{n_k}(s) = \int_{-\infty}^{\infty} e^{-ts} \de \mu(s)
        \label{eq:stahl16}
    \end{equation}
    for $t \in [-1,1]$. If $t \in (-1,1)$ this is implied by \eqref{eq:stahl14} since $e^{-ts}$ is dominated by $\cosh(s)$ for $s\in (-\infty,\infty)$. An application of the dominated convergence theorem and \eqref{eq:stahl15} show that the right-hand side of \eqref{eq:stahl16} is continuous for all $t \in [-1,1]$. To prove the claim, it therefore suffices to show that $t \mapsto Z(t)$ is continuous\footnote{It is not difficult to see that $Z(t)$ is convex on $[-1,1]$. This, however, only implies continuity on $(-1,1)$ and not on $[-1,1]$.} at $t = \pm 1$. However, this is a consequence of the continuity of the eigenvalues of $A+tB$ and another application of dominated convergence. The continuity of the eigenvalues of $A+tB$ follows from our assumptions and the min-max principle.
\end{proof}

\subsection{Proof of Theorem~\ref{thm:correlation-intro}}

In this subsection we discuss our first new abstract correlation inequality for Gibbs states. We call it a second order correlation inequality because it provides a bound for the expectation of the square of an operator. It is inspired by the recent inequality of Lewin, Nam, and Rougerie in \cite[Theorem~7.1]{LewNamRou-21}, where the second moment of an observable in a Gibbs state is bounded by the expectation of the observable in a family of perturbed Gibbs states, plus an additional term that involves a fourth commutator of the observable and the Hamiltonian. Our new observation is that Stahl's theorem, see Section~\ref{sec:StahlsTheorem} above, can be used to obtain a bound on the Duhamel two-point function, which allows us to control the quantum variance in a much more efficient way. 

Theorem~\ref{thm:correlation-intro} in Section~\ref{sec:DiscussionOfProof} is a special case of the following result.

\begin{theorem}[Abstract second order correlation inequality] \label{thm:correlation}
Let $A$ and $B$ be two self-adjoint operators on a separable complex Hilbert space. We assume that $A + tB$ is self-adjoint on the domain of $A$ and that $\exp(-(A + t B))$ is trace-class $t \in [-1,1]$. We also assume that the Gibbs state
\begin{equation}\label{eq:Gibbs-t}
\Gamma_t = \frac{\exp(-(A+tB))}{Z(t)} \quad \text{ with } \quad Z(t) = \Tr [ \exp(-(A+tB)) ]
\end{equation}
is such that $B\Gamma_t$ is trace class for all $t \in [-1,1]$ and that it satisfies
\begin{align}\label{eq:CRI-condition}
\sup_{t\in [-1,1]}|\Tr (B \Gamma_t )| \le a. 
\end{align}
Then we have
\begin{equation}
\Tr [ \Gamma_0^{1/2} B^2 \Gamma_0^{1/2} ] \leq a e^{a} + \frac{1}{4} \sum_{\substack{\alpha,\beta \in \mathbb{N} \\ \lambda_{\alpha} \neq \lambda_{\beta}} } | \langle \psi_{\alpha}, B \psi_{\beta} \rangle |^2 (\lambda_{\beta}-\lambda_{\alpha})(\gamma_{\alpha} - \gamma_{\beta})
\label{eq:StahlA2}
\end{equation}
with the eigenvalues $\{ \lambda_{\alpha} \}_{\alpha=1}^{\infty}$ and the eigenvectors $\{ \psi_{\alpha} \}_{\alpha=1}^{\infty}$ of $A$ and $\gamma_{\alpha} = \exp(-\lambda_{\alpha})/Z(0)$ for $\alpha \in \mathbb{N}$.
\end{theorem}

\begin{remark}
\begin{enumerate}[label=(\alph*)]
    \item The second term on the right-hand side of \eqref{eq:StahlA2} is nonnegative. It may be finite or infinite. 

    \item Under the assumption of Theorem~\ref{thm:correlation-intro} the term on the left-hand side of \eqref{eq:StahlA2} equals $\Tr[B^2 \Gamma_0]$ and the second term on the right-hand side equals $\frac 1 4 \Tr([[B,A],B]\Gamma_0)$.

\end{enumerate}
\end{remark}

Before give the proof of Theorem~\ref{thm:correlation} we state and prove the following lemma. 
\begin{lemma}
    \label{lem:FalkBruch}
    Let $A$ and $B$ be two self-adjoint operators on a complex separable Hilbert space. We assume that $\exp(-A)$ is trace-class and that $\mathds{1}(A \leq n) B \mathds{1}(A \leq n)$ is a bounded operator for any $n \in \mathbb{N}$. With the notation of Theorem~\ref{thm:correlation} we have
    \begin{equation}
        \Tr[\Gamma_0^{1/2} B^2 \Gamma_0^{1/2}] \leq \int_0^1 \Tr[ \Gamma_0^{(1-s)/2} B \Gamma_0^{s} B \Gamma_0^{(1-s)/2} ] \de s + \frac{1}{4} \sum_{\substack{\alpha,\beta \in \mathbb{N} \\ \lambda_{\alpha} \neq \lambda_{\beta} }} | \langle \psi_{\alpha}, B \psi_{\beta} \rangle |^2 (\lambda_{\beta}-\lambda_{\alpha})(\gamma_{\alpha} - \gamma_{\beta}). \label{eq:StahlB2A}
    \end{equation}
\end{lemma}
\begin{proof}[Proof of Lemma~\ref{lem:FalkBruch}]
    The proof of the case when $B$ is bounded is a direct consequence of \cite[Theorem~7.2]{LewNamRou-21}. To see this, one needs to note that 
    \begin{equation}
        \sum_{\substack{\alpha,\beta \in \mathbb{N} \\ \lambda_{\alpha} \neq \lambda_{\beta}} } | \langle \psi_{\alpha}, B \psi_{\beta} \rangle |^2 (\lambda_{\beta}-\lambda_{\alpha})(\gamma_{\alpha} - \gamma_{\beta}) = \Tr([[B,A],B]\Gamma_0)
        \label{eq:Andi1}
    \end{equation}
    holds for bounded $B$. We can therefore assume that \eqref{eq:StahlB2A} holds with $B$ replaced by $P_n B P_n$ with $P_n = \mathds{1}(A \leq n)$. 

    When we evaluate the trace in the first term on the right-hand side of \eqref{eq:Andi1} in terms of the eigenfunctions of $A$ and use the spectral decomposition of $\exp(-A(1-s))$, we see that it equals
    \begin{equation}
        \int_0^{1} \sum_{\alpha,\beta =1 }^{\infty} \gamma_{\alpha}^{1-s} \gamma_{\beta}^s | \langle \psi_{\alpha}, B \psi_{\beta} \rangle |^2 \de s.
        \label{eq:Andi2}
    \end{equation}
    Replacing $B$ by $B_n$ in the above equation amounts to multiplying the integrand by the characteristic functions $\mathds{1}(\alpha \leq n) \mathds{1}(\beta \leq n)$. We therefore have the upper bound
    \begin{equation}
        \int_0^1 \Tr[ \Gamma_0^{(1-s)/2} B_n \Gamma_0^{s} B_n \Gamma_0^{(1-s)/2} ] \de s \leq \int_0^1 \Tr[ \Gamma_0^{(1-s)/2} B \Gamma_0^{s} B \Gamma_0^{(1-s)/2} ] \de s.
        \label{eq:Andi3}
    \end{equation}
    
    Next, we consider the second term on the right-hand side of \eqref{eq:StahlB2A}. Since the effect of $P_n$ is again only a restriction of the number of (positive) terms in the sum, we have a similar upper bound here, too:
    \begin{equation}
        \sum_{\substack{\alpha,\beta \in \mathbb{N} \\ \lambda_{\alpha} \neq \lambda_{\beta}} } | \langle \psi_{\alpha}, B_n \psi_{\beta} \rangle |^2 (\lambda_{\beta}-\lambda_{\alpha})(\gamma_{\alpha} - \gamma_{\beta}) 
        \leq \sum_{\substack{\alpha,\beta \in \mathbb{N} \\ \lambda_{\alpha} \neq \lambda_{\beta} }} | \langle \psi_{\alpha}, B \psi_{\beta} \rangle |^2 (\lambda_{\beta}-\lambda_{\alpha})(\gamma_{\alpha} - \gamma_{\beta}).
        \label{eq:Andi4}
    \end{equation}
    In combination, \eqref{eq:StahlB2A} with $B$ replaced by $B_n$, \eqref{eq:Andi3}, and \eqref{eq:Andi4} show that
    \begin{equation}
        \Tr[\Gamma_0^{1/2} B_n^2 \Gamma_0^{1/2}] \leq \int_0^1 \Tr[ \Gamma_0^{(1-s)/2} B \Gamma_0^{s} B \Gamma_0^{(1-s)/2} ] \de s + \frac{1}{4}  \sum_{\substack{\alpha,\beta \in \mathbb{N} \\ \lambda_{\alpha} \neq \lambda_{\beta}} } | \langle \psi_{\alpha}, B \psi_{\beta} \rangle |^2 (\lambda_{\beta}-\lambda_{\alpha})(\gamma_{\alpha} - \gamma_{\beta})
        \label{eq:Andi5}
    \end{equation}
    holds for all $n \in \mathbb{N}$.

    It remains to show that the left-hand side of \eqref{eq:Andi5} converges to the left-hand side of \eqref{eq:StahlB2A} as $n \to \infty$. We evaluate the trace in terms of the basis $\{ \psi_{\alpha} \}_{\alpha=1}^{\infty}$, use the identity $P_n = \sum_{\beta=1}^n | \psi_{\beta} \rangle \langle \psi_{\beta} |$, and find
    \begin{equation}
        \Tr[\Gamma_0^{1/2} (B_n)^2 \Gamma_0^{1/2}] = \sum_{\alpha,\beta =1}^{\infty} \mathds{1}(\alpha \leq n) \mathds{1}(\beta \leq n)  \lambda_{\alpha} | \langle \psi_{\alpha}, B \psi_{\beta} \rangle |^2.
    \end{equation}
    An application of the monotone convergence theorem shows that the right-hand side converges to 
    \begin{equation}
        \sum_{\alpha,\beta =1 }^{\infty} \lambda_{\alpha} | \langle \psi_{\alpha}, B \psi_{\beta} \rangle |^2 = \Tr[\Gamma_0^{1/2} B^2 \Gamma_0^{1/2}],
    \end{equation}
    which proves \eqref{eq:StahlB2A}.
\end{proof}

We are now prepared to give the proof of Theorem~\ref{thm:correlation}.

\begin{proof}[Proof of Theorem~\ref{thm:correlation}] 
Using Duhamel's formula, we check that
\begin{equation}
    Z'(t) = -\Tr[B \exp(-(A+tB))] 
    \label{eq:correlationInequalityA2_1}
\end{equation}
holds for $t \in [-1,1]$ and
\begin{equation}
    Z''(0) = \int_0^{1} \sum_{\alpha,\beta =1 }^{\infty} \gamma_{\alpha}^{1-s} \gamma_{\beta}^s | \langle \psi_{\alpha}, B \psi_{\beta} \rangle |^2 \de s
    \label{eq:correlationInequalityA2_2}.
\end{equation}
Next, we use our assumption
\begin{equation}
    a \geq |\Tr (B \Gamma_t)| = |\partial_t \ln (Z(t))|, \quad t\in [-1,1] 
    \label{eq:correlationInequalityA2}
\end{equation}
and Gr\"onwall's inequality to prove the bound
\begin{align}\label{eq:1-asum-used}
e^{a} \ge \frac{Z(t)}{Z(0)} \ge e^{-a}, \quad t\in [-1,1]. 
\end{align}

From Theorem~\ref{prop:stahlInfinite} we know that 
\begin{equation}
    Z''(t) = \int_{-\infty}^{\infty} s^{2} e^{-ts} \de \mu(s) \geq 0
\end{equation}
holds for all $t \in (-1,1)$, and hence $t\mapsto Z''(t)$ is convex on that interval. Using this, we estimate
\begin{align}\label{eq:new-correlation-proof-1}
Z''(0) \le \frac{1}{2}\int_{-1}^1 Z''(s) \de s = \frac{1}{2} ( Z'(1) - Z'(-1) ) \leq \sup_{t \in [-1,1]} | \Tr [B \Gamma_t] | Z(t)  \leq a e^{a} Z(0),
\end{align}
where the last estimate follows from \eqref{eq:1-asum-used}. In combination, \eqref{eq:correlationInequalityA2_2}, \eqref{eq:correlationInequalityA2} and \eqref{eq:new-correlation-proof-1} show
\begin{equation}
    \int_0^{1} \sum_{\alpha,\beta =1 }^{\infty} \gamma_{\alpha}^{1-s} \gamma_{\beta}^s | \langle \psi_{\alpha}, B \psi_{\beta} \rangle |^2 \de s \leq a e^a.
    \label{eq:new-correlation-proof-2}
\end{equation}
Eq.~\eqref{eq:StahlA2} now follows from \eqref{eq:StahlB2A} and \eqref{eq:new-correlation-proof-2}.
\end{proof}

\subsection{Second order estimates for the Gibbs state} 
\label{sec:secondOrderEstimates}
In this section we apply the abstract result of the previous section to the Gibbs state $G_{h,\eta}(\beta,\mu)$ in \eqref{eq:GeneralGibbsState}. We will use the notation that has been introduced in Section~\ref{sec:FirstOrderEstimates}. In particular, we recall the definitions of the one-particle Hamiltonians $h$ in \eqref{eq:generalizedOneParticleHamiltonian} and $w$ in \eqref{eq:definitionOfW} as well as that $B_p = \de \Upsilon( Q \cos(p \cdot x) Q )$ with $Q$ in \eqref{eq:excitationFockSpace}.

\begin{theorem} (Second order estimates) 
\label{thm:secondOrderEstimates}
Let $v$ satisfy the assumptions of Theorem~\ref{thm:norm-approximation}. We consider the limit $\eta \to \infty$, $\beta \eta^{2/3} \to \kappa  \in (0,\infty)$. The chemical potential $\mu$, which may depend on $\eta$, is assumed to satisfy $-\eta^{2/3} \lesssim \mu \lesssim 1$. Let $\widetilde \mu < 0$ be the unique solution to the equation
    \begin{equation}
        \sum_{p \in \Lambda^*} \frac{1}{e^{\beta(p^2 - \widetilde{\mu})}-1} = \frac{(\mu - \widetilde{\mu})\eta}{\hat{v}(0)}
        \label{eq:GrantCanonicalEffectiveIddealGasChemPotb}
    \end{equation}
and define
    \begin{equation}
        M(\beta,\widetilde{\mu}) = \sum_{p \in \Lambda^*} \frac{1}{e^{\beta(p^2 - \widetilde{\mu})}-1}.
        \label{eq:particleNumberb}
    \end{equation}
Moreover, let $\mu_0(\beta,M)$ and $N_0(\beta,M)$ be defined as in \eqref{eq:idealgase1pdmchempot} and \eqref{eq:crittemp}, respectively (note that $\mu_0(\beta,M) = \widetilde{\mu}$). Then the following holds: 
\begin{enumerate}[label=(\alph*)]
\item Assume that $h$ is diagonal in momentum space. For $p \in \Lambda_+^*$ we have
\begin{equation} \label{eq:aa-G}
\Tr [ (a_p^*a_p)^2  G_{h,\eta}(\beta,\mu) ] \lesssim \left( \frac{1}{\beta ( p^2 - \mu_0(\beta,M)) } \right)^2 + 1. 
\end{equation}
If $N_0(\beta,M) \lesssim \eta^{2/3}$ the bound also holds for $p = 0$.
\item For $p \in \Lambda_+^*$ we have
\begin{equation}
    \Tr [ B_p^2  G_{h,\eta}(\beta,\mu) ] \lesssim \eta^{4/3} ( 1 + p^2 \eta^{-1} )
    \label{eq:cN-cN}
\end{equation}
with $B_p = \de \Upsilon( Q \cos(p \cdot x) Q )$ and $Q$ in \eqref{eq:excitationFockSpace}.
\item We have
\begin{align}
\Tr [ (\cN_{+}-N_{+}^{\mathrm{G}})^2 G_{h,\eta}(\beta,\mu) ] &\lesssim \eta^{4/3},
\label{eq:varianceNPlus}
\end{align}
where $N_{+}^{\mathrm{G}} = \Tr[\mathcal{N}_+ G_{h,\eta}(\beta,\mu)]$.
\item We have
    \begin{equation}
    \Tr [\cN^4 G_{h,\eta}(\beta,\mu) ]  \lesim  \eta^4.
    \label{eq:particleNumberBoundsPerturbedState-k}
    \end{equation}
\end{enumerate}

\end{theorem} 

\begin{remark}
    The bounds in the above theorem cannot be expected to hold for general approximate minimizers of the free energy. For example, \eqref{eq:lowerBoundFE2} below remains true if we add the nonnegative term 
    \begin{equation}
        \frac{1}{8 \eta} \sum_{p \in \Lambda^*_+} \hat{v}(p) \tr[ B^2_p G]
    \end{equation}
    on the right-hand side. To see this we apply \eqref{eq:firstOrderBoundsAndiNew8}. When we keep this term until the end of the lower bound for the free energy, we can conclude that it is bounded from above by a constant times $\eta^{5/3}$. Therefore the bound in \eqref{eq:cN-cN} is stronger than what one would obtain from the use of coercivity, which is crucial to resolve the free energy to the precision that justifies Bogoliubov theory.  
\end{remark}

\begin{proof} {\textbf {Proof of \eqref{eq:cN-cN}.}} We apply Theorem~\ref{thm:correlation} with the choices 
\begin{equation}
    A=\beta( \mathcal{H}_{h,\eta} - \mu \mathcal{N} ), \quad B=  \lambda \beta B_p
\end{equation}
with $\mathcal{H}_{h,\eta}$ in \eqref{eq:perturbedFockSpaceHamiltonian}, $p \in \Lambda_+^*$, and $B_p$ in \eqref{eq:cN-cN}. We choose $|\lambda|$ sufficiently small such that $h + \lambda Q \cos(p \cdot x) Q \gesssim - \Delta$. Condition \eqref{eq:CRI-condition}  in Theorem \ref{thm:correlation} is satisfied with a constant $a>0$ because of the second estimate in \eqref{eq:unperturbedFirstOrderBounds}. Thus an application of Theorem~\ref{thm:correlation} shows
\begin{equation}
\Tr[ B_p^2 G_{h,\eta} ] \leq \frac{a e^a}{\beta^2} + \frac{\beta}{4} \Tr( [[B_p, \mathcal{H}_{h,\eta} - \mu \mathcal{N}], B_p]  G_{h,\eta} ).
\label{eq:section63A9}
\end{equation}

To compute the double commutator of $B_p$ and $\de \Gamma(h)$ we use
\begin{equation}
    [\de \Upsilon (X), \de \Upsilon (Y)]= \de \Upsilon ([X,Y])
    \label{eq:section63A4}
\end{equation}
for two operators $X,Y \in \mathcal{B}(L^2(\Lambda))$. Next we write $h = -\Delta + h + \Delta$ and apply the formula  
\begin{equation}
    [[\varphi(x),-\Delta],\varphi(x)]  = (\nabla \varphi(x))^2,
    \label{eq:section63A9b}
\end{equation}
where $\varphi$ is a multiplication operator with the real-valued function $\varphi(x)$. These considerations show
\begin{equation}
    [[B_p, \de \Upsilon(h), B_p] = \de \Upsilon( Q p^2 \cos^2(p \cdot x) Q ) + \de \Upsilon( Q [[\cos^2(p \cdot x), (h+\Delta)], \cos^2(p \cdot x)] Q ).
    \label{eq:section63A9c}
\end{equation}
When combined with $\Vert h + \Delta \Vert < +\infty$, \eqref{eq:particleNumberBoundsPerturbedState}, and the bound $\Vert \de \Upsilon(X) \psi \Vert \leq \Vert X \Vert_{\infty} \Vert \mathcal{N} \psi \Vert$ with a bounded operator $X$ acting on $L^2(\Lambda)$, this also implies 
\begin{equation}
    \Tr([[B_p, \de \Upsilon(h), B_p] G_{h,\eta}) \lesssim (1 + p^2) \Tr[ \mathcal{N} G_{h,\eta} ] \lesssim (1 + p^2) \eta.
    \label{eq:section63A9d}
\end{equation}

To obtain a bound for second commutator on the right-hand side of \eqref{eq:section63A9} involving the interaction term $\mathcal{V}_{\eta}$, we use again the bound for $\Vert \de \Upsilon(X) \psi \Vert$ from above. It is not difficult to see that
\begin{equation}
    \Tr( [[B_p, \mathcal{V}_{\eta}], B_p]  G_{h,\eta} ) \lesssim \eta. 
    \label{eq:section63A10}
\end{equation}
Putting \eqref{eq:section63A9}, \eqref{eq:section63A9d}, and \eqref{eq:section63A10} together, we find
\begin{equation}
    \Tr[ B_p^2 G_{h,\eta} ] \lesssim \beta^{-2} + (1+p^2) \beta \eta,
    \label{eq:section63A10b}
\end{equation}
which proves \eqref{eq:cN-cN}.

{\textbf {Proof of \eqref{eq:aa-G}.}} Let us first prove \eqref{eq:aa-G} with $\mu_0(\beta,M)$ replaced by zero on the right-hand side. This is sufficient if $N_0(\beta,M) \gesssim \eta^{2/3}$ as it is equivalent to $-\mu_0(\beta,M) \lesssim 1$. We apply Theorem~\ref{thm:correlation} with the choices
\begin{equation}
    A=\beta( \mathcal{H}_{h,\eta} - \mu \mathcal{N} ) \quad \text{ and } \quad B= \lambda \beta p^2 a_p^* a_p
    \label{eq:section63A1}
\end{equation}
 with $p \in \Lambda_+^*$ and a constant $\lambda \in \mathbb{R}$ that we allow ourselves to choose as small as we wish. We recall that $a_p^* a_p = \de \Upsilon(| \varphi_p \rangle \langle \varphi_p |)$ with $\varphi_p(x) = e^{\mathrm{i}p \cdot x}$. Our assumptions on $h$ imply that $h + \lambda p^2 |\varphi_p \rangle \langle \varphi_p| \gesssim -\Delta$ holds provided $|\lambda|$ is small enough. From \eqref{eq:unperturbedFirstOrderBounds} we know that condition \eqref{eq:CRI-condition} in Theorem \ref{thm:correlation} is satisfied with a constant $a > 0$ that does not depend on $\eta$. An application of Theorem \ref{thm:correlation} therefore gives
\begin{equation}
     \Tr[ (a_p^*a_p)^2 G_{h,\eta}(\beta,\mu) ] \leq \frac{ a e^a }{( \lambda \beta p^2 )^2} + \frac{\beta}{4} \Tr( [[a_p^*a_p,\mathcal{H}_{h,\eta} - \mu \mathcal{N}],a_p^* a_p ]  G_{h,\eta} ).  
    \label{eq:section63A2}
\end{equation}

Since $h$ is diagonal in momentum space we have $[a_p^*a_p, \de \Upsilon(h) - \mu \mathcal{N}] = 0$. A short computation also shows
\begin{align}\label{eq:commutator-appl-B1-B>0-2Andi}
\sum_{k,q,r\in \Lambda^*} \hat v(k) \Big[a_p^* a_p,  \Big[ a_p^* a_p, a^*_{r+k} a^*_{q-k} a_r a_q \Big]\Big] 
=\sum_{k,q,r\in \Lambda^*} \hat v(k) (\delta_{p,r+k} +\delta_{p,q-k}  -\delta_{p,r} -\delta_{p,q})^2 a^*_{r+k} a^*_{q-k} a_r a_q. 
\end{align}
Let us have a closer look at the term proportional to $\delta_{p,r+k}$, which reads
\begin{align}
    \sum_{q,r \in \Lambda^*} \hat v(p-r) \tr[ a^*_{p} a^*_{q-p+r} a_r a_q G_{h,\eta} ] =& \sum_{q \in \Lambda^*, r \in \Lambda_+^*} \hat v(p-r) \tr[ a^*_{p} a^*_{q-p+r} a_r a_q G_{h,\eta} ] \nn \\
    &+ \hat v(p) \sum_{q \in \Lambda^*}  \tr[ a^*_{p} a^*_{q-p} a_0 a_q G_{h,\eta} ].
    \label{eq:firstOrderBoundsAndiNew1}
\end{align}
The first term on the right-hand side can be written as
\begin{align}
    \sum_{q \in \Lambda^*, r \in \Lambda_+^*} \hat v(p-r) \tr[ a^*_{p} a^*_{q-p+r} a_r a_q G_{h,\eta} ] =& \sum_{r \in \Lambda_+^*} \hat v(p-r) \tr\left[ a^*_{p} a_r \left( \sum_{ q \in \Lambda^* } a^*_{q-p+r} a_q \right) G_{h,\eta} \right] \nn \\
    &- \sum_{r \in \Lambda_+^*} \hat{v}(p-r) \tr[ a_p^* a_p G_{h,\eta} ].
    \label{eq:firstOrderBoundsAndiNew2}
\end{align}
The quadratic operator in the bracket in the first term is the second quantization of multiplication with the function $\varphi_{-p+r}(x) = e^{\mathrm{i}(-p+r) \cdot x}$. An application of the Cauchy-Schwarz inequality shows
\begin{align}
    \sum_{r \in \Lambda_+^*} \hat v(p-r) \left| \tr\left[ a^*_{p} a_r \de \Upsilon(  \varphi_{-p+r} ) G_{h,\eta} \right] \right| \leq& \sum_{r \in \Lambda_+^*} \hat{v}(p-r) \left( \tr[ (a^*_p a_p)^2 G_{h,\eta} ] \tr[ (a^*_r a_r + 1)^2 G_{h,\eta} ] \right)^{1/4} \nn \\ 
    &\times \left( \tr[ \de \Upsilon^*(\varphi_{-p+r}) \de \Upsilon(\varphi_{-p+r}) G_{h,\eta} ] \right)^{1/2}. 
    \label{eq:firstOrderBoundsAndiNew3}
\end{align}
With $|\varphi_{-p+r}|=1$, the bound for the second quantization of an operator below \eqref{eq:section63A9c}, and \eqref{eq:particleNumberBoundsPerturbedState}, we check that the expectation in the second line is bounded by $\tr[ \mathcal{N}^2 G_{h,\eta} ] \lesssim \eta^2$. Accordingly, the right-hand side of \eqref{eq:firstOrderBoundsAndiNew3} is bounded by a constant times
\begin{equation}
    \Vert \hat{v} \Vert_1 \eta \left( \sup_{r \in \Lambda_+^*} \tr[ (a^*_r a_r)^2 G_{h,\eta} ] + 1 \right)^{1/2} 
    \label{eq:firstOrderBoundsAndiNew4}
\end{equation}
and we have
\begin{equation}
    \left| \sum_{q \in \Lambda^*, r \in \Lambda_+^*} \hat v(p-r) \tr[ a^*_{p} a^*_{q-p+r} a_r a_q G_{h,\eta} ] \right| \lesssim \Vert \hat{v} \Vert_1 \eta \left( \sup_{r \in \Lambda_+^*} \tr[ (a^*_r a_r)^2 G_{h,\eta} ] + 1 \right)^{1/2}.  
    \label{eq:firstOrderBoundsAndiNew5}
\end{equation}

Next, we consider the second term on the right-hand side of \eqref{eq:firstOrderBoundsAndiNew1}. If $q = 0$ or $q = p$ it is bounded by a constant times
\begin{equation}
    \Vert \hat{v} \Vert_{\infty} \eta \left( \sup_{r \in \Lambda_+^*} \tr[ (a^*_r a_r)^2 G_{h,\eta} ] \right)^{1/2}.
    \label{eq:firstOrderBoundsAndiNew6}
\end{equation}
If $q \notin \{0,p\}$ we introduce the notation $\widetilde{B}_p = \sum_{q,q-p \in \Lambda_+^*} a^*_{q-p} a_q$ and estimate
\begin{align}
    \left| \sum_{q \in \Lambda^*_+ \backslash \{ p \}}  \tr[ a^*_{p} a^*_{q-p} a_0 a_q G_{h,\eta} ] \right| \leq \left( \tr[ a_{p} a_p^* a^*_0 a_0  G_{h,\eta} ] \tr[ \widetilde{B}_p^* \widetilde{B}_p ] \right)^{1/2}.
    \label{eq:firstOrderBoundsAndiNew7}
\end{align}
We note that
\begin{equation}
    \frac{1}{2} (\widetilde{B}_p^* + \widetilde{B}_p) = B_p   
    \label{eq:lowerBoundFE3}
\end{equation}
with $B_p$ below \eqref{eq:cN-cN}. Using the translation-invariance of $G_{h,\eta}$, which implies momentum conservation for the matrix elements of its reduced density matrices in momentum space, we check that
\begin{equation}
    \Tr [ \widetilde{B}_p^2 G ] =  0 \quad \text{ and } \quad [\widetilde{B}_p,\widetilde{B}_p^*] = 0
    \label{eq:lowerBoundFE4}
\end{equation}
hold for $p \in \Lambda_+^*$. But this implies 
\begin{equation}
    \Tr [ \widetilde{B}_p^* \widetilde{B}_p G ] = \frac{1}{2} \Tr[ B_p^2 G ]. 
    \label{eq:firstOrderBoundsAndiNew8}
\end{equation}
When we combine \eqref{eq:particleNumberBoundsPerturbedState}, \eqref{eq:cN-cN}, \eqref{eq:firstOrderBoundsAndiNew7}, and \eqref{eq:firstOrderBoundsAndiNew8}, we find
\begin{equation}
    \left| \sum_{q \in \Lambda^*_+ \backslash \{ p \}}  \tr[ a^*_{p} a^*_{q-p} a_0 a_q G_{h,\eta} ] \right| \lesssim \eta^{7/6} \left( 1 + (\tr[ (a_p^* a_p)^2 G_{h,\eta} ])^{1/4} \right) ( 1 + |p| ), 
    \label{eq:firstOrderBoundsAndiNew99}
\end{equation}
and hence
\begin{align}
    \hat v(p) \left| \sum_{q \in \Lambda^*}  \tr[ a^*_{p} a^*_{q-p} a_0 a_q G_{h,\eta} ] \right| \lesssim& \eta \Vert \hat{v} \Vert_{\infty} \left( \sup_{r \in \Lambda_+^*} \tr[ (a^*_r a_r)^2 G_{h,\eta} ] \right)^{1/2} \nn \\
    &+ \eta^{7/6} \left( \sup_{r \in \Lambda^*} \hat{v}(r)(1+|r|) \right) \left( 1 + (\tr[ (a_p^* a_p)^2 G_{h,\eta} ])^{1/4} \right) .
\label{eq:firstOrderBoundsAndiNew9}
\end{align}

In combination, \eqref{eq:firstOrderBoundsAndiNew1}, \eqref{eq:firstOrderBoundsAndiNew5}, and \eqref{eq:firstOrderBoundsAndiNew9} show
\begin{align}
    \left| \sum_{q,r \in \Lambda^*} \hat v(p-r) \tr[ a^*_{p} a^*_{q-p+r} a_r a_q G_{h,\eta} ] \right| \lesssim& \eta \Vert \hat{v} \Vert_{1} \left( 1 + \left( \sup_{r \in \Lambda_+^*} \tr[ (a^*_r a_r)^2 G_{h,\eta} ] \right)^{1/2} \right) \nn \\
    &+ \eta^{7/6}\left( \sup_{r \in \Lambda^*} \hat{v}(r) (1+|r|) \right) \left( 1 + \left( \tr[ (a^*_p a_p)^2 G_{h,\eta} ] \right)^{1/4} \right) .
    \label{eq:firstOrderBoundsAndiNew10}
\end{align}
For the other terms on the right-hand side of \eqref{eq:commutator-appl-B1-B>0-2Andi} that are proportional to $\delta_{p,q-k}$, $\delta_{p,r}$, or $\delta_{p,q}$, we can obtain similar bounds using the symmetry of the relevant labels. The cross terms, which involve the product of two Kronecker deltas, are even simpler to estimate since the corresponding sum is taken over a more restricted domain. Thus, we conclude that all terms on the right-hand side of \eqref{eq:commutator-appl-B1-B>0-2Andi} are bounded by the right-hand side of \eqref{eq:firstOrderBoundsAndiNew10}. 

We conclude the bound for the double commutator of the form
\begin{align} 
&\beta \Tr( [[a_p^*a_p,\mathcal{H}_{h,\eta} - \mu \mathcal{N}],a_p^* a_p ]  G_{h,\eta} )=\frac{\beta}{8\eta}\sum_{k,q,r\in \Lambda^*} \hat v(k) \tr( [a_p^* a_p,  [ a_p^* a_p, a^*_{r+k} a^*_{q-k} a_r a_q ] ] G_{h,\eta}) \nonumber
\\    &\lesssim \beta \Vert \hat{v} \Vert_{1} \left( 1 + \left( \sup_{r \in \Lambda_+^*} \tr[ (a^*_r a_r)^2 G_{h,\eta} ] \right)^{1/2} \right) 
    + \beta\eta^{1/6}\left( \sup_{r \in \Lambda^*} \hat{v}(r) (1+|r|) \right) \left( 1 + \left( \tr[ (a^*_p a_p)^2 G_{h,\eta} ] \right)^{1/4} \right).
    \label{eq:doublecommboundapapapap}
\end{align}
Inserting \eqref{eq:doublecommboundapapapap} in  \eqref{eq:section63A2}, we find
\begin{equation}
    \sup_{r \in \Lambda_+^*} \tr[ (a_r^* a_r)^2 G_{h,\eta} ] \lesssim \eta^{4/3}.
    \label{eq:firstOrderBoundsAndiNew11}
\end{equation}
Finally, when we use \eqref{eq:firstOrderBoundsAndiNew11} in the bounds above, we also obtain
\begin{equation}
    \tr[ (a_p^* a_p)^2 G_{h,\eta} ] \lesssim \left( \frac{1}{\beta p^2} \right)^2 + 1,\quad p\in \Lambda_+^*. 
    \label{eq:firstOrderBoundsAndiNew12}
\end{equation}

It remains to consider the case $N_0(\beta,M) \lesssim \eta^{2/3}$. In this case we apply Theorem~\ref{thm:correlation} with the choices
\begin{equation}
    A=\beta( \mathcal{H}_{h,\eta} - \mu \mathcal{N} ) \quad \text{ and } \quad B= \lambda \beta (p^2-\mu_0(\beta,M)) a_p^* a_p, \qquad p\in \Lambda^*
    \label{eq:section63A1p=0}.
\end{equation}
By \eqref{eq:firstorderapapbound}, we know that condition \eqref{eq:CRI-condition} in Theorem \ref{thm:correlation} is satisfied with a constant $a > 0$ that does not depend on $\eta$, and hence we obtain 
\begin{align}
     &\Tr[ (a_p^*a_p)^2 G_{h,\eta}(\beta,\mu) ] \leq \frac{ a e^a }{( \lambda \beta (p^2-\mu_0(\beta,M) )^2} + \frac{\beta}{4} \Tr( [[a_p^*a_p,\mathcal{H}_{h,\eta} - \mu \mathcal{N}],a_p^* a_p ]  G_{h,\eta} )\nonumber \\
     &\lesssim \frac{ 1 }{( \beta (p^2-\mu_0(\beta,M) )^2} + \beta \left( 1 + \left( \sup_{r \in \Lambda^*} \tr[ (a^*_r a_r)^2 G_{h,\eta} ] \right)^{1/2} \right) 
    +\beta\eta^{1/6}  \left( 1 + \left( \tr[ (a^*_p a_p)^2 G_{h,\eta} ] \right)^{1/4} \right).
    \label{eq:section63A2N0small}
\end{align}
Here the estimate for the double double commutator is obtained by following the proof of  \eqref{eq:doublecommboundapapapap}. Note that  \eqref{eq:section63A2N0small} holds for all $p\in \Lambda^*$. Therefore, we arrive at the bound
\begin{equation}
    \sup_{p\in \Lambda^*} \Tr[ (a_p^*a_p)^2 G_{h,\eta}(\beta,\mu) ] \lesssim \frac{1}{( \beta \mu_0(\beta,M) )^2} + 1 \lesim \frac{1}{\beta^2}
\end{equation}
where we used $-\mu_0(\beta,M) \gesim 1$ when $N_0(\beta,M) \lesim \eta^{2/3}$. Inserting this bound in \eqref{eq:section63A2N0small} we conclude that   
\begin{equation}
     \Tr[ (a_p^*a_p)^2 G_{h,\eta}(\beta,\mu) ] \lesssim \frac{ 1 }{(  \beta (p^2-\mu_0(\beta,M) )^2} + 1,\quad p\in \Lambda^* 
\end{equation}
if $N_0(\beta,M) \lesssim \eta^{2/3}$. The proof of \eqref{eq:aa-G} is complete.

{\textbf {Proof of \eqref{eq:varianceNPlus}}.}  We apply Theorem \ref{thm:correlation} with $A=\beta( \mathcal{H}_{h,\eta} - \mu \mathcal{N} )$ and $B= \lambda \beta (\cN_+-N_+^{\mathrm{G}})$ with $|\lambda|$ chosen small enough such that $h + \lambda Q \gesim - \Delta$ holds. Condition \eqref{eq:CRI-condition} in Theorem \ref{thm:correlation} with a constant $a > 0$ is justified by Theorem~\ref{thm:firstOrderAPriori}. An application of Theorem \ref{thm:correlation} therefore shows 
\begin{equation}
\Tr( (\cN_+-N_+)^2 G_{h,\eta} ) \leq \frac{ a e^a }{\beta^2} + \frac{\beta}{4} \Tr( [[\cN_+, \mathcal{H}_{h,\eta} - \mu \mathcal{N}],\cN_+] G_{h,\eta} ). 
\label{eq:section63A10c}
\end{equation}
The second commutator on right-hand side involving the kinetic term vanishes. The one involving the interaction term can be estimated similarly as the term in \eqref{eq:section63A10}, which yields 
\begin{equation}
    \Tr( [[\cN_+, \mathcal{V}_{\eta},\cN_+] G_{h,\eta} ) \lesssim \eta.
    \label{eq:section63A10d}
\end{equation}
When we combine \eqref{eq:section63A10c} and \eqref{eq:section63A10d}, this proves \eqref{eq:varianceNPlus}.

{\textbf {Proof of \eqref{eq:particleNumberBoundsPerturbedState-k}.}} 
We apply Theorem \ref{thm:correlation} with $A=\beta( \mathcal{H}_{h,\eta} - \mu \mathcal{N} )$ and $B= \eta^{-2} \cN^2$ (considering the operator $A+tB$ corresponds to shifting $\hat v(0)$ by $2t \beta^{-1}\eta^{-1}=O(\eta^{-1/3})$). We highlight that $A$ and $B$ commute. An application of Theorem~\ref{thm:correlation} shows
\begin{equation}
  \eta^{-4}  \Tr[\mathcal{N}^4 G_{h,\eta}] \lesssim 1.
\end{equation}
\end{proof}

\section{Sharp lower bound for the free energy} 
\label{sec:lowerBound}
In this section we prove the lower bound in \eqref{eq:ThmFreeEnergyBounds} and thereby finish the proof of Theorem~\ref{thm:main1}. We will also provide the proof of Corollary~\ref{cor:main1}.

\subsection{Lower bound for the energy in terms of a simplified Hamiltonian}
\label{sec:lowerBoundSimpliefiedHamiltonian}
Our goal is to obtain a lower bound for 
\begin{equation}
    F(\beta,N) = \Tr [\cH_N G] - \frac{1}{\beta} S(G)    
    \label{eq:lowerBoundFE0}
\end{equation}
with the Gibbs state $G$ in \eqref{eq:interactingGibbsstate}. Note that, for the sake of simplicity and because we will have to introduce other subscripts later, we omit the subscripts $\beta,N$ and write $G$ instead of $G_{\beta,N}$ in Section~\ref{sec:lowerBound}. 

In this section we replace the Hamiltonian $\mathcal{H}_N$ in \eqref{eq:lowerBoundFE0} by an operator that is easier to handle in a controlled way. To that end, we first recall the decomposition of $\mathcal{H}_N$ in \eqref{eq:decompH}. 

We start with the term in the last line of \eqref{eq:decompH}, which we write as
\begin{equation}
    \frac{1}{2N} \sum_{u,v,p,u+p,v-p \in \Lambda^*_+} \hat{v}(p) a^*_{u+p} a^*_{v-p} a_u a_v = \frac{1}{2N} \sum_{p \in \Lambda^*_+} \hat{v}(p) \widetilde{B}_p^* \widetilde{B}_p -  \frac{1}{2N} \left( \sum_{p \in \Lambda^*_+} \hat{v}(p) \right) \cN_+
    \label{eq:lowerBoundFE1}
\end{equation}
with $\widetilde{B}_p=\sum_{v, v-p\in\Lambda^*_+}a^*_{v-p}a_v$. Using \eqref{eq:lowerBoundFE1}, $\widetilde{B}_p^* \widetilde{B}_p \ge 0$ and $\hat v(p)\ge 0$ for $p \in \Lambda_+^*$, and $\Tr (\cN_+ G) \leq N$, we find
\begin{equation}
\frac{1}{2N} \sum_{u,v,p,u+p,v-p \in \Lambda^*_+} \hat{v}(p) \Tr [ a^*_{u+p} a^*_{v-p} a_u a_v G] \ge -  \frac{v(0)}{2}.
\label{eq:lowerBoundFE2}
\end{equation} 

Next, we consider the term that is cubic in creation and annihilation operators $a_p^*$, $a_p$ with $p \in \Lambda_+^*$, that is, the fourth term on the right-hand side of \eqref{eq:decompH}. We first recall \eqref{eq:firstOrderBoundsAndiNew8}. In combination with Lemma~\ref{lem:roughBoundChemicalPotential} and \eqref{eq:cN-cN} in Theorem~\ref{thm:secondOrderEstimates}, this implies
\begin{equation}
    \Tr [ \widetilde{B}_p^* \widetilde{B}_p G ] = \frac{1}{2} \Tr[ B_p^2 G ] \lesssim N^{4/3} (1+p^2 N^{-1}).
    \label{eq:AA-G}
\end{equation}
Using \eqref{eq:AA-G}, \eqref{eq:aa-G} in Theorem~\ref{thm:secondOrderEstimates}, and the Cauchy--Schwarz inequality, we can estimate the term that is cubic in $a_p^*, a_p$ with $p \in \Lambda_+^*$ as follows:
\begin{align}
\frac{1}{N} &\sum_{p,k,p+k \in \Lambda_+^*} \hat{v}(p)  \Tr [ (a^*_{k+p} a^*_{-p} a_k a_0 + \mathrm{h.c.}) G ] =  \frac{1}{N} \sum_{p  \in \Lambda_+^*} \hat{v}(p)  \Tr [ ( a_{-p}^* B_{-p} a_0 + \mathrm{h.c.})G ] \nonumber \\
&\geq  -\frac{2}{N}  \sqrt{\sum_{p \in \Lambda_+^*} \hat{v}(p) |p| \Tr [ a_{-p}^* (\cN+1) a_{-p} G ] }   \sqrt{  \sum_{p \in \Lambda_+^*} \hat{v}(p) |p|^{-1} \Tr [ a_0^* B_{-p}^* (\cN+1)^{-1} B_{-p} a_0 G ]} \nonumber \\
&=  -\frac{2}{N}  \sqrt{\sum_{p \in \Lambda_+^*} \hat{v}(p) |p| \Tr [ \cN a_{-p}^* a_{-p} G] } \sqrt{  \sum_{p \in \Lambda_+^*} \hat{v}(p) |p|^{-1} \Tr [ \cN^{-1} a_0^* a_0  \widetilde{B}_p^* \widetilde{B}_p G ]} \nonumber \\
&\geq -\frac{2}{N}  \sqrt{\sum_{p \in \Lambda_+^*} \hat{v}(p) |p| \sqrt{ \Tr [ \cN^2 G ] \Tr [ (a_{-p}^* a_{-p})^2 G ] } } \sqrt{  \sum_{p \in \Lambda_+^*} \hat{v}(p) |p|^{-1} \Tr [ \widetilde{B}_p^* \widetilde{B}_p G ]} \nonumber \\
&\gesssim - \frac 2 N \sqrt{ \sum_{p \in \Lambda_+^*} \hat{v}(p) |p| N^{5/3} } \sqrt{  \sum_{p \in \Lambda_+^*} \hat{v}(p) |p|^{-1} N^{4/3} p^2 } = -N^{1/2} \sum_{p \in \Lambda_+^*} |p| \hat{v}(p). \label{eq:lowerBoundFE5}
\end{align}

Now we consider the term
\begin{equation}
    \frac{\hat{v}(0)}{2N} \sum_{u,v \in \Lambda^*} \Tr [ a_u^* a_v^* a_u a_v G ] = \frac{\hat{v}(0)}{2N} \Tr [ \cN (\cN-1) G ] = \frac{\hat{v}(0)}{2} N +  \frac{\hat{v}(0)}{2N} \Tr [ (\cN - N) ^2 G ] - \frac{\hat{v}(0)}{2}. 
\label{eq:lowerBoundFE6}
\end{equation}
An application of the Cauchy--Schwarz inequality shows that the variance in the above equation satisfies the operator inequality
\begin{equation}
    (\cN - N)^2=(\cN_0 - N^{\mathrm{G}}_0 + \cN_+  - N^{\mathrm{G}}_+)^2 \ge (1-\delta) (\cN_0 - N^{\mathrm{G}}_0 )^2 - \delta^{-1} (\cN_+ - N^{\mathrm{G}}_+)^2      
    \label{eq:lowerBoundFE7}
\end{equation}
for any $0 < \delta < 1$. We recall the notations $\mathcal{N}_0 = a_0^*a_0$, $\mathcal{N}_+ = \mathcal{N} - \mathcal{N}_0$, $N_0^{\mathrm{G}} = \Tr[ \mathcal{N}_0 G ]$ and $N_+^{\mathrm{G}} = \Tr[ \mathcal{N}_+ G ]$. Moreover, an application of \eqref{eq:varianceNPlus} in Theorem~\ref{thm:secondOrderEstimates} shows that the expectation of the second term on the right-hand side of \eqref{eq:lowerBoundFE7} with respect to the state $G$ is bounded from below by a constant times $\delta^{-1} N^{4/3}$. When we put these considerations together, we find the lower bound
\begin{equation}
\frac{\hat{v}(0)}{2N} \sum_{u,v \in \Lambda^*} \Tr [ a_u^* a_v^* a_u a_v G ] \geq \frac{\hat{v}(0) N}{2} +  (1-\delta) \frac{\hat{v}(0)}{2N} \Tr [ (\cN_0 - N_0^{\mathrm{G}})^2 G ] - C N^{1/3} ( 1 + \delta^{-1}  ). 
\label{eq:lowerBoundFE8}
\end{equation}

Let us summarize our findings. In combination, \eqref{eq:lowerBoundFE0}, \eqref{eq:lowerBoundFE2}, \eqref{eq:lowerBoundFE5}, and \eqref{eq:lowerBoundFE8} show  
\begin{equation} 
F(\beta,N) \ge \Tr [ \widetilde{\cH}_N G ] - \frac{1}{\beta} S(G) + \frac{\hat v(0) N}{2} - C \left( N^{1/2} + N^{1/3} \delta^{-1} \right)
\label{eq:free-energy-lb-Step1}
\end{equation}
with the simplified Hamiltonian 
\begin{equation}
	\widetilde{\mathcal{H}}_N =  \de \Upsilon\big(-\Delta\big)+ \frac{1}{2N} \sum_{p \in \Lambda^*_+} \hat{v}(p) \left\{ 2 a_p^* a_0^* a_p a_0 + a_0^* a_0^* a_p a_{-p} + a_p^* a_{-p}^* a_0 a_0 \right\} 
	 + (1-\delta) \frac{\hat{v}(0)}{2N} (\cN_0- N^{\mathrm{G}}_0)^2.  
	\label{eq:decompHlambda-lb-simplified}
\end{equation}
In the next section we analize the first two terms on the right-hand side of \eqref{eq:free-energy-lb-Step1} with a c-number substitution in the spirit of \cite{DeuSei-20,LieSeiYng-05}. 

\subsection{C-number substitution}
\label{sec:cNumberSubstituion}
In the following we briefly introduce the c-number substitution, which also allows us to set the notation. We start by recalling the resolution of identity
\begin{equation}
	\int_{\mathbb{C}} | z \rangle \langle z | \de z = \mathds{1}_{\mathscr{F}_0} 
	\label{eq:resolutionofidentity}
\end{equation}
on the Fock space $\mathscr{F}_0$ over the $p=0$ mode with $|z \rangle $ in \eqref{eq:coherentstate} and the measure $\de z$ below \eqref{BogoliubovHamiltonianAndi}. Given any state $\Gamma \in \mathcal{S}_N$ with $\mathcal{S}_N$ in \eqref{eq:states}, we define the operator $\widetilde{\Gamma}_z$ acting on the excitation Fock space $\mathscr{F}_+$ in \eqref{eq:excitationFockSpace} by
\begin{equation}
	\widetilde{\Gamma}_z = \tr[ | z \rangle \langle z | \Gamma ] = \langle z, \Gamma z \rangle.
    \label{eq:cnumber1}
\end{equation}
We also denote 
\begin{equation}
	\zeta_{\Gamma}(z) = \tr_+[ \widetilde{ \Gamma }_z ],
    \label{eq:cnumber2}
\end{equation}
where $\tr_+$ denotes the trace over $\mathscr{F}_+$. 
Here we used the isometry \eqref{eq:exponentialProperty} to identify states on $\mathscr{F}$ and states on $\mathscr{F}_0 \otimes \mathscr{F}_+$. 
Since $\Gamma$ is a state, $\zeta_{\Gamma}$ defines a probability measure on $\mathbb{C}$. By $S(\zeta_{\Gamma})$ we denote the entropy of the classical probability distribution $\zeta_{\Gamma}$, that is,
\begin{equation}
	S(\zeta) = - \int_{\mathbb{C}} \zeta_{\Gamma}(z) \ln\left( \zeta_{\Gamma}(z) \right) \de z.
    \label{eq:cnumber3}
\end{equation} 
We also define the state 
\begin{equation}
	\Gamma_z = \frac{\widetilde{\Gamma}_z}{\tr_+[ \widetilde{\Gamma}_z ]}
 \label{eq:cnumber4}
\end{equation}
on $\mathscr{F}_+$. The following lemma, whose proof can be found in \cite[Lemma~3.2]{DeuSei-20}, provides us with an upper bound for the entropy of $\Gamma$ in terms of the ones of $\Gamma_z$ and $\zeta_{\Gamma}$.
\begin{lemma}
	\label{lem:entropyinequality}
	Let $\Gamma$ be a state on $\mathscr{F}$. The entropy of $\Gamma$ is bounded in the following way:
	\begin{equation}
		S(\Gamma) \leq \int_{\mathbb{C}} S(\Gamma_z) \zeta_{\Gamma}(z) \de z + S(\zeta_{\Gamma}).
    \label{eq:cnumber5}
	\end{equation}
\end{lemma}
The above Lemma allows us to replace the entropy of $\Gamma$ in the Gibbs free energy functional by the ones of $\Gamma_z$ and $\zeta_{\Gamma}$ for a lower bound. In order to express also the energy in terms of the ones of $\Gamma_z$ and $\zeta_{\Gamma}$, we introduce the upper and the lower symbol $\mathcal{H}^{\mathrm{s}}$ and $\mathcal{H}_{\mathrm{s}}$ related to a general Hamiltonian $\mathcal{H}$. 
They are defined by the relations
\begin{equation}
	\mathcal{H} = \int_{\mathbb{C}} \mathcal{H}^{\mathrm{s}}(z) | z \rangle \langle z | \de z \quad \text{ and } \quad 
	\mathcal{H}_{\mathrm{s}}(z) = \langle z, \mathcal{H} z \rangle, \label{eq:cnumber6}
\end{equation}
respectively. More information on the upper and the lower symbol can be found e.g. in \cite{LieSeiYng-05}. Using the defining equation for the upper symbol we can write the expectation of the energy with respect to a state $\Gamma \in \mathcal{S}_N$ as
\begin{equation}
	\tr[\mathcal{H} \Gamma] = \int_{\mathbb{C}} \tr[ \mathcal{H}^{\mathrm{s}}(z) | z \rangle \langle z | \Gamma ] \de z = \int_{\mathbb{C}} \tr[ \mathcal{H}^{\mathrm{s}}(z) \Gamma_z ] \zeta_{\Gamma}(z) \de z.
	\label{eq:cnumber7}
\end{equation}
However, the upper symbol has the disadvantage that it is not necessarily nonnegative even if $\cH$ is nonnegative. We therefore prefer to work with the lower symbol and replace $\mathcal{H}^{\mathrm{s}}(z)$ by $\mathcal{H}_{\mathrm{s}}(z)$ in \eqref{eq:cnumber7}. 

We make the choice $\mathcal{H} = \widetilde{\mathcal{H}}_N$ with $\widetilde{\mathcal{H}}_N$ in \eqref{eq:decompHlambda-lb-simplified} and denote the upper and the lower symbols of $\widetilde{\mathcal{H}}_N$ by $\widetilde{\mathcal{H}}^{\mathrm{s}}$ and $\widetilde{\mathcal{H}}_{\mathrm{s}}$, respectively. A short computation shows that the difference between the upper and the lower symbols $\Delta \mathcal{H}(z) = \widetilde{\mathcal{H}}_{\mathrm{s}}(z) - \widetilde{\mathcal{H}}^{\mathrm{s}}(z)$ is given by
\begin{align}
	\Delta \mathcal{H}(z) =  -\frac{1}{N} \sum_{p \in \Lambda_+^*} \hat{v}(p) + \frac{(1-\delta) \hat{v}(0)}{2N} \left( 4 |z|^2 - 1 - 2 N_0^{\mathrm{G}} \right) \leq \frac{2(1-\delta) \hat{v}(0) |z|^2}{N}.
    \label{eq:cnumber8}
\end{align}
To quantify the cost of replacing $\widetilde{\mathcal{H}}^{\mathrm{s}}$ by $\widetilde{\mathcal{H}}_{\mathrm{s}}$, we estimate 
\begin{equation}
    \int_{\mathbb{C}} |z|^2 \zeta_{\Gamma}(z) \de z = \int_{\mathbb{C}} \Tr[ a_0   |z \rangle \langle z | a_0^* \Gamma ]  \de z = \Tr[ (a^*_0 a_0 + 1) \Gamma ] \leq N+1,
    \label{eq:cnumber9}
\end{equation}
which implies
\begin{equation}
	\int_{\mathbb{C}} \tr_+ \left[ \Delta \mathcal{H}(z) \Gamma_z \right] \zeta_{\Gamma}(z) \de z \leq 4 \hat{v}(0)
	\label{eq:cnumber10}
\end{equation}
provided $N \geq 1$.

When we put Lemma~\ref{lem:entropyinequality}, \eqref{eq:cnumber7}, and \eqref{eq:cnumber10} together, and also add and subtract the chemical potential $\mu_0(\beta,N)$ of the ideal gas times $\mathcal{N}_+$ in $\widetilde{\mathcal{H}}_{\mathrm{s}}$, we find
\begin{equation}
    \Tr [ \widetilde{\cH}_N G ] - \frac{1}{\beta} S(G) \geq \mu_0(\beta,N) N_+^{\mathrm{G}} + \int_{\mathbb{C}} \left\{ \tr_+ \left[ \mathcal{H}_{\mathrm{s}}^{\mu_0}(z) G_z \right] - \frac{1}{\beta} S(G_z) \right\} \zeta_G(z) \de z - \frac{1}{\beta}S(\zeta_G) - 4\hat{v}(0)
    \label{eq:cnumber11}
\end{equation}
with
\begin{align} 
    \mathcal{H}_{\mathrm{s}}^{\mu_0}(z) =& \de \Upsilon\big(-Q(\Delta +\mu_0(\beta,N)) \big)+ \frac{1}{2N} \sum_{p \in \Lambda^*_+} \hat{v}(p) \left\{ 2 |z|^2 a_p^* a_p  + \overline{z}^2 a_p a_{-p} + z^2 a_p^* a_{-p}^*   \right\} \nonumber \\
	&+ (1-\delta) \frac{\hat{v}(0)}{2N} ( |z|^2 - N_0^{\mathrm{G}})^2.
    \label{eq:Bog-Hamil-z}
\end{align}
In the following we denote the sum of the first two terms on the right-hand side by $\mathcal{H}^{\mathrm{Bog}}(z)$. 
\subsection{Final lower bound}
\label{sec:finalLowerBound}
The results from the previous two sections, that is, \eqref{eq:free-energy-lb-Step1} and \eqref{eq:cnumber11} show the lower bound
\begin{align}
    F(\beta,N) \geq& \frac{\hat{v}(0) N}{2} + \mu_0(\beta,N) N_+^{\mathrm{G}} + \int_{\mathbb{C}} \left\{ \tr_+ \left[ \mathcal{H}_{\mathrm{s}}^{\mu_0}(z) G_z \right] - \frac{1}{\beta} S(G_z) \right\} \zeta_G(z) \de z - \frac{1}{\beta}S(\zeta_G) \nonumber \\
    &- C \left( N^{1/2} + N^{1/3} \delta^{-1} \right)
    \label{eq:finalLowerBound1}
\end{align}
for the free energy.

By the Gibbs variational principle we have
\begin{align}
    \tr_+[ \mathcal{H}^{\mathrm{Bog}}(z) G_z ] - \frac{1}{\beta} S(G_z) \geq -\frac{1}{\beta} \ln\left( \tr_+ \exp(-\beta \mathcal{H}^{\mathrm{Bog}}(z)) \right) = \frac{1}{\beta} \sum_{p \in \Lambda_+^*} \ln \left( 1-\exp(-\beta \epsilon(z,p)) \right),    \label{eq:finalLowerBound2}
\end{align}
where $\epsilon(z,p)$ equals $\epsilon(p)$ in \eqref{eq:BogoliubovDispersion} except that $N_0$ is replaced by $|z|^2$. To obtain the second identity, we applied Lemma~\ref{lem:Diaggamma}. Next, we replace $|z|^2$ by $N_0$ on the right-hand side of \eqref{eq:finalLowerBound2}. To that end, we define the function
\begin{equation}
    g(x) = \frac{1}{\beta} \sum_{p \in \Lambda_+^*} \ln \left( 1-\exp(-\beta \epsilon(\sqrt{x},p)) \right).
    \label{eq:finalLowerBound2_1}
\end{equation}
Its first and second derivatives satisfy
\begin{equation}
    0 \leq g'(x) = \frac{1}{N} \sum_{p \in \Lambda_+^*} \frac{\hat{v}(p)}{\exp(\beta \epsilon(\sqrt{x},p))-1} \frac{p^2-\mu_0}{\epsilon(\sqrt{x},p)} \leq \frac{1}{\beta N} \sum_{p \in \Lambda_+^*} \frac{\hat{v}(p)}{p^2}
    \label{eq:finalLowerBound3a}
\end{equation}
and
\begin{align}
    0 \geq g''(x) &= - \frac{1}{N^2} \sum_{p\in \Lambda_+^*} \frac{ \hat{v}^2(p) (p^2-\mu_0)^2}{\epsilon^2(\sqrt{x},p)} \left( \frac{\beta}{4 \sinh^2\left(\frac{\beta \epsilon(\sqrt{x},p)}{2} \right)} + \frac{1}{\exp(\beta \epsilon(\sqrt{x},p)) - 1} \frac{1}{\epsilon(\sqrt{x},p)} \right) \nonumber \\
    &\geq - \frac{2}{\beta N^2} \sum_{p \in \Lambda_+^*} \frac{\hat{v}^2(p)}{|p|^4},
    \label{eq:finalLowerBound3b}
\end{align}
respectively. To obtain these bounds we used $\exp(x)-1 \geq x$ and $\sinh^2(x) \geq x^2$ for $x \geq 0$, as well as $\epsilon(\sqrt{x},p) \geq p^2-\mu_0$ and $\mu_0 < 0$. As a consequence the function
\begin{equation}
    f(x) = g(x) + \delta  \frac{\hat{v}(0)}{2N}(x - N_0^{\mathrm{G}})^2
    \label{eq:finalLowerBound5}
\end{equation}
is convex provided 
\begin{equation}
  \delta > \frac{2}{\beta N \hat{v}(0))} \sum_{p \in \Lambda_+^*} \frac{\hat{v}^2(p)}{|p|^4}
    \label{eq:finalLowerBound6}
\end{equation}
holds, which we assume from now on. An application of Jensen's inequality therefore shows
\begin{equation}
    \int_{\mathbb{C}} f(|z|^2) \zeta_G(z) \de z \geq f\left( \int_{\mathbb{C}} |z|^2 \zeta_G(z) \de z \right).
    \label{eq:finalLowerBound7}
\end{equation}
From Proposition~\ref{prop:roughapriori} and \eqref{eq:cnumber9} we know that 
\begin{equation}
    \left| \int_{\mathbb{C}} |z|^2 \zeta_G(z) \de z - N_0(\beta,N) \right| \lesssim N^{2/3} \ln(N)
    \label{eq:finalLowerBound8}
\end{equation}
with $N_0$ in \eqref{eq:crittemp}. Together with \eqref{eq:finalLowerBound3a} this implies
\begin{equation}
    f\left( \int_{\mathbb{C}} |z|^2 \zeta_G(z) \de z \right) \geq g(N_0) - C N^{1/3} \ln(N). 
    \label{eq:finalLowerBound9}
\end{equation}
When we put the above considerations together, we obtain 
\begin{align}
    \int_{\mathbb{C}} \left\{ \tr_+[ \mathcal{H}^{\mathrm{Bog}}(z) G_z ] - \frac{1}{\beta} S(G_z) \right\} \zeta_G(z) \de z \geq& \frac{1}{\beta} \sum_{p \in \Lambda_+^*} \ln \left( 1-\exp(-\beta \epsilon(p)) \right) \nonumber \\
    & + (1- 2\delta)  \frac{\hat{v}(0)}{2N}  \int_{\mathbb{C}}\left( |z|^2 - N_0^{\mathrm{G}} \right)^2 \zeta_G(z) \de z - C N^{1/3} \ln(N) \label{eq:finalLowerBound10}
\end{align}
with $\epsilon(p)$ in \eqref{eq:BogoliubovDispersion}. 

Let us consider now the term $\mu_0(\beta,N) N_+^{\mathrm{G}}$. An application of Proposition~\ref{prop:roughapriori} shows 
\begin{equation}
    \mu_0(\beta,N) N_+^{\mathrm{G}} \geq \mu_0 (N-N_0(\beta,N)) - C |\mu_0(\beta,N)| N^{2/3} \ln(N).
\end{equation}
Moreover, from \eqref{eq:preplemmas31} and \eqref{eq:preplemma32}--\eqref{eq:preplemma33b} we know that $\big| \sum_{p\in\Lambda_+^*} \gamma_p - (N-N_0(\beta,N)) \big| \lesssim N^{2/3}$, and hence
\begin{equation}
    \mu_0(\beta,N) N_+^{\mathrm{G}} \geq \mu_0(\beta,N) \sum_{p\in\Lambda_+^*} \gamma_p - C |\mu_0(\beta,N)| N^{2/3} \ln(N).
    \label{eq:finalLowerBound11}
\end{equation}

It remains to consider the terms related to the condensate, which read 
\begin{equation}
    \frac{(1-2\delta) \hat{v}(0)}{2N} \int_{\mathbb{C}} ( |z|^2 - N_0^{\mathrm{G}} )^2 \zeta_{G} \de z - \frac{1}{\beta} S(\zeta_G) \geq \tilde{F}_{\mathrm{c}}^{\mathrm{BEC}}(\beta,N_0^{\mathrm{G}}).
    \label{eq:finalLowerBound12}
\end{equation}
Here $\tilde{F}_{\mathrm{c}}^{\mathrm{BEC}}(\beta,N_0^{\mathrm{G}})$ equals $F_{\mathrm{c}}^{\mathrm{BEC}}(\beta,N_0^{\mathrm{G}})$ except that $\hat{v}(0)$ is replaced by $(1-2\delta)\hat{v}(0)$. Using parts (c) and (d) of Lemma~\ref{prop:FreeEnergyBEC}, we see that 
$$\tilde{F}_{\mathrm{c}}^{\mathrm{BEC}}(\beta,N_0^{\mathrm{G}}) \geq F_{\mathrm{c}}^{\mathrm{BEC}}(\beta,N_0^{\mathrm{G}}) - C \delta N^{2/3}.$$
Moreover, an application of Lemma~\ref{lem:replaceParticleNumberInCondensateEnergy} show 
$$F_{\mathrm{c}}^{\mathrm{BEC}}(\beta,N_0^{\mathrm{G}}) \geq F_{\mathrm{c}}^{\mathrm{BEC}}(\beta,N_0(\beta,N)) - C f^{\mathrm{BEC}}(N_0(\beta,N),N_0^{\mathrm{G}})$$
 with $f^{\mathrm{BEC}}(N_0(\beta,N),N_0^{\mathrm{G}})$ in \eqref{eq:perturbationTheoryChemicalPotential1b}. Finally, an application of Lemma~\ref{lem:comparisonContinuousDiscreteCondensateFreeEnergy} shows 
 $$F_{\mathrm{c}}^{\mathrm{BEC}}(\beta,N_0) \geq F^{\mathrm{BEC}}(\beta,N_0) - C N^{1/3}.$$
 Putting these considerations together, we obtain
\begin{equation}
    \frac{(1-2\delta) \hat{v}(0)}{2N} \int_{\mathbb{C}} ( |z|^2 - N_0^{\mathrm{G}} )^2 \zeta_{G} \de z - \frac{1}{\beta} S(\zeta_G) \geq F^{\mathrm{BEC}}(\beta,N_0(\beta,N)) - C \left(N^{1/3} + \delta N^{2/3} + f^{\mathrm{BEC}}(N_0(\beta,N),N_0^{\mathrm{G}}) \right).
    \label{eq:finalLowerBound13}
\end{equation}

Let us collect the above estimates. We insert \eqref{eq:finalLowerBound10}, \eqref{eq:finalLowerBound11}, and \eqref{eq:finalLowerBound13} into \eqref{eq:finalLowerBound1}, and additionally use \eqref{eq:BogoliubovFreeEnergy}, which gives 
\begin{align}
    F(\beta,N) \geq& F^{\mathrm{Bog}}(\beta,N) + \frac{\hat{v}(0)N}{2} + F^{\mathrm{BEC}}(\beta,N_0(\beta,N)) \nonumber \\ 
    &- C \left( N^{1/2} + |\mu_0| N^{2/3} \ln(N) + \delta N^{2/3} + N^{1/3} \delta^{-1} + f^{\mathrm{BEC}}(N_0(\beta,N),N_0^{\mathrm{G}}) \right). 
    \label{eq:finalLowerBound14}
\end{align}
During the derivation of this result we assumed that $\delta$ satisfies \eqref{eq:finalLowerBound6}. The optimal choice for $\delta$ is $\delta = N^{-1/6}$, which results in an error term that is bounded by a constant times $N^{1/2}$. Proposition~\ref{sec:proofFirstOrderApriori1} allows us to obtain a bound for $f^{\mathrm{BEC}}(N_0(\beta,N),N_0^{\mathrm{G}})$. The above bound yields the desired result as long as $N_0(\beta,N) \geq N^{2/3+\epsilon}$ with $\epsilon > 0$ which ensures that $|\mu_0| \lesim  N^{-\eps}$. We will therefore combine it with another bound that we derive now.

Using \eqref{eq:firstOrderLowerBound2} and $\tr[\mathcal{N}^2 \Gamma] \geq (\tr[\mathcal{N} \Gamma])^2$, we check that
\begin{align}
    F(\beta,N) \geq& \inf_{\Gamma \in \mathcal{S}_N} \left\{ \Tr[\de \Upsilon(-\Delta) \Gamma] - \frac{1}{\beta}S(\Gamma) \right\} + \frac{\hat{v}(0)N}{2} - \frac{v(0)}{2} \nonumber \\
    \geq& F_0(\beta,N) + \frac{\hat{v}(0)N}{2} - \frac{v(0)}{2}
    \label{eq:finalLowerBound15}
\end{align}
with $F_0(\beta,N)$ in \eqref{eq:freeEnergyIdealGasBEC}. The terms on the right-hand side of \eqref{eq:finalLowerBound15} need to be bounded in terms of those on the right-hand side of \eqref{eq:finalLowerBound14}. We start by noting that
\begin{align}
    \frac{1}{\beta} \sum_{p \in \Lambda_+^*} \ln\left( 1 -\exp(-\beta \epsilon(p)) \right) \leq& \frac{1}{\beta} \sum_{p \in \Lambda_+^*} \ln\left( 1 -\exp(-\beta (p^2-\mu_0))) \right) \nonumber \\
    &+ \sum_{p \in \Lambda_+^*} \frac{1}{\exp(-\beta(p^2-\mu_0))-1} \left( \sqrt{1 + 2 (N_0/N) \frac{\hat{v}(p)}{p^2-\mu_0}} \right) \nonumber \\
    \leq& \frac{1}{\beta} \sum_{p \in \Lambda_+^*} \ln\left( 1 -\exp(-\beta (p^2-\mu_0))) \right) + \frac{C N_0}{\beta N} \sum_{p \in \Lambda_+^*} \frac{\hat{v}(p)}{p^4}.
    \label{eq:finalLowerBound16}
\end{align}
Moreover, an application of Lemma~\ref{lem:BoundN0} shows
\begin{equation}
    \mu_0(\beta,N) (N-N_0(\beta,N)) \geq \mu_0(\beta,N) \sum_{p \in \Lambda_+^*} \gamma_p - C |\mu_0(\beta,N)| \left[ \left( 1 + \frac{1}{\beta} \right) \frac{N_0^2(\beta,N)}{N^2} + \frac{N_0(\beta,N)}{\beta N} \right]
\end{equation}
with $\gamma_p$ in \eqref{eq:1pdm}.

Next, we insert the trial state 
\begin{equation}
    p(n) = \exp(\beta \mu_0(\beta,N) n) (1-\exp(\beta \mu_0(\beta,N)))
    \label{eq:finalLowerBound16_0}
\end{equation}
into \eqref{eq:condensateFreeEnergy} and find
\begin{align}
    F^{\mathrm{BEC}}(\beta,N_0(\beta,N)) \leq& \frac{\hat{v}(0)}{2N} \sum_{n=0}^{\infty} n^2 p(n) - \frac{1}{\beta} S(p) - \frac{N_0^2(\beta,N) \hat{v}(0)}{2N} \nonumber \\
    =& \frac{1}{\beta} \ln(1-\exp(\beta \mu_0(\beta,N))) + \mu_0(\beta,N) N_0(\beta,N) + \frac{\hat{v}(0)}{2N} \frac{1}{4 \sinh\left( \frac{-\beta \mu_0(\beta,N)}{2} \right)} \nonumber \\
    \leq& F_0^{\mathrm{BEC}}(\beta,N) + C \frac{N^2_0(\beta,N)}{N}
    \label{eq:finalLowerBound17}
\end{align}
with $F_0^{\mathrm{BEC}}$ in \eqref{eq:freeEnergyIdealGasBEC}.

Putting these considerations and $|\mu_0| \lesssim (\beta N_0)^{-1}$ together, we obtain the lower bound
\begin{equation}
    F(\beta,N) \geq F^{\mathrm{Bog}}(\beta,N) + \frac{\hat{v}(0)N}{2} + F^{\mathrm{BEC}}(\beta,N_0(\beta,N)) - C \left[ N^{1/3} + N^{-1/3} N_0(\beta,N) + \frac{N_0^2(\beta,N)}{N} \right].
    \label{eq:finalLowerBound18}
\end{equation}
We apply \eqref{eq:finalLowerBound14} if $N_0(\beta,N) \geq N^{19/24}$ and \eqref{eq:finalLowerBound18} if $N_0(\beta,N) < N^{19/24}$, which yields the final lower bound
\begin{equation}
    F(\beta,N) \geq F^{\mathrm{Bog}}(\beta,N) + \frac{\hat{v}(0)N}{2} + F^{\mathrm{BEC}}(\beta,N_0(\beta,N)) - C N^{5/8} \ln(N)
    \label{eq:finalLowerBound19}
\end{equation}
for the free energy. In combination with the matching upper bound in \eqref{eq:upperBoundFreeEnergy}, this proves Theorem~\ref{thm:main1} with an additional logarithmic factor in the remainder term. In Remark~\ref{rem:removeLog} below we explain how this factor can be removed.

Corollary~\ref{cor:main1} is a direct consequence of Theorem~\ref{thm:main1} and Proposition~\ref{prof:freeEnergyEffectiveCondensateTheoryMain}. Part~(a) of Proposition~\ref{prof:freeEnergyEffectiveCondensateTheoryMain} follows from part~(a) of Lemma~\ref{prop:FreeEnergyBEC} and Lemma~\ref{lem:comparisonContinuousDiscreteCondensateFreeEnergy}. Part~(b) of the proposition follows from \eqref{eq:finalLowerBound17} and the lower bound
\begin{equation}
    F^{\mathrm{BEC}}(\beta,N_0(\beta,N)) \geq F_0^{\mathrm{BEC}}(\beta,N) - \frac{\hat{v}(0) N^2_0}{2N},
    \label{eq:finalLowerBound20}
\end{equation}
which follows from its variational characterization in \eqref{eq:condensateFreeEnergy} and $\hat{v}(0) \geq 0$.

\section{The Gibbs state part III: trace norm approximation, 1-pdm, and condensate distributions} \label{sec:Gibbs-state-III}
In this section we establish the trace-norm approximation of the Gibbs state in \eqref{eq:interactingGibbsstate}. Using this result we prove pointwise bounds for its 1-pdm and we study the asymptotic behavior of two probability distributions related to the Bose--Einstein condensate. With a Griffith argument we also prove a trace-norm bound for its 1-pdm. 
\subsection{Trace norm bound for the Gibbs state}
\label{sec:traceNormBoundGibsState}
In this section we prove a bound on the trace norm distance of the Gibbs state and the state 
\begin{equation}
    \Gamma_{\beta,N} = \begin{cases}
         \int_{\mathbb{C}} |z \rangle \langle z | \otimes G^{\mathrm{Bog}}(z) g^{\mathrm{BEC}}(z) \de z & \ \text{ if } N_0(\beta,N) \geq N^{2/3}, \\
        G_{\beta,N}^{\mathrm{id}} & \ \text{ if } N_0(\beta,N) < N^{2/3}. \\
    \end{cases}
    \label{eq:referenceState}
\end{equation}
The notation appearing in the first line of the above formula has been explained below \eqref{eq:referenceState-intro1} and the Gibbs state $G_{\beta,N}^{\mathrm{id}}$ of the ideal gas has been defined in \eqref{eq:GibbsStateIdealGas}. We start by proving a bound for the relative entropy of $\Gamma_{\beta,N}$ with respect to $G_{\beta,N}$, which is defined by
\begin{equation}
    S(\Gamma_{\beta,N},G_{\beta,N}) = \Tr[ \Gamma_{\beta,N}(\ln(\Gamma_{\beta,N}) - \ln(G_{\beta,N})) ].
    \label{eq:relativeEntropy}
\end{equation}
A short computation shows $\beta^{-1}S(\Gamma_{\beta,N},G_{\beta,N}) = \mathcal{F}(\Gamma_{\beta,N}) - \mathcal{F}(G_{\beta,N}) \geq 0$. With the upper bounds for the free energy of $\Gamma_{\beta,N}$ in \eqref{eq:UpperBoundAndi10} and \eqref{eq:UpperBoundAndi16} and the lower bound for that of $G_{\beta,N}$ in \eqref{eq:finalLowerBound19}, we find
\begin{equation}
    S(\Gamma_{\beta,N},G_{\beta,N}) \lesssim \beta N^{5/8} \ln(N) \lesssim N^{-1/24} \ln(N).
    \label{eq:relativeEntropyBound}
\end{equation}
An application of Pinsker's inequality, 
see e.g. \cite[Theorem~1.15]{OhyPer-93}, allows us to conclude that
\begin{equation}
    \left\Vert \Gamma_{\beta,N} - G_{\beta,N} \right\Vert_1^2 \leq 2 S(\Gamma_{\beta,N},G_{\beta,N}) \lesssim  N^{-1/24} \ln(N)
    \label{eq:traceNormBound}
\end{equation}
holds. This proves the trace norm approximation for the Gibbs state in Theorem~\ref{thm:norm-approximation} with an additional logarithmic factor in the rate. In Remark~\ref{rem:removeLog} below we explain how this factor can be removed. 
\subsection{Trace norm bound for the 1-pdm}
\label{sec:1pdm}
 In this section, we apply a Griffiths argument to derive bounds for the 1-pdm of the Gibbs state in \eqref{eq:interactingGibbsstate}. The main reason we do not rely on an approach based on relative entropy bounds, as in Section~\ref{sec:condensate Fraction} and \cite{DeuSeiYng-19,DeuSei-20,LewNamRou-21}, is that the state $\Gamma_{\beta,N}$ is not quasi-free when $N_0 \geq N^{2/3}$. However, the quasi-freeness of the state approximating the Gibbs state plays a crucial role in this approach.

The proofs of the upper and lower bounds for the free energy in Sections \ref{sec:upperBoundFreeEnergy} and \ref{sec:lowerBound} apply in the same way if we replace $-\Delta$ by a general translation-invariant one-particle Hamiltonian $h$ that satisfies \eqref{eq:generalizedOneParticleHamiltonian}. In the following, we use this bound with the operator
\begin{equation}
    h = \sum_{p\in \Lambda^*} h(p) | \varphi_p \rangle \langle \varphi_p | \quad \text{ with } \quad h(p) = p^2 + \lambda f(p).
    \label{eq:bound1pdm1}
\end{equation}
Here $\varphi_p(x) = e^{\mathrm{i} p \cdot x}$ and $f : \Lambda^* \to \mathbb{R}$ is a function that satisfies $f(0) = 0$ and $\sup_{p \in \Lambda^*} | f(p ) | \leq 1$. The absolute value of the parameter $\lambda \in \mathbb{R}$ is chosen small enough such that $h$ satisfies \eqref{eq:generalizedOneParticleHamiltonian}. The result for the free energy reads
\begin{equation}
    F(\beta,N,\lambda) = \tilde{F}^{\mathrm{Bog}}(\beta,N,\lambda) + \frac{\hat{v}(0) N}{2} + F^{\mathrm{BEC}}(\beta,N_0(\beta,N,\lambda)) + O\left( N^{5/8} \ln(N) \right),
    \label{eq:FreeEnergyBoundWithPerturbation}
\end{equation}
where $F(\beta,N,\lambda)$ and $\tilde{F}^{\mathrm{Bog}}(\beta,N,\lambda)$ denote the free energies in \eqref{eq:freeenergy} and \eqref{eq:BogoliubovFreeEnergy} with $-\Delta$ replaced by $h$, respectively. By $N_0(\beta,N,\lambda)$ we denoted the expected particle number in the condensate of the ideal gas related to $h$. We highlight that the operator $h$ appears in the definition of $\tilde{F}^{\mathrm{Bog}}(\beta,N,\lambda)$ in two ways. First, we need to replace $p^2$ by $h(p)$ in the definitions of $\epsilon(p)$ in \eqref{eq:BogoliubovDispersion} and $u_p$, $v_p$ in \eqref{eq:coefficientsBogtrafoAndi}. Second, we also need to replace $N_0(\beta,N)$ by $N_0(\beta,N,\lambda)$ and $\mu_0(\beta,N)$ by $\mu_0(\beta,N,\lambda)=-\beta^{-1} \ln(1+N_0^{-1}(\beta,N,\lambda))$ in these quantities. In the following we consider two parameter regimes separately. We start with the parameter regime, where $N_0 \geq N^{19/24}$ holds. The case $N_0 < N^{19/24}$ will be discussed afterwards.
\subsubsection{The condensed phase}
In this section we assume $N_0 \geq N^{19/24}$. In the first step of our analysis we replace $N_0(\beta,N,\lambda)$ by $N_0(\beta,N)$ and $\mu_0(\beta,N,\lambda)$ by $\mu_0(\beta,N)$ on the right-hand side of \eqref{eq:FreeEnergyBoundWithPerturbation} in all places, where they appear. From Lemma~\ref{lem:idealGasWithVariable1ParticleHamiltonian} in Appendix~\ref{app:idealGas}, we know that $|N_0(\beta,N,\lambda) - N_0(\beta,N) | \lesssim |\lambda| \ \Vert f \Vert_{\infty}/\beta \lesssim N^{2/3}$. Using this and Lemmas~\ref{lem:replaceParticleNumberInCondensateEnergy} and \ref{lem:comparisonContinuousDiscreteCondensateFreeEnergy} in Appendix~\ref{app:effcondensate}, we check that
\begin{equation}
    | F^{\mathrm{BEC}}(\beta,N_0(\beta,N,\lambda)) - F^{\mathrm{BEC}}(\beta,N_0(\beta,N)) | \lesssim N^{5/8} 
    \label{eq:bound1pdm2_0}
\end{equation}
holds. With the bound for $N_0(\beta,N,\lambda)$ we can also replace $N_0(\beta,N,\lambda)$ by $N_0(\beta,N)$ in $F^{\mathrm{Bog}}(\beta,N,\lambda)$. A straightforward computation that we leave to the reader shows that this replacements gives rise to an error that is bounded by a constant times $N^{11/24}$. 

Next, we replace $\mu_0(\beta,N,\lambda)$ by $\mu_0(\beta,N)$ in $\tilde{F}^{\mathrm{Bog}}(\beta,N,\lambda)$. A second order Taylor expansion allows us to write
\begin{align}
    &-\frac{1}{\beta} \ln\left( \tr_{\mathscr{F}_+} \exp\left( -\beta \left( \mathcal{H}^{\mathrm{Bog}}_{\mu_0(\beta,N,\lambda)} + \lambda \de \Upsilon(f) \right) \right) \right) = -\frac{1}{\beta} \ln\left( \tr_{\mathscr{F}_+} \exp\left( -\beta \left( \mathcal{H}^{\mathrm{Bog}} + \lambda \de \Upsilon(f) \right) \right) \right) \label{eq:bound1pdm2b} \\
    &\hspace{1cm}  + (\mu_0(\beta,N) - \mu_0(\beta,N,\lambda)) \sum_{p \in \Lambda_+^*} \gamma_p(\lambda,\mu_0(\beta,N))  + \frac{(\mu_0(\beta,N) - \mu_0(\beta,N,\lambda))^2}{2}  \sum_{p \in \Lambda_+^*} \left[ \frac{\partial \gamma_p(\lambda,\mu))}{\partial \mu} \right]_{\mu = \sigma} \nonumber
\end{align}
with some $\sigma \in \{ (1-t) \mu_0(\beta,N,\lambda) + t \mu_0(\beta,N) \ | \ t \in [0,1] \}$. By $\mathcal{H}^{\mathrm{Bog}}_{\mu_0(\beta,N,\lambda)}$ we denoted the Bogoliubov Hamiltonian in \eqref{BogoliubovHamiltonianAndi} with $\mu_0(\beta,N)$ replaced by $\mu_0(\beta,N,\lambda)$ and $\de \Upsilon(f)$ denotes the second quantization of multiplication with the function $f(p)$ in Fourier space. For the sake of readability, we also included the dependence of $\gamma_p$ on $\lambda$ and the chemical potential in our notation. It is not difficult to see that $\mu_0(\beta,N)$ and $\mu_0(\beta,N,\lambda)$ share the same leading order asymptotics as $N \to \infty$. We use this, $\mu_0 = -\beta^{-1} \ln(1+N_0^{-1})$, $N_0 \geq N^{19/24}$, and Lemma~\ref{lem:perturbedBogoliubov} in Appendix~\ref{app:perturbedBogoliubov} to check that the absolute value of the second term in the second line of \eqref{eq:bound1pdm2b} is bounded by a constant times $N^{5/12}$. This also shows
\begin{align}
    -\frac{1}{\beta} \ln&\left( \tr_{\mathscr{F}_+} \exp\left( -\beta \left( \mathcal{H}^{\mathrm{Bog}}_{\mu_0(\beta,N,\lambda)} + \lambda \de \Upsilon(f) \right) \right) \right) + \mu_0(\beta,N,\lambda) \sum_{p \in \Lambda_+^*} \gamma_p(\lambda,\mu_0(\beta,N,\lambda)) \nonumber \\
    =& -\frac{1}{\beta} \ln\left( \tr_{\mathscr{F}_+} \exp\left( -\beta \left( \mathcal{H}^{\mathrm{Bog}}_{\mu_0(\beta,N)} + \lambda \de \Upsilon(f) \right) \right) \right) + \mu_0(\beta,N) \sum_{p \in \Lambda_+^*} \gamma_p(\lambda,\mu_0(\beta,N)) \nonumber \\
    &+ \mu_0(\beta,N,\lambda) \sum_{p \in \Lambda_+^*} \left( \gamma_p(\lambda,\mu_0(\beta,N,\lambda)) - \gamma_p(\lambda,\mu_0(\beta,N) \right) + O\left(N^{5/12} \right). \label{eq:bound1pdm2c}
\end{align}
With a first order Taylor expansion, we check that the absolute value of the term in the last line is bounded from above by
\begin{equation}
    \left| \mu_0(\beta,N,\lambda) \sum_{p \in \Lambda_+^*} \left[ \frac{\partial \gamma_p(\lambda,\mu))}{\partial \mu} \right]_{\mu = \sigma} \right| \lesssim \frac{|\lambda| \ \Vert f \Vert_{\infty}}{\beta^2 N_0(\beta,N)} \lesssim N^{13/24}.
    \label{eq:bound1pdm2d}
\end{equation}
Here $\sigma \in \{ (1-t) \mu_0(\beta,N,\lambda) + t \mu_0(\beta,N) \ | \ t \in [0,1] \}$ and we used Lemma~\ref{lem:perturbedBogoliubov} to obtain the second bound. In combination, \eqref{eq:bound1pdm2c} and \eqref{eq:bound1pdm2d} prove
\begin{equation}
    F(\beta,N,\lambda) = F^{\mathrm{Bog}}(\beta,N,\lambda) + \frac{\hat{v}(0) N}{2} + F^{\mathrm{BEC}}(\beta,N_0(\beta,N)) + O \left( N^{5/8} \ln(N) \right).
    \label{eq:bound1pdm2e}
\end{equation}
Here $F^{\mathrm{Bog}}(\beta,N,\lambda)$ denotes the Bogoliubov free energy in \eqref{eq:BogoliubovFreeEnergy} with $p^2$ replaced by $h(p)$. In contrast to $\tilde{F}^{\mathrm{Bog}}(\beta,N,\lambda)$, this free energy depends on $N_0(\beta,N)$ and $\mu_0(\beta,N)$ instead of $N_0(\beta,N,\lambda)$ and $\mu_0(\beta,N,\lambda)$.

From \eqref{eq:bound1pdm1} and \eqref{eq:bound1pdm2e} we know that
\begin{equation}
    \lambda \tr[ \de \Upsilon(f) G_{\beta,N}] + \mathcal{F}(G_{\beta,N}) \geq F^{\mathrm{Bog}}(\beta,N,\lambda) + \frac{\hat{v}(0) N}{2} + F^{\mathrm{BEC}}(\beta,N_0(\beta,N)) + O\left( N^{5/8} \ln(N) \right).
    \label{eq:bound1pdm2}
\end{equation}
In combination with the upper bound for the free energy in \eqref{eq:bound1pdm2e} with $\lambda = 0$, this implies
\begin{equation}
    \tr[ \de \Upsilon(f) G_{\beta,N}] \geq \frac{F^{\mathrm{Bog}}(\beta,N,\lambda) - F^{\mathrm{Bog}}(\beta,N,0)}{\lambda} - \frac{C N^{5/8} \ln(N)}{\lambda}
    \label{eq:bound1pdm3}
\end{equation}
if $\lambda > 0$ and
\begin{equation}
    \tr[ \de \Upsilon(f) G_{\beta,N}] \leq \frac{F^{\mathrm{Bog}}(\beta,N,\lambda) - F^{\mathrm{Bog}}(\beta,N,0)}{\lambda} + \frac{C N^{5/8} \ln(N)}{|\lambda|}
    \label{eq:bound1pdm4}
\end{equation}
if $\lambda < 0$.

In the following we have a closer look at the first term on the right-hand side of \eqref{eq:bound1pdm3} and \eqref{eq:bound1pdm4}. A second order Taylor expansion allows us to write
\begin{equation}
    F^{\mathrm{Bog}}(\beta,N,\lambda) = F^{\mathrm{Bog}}(\beta,N,0) + \frac{\partial F^{\mathrm{Bog}}(\beta,N,\lambda)}{\partial \lambda} \bigg|_{\lambda = 0} \lambda + \frac{1}{2} \frac{ \partial^2 F^{\mathrm{Bog}}(\beta,N,\lambda)}{\partial \lambda^2} \bigg|_{\lambda = \widetilde{\lambda}} \lambda^2
    \label{eq:bound1pdm5}
\end{equation}
with some $\widetilde{\lambda} \in \{ t \lambda \ | \ t \in [0,1] \}$. The first derivative in the above equation reads
\begin{equation}
    \frac{\partial F^{\mathrm{Bog}}(\beta,N,\lambda)}{\partial \lambda} \bigg|_{\lambda = 0} = \sum_{p \in \Lambda_+^*} f(p) \gamma_p \bigg|_{\lambda = 0} + \mu_0(\beta,N) \sum_{p \in \Lambda_+^*} \frac{\partial \gamma_p}{\partial \lambda} \bigg|_{\lambda = 0}.
    \label{eq:bound1pdm6}
\end{equation}
Using $N_0 \geq N^{19/24}$ and Lemma~\ref{lem:perturbedBogoliubov} in Appendix~\ref{app:perturbedBogoliubov}, we check that the absolute value of the second term on the right-hand side is bounded by a constant times $N^{13/24}\ \Vert f \Vert_{\infty}$.  That is, we have
\begin{equation}
    \left| \left[ \frac{\partial F^{\mathrm{Bog}}(\beta,N,\lambda)}{\partial \lambda}  - \sum_{p \in \Lambda_+^*} f(p) \gamma_p \right]_{\lambda = 0} \right| \lesssim N^{13/24} . 
    \label{eq:bound1pdm8}
\end{equation}
The second derivative of the Bogoliubov free energy with respect to $\lambda$ reads
\begin{equation}
    \frac{ \partial^2 F^{\mathrm{Bog}}(\beta,N,\lambda)}{\partial \lambda^2} = \sum_{p \in \Lambda_+^*} \left[ f(p) \frac{\partial \gamma_p}{\partial \lambda} + \mu_0(\beta,N) \frac{\partial^2 \gamma_p}{\partial \lambda^2} \right].
    \label{eq:bound1pdm9}
\end{equation}
Another application of Lemma~\ref{lem:perturbedBogoliubov} shows
\begin{equation}
    \left| \left[ \frac{ \partial^2 F^{\mathrm{Bog}}(\beta,N,\lambda)}{\partial \lambda^2} \right]_{\lambda = \widetilde{\lambda}} \right| \lesssim \frac{\Vert f \Vert_{\infty}^2}{\beta} \left( 1 + \frac{1}{\beta N_0(\beta,N)} \right) \lesssim N^{2/3}.
    \label{eq:bound1pdm10}
\end{equation}

Putting the above considerations together, we find
\begin{equation}
\left\| Q (\gamma_{\beta,N} - \sum_{p \in \Lambda^*_+} \gamma_p | \varphi_p \rangle \langle \varphi_p | ) Q \right\|_1 = \sup_{\Vert f \Vert_{\infty} \leq 1} \left| \sum_{p \in \Lambda_+^*} f(p) ( \gamma_{\beta,N}(p) - \gamma_p ) \right| \lesssim N^{13/24} + N^{2/3} | \lambda | + \frac{N^{5/8} \ln(N)}{|\lambda|}.
    \label{eq:bound1pdm11}
\end{equation}
Here $\gamma_{\beta,N}(p)$ denote the eigenvalues of $\gamma_{N,\beta}$ with eigenfunctions $\varphi_p(x) = e^{\mathrm{i} p \cdot x}$ with $p \in \Lambda^*$. The optimal choice for $|\lambda|$ is $N^{-1/48} \sqrt{\ln(N)}$, which yields an error of the order $N^{2/3-1/48} \sqrt{\ln(N)}$. In combination, \eqref{eq:bound1pdm11} and $\trs \gamma_{\beta,N} = N$ allow us to conclude that
\begin{equation}
    \left\| \gamma_{\beta,N} - \widetilde{N}_0(\beta,N) | \varphi_0 \rangle \langle \varphi_0 | - \sum_{p \in \Lambda^*_+} \gamma_p | \varphi_p \rangle \langle \varphi_p | \ \right\|_1 \lesssim N^{2/3-1/48} \sqrt{\ln(N)}
    \label{eq:bound1pdm12}
\end{equation}
holds with $\widetilde{N}_0(\beta,N) = N - \sum_{p \in \Lambda_+^*} \gamma_p$. This bound has been derived under the assumption $N_0(\beta,N) \geq N^{19/24}$ and it remains to consider the case $N_0(\beta,N) < N^{19/24}$.

\subsubsection{The non-condensed phase}
\label{sec:1pdmNonCondensedPhase}

We use \eqref{eq:UpperBoundAndi11} and \eqref{eq:finalLowerBound15} with $p^2$ replaced by $h(p)$ and $\mu_0(\beta,N)$ replaced by $\mu_0(\beta,N,\lambda)$ to see that
\begin{equation}
    F(\beta,N,\lambda) = F_0(\beta,N,\lambda) + \frac{\hat{v}(0) N}{2} + O\left( N^{1/3} \right).
    \label{eq:bound1pdm13}
\end{equation}
Here $F_0(\beta,N,\lambda)$ denotes the free energy of the ideal gas below \eqref{eq:CorollaryNoBEC} with $p^2$ replaced by $h(p)$ and $\mu_0(\beta,N)$ replaced by $\mu_0(\beta,N,\lambda)$. With \eqref{eq:bound1pdm13} we find
\begin{equation}
    \tr[ \de \Upsilon(f) G_{\beta,N}] \geq \frac{F_0(\beta,N,\lambda) - F_0(\beta,N,0)}{\lambda} - \frac{C N^{1/3}}{\lambda}
    \label{eq:bound1pdm14}
\end{equation}
if $\lambda > 0$ and
\begin{equation}
    \tr[ \de \Upsilon(f) G_{\beta,N}] \leq \frac{F_0(\beta,N,\lambda) - F_0(\beta,N,0)}{\lambda} + \frac{C N^{1/3}}{|\lambda|}
    \label{eq:bound1pdm15}
\end{equation}
if $\lambda < 0$. 

As before, we expand $F_0(\beta,N,\lambda)$ with a Taylor expansion to second order. In combination with
\begin{equation}
    \frac{\partial F_0(\beta,N,\lambda)}{\partial \lambda} = \sum_{p \in \Lambda^*} \frac{f(p)}{\exp(\beta (p^2 + \lambda f(p) - \mu_0(\beta,N,\lambda)))-1}    
    \label{eq:bound1pdm16}
\end{equation}
and
\begin{equation}
    \frac{\partial^2 F_0(\beta,N,\lambda)}{\partial \lambda^2} = \sum_{p \in \Lambda^*} \frac{- \beta f(p)(f(p) - \partial \mu_0(\beta,N,\lambda)/\partial \lambda)}{4 \sinh^2\left( \frac{\beta(p^2 + \lambda f(p) - \mu_0(\beta,N,\lambda))}{2} \right)},
    \label{eq:bound1pdm17}
\end{equation}
this gives
\begin{equation}
    \Vert \gamma_{\beta,N} - \gamma_0(\beta,N) \Vert_1 \lesssim \sup_{\Vert f \Vert_{\infty} \leq 1} \left\{  \frac{N^{1/3}}{|\lambda|} + | \lambda| \left[ \sum_{p \in \Lambda^*} \frac{\beta \Vert f \Vert_{\infty} \left( \Vert f \Vert_{\infty} + | \partial \mu_0(\beta,N,\lambda )/\partial \lambda | \right) }{\sinh^2\left( \frac{\beta(p^2 + \lambda f(p) - \mu_0(\beta,N,\lambda))}{2} \right)} \right]_{\lambda = \widetilde{\lambda}} \right\}
    \label{eq:bound1pdm18}
\end{equation}
with some $\widetilde{\lambda} \in \{ t \lambda \ | \ t \in [0,1] \}$. From Lemma~\ref{lem:derivativesChemPot} in Appendix~\ref{app:idealGas} we know that $| \partial \mu_0(\beta,N,\lambda)/\partial \lambda | \lesssim \Vert f \Vert_{\infty}$ holds. We conclude that there exists a constant $c>0$ such that sum in the second term on the right-hand side of \eqref{eq:bound1pdm18} satisfies
\begin{equation}
    \left[ \sum_{p \in \Lambda^*} \frac{\beta \Vert f \Vert_{\infty} \left( \Vert f \Vert_{\infty} + | \partial \mu_0(\beta,N,\lambda)/\partial \lambda | \right) }{\sinh^2\left( \frac{\beta(p^2 + \lambda f(p) - \mu_0(\beta,N,\lambda))}{2} \right)} \right]_{\lambda = \widetilde{\lambda}} \lesssim  \sum_{p \in \Lambda^*} \frac{\Vert f \Vert_{\infty}^2}{\beta(c p^2 - \mu_0(\beta,N,\widetilde{\lambda}))^2} \lesssim \Vert f \Vert_{\infty}^2 \left( \beta N_0^2(\beta,N,\widetilde{\lambda}) + \beta^{-1} \right)
    \label{eq:bound1pdm19}
\end{equation}
An application of Lemma~\ref{lem:idealGasWithVariable1ParticleHamiltonian} shows $N_0(\beta,N,\lambda) \leq N_0(\beta,N) + C | \lambda | \ \Vert f \Vert_{\infty}/\beta$. Accordingly, the right-hand side of \eqref{eq:bound1pdm19} is bounded from above by a constant times $N^{11/12}$. We insert this bound and the optimal choice $\lambda = N^{-7/24}$ into \eqref{eq:bound1pdm18} and find
\begin{equation}
    \Vert \gamma_{\beta,N} - \gamma_0(\beta,N) \Vert_1 \lesssim N^{5/8}.
    \label{eq:bound1pdm20}
\end{equation}

It remains to replace $\gamma_0(\beta,N)$ by $\widetilde{N}_0(\beta,N) | \varphi_0 \rangle \langle \varphi_0 | + \gamma^{\mathrm{Bog}}$ in \eqref{eq:bound1pdm20}. A straightforward computation that we leave to the reader shows 
\begin{equation}
    \Vert \gamma_0(\beta,N) - \widetilde{N}_0(\beta,N) | \varphi_0 \rangle \langle \varphi_0 | - \sum_{p \in \Lambda^*_+} \gamma_p | \varphi_p \rangle \langle \varphi_p | \ \Vert_1 \lesssim N^{11/24}.
    \label{eq:bound1pdm21_0}
\end{equation}
We conclude that 
\begin{equation}
    \Vert \gamma_{\beta,N} - \widetilde{N}_0(\beta,N) | \varphi_0 \rangle \langle \varphi_0 | - \sum_{p \in \Lambda^*_+} \gamma_p | \varphi_p \rangle \langle \varphi_p | \ \Vert_1 \lesssim N^{5/8}
    \label{eq:bound1pdm21}
\end{equation}
holds provided $N_0(\beta,N) < N^{19/24}$. In combination, \eqref{eq:bound1pdm12} and \eqref{eq:bound1pdm21} prove \eqref{eq:theorem21pdm} with an additional log term in the rate. In the following remark we explain how this additional factor in the rate can be removed.
\begin{remark}
    \label{rem:removeLog}
    From \eqref{eq:bound1pdm12}, \eqref{eq:bound1pdm21}, and Lemma~\ref{lem:BoundN0} we learn that
    \begin{equation}
        | \Tr[a_0^* a_0 G_{\beta,N}] - N_0(\beta,N) | \lesssim N^{2/3}
        \label{eq:removelog}
    \end{equation}
    holds. This supersedes the bound in Proposition~\ref{prop:roughapriori}. Using \eqref{eq:removelog} to get a bound for $f^{\mathrm{BEC}}(N_0(\beta,N),N_0^{\mathrm{G}})$ in \eqref{eq:finalLowerBound14}, we obtain \eqref{eq:finalLowerBound19} without the logarithmic factor in the rate. Using the improved bound for the free energy we also obtain \eqref{eq:traceNormBound} and \eqref{eq:bound1pdm12} without a logarithmic factor. For the sake of simplicity we will in the following refer to equations with the logarithmic factor and use them without it. 
\end{remark}
\subsection{Pointwise bounds for the 1-pdm in Fourier space}
\label{sec:pwBound1pdm}
In this section we prove pointwise bounds for the integral kernel of the 1-pdm in Fourier space. For the sake of simplicity, we restrict attention to the two cases $\kappa > \frac{1}{4 \pi}[\upzeta(3/2)]^{2/3}$ (condensed phase) and $\kappa < \frac{1}{4 \pi}[\upzeta(3/2)]^{2/3}$ (non-condensed phase). 

To treat both cases we need the following lemma. The second part of the lemma is not needed here but stated for later reference.

\begin{lemma} \label{lem:fundamentalLemmaForPWBounds}
Let $G,G'$ be two non-negative trace class operators on a separable complex Hilbert space. Let $B$ be a self-adjoint operator and assume that $B^2 G$ and $B^2 G'$ are trace class. Then we have 
\begin{align}\label{eq:B-G-G'-1}
|\Tr [B(G-G')] |\le 2\sqrt{\Tr [B^2 (G+G')]}  \sqrt{\Tr |G-G'|}.
\end{align}
More generally, for any $\theta \in (1,\infty)$ and if $|B|^{\theta} G$ and $|B|^{\theta} G$ are trace class we have 
\begin{align}\label{eq:B-G-G'-2}
|\Tr [B(G-G')] | \lesssim_\theta ( \Tr [|B|^{\theta} (G+G')] )^{1/\theta}  (\Tr |G-G'|)^{1-1/\theta}.
\end{align}
\end{lemma}
\begin{proof} 
For any $\eps>0$, we can bound 
\begin{align}
|\Tr [B (G-G')] | &= |  \Tr [B \mathds{1}(|B|\le \eps^{-1} ) (G-G')] + \Tr [B \mathds{1}(|B|> \eps^{-1}) (G-G')]  | \nn \\
&\le | \Tr [B \mathds{1}(|B|\le \eps^{-1} ) (G-G')] | + |\Tr [B \mathds{1}(|B|> \eps^{-1}) G]| + |\Tr [B \mathds{1}(|B|> \eps^{-1}) G']| \nn \\
&\le \eps^{-1} \Tr |G-G'| +  \eps  \Tr [B^2 (G+G') ].  
\end{align}
Optimizing over $\eps>0$ we obtain \eqref{eq:B-G-G'-1}. 

More generally, for any $\theta\in (1,\infty]$ and any $\eps>0$, by the same argument, we have
\begin{align}
|\Tr [B (G-G')] | \le  \eps^{-1}  \Tr |G-G'| +  \eps^{\theta-1} \Tr [|B|^\theta (G+G')].
\end{align}
Optimization over $\eps>0$ yields \eqref{eq:B-G-G'-2}. 
\end{proof}

If $\kappa > \frac{1}{4 \pi}[\upzeta(3/2)]^{2/3}$ we use Lemma~\ref{lem:fundamentalLemmaForPWBounds}, \eqref{eq:aa-G}, which also holds with $G_{\beta,N}$ replaced by $\Gamma_{\beta,N}$, and \eqref{eq:traceNormBound} to see that
\begin{equation}
    | \Tr[a_p^* a_p (G_{\beta,N} - \Gamma_{\beta,N})] \leq 2 \sqrt{ \Tr[ (a_p^* a_p)^2 (G_{\beta,N} + \Gamma_{\beta,N}) ] } N^{-1/96} \lesssim \left( \frac{1}{\beta(p^2 - \mu_0(\beta,N))} + 1 \right) N^{-1/96}
\end{equation}
holds for $p \in \Lambda^*_+$. Using this and Lemma~\ref{lem:1pdmAndPairingFunction}, we obtain
\begin{equation}
    | \gamma_{\beta,N}(p) - \gamma_p | \lesssim \left( \frac{1}{\beta(p^2 - \mu_0(\beta,N))} + 1 \right) N^{-1/96}
\end{equation}
for $p \in \Lambda^*_+$. Here $\gamma_{\beta,N}(p)$, $p \in \Lambda^*$ denote the eigenvalues of $\gamma_{\beta,N}$. The fact that this operator is diagonal in momentum space follows from the translation-invariance of $G_{\beta,N}$. If $p = 0$ we apply \eqref{eq:theorem21pdm} to show that
\begin{equation}
    | \gamma_{\beta,N}(0) - \widetilde{N}_0(\beta,N) | \lesssim N^{2/3-1/48}.
    \label{eq:1pdmBEC}
\end{equation}

It remains to consider the case $\kappa < \frac{1}{4 \pi}[\upzeta(3/2)]^{2/3}$. Using Lemma~\ref{lem:fundamentalLemmaForPWBounds}, \eqref{eq:aa-G}, \eqref{eq:traceNormBound}, and $-\beta \mu_0(\beta,N) \sim 1$, we see that 
\begin{equation}
    |\Tr[a_p^* a_p (G_{\beta,N} - G^{\mathrm{id}}_{\beta,N})]| \leq 2 \sqrt{ \Tr[(a_p^* a_p)^2(G_{\beta,N} + G^{\mathrm{id}}_{\beta,N})] } \sqrt{\Vert G_{\beta,N} - G^{\mathrm{id}}_{\beta,N} \Vert_1} \lesssim N^{-1/96}
    \label{eq:pointwiseBounds1pdm5}
\end{equation}
holds for all $p \in \Lambda^*$. To obtain the second bound we also used that the state $G^{\mathrm{id}}_{\beta,N}$ satisfies the bound in \eqref{eq:aa-G}, too. Eq.~\eqref{eq:pointwiseBounds1pdm5} allows us to conclude that
\begin{equation}
    \gamma_{\beta,N}(p) = \frac{1}{\exp(\beta(p^2 - \mu_0(\beta,N)))-1} + O(N^{-1/96})
    \label{eq:pointwiseBounds1pdm6}
\end{equation}
holds for $p \in \Lambda^*$. This ends the proof of Theorem~\ref{thm:1-pdm}. 
\subsection{Condensate distributions}
In this section we investigate the coherent state condensate distribution $\zeta_{G}(z)$ in \eqref{eq:condensateDistributionGibbsStateContinuous} and the  random variable $\mathbf{N}_0$ describing the number of particles in the condensate, which has been defined in \eqref{eq:condensateDistributionGibbsStateDiscreteintro}. We start our analysis with a bound relating $\zeta_{G}$ to a Gaussian distribution in $L^1(\mathbb{C})$-norm. Afterwards, we investigate the different asymptotic regimes of the distribution of $\mathbf{N}_0$. 
\subsubsection{$L^1(\mathbb{C})$-norm bound for $\zeta_{G}$}
In this subsection we assume that $N_0(\beta,N) \gesssim N^{5/6 + \epsilon}$ holds with some fixed $0 < \epsilon \leq 1/6$. The inequality in \eqref{eq:finalLowerBound12} becomes an equality if we add $\beta^{-1}$ times the classical relative entropy
\begin{equation}
    S(\zeta_G,\tilde{g}) = \int_{\mathbb{C}} \zeta_{G}(z) \ln \left( \frac{\zeta_{G}(z)}{\tilde{g}(z) } \right) \de z
    \label{eq:classicalRelativeEntropy}
\end{equation}
on the right-hand side. Here $\tilde{g}$ denotes the probability distribution in \eqref{eq:GibbsDistributionDiscrete} with $\hat{v}(0)$ replaced by $(1-\delta)\hat{v}(0)$, $0 < \delta < 1$ and $N_0(\beta,N)$ replaced by $N_0^{\mathrm{G}}(\beta,N)$ defined below \eqref{eq:lowerBoundFE7}. When we keep the nonnegative term in \eqref{eq:classicalRelativeEntropy} in the lower bound for the free energy, insert the optimal choice $\delta = N^{-1/6}$ we found below \eqref{eq:finalLowerBound14}, and use the upper bound for the free energy in \eqref{eq:UpperBoundAndi10}, we find 
\begin{equation}
    S(\zeta_G,\tilde{g}) \lesssim \beta N^{5/8} \lesssim N^{-1/24}.
    \label{eq:boundClassicalRelativeEntropy}
\end{equation}
The classical relative entropy satisfies Pinsker's inequality $S(\zeta_G,\tilde{g}) \geq \frac{1}{2} \Vert \zeta_G - \tilde{g} \Vert_1^2$, and hence we conclude that
\begin{equation}
    \Vert \zeta_G - \tilde{g} \Vert_1 \lesssim N^{-1/48}.
    \label{eq:condensateDistribution1}
\end{equation}

In the next step we compare $\tilde{g}$ in $L^1(\mathbb{C})$-norm to a Gaussian distribution. To that end, we write
\begin{equation}
    \tilde{g}(z) = \frac{\exp\left( - \frac{\beta \hat{v}(0)(1-\delta)}{2 N} \left( |z|^2 - \frac{\mu^{\mathrm{G}}(\beta,N) N}{\hat{v}(0)(1-\delta)} \right)^2 \right)}{\int_{\mathbb{C}} \exp\left( - \frac{\beta \hat{v}(0)(1-\delta)}{2 N} \left( |w|^2 - \frac{\mu^{\mathrm{G}}(\beta,N) N}{\hat{v}(0)(1-\delta)} \right)^2 \right) \de w},
    \label{eq:condensateDistribution2}
\end{equation}
where $\mu^{\mathrm{G}}(\beta,N)$ is chose such that $\int_{\mathbb{C}} |z|^2 \tilde{g}(z) \de z = N_0^{\mathrm{G}}(\beta,N)$ holds. Let us derive a bound for $\mu^{\mathrm{G}}(\beta,N)$. From \eqref{eq:removelog} we know that $| N_0^{\mathrm{G}}(\beta,N) - N_0(\beta,N) | \lesssim N^{2/3}$. In combination with $N_0(\beta,N) \gesssim N^{5/6 + \epsilon}$ and \eqref{eq:LemChemPotBEC1}, we conclude that 
\begin{equation}
    \mu^{\mathrm{G}}(\beta,N) = \frac{\hat{v}(0)(1-\delta) N_0(\beta,N)}{N} + O(N^{-1/3}).
    \label{eq:condensateDistribution3}
\end{equation}
Using \eqref{eq:condensateDistribution3}, our assumption for $N_0(\beta,N)$, the coordinate transformation in \eqref{eq:FreeEnergyBECApriori21} in Appendix~\ref{app:effcondensate}, and $\delta=N^{-1/6}$, it is not difficult to see that 
\begin{align}
    \int_{\mathbb{C}} \exp\left( - \frac{\beta \hat{v}(0)(1-\delta)}{2 N} \left( |w|^2 - \frac{\mu^{\mathrm{G}}(\beta,N) N}{\hat{v}(0)(1-\delta)} \right)^2 \right) \de w &= \sqrt{\frac{2 \pi N}{\beta \hat{v}(0) (1-\delta)} } + O\left( \exp\left(-c N^{2\epsilon} \right) \right)  \nonumber \\
    &=\sqrt{ \frac{2 \pi N}{\beta \hat{v}(0) } } \left( 1 + O\left(N^{-1/6} \right) \right)
    \label{eq:condensateDistribution4}
\end{align}
holds. 

Using the same coordinate transformation and \eqref{eq:condensateDistribution3} again we also obtain the bound
\begin{align}
    \int_{\mathbb{C}} &\left| \exp\left( - \frac{\beta \hat{v}(0)(1-\delta)}{2 N} \left( |z|^2 - \frac{\mu^{\mathrm{G}}(\beta,N) N}{\hat{v}(0)(1-\delta)} \right)^2 \right) - \exp\left( - \frac{\beta \hat{v}(0)}{2 N} \left( |z|^2 - N_0(\beta,N) \right)^2 \right) \right| \de z \nonumber \\
    &\hspace{1cm}\leq \int_{0}^{\infty} \exp\left( - \frac{\beta \hat{v}(0)}{2 N} \left( x - \frac{\mu^{\mathrm{G}}(\beta,N) N}{\hat{v}(0)(1-\delta)} \right)^2 \right) \left| 1 - \exp\left( \frac{\beta \hat{v}(0)\delta}{2 N} \left( x - \frac{\mu^{\mathrm{G}}(\beta,N) N}{\hat{v}(0)(1-\delta)} \right)^2 \right) \right| \de x \nonumber \\ 
    &\hspace{1.35cm}+\int_{0}^{\infty} \exp\left( - \frac{\beta \hat{v}(0)}{2 N} \left( x - N_0(\beta,N) \right)^2 \right) \left| 1 - \exp\left( \frac{C |x-N_0(\beta,N)| }{N} \right) \right| \de x.
    \label{eq:condensateDistribution5}
\end{align}
It is straightforward to check that the terms on the right-hand side are bounded from above by a constant times
\begin{equation}
    N^{-1/6} \int_{\mathbb{C}} \exp\left( - \frac{\beta \hat{v}(0)}{2 N} \left( |z|^2 - N_0(\beta,N) \right)^2 \right) \de z. 
    \label{eq:condensateDistribution6}
\end{equation}
Putting \eqref{eq:condensateDistribution4}--\eqref{eq:condensateDistribution6} together, we find
that 
\begin{equation}
    \int_{\mathbb{C}} | \tilde{g}(z) - g(|z|^2) | \de z \lesssim N^{-1/6}
    \label{eq:condensateDistribution7}
\end{equation}
holds with the function
\begin{equation}
    g(x) = \sqrt{ \frac{\beta \hat{v}(0)}{2 \pi N} } \exp\left( - \frac{\beta \hat{v}(0)}{2 N} \left( x - N_0(\beta,N) \right)^2 \right).
    \label{eq:condensateDistribution8}
\end{equation}
The function $g$ is a normal distribution with mean $N_0(\beta,N)$ and variance $N/(\beta \hat{v}(0))$. 

Finally, \eqref{eq:condensateDistribution1} and \eqref{eq:condensateDistribution8} show
\begin{equation}
    \int_{\mathbb{C}} | \zeta_G(z) - g(|z|^2) | \de z \lesssim N^{-1/48},
    \label{eq:condensateDistribution9}
\end{equation}
which proves Theorem~\ref{thm:particleNumberDistributionBEC1}.
\subsubsection{Limiting distributions for $\mathbf{N}_0$}
\label{sec:asymptoticsCondensateDistribution}
In this section we study the asymptotics of the distribution of the random variable $\mathbf{N}_0$ in \eqref{eq:condensateDistributionGibbsStateDiscreteintro} in four different parameter regimes. From \eqref{eq:traceNormBound} we know that it satisfies
\begin{equation}
    \left| \sum_{n=0}^{\infty} f(n) \left( \mathbf{P}(\mathbf{N}_0 = n) - \Tr[|n \rangle \langle n | \ \Gamma_{\beta,N}] \right) \right| \leq \Vert f \Vert_{\infty} \ \Vert G_{\beta,N} - \Gamma_{\beta,N} \Vert_1 \lesssim \Vert f \Vert_{\infty} \ N^{-1/48}
    \label{eq:condensateDistribution10}
\end{equation}
with the state $\Gamma_{\beta,N}$ in \eqref{eq:referenceState} and any bounded function $f : \mathbb{N}_0 \to \mathbb{C}$. 

To evaluate the expression in \eqref{eq:condensateDistribution10} involving $\Gamma_{\beta,N}$ further, we have to distuingish between two cases. If $N_0(\beta,N) \geq N^{2/3}$ we have 
\begin{equation}
    \Tr[|n \rangle \langle n | \ \Gamma_{\beta,N}] = \int_{\mathbb{C}} | \langle z , n \rangle |^2 g(z) \de z = \frac{1}{n!} \int_{\mathbb{C}} |z|^{2n} \exp(-|z|^2) g(z) \de z = \frac{1}{n!} \int_{0}^{\infty} x^{n} \exp(-x) g(\sqrt{x}) \de x
    \label{eq:condensateDistribution11}
\end{equation}
with $g(z)$ in \eqref{eq:GibbsDistributionDiscrete}. To obtain the second equality, we applied the identities $a_0 |z \rangle = z | z \rangle$ and $\langle z , \Omega_0 \rangle = \exp(-|z|^2/2)$, which hold for all $z \in \mathbb{C}$. By $\Omega_0$ we denoted the vacuum vector of the Fock space over the $p=0$ mode. We also recall the definition of $|n\rangle$ below \eqref{eq:condensateDistributionGibbsStateDiscreteintro}. The last equality follows from the coordinate transformation in \eqref{eq:FreeEnergyBECApriori21} in Appendix~\ref{app:effcondensate}. The probability distribution in \eqref{eq:condensateDistribution11} is called a randomized Poisson distribution. It is defined via a Poisson random variable, whose rate parameter is again a random variable. Here the rate is given by $X_{\beta,N}$, the random variable related to the probability distribution $g(\sqrt{x})$. In contrast, if $N_0 < N^{2/3}$ then
\begin{equation}
    \Tr[|n \rangle \langle n | \ \Gamma_{\beta,N}] = \exp(\beta \mu_0(\beta,N)n)(1-\exp(\beta \mu_0(\beta,N))).
    \label{eq:condensateDistribution12}
\end{equation}

We recall the definition of the random variable $\widetilde{\mathbf{N}}_0$ in \eqref{eq:condensateDistributionGibbsStateDiscreteintro}. 
We also define $\mathbf{M}_0$ and $\widetilde{\mathbf{M}}_{0}$ by 
\begin{equation}
    \mathbf{P}(\mathbf{M}_0 = n) = \Tr[|n \rangle \langle n | \ \Gamma_{\beta,N}] \quad \text{ and } \quad \widetilde{\mathbf{M}}_0 = \frac{\mathbf{M}_0 - \widetilde{N}_0(\beta,N)}{\sqrt{\textbf{Var}(\beta,N)}},
    \label{eq:condensateRVb}
\end{equation}
where $\widetilde{N}_0(\beta,N)$ and $\textbf{Var}(\beta,N)$ denote the expectation and the variance of $X_{\beta,N}$, respectively. In the following we denote the characteristic function of a random variable $X$ by
\begin{equation}
    \phi_{X}(t) = \mathbf{E}(e^{\mathrm{i}t X}).
    \label{eq:characteristicFunction}
\end{equation}

To discuss three of our four parameter regimes, we need the following abstract lemma about randomized Poisson distributions, which we discuss first. 

\begin{lemma}\label{lem:abstractLemmaRandomizedPoissonDistributions}
    Let $\{ X_j \}_{j \in \mathbb{N}}$ be a sequence of real-valued nonnegative random variables with means $m_j$ and variances $\sigma_j$ and denote $\widetilde{X}_j = (X_j - m_j)/\sigma_j$. We assume that $\sigma_j \to +\infty$ and $m_j/\sigma_j^2 \to 0$ for $j \to \infty$. We also assume that there are constants $\lambda_0, C > 0$ such that  
    \begin{equation}
        \lim_{j \to \infty} \mathbf{E}(\exp(\lambda |\widetilde{X}_j|)) \leq C
        \label{eq:assumptionMomentGeneratingFunction}
    \end{equation}
    holds for all $0 < \lambda < \lambda_0$. Let $\{ Y_j \}_{j \in \mathbb{N}}$ be a sequence of randomized Poisson random variables with laws
    \begin{equation}
        \mathbf{P}(Y_j = n) = \frac{1}{n!} \mathbf{E}( X_j^n \exp(-X_j) )
        \label{eq:generalRandomizedPoisson}
    \end{equation}
    and define $\widetilde{Y}_j = (Y_j - m_j)/\sigma_j$. Then we have
     \begin{equation}
        \lim_{j \to \infty} | \phi_{\widetilde{Y}_j}(t) - \phi_{\widetilde{X}_j}(t) | = 0
    \end{equation}
    for all $t \in \mathbb{R}$.
\end{lemma}
\begin{proof}
    A short computation shows that the characteristic function of $\widetilde{Y}_j$ reads
    \begin{equation}
        \phi_{\widetilde{Y}_j}(t) = \mathbf{E}\left( \sum_{n=0}^{\infty} e^{\mathrm{i} (t/\sigma_j) n} \frac{X_j^n}{n!} \exp(-X_j) \right) e^{-\mathrm{i} t m_j/\sigma_j} =  \mathbf{E}\left( \exp\left( X_j \left( e^{\mathrm{i}(t/\sigma_j)} - 1 \right) \right) \right) e^{-\mathrm{i} t m_j/\sigma_j}. 
        \label{eq:condensateDistribution13}
    \end{equation}
    A second order Taylor expansion allows us to write $e^{\mathrm{i}(t/\sigma_j)} - 1 = it/\sigma_j + O(t^2/\sigma_j^2)$. We insert this and $X_j = m_j + \widetilde{X}_j \sigma_j$ into \eqref{eq:condensateDistribution13} and find
    \begin{equation}
        \phi_{\widetilde{Y}_j}(t) = \left[ \phi_{\widetilde{X}_j}(t) + \mathbf{E}\left( e^{\mathrm{i}t \widetilde{X}_j} \left[ \exp(\widetilde{X}_j \cdot O(t^2/\sigma_j)) - 1 \right] \right) \right] \left( 1 + O\left( m_j/\sigma_j^2 \right) \right).
        \label{eq:condensateDistribution14}
    \end{equation}
   When we rearrange the terms and take the limit $j \to \infty$ on both sides, we find
    \begin{equation}
        \lim_{j \to \infty} | \phi_{\widetilde{Y}_j}(t) - \phi_{\widetilde{X}_j}(t) | \leq \lim_{j \to \infty} \mathbf{E}\left( | \exp( \widetilde{X}_j \cdot O(t^2/\sigma_j)) - 1 | \right).
        \label{eq:condensateDistribution15}
    \end{equation}
    Our assumptions on $\widetilde{X}_j$ guarantee that the right-hand side equals zero, which proves the claim.
\end{proof}
We are now prepared to discuss the first parameter regime. We recall that $N_0$ and $\widetilde{N}_0$ have the same asymptotic behavior as $N \to \infty$. This is guaranteed by Lemma~\ref{lem:BoundN0} and allows us to describe the different parameter regimes below in terms of $N_0$.
\subsubsection*{Parameter regime 1: $N_0(\beta,N) \gg N^{5/6}$}
The first parameter regime corresponds to the condensed phase as well as to the temperature regime, where the critical point is approached from the condensed phase in such a way that $N_0(\beta,N) \gg N^{5/6}$ still holds.

We recall the definition of $X_{\beta,N}$ below \eqref{eq:condensateDistribution11} as well as that its expectation is given by $\widetilde{N}_0(\beta,N)$, see \eqref{eq:GibbsDistributionDiscrete}. Its variance is denoted by $\mathbf{Var}(\beta,N)$. We also define the centered and rescaled random variable 
\begin{equation}
    \widetilde{X}_{\beta,N} = \frac{X_{\beta,N} - \widetilde{N}_0(\beta,N)}{\sqrt{\mathbf{Var}(\beta,N)}}.
    \label{eq:widetildeX}
\end{equation}
In the parameter regime we are currently interested in, the random variable $\widetilde{\mathbf{M}}_{0}$ is centered because
\begin{equation}
    \mathbf{E}(\mathbf{M}_0) = \sum_{n=1}^{\infty} \frac{1}{(n-1)!} \mathbf{E}\left( X_{\beta,N}^n \exp(-X_{\beta,N}) \right) = \mathbf{E} \left( X_{\beta,N} \exp(-X_{\beta,N}) \sum_{n=1}^{\infty} \frac{1}{(n-1)!}  X_{\beta,N}^{n-1}  \right) = \mathbf{E}(X_{\beta,N}).
    \label{eq:condensateDistribution15_0}
\end{equation}
It should, however, be noted that the variance of $\mathbf{M}_{0}$ does not equal $\mathbf{Var}(\beta,N)$. 

We apply Lemma~\ref{lem:abstractLemmaRandomizedPoissonDistributions} to the characteristic function of $\widetilde{\mathbf{M}}_{0}$, which gives
\begin{equation}
    \lim_{N \to \infty} | \phi_{\widetilde{\mathbf{M}}_{0}}(t) - \phi_{\widetilde{X}_{\beta,N}}(t) | = 0. 
    \label{eq:condensateDistribution16}
\end{equation}
That $\widetilde{X}_{\beta,N}$ satisfies the assumptions of Lemma~\ref{lem:abstractLemmaRandomizedPoissonDistributions} is guaranteed by Lemmas~\ref{lem:asymptoticsVariances} and \ref{lem:BoundCenteredCondensateDistribution} in Appendix~\ref{app:effcondensate}. The pointwise limit of the characteristic function of $\widetilde{X}_{\beta,N}$ has been computed in part~(a) of Lemma~\ref{lem:lemmaCharacteristicFunction}, which allows us to conclude that
\begin{equation}
    \lim_{N \to \infty} \phi_{\widetilde{\mathbf{M}}_{0}}(t) = \exp(-t^2/2)
    \label{eq:condensateDistribution17}
\end{equation}
holds. Finally, applications of \eqref{eq:condensateDistribution10} and \eqref{eq:condensateDistribution17} show that $\lim_{N \to \infty} \phi_{\widetilde{\mathbf{N}}_0}(t)$ equals the right-hand side of \eqref{eq:condensateDistribution17}, too. In particular, $\widetilde{\mathbf{N}}_0$ converges in distribution to a standard normal random variable.
\subsubsection*{Parameter regime 2: $N_0(\beta,N) = t N^{5/6}$ with some fixed $t \in \mathbb{R}$}
To be in this parameter regime we need to approach the critical point such that $N_0(\beta,N) = t N^{5/6}$ holds with some fixed $t \in \mathbb{R}$. The same analysis as in the previous parameter regime shows that $\widetilde{\mathbf{N}}_0$ converges in distribution to a random variable with distribution $f_{\sigma,A,B}$ in \eqref{eq:lemmaCharacteristicFunctiond}. 
\subsubsection*{Parameter regime 3: $1 \ll N_0(\beta,N) \ll N^{5/6}$}
This parameter regime again can be obtained when we approach the critical point in a specific way. If $N^{2/3} \leq N_0(\beta,N) \ll N^{5/6}$ we argue as in the previous two parameter regimes to see that $\widetilde{\mathbf{N}}_0$ converges in distribution to a random variable with distribution function given by $f$ in \eqref{eq:lemmaCharacteristicFunctione}. If $1 \ll N_0(\beta,N) < N^{2/3}$ the probability distribution of $\mathbf{M}_0$ is given by the right-hand side of \eqref{eq:condensateDistribution12}. The characteristic function of $\widetilde{\mathbf{M}}_{0}$ therefore reads
\begin{align}
    \phi_{\widetilde{\mathbf{M}}_{0}}(t) &= \mathbf{E}\left( \exp\left( \mathrm{i}t \left( \frac{\mathbf{M}_{0} - \widetilde{N}_0}{\sqrt{\mathbf{Var}}} \right) \right) \right) = e^{-\mathrm{i} t \widetilde{N}_0/\sqrt{\mathbf{Var}}} \sum_{n=0}^{\infty} e^{\mathrm{i} (t/\sqrt{\mathbf{Var}}) n} \exp(\beta \mu_0(\beta,N)n)(1-\exp(\beta \mu_0(\beta,N))) \nonumber \\
    &= e^{-\mathrm{i} t \widetilde{N}_0/\sigma} \frac{1-\exp(\beta \mu_0(\beta,N))}{1-\exp(\beta \mu_0(\beta,N) + \mathrm{i}(t/\sqrt{\mathbf{Var}})}.
    \label{eq:condensateDistribution19}
\end{align}
Using \eqref{eq:condensateDistribution19}, $N_0 \simeq \widetilde{N}_0$, $-\beta \mu_0(\beta,N) \simeq 1/N_0$, and $\sqrt{\mathbf{Var}(\beta,N)} \simeq N_0(\beta,N)$, which follows from part~(b) of Lemma~\ref{lem:ChemPotBECCont} and part~(c) of Lemma~\ref{lem:asymptoticsVariances}, we see that
\begin{equation}
    \lim_{N \to \infty} \phi_{\widetilde{\mathbf{M}}_{0}}(t) = \frac{e^{-\mathrm{i}t}}{1-\mathrm{i}t} 
    \label{eq:condensateDistribution20}
\end{equation}
holds. The right-hand side of \eqref{eq:condensateDistribution20} is the characteristic function of the probability distribution $f$ in \eqref{eq:lemmaCharacteristicFunctione}. As before we check with \eqref{eq:condensateDistribution10} that $\phi_{\widetilde{\mathbf{N}}_{0}}(t)$ converges to the same limit as $N \to \infty$.  
\subsubsection*{Parameter regime 4: $N_0(\beta,N) = t$ with some fixed $t > 0$}
In this parameter regime we are above the critical point, where $\beta \mu_0(\beta,N)$ does not depend on $N$ (or it has a limit as $N \to \infty$) and an application of \eqref{eq:condensateDistribution10} shows that $\mathbf{N}_0$ converges in distribution to a random variable with law given by the right-hand side of \eqref{eq:condensateDistribution12}. This ends the proof of Theorem~\ref{thm:particleNumberDistributionBEC2}.

\section{The Gibbs state part IV: higher order correlation inequalities}
\label{sec:higherOrderCorrelationInequalities}
In this section we prove Theorem~\ref{thm:higher-moments}. Then we apply our abstract result to the Gibbs state $G_{\beta,N}$. 

\subsection{Proof of Theorem~\ref{thm:higher-moments}}
The goal of this section is to prove Theorem~\ref{thm:higher-moments}. We will prove that for all even $k\in \mathbb{N}$ and all $\delta\in (0,1)$,  
\begin{align}\label{eq:corr-thm-conclusion}
\max_{t\in [-1+\delta, 1-\delta]} \Tr [B^k \exp(-A+tB)] \lesssim_{k,\delta}   \sup_{t\in [-1,1]}  \Tr[ ( 1+b^2 X^{2\alpha+k-2})  \exp(-A+tB) ]. 
\end{align}
Moreover, if $B\ge 0$, then for all $k\in \mathbb{N}$ and all $\delta\in (0,1)$ we have
 \begin{align} \label{eq:corr-thm-conclusion-B>0} 
\max_{t\in [-1+\delta, 1-\delta]}\Tr[ B^{k} \exp(-A+tB) ] \lesssim_{k,\delta} \sup_{|t| \le 1}  \Big( \Tr[ \exp(-A+tB) ]  + \sum_{\ell=1}^k b |\Tr ( [B, [B, A]] X^ {\alpha+\ell-3}\exp(-A+tB) )| \Big). 
\end{align}
The statements  of Theorem~~\ref{thm:higher-moments} follow from these bounds and \eqref{eq:1-asum-used}. 

In the following we use the notation $A_t=A-tB$. We observe that for all $\ell \in \mathbb{N}$ we have  
\begin{equation}\label{eq:Tr-B-ell-Duhamel}
    \partial_t \Tr [B^\ell e^{-A_t}] = \int_0^1 \Re \Tr \left[ B^\ell e^{-sA_t} B e^{-(1-s)A_t} \right] \de s. 
\end{equation}
We use \cite[Theorem 7.2 (iii)]{LewNamRou-21} and the fact that $X$ commutes with both, $A$ and $B$, to estimate
\begin{align} \label{eq:FB-k-2}
\Big| & \Re \Tr[ B^\ell e^{-sA_t} B e^{-(1-s)A_t} ]  - \Tr[B^{\ell+1}e^{-A_t}] \Big| \nonumber \\  
&\hspace{2cm }\le  \frac{1}{4} \sqrt{ |\Tr ( [B, [B, A]] X^qe^{-A_t} )|}  \sqrt{|\Tr ( [B^\ell,[ B^\ell, A]] X^{-q}e^{-A_t})|},
\end{align}
which holds for all $s\in [0,1]$ and all $q\in \mathbb{R}$. Since strictly speaking the results in \cite[Theorem 7.2]{LewNamRou-21} are stated with bounded operators $A,B$ and without the operator $X$, let us briefly explain how to obtain \eqref{eq:FB-k-2} in our more general setting. First, for any self-adjoint operator $Y$ such that $Y e^{-sA_t}$ and $[Y,A] e^{-sA_t}$ are Hilbert-Schmidt for all $s\in (0,1)$, the function
\begin{equation}
    f(s)=\Tr [ Y e^{-sA_t}  Y e^{ -(1-s)A_t} ],\quad s\in [0,1]    
\end{equation}
satisfies 
\begin{equation}
    f''(s)= - \Tr ( [Y,A] e^{-sA_t}   [Y,A] e^{-(1-s)A_t}   ) = \Tr \Big( \Big| e^{-sA_t/2}  [Y,A] e^{-(1-s)A_t/2}   \Big|^2  \Big) \ge 0    
\end{equation}
for all $s\in (0,1)$. That is, $f$ is convex. This key observation, which goes back to \cite[Eq. (7.12)]{LewNamRou-21}, also holds for unbounded operators. With the convexity of $f$ and the symmetry $f(s)=f(1-s)$ we find 
\begin{equation}
    0 \le f(0)-f(s) \le f(0)-f(1/2) \le - \frac{1}{2} \lim_{s\to 0}f'(s),    
\end{equation}
which gives
\begin{align}    \label{eq:key-correlation-estimate-Y-At}
0 \le \Tr [Y^2 e^{-A_t} ] -  \Tr [ Y e^{-sA_t}  Y e^{-(1-s)A_t} ] \le - \frac{1}{2}\Tr (Y[Y,A_t] e^{-A_t} ) = \frac{1}{4}\Tr ([[Y,A_t],Y] e^{-A_t} ). 
\end{align}
Now we apply \eqref{eq:key-correlation-estimate-Y-At} four times with $Y\in \{B_1,B_2,B_1\pm B_2\}$ with $B_1= \eps B X^{q/2}$ and $B_2=\eps^{-1} B^{\ell}X^{-q/2}$ for some $\eps>0$. With these choices, the technical requirement that $Y e^{-sA_t}$ and $[Y,A] e^{-sA_t}$ are Hilbert--Schmidt follow from our assumptions on $A,B$ and $X$. We find
\begin{align}   
  &\Tr [ (B_1\pm B_2) e^{-sA_t} (B_1\pm B_2)  e^{-(1-s)A_t} ] - \sum_{j=1}^2 \Tr [ B_j e^{-sA_t}  B_j  e^{-(1-s)A_t} ] \nn \\
  &\hspace{1.5cm} \le \Tr [(B_1 \pm B_2)^2 e^{-A_t}] - \sum_{j=1}^2 \Big(  \Tr [ B_j^2 e^{-A_t} ] - \frac{1}{4}   \Tr ([[B_j,A_t],B_j] e^{-A_t} )  \Big),  
\end{align}
which is equivalent to 
\begin{align*}   
  \pm 2 \Big( \Re \Tr (B^\ell e^{-sA_t} B e^{-(1-s)A_t} ) -  \Tr (B^{\ell+1} e^{-A_t}  ) \Big) \le \frac{\eps^{2} }{4} \Tr ([[B,A_t],B] X^{q}e^{-A_t})  + \frac{\eps^{-2}}{4}    \Tr ([[B^\ell,A_t],B^\ell] X^{-q}e^{-A_t}). 
\end{align*}
Optimizing the latter bound over $\eps>0$ leads to \eqref{eq:FB-k-2}. 

To estimate further the right-hand side of \eqref{eq:FB-k-2}, we use our assumption in \eqref{eq:corr-thm-ass2}. For the first term on the right-hand side of \eqref{eq:FB-k-2}, we can bound immediately 
\begin{align} \label{eq:FB-k-2a}
\pm [B, [B, A]] X^q  \le b X^{q+\alpha}. 
\end{align}
For the second term, we have the identity 
\begin{align}\label{eq:doub-comm-expansion}
[B^\ell, [B^\ell,A]] = \sum_{i=0}^{\ell-1} [B^\ell, B^i [B,A] B^{\ell-1-i}] =  \sum_{i,j=0}^{\ell-1} B^{i+j} [B,[B,A]] B^{2\ell-2-i-j}.
\end{align}
By the Cauchy--Schwarz inequality, 
\begin{align}\label{eq:doub-comm-B-CS}
\pm ( B^{2n} [B,[B,A]] + [B,[B,A]] B^{2n} ) 
\le b \Big| (B^{n} X^{\alpha/2} )^* \Big|^2 + b^{-1} \Big| X^{-\alpha/2} B^n [B,[B,A]] \Big|^2 \le 2b B^{2n} X^{\alpha}
\end{align}
for all integer $n\ge 0$. Here we used again \eqref{eq:corr-thm-ass2}. Inserting   \eqref{eq:doub-comm-B-CS} in \eqref{eq:doub-comm-expansion} we find 
\begin{align} \label{eq:FB-k-2b1}
\pm [B^\ell, [B^\ell,A]] X^{-q}\lesssim_\ell b B^{2\ell-2} X^{\alpha-q}. 
\end{align}
In summary, from \eqref{eq:Tr-B-ell-Duhamel},  \eqref{eq:FB-k-2} and \eqref{eq:FB-k-2b1} we deduce that
\begin{align} \label{eq:FB-k-3}
\Big|  \partial_t \Tr[ B^\ell e^{-A_t}]  - \Tr[B^{\ell+1}e^{-A_t}] \Big| \lesssim_\ell  b \sqrt{ \Tr [ X^{\alpha+q}e^{-A_t} ]}  \sqrt{  \Tr [ B^{2\ell-2} X^{\alpha-q}e^{-A_t}]}, \quad \forall \ell = 0,1,2,...
\end{align}

Next, we apply \eqref{eq:FB-k-3} with $\ell\in \{k-2,k-1\}$, $q=\alpha+k-4$ and use the Cauchy--Schwarz inequality to check that
\begin{align}
\Big| \partial_t \Tr [B^{k-2} e^{-A_t}] - \Tr[ B^{k-1} e^{-A_t}  ] \Big| &\lesssim_k b  \sqrt{  \Tr[X^{2\alpha+k-4}e^{-A_t}] } \sqrt{  \Tr [B^{k-2}   e^{-A_t}] } \nn\\
& \le \Tr [B^{k-2}   e^{-A_t}] + b^2 \Tr[X^{2\alpha+k-4}e^{-A_t}], \label{eq:FB-k-4a} \\
 \partial_t \Tr [B^{k-1} e^{-A_t}] &\ge  \Tr[ B^{k} e^{-A_t} ]  - C_k  b  \sqrt{  \Tr[X^{2\alpha+k-4}e^{-A_t} ]} \sqrt{  \Tr [B^{k}   e^{-A_t}] } \nn\\
& \ge \frac{1}{2} \Tr [B^{k}   e^{-A_t}] - C_k  b^2 \Tr[X^{2\alpha+k-4}e^{-A_t}] . \label{eq:FB-k-4b}
\end{align}
Here we used $\pm B\le X$ and in \eqref{eq:FB-k-4b} we also used that $k$ is even. Similarly, applying \eqref{eq:FB-k-3} with $\ell\in \{k,k+1\}$, $q=\alpha+k-2$, we get
\begin{align}
\Big| \partial_t \Tr [B^{k} e^{-A_t}] - \Tr[ B^{k+1} e^{-A_t}  ] \Big| &\lesssim_k b  \sqrt{  \Tr(X^{2\alpha+k-2}e^{-A_t}] } \sqrt{  \Tr [B^{k}   e^{-A_t}] } \nn\\
& \le \Tr [B^{k}   e^{-A_t}] + b^2 \Tr(X^{2\alpha+k-2}) e^{-A_t}], \label{eq:FB-k-4c} \\
 \partial_t \Tr [B^{k+1} e^{-A_t}] &\ge  \Tr[ B^{k+2} e^{-A_t} ]  - C_k  b  \sqrt{  \Tr[X^{2\alpha+k-2}e^{-A_t} ]} \sqrt{  \Tr [B^{k+2}   e^{-A_t}] } \nn\\
& \ge \frac{1}{2} \Tr [B^{k+2}   e^{-A_t}] - C_k  b^2 \Tr[X^{2\alpha+k-2} e^{-A_t}] . \label{eq:FB-k-4d}
\end{align}

We can put \eqref{eq:FB-k-4a}, \eqref{eq:FB-k-4b}, \eqref{eq:FB-k-4c} and \eqref{eq:FB-k-4d} in good use by the following elementary  lemma, which is an approximate version of a second order Taylor expansion of a convex function.  

\begin{lemma}\label{lem:f-f4} Let $I\subset \R$ be an open interval. Let $f_j,e_j: I \to \R$, $j \in \{ 0,1,2 \}$ be such that for all $j$ the functions $f_j$ are continuously differentiable and the functions $e_j$ are continuous. We also assume that 
\begin{equation}\label{eq:elem-lem-0}
    f_0'=f_1+e_1,\quad f_1' \ge f_2+e_2, \quad f_0\ge 0, \quad f_2\ge 0. 
    \end{equation}
For $\delta>0$ and $(x-2\delta,x+2\delta)\subset I$ we have
\begin{align} \label{eq:elem-lem-1}
 \frac{\delta}{2}   \int_{x-\delta/2}^{x+\delta/2} f_2(t) \de t \le f_0(x+\delta)+f_0(x-\delta)-2f_0(x) + \int_{I} (|e_1(t)|+|e_2(t)|) \de t 
\end{align}
and 
\begin{align}\label{eq:elem-lem-2}
f_0(x) \le 2\delta^{-1} \int_{I} (|f_0(t)|+|e_1(t)|+|e_2(t)|) \de t.  
\end{align}
\end{lemma}

\begin{proof}[Proof of Lemma~\ref{lem:f-f4}] If $\eps\ge \delta>0$ are such that $[x-\eps,x+\eps]\subset I$, then from \eqref{eq:elem-lem-0} we have
\begin{align}\label{eq:elem-lem-3}
f_0(x+\eps) + f_0(x-\eps) -2f_0(x) &= \int_x^{x+\eps} ( f_0'(s) -f_1(x) ) \de s  + \int_{x-\eps}^x  ( f_1(x) - f_0'(s)   ) \de s \nn\\
&\ge \int_x^{x+\eps} (f_1(s) - f_1(x)) \de s  + \int_{x-\eps}^x ( f_1(x) - f_1(s)  ) \de s -  \int_{I} |e_1(s)| \de s \nn \\
& = \int_x^{x+\eps} \int_x^s  f_1' (\xi) \de \xi \de s + \int_{x-\eps}^x  \int_{s}^x  f_1'(\xi) \de \xi \de s -   \int_{I} |e_1(s)| \de s \nn\\
&\ge \int_x^{x+\eps} \int_x^s f_2(\xi) \de \xi \de s + \int_{x-\eps}^x \int_{s}^x f_2(\xi) \de \xi \de s -  \int_{I} (|e_1(s)|+|e_2(s)|) \de s  \nn\\
&\ge \frac{\delta}{2} \int_{x-\delta/2}^{x+\delta/2}  f_2 (\xi) \de \xi - \int_{I} ( |e_1(s)|+|e_2(s)|) \de s. 
\end{align}
Here we obtained the last lower bound by using $f_2\ge 0$ and restricting the integration to the domain $\delta/2\le |s-x|  \le 3\delta/2$. Choosing $\eps=\delta$ in \eqref{eq:elem-lem-3} gives \eqref{eq:elem-lem-1}. Moreover, since $f_2\ge 0$ we deduce from  \eqref{eq:elem-lem-3} that
\begin{align}\label{eq:elem-lem-4}
2f_0(x) \le f_0(x+\eps) + f_0(x-\eps)  +  \int_{I} ( |e_1(s)|+|e_2(s)|) \de s.
\end{align}
Averaging \eqref{eq:elem-lem-4} over $\eps\in (\delta,2\delta)$, we obtain \eqref{eq:elem-lem-2}. 
\end{proof}

With \eqref{eq:FB-k-4a}, \eqref{eq:FB-k-4b}, and \eqref{eq:elem-lem-1} from Lemma~\ref{lem:f-f4}, we find that 
\begin{align} \label{eq:eq:elem-lem-1-appl}
\int_{-1+2\delta}^{1-2\delta}  \Tr [B^{k} e^{-A_t}] \de t &\lesssim_\delta \sup_{|x|\le 1-\delta } \int_{x-\delta/2}^{x+\delta/2} \Tr [B^{k} e^{-A_t}] \de t \nn\\
&\lesssim_{k,\delta} \sup_{|t|\le 1-\delta/2} \Tr [B^{k-2} e^{-A_t}] + \sup_{|t|\le 1-\delta/2}   b^2 \Tr[ X^{2\alpha+k-4}e^{-A_t}]
\end{align}
holds for all $\delta\in (0,1/4)$. Moreover, using \eqref{eq:FB-k-4c}, \eqref{eq:FB-k-4d}, \eqref{eq:elem-lem-1}  in Lemma \ref{lem:f-f4} and \eqref{eq:eq:elem-lem-1-appl}, we get  
\begin{align}
\sup_{|t|\le 1-4\delta}  \Tr [B^{k} e^{-A_t}]  &\lesssim_\delta \int_{-1+2\delta}^{1+2\delta} \Tr [B^{k} e^{-A_t}] \de t+ \sup_{|t|\le 1-\delta/2}   b^2 \Tr[ X^{2\alpha+k-2}e^{-A_t}] \nn\\
&\lesssim_{k,\delta} \sup_{|t|\le 1-\delta/2} \Tr [B^{k-2} e^{-A_t}] + \sup_{|t|\le 1-\delta/2}   b^2 \Tr[ X^{2\alpha+k-2}e^{-A_t}].
\end{align}
This holds for all even $k\ge 2$ and all $\delta \in (0,1/4)$. The conclusion of  \eqref{eq:corr-thm-conclusion} follows by induction. 

Finally, let us additionally assume that $B\ge 0$. In this case, we can use induction from $k-1$ to $k$, instead of $k-2$ to $k$. To be precise, using \eqref{eq:Tr-B-ell-Duhamel}, \eqref{eq:FB-k-2}, \eqref{eq:doub-comm-B-CS} (but not  \eqref{eq:FB-k-2a}) with $q=\alpha+\ell-3$, and the Cauchy--Schwarz inequality we have 
\begin{align} \label{eq:FB-k-2-B>0-a}
\partial_t  \Tr[ B^{\ell} e^{-A_t} ]  &\ge \Tr[B^{\ell+1}e^{-A_t}] - C_\ell  \sqrt{ |\Tr ( [B, [B, A]] X^qe^{-A_t} )|}  \sqrt{b \Tr [ B^{2\ell-2} X^{-q+\alpha}e^{-A_t}]} \nn\\
&\ge \Tr[B^{\ell+1}e^{-A_t}] - C_\ell  \sqrt{ |\Tr ( [B, [B, A]] X^qe^{-A_t} )|}  \sqrt{b \Tr [ B^{\ell+1} X^{-q+\alpha+\ell-3}e^{-A_t}]} \nn\\
&\ge \frac{1}{2}\Tr[B^{\ell+1}e^{-A_t}] - C_\ell b |\Tr ( [B, [B, A]] X^ {\alpha+\ell-3} e^{-A_t} )|.  
\end{align}
Applying \eqref{eq:FB-k-2-B>0-a} for $\ell=k-1$, we have
\begin{align} \label{eq:FB-k-2-B>0-b}
\int_{-1+\delta}^{1-\delta}  \Tr[ B^{k} e^{-A_t} ]  \de t  \lesssim_k \sup_{|t| \le 1}  \Big( \Tr[ B^{k-1} e^{-A_t} ] + b |\Tr ( [B, [B, A]] X^ {\alpha+k-4} e^{-A_t} )|  \Big). 
\end{align}
Then applying \eqref{eq:FB-k-2-B>0-a} for $\ell=k$ and using $\Tr[B^{k+1}e^{-A_t}]  \ge 0$, we have for every $-1 \le s\le t \le 1$,
\begin{align} \label{eq:FB-k-2-B>0-c}
 \Tr[ B^{k} e^{-A_s} ] - \Tr[ B^{k} e^{-A_t} ]  = - \int_{s}^t \partial_\xi \Tr[ B^{k} e^{-A_\xi} ]  \de   \xi \lesssim_k  b |\Tr ( [B, [B, A]] X^ {\alpha+k-3} e^{-A_t} )| . 
\end{align}
Combining \eqref{eq:FB-k-2-B>0-b} and \eqref{eq:FB-k-2-B>0-c}, we conclude that for all $|s|\le 1-2\delta$,
\begin{align} \label{eq:FB-k-2-B>0-d} 
\Tr[ B^{k} e^{-A_s} ] &= \delta^{-1} \int_{s}^{s+\delta}  \Tr[ B^{k} e^{-A_t} ] \de t - \delta^{-1} \int_{s}^{s+\delta}  \Big( \Tr[ B^{k} e^{-A_t} ]  - \Tr[ B^{k} e^{-A_s} ]  \Big) \de t \nn \\
&\lesssim_{k,\delta} \sup_{|t| \le 1-\delta}  \Big( \Tr[ B^{k-1} e^{-A_t} ] + b |\Tr ( [B, [B, A]] X^ {\alpha+k-4}e^{-A_t} )| +  b |\Tr ( [B, [B, A]] X^ {\alpha+k-3}e^{-A_t} )|  \Big). 
\end{align}  
This holds for all $k\ge 1$. By induction, we get the desired estimate \eqref{eq:corr-thm-conclusion-B>0}. The proof of Theorem~\ref{thm:higher-moments}  is complete. 

\subsection{Higher order moment bounds for the Gibbs state} 
\label{sec:higherOrderEstimates}

As an application of the above abstract theorem, we obtain the following estimates. 

\begin{theorem}[Higher moment estimates] \label{thm:higherOrderEstimates} 
Under the conditions of Theorem \ref{thm:secondOrderEstimates}, the following holds: 
\begin{enumerate}[label=(\alph*)]
\item Assume that $h$ is diagonal in momentum space. For $p \in \Lambda_+^*$ we have
\begin{equation} \label{eq:aa-G-higher}
\Tr [ (a_p^*a_p)^4  G_{h,\eta}(\beta,\mu) ] \lesssim \beta^{-4}. 
\end{equation}
\item We have
\begin{align}
\Tr [ (\cN_{+}-N_{+}^{\mathrm{G}})^4 G_{h,\eta}(\beta,\mu) ] &\lesssim \beta^{-4}
\label{eq:varianceNPlus-higher}
\end{align}
with $N_{+}^{\mathrm{G}} = \Tr[\mathcal{N}_+ G_{h,\eta}(\beta,\mu)]$.
\item For all $k \in \mathbb{N}$, we have
    \begin{equation}
    \Tr [\cN^k G_{h,\eta}(\beta,\mu) ]  \lesim_k  \eta^k. 
    \label{eq:particleNumberBoundsPerturbedState-k-higher}
    \end{equation}

\item If $\beta=\kappa \beta_c$ with $\kappa\in (0,1)$ fixed (the non-condensed phase), then for all $p\in \Lambda^*$, we have
\begin{equation}
    \Tr [ (a_p^*a_p)^3  G_{h,\eta}(\beta,\mu)  ]  \lesim 1. 
    \label{eq:apap-third-moment_0}
    \end{equation}
\end{enumerate}
\end{theorem} 

\begin{proof} We will apply Theorem~\ref{thm:higher-moments} to prove Theorem~\ref{thm:higherOrderEstimates}. In all cases we choose $A=\beta (\mathcal{H}_{h,\eta} - \mu \mathcal{N})$, the Hamiltonian related to the Gibbs state $G_{h,\eta}(\beta,\mu)$ in \eqref{eq:GeneralGibbsState}. Since $\cN$ commutes with the Hamiltonian, the bound in \eqref{eq:particleNumberBoundsPerturbedState-k-higher} follows immediately. 

To prove \eqref{eq:aa-G-higher} and \eqref{eq:varianceNPlus-higher}, we apply Theorem~\ref{thm:higher-moments} with the choices $X=1+\beta \cN$, $B \in \{\lambda \beta a_p^* a_p,\lambda \beta (\cN_+-N_+^{\mathrm{G}})\}$ with a constant $\lambda \in \mathbb{R}$, whose absolute value is chosen sufficiently small. In both cases, we have $[B,X]=0$, $\pm  B\le X$, and 
\begin{align}\label{eq:commutator-appl-B1}
\quad \pm [B,[B,A]] \lesssim \beta^{3} \eta^{-1} \cN^2 \le b X^{\alpha} \quad \text{ with } \quad b \sim \beta \eta^{-1}, \ \alpha=2.
\end{align}
The above bound follows from \eqref{eq:commutator-appl-B1-B>0-2Andi}. The bound in \eqref{eq:particleNumberBoundsPerturbedState-k-higher} also holds for the perturbed Gibbs state $\Gamma_t=Z_t^{-1}e^{-A+tB}$ with $Z_t=\Tr [e^{-A+tB}]$, that is, we have
    \begin{equation}
   \sup_{t\in [-1,1]} \Tr [\cN^k \Gamma_t ]  \lesim_k  \eta^k. 
    \label{eq:particleNumberBoundsPerturbedState-k-higher-applyB1}
    \end{equation}
    The first moment bound $ \sup_{t\in [-1,1]}  |\Tr [B \Gamma_t]| \lesssim 1$ follows from Theorem \ref{thm:firstOrderAPriori}.    
Therefore,  \eqref{eq:corr-thm-conclusion-B>0} follows from an application of \eqref{eq:corr-thm-conclusion-Gamma} with $k=4$ and $\beta \sim \eta^{-2/3}$, namely 
\begin{align}
\Tr[ B^4 G_{h,\eta}(\beta,\mu)]= \Tr [ B^4 \Gamma_0 ] \lesssim 1 + \sup_{|t|\le 1} b^2 \Tr [ X^{6} \Gamma_t] \lesssim 1 + ( \beta \eta^{-1}) ^2 (\beta \eta)^6  = 1 + \beta^8 \eta^5 \lesssim 1.
\end{align}
It remains to prove part (d).

Since $\kappa < \frac{1}{4 \pi}[\upzeta(3/2)]^{2/3}$ we have $\mu_0 \sim - N^{2/3}$. We choose $B=\lambda a_p^*a_p$ with a constant $\lambda \in \mathbb{R}$ that satisfies $0< |\lambda| < \min\{1, - \beta \mu_0/2 \}$. From~\eqref{eq:unperturbedFirstOrderBounds} in Theorem~\ref{thm:firstOrderAPriori}, we have the first moment estimate $ \sup_{t\in [-1,1]}  |\Tr [B \Gamma_t]| \lesssim 1$, with the perturbed Gibbs state $\Gamma_t=Z_t^{-1}\exp(-A+tB)$, $Z_t=\Tr[\exp(-A+tB)]$.

This allows us to apply Theorem~\ref{thm:higher-moments} with $X=1+ \cN$. Let us we derive a bound for the commutator appearing in \eqref{eq:corr-thm-conclusion-B>0-Gamma}:
\begin{align}\label{eq:commutator-appl-B1-B>0-2}
\pm [B,[B,A]] &= \pm \beta \eta^{-1} \lambda^2 \sum_{k,q,r\in \Lambda^*} \hat v(k) \Big[a_p^* a_p,  \Big[ a_p^* a_p, a^*_{r+k} a^*_{q-k} a_r a_q \Big]\Big] \nn\\
&=  \pm \beta \eta^{-1} \lambda^2 \sum_{k,q,r\in \Lambda^*} \hat v(k) (\delta_{p,r+k} +\delta_{p,q-k}  -\delta_{p,r} -\delta_{p,s})^2 a^*_{r+k} a^*_{q-k} a_r a_q \nn\\
&\lesssim \beta \eta^{-1} \sum_{k \in \Lambda^*} \hat v(k) (a_p^* a_p + a^*_{p+k}a_{p+k} + a^*_{p-k}a_{p-k}) \cN. 
\end{align}
To obtain the first equality we used $[a_p^*a_p, a_r^*]= \delta_{p,r} a_r^*$ and $[a_p^*a_p, a_r] = - \delta_{p,r} a_r$. In the following use of \eqref{eq:commutator-appl-B1-B>0-2}, the sum over $k$ is compensated by the fact that $\hat v$ is summable. For the term in the sum with $\ell = k = 3$, $\alpha=2$, $b=\beta/\eta$ we have
\begin{align}
   \sup_{|t|\le 1} |\Tr [ [B,[B,A]]]X^\alpha \Gamma_t]| &\lesssim \beta^2 \eta^{-2}  \sup_{|t|\le 1} \sup_{q\in \Lambda^*} \Tr [a_q^* a_q \cN^3\Gamma_t] \nn \\
    &\leq \beta^2 \eta^{-2}   \sup_{|t|\le 1}  \sup_{q\in \Lambda^*} \sqrt{ \Tr [ (a_q^* a_q)^2\Gamma_t] }   \sqrt{ \Tr [ \cN^6 \Gamma_t] }   \nn\\
    &\lesssim \beta^2 \eta \lesssim 1. 
\label{eq:apap-third-moment}
\end{align}
In the last step we used \eqref{eq:aa-G} to obtain a bound for $\Tr[(a_p^* a_p)^2 \Gamma_t]$. We also used \eqref{eq:particleNumberBoundsPerturbedState-k-higher} to obtain a bound for the expectation of the sixths power of the number operator. To be precise, we stated  \eqref{eq:particleNumberBoundsPerturbedState-k-higher} only for the unperturbed Gibbs state but the same proofs apply if we replace it by the perturbed one. We use similar estimates than the one in \eqref{eq:apap-third-moment} to bound the terms with $\ell = 1,2$. Putting everything together, we find 
    \begin{align}
        \sup_{|t|\le 1-2\delta} \Tr [ B^3  G_{h,\eta}(\beta,\mu)  ]  &\lesim  1,
    \label{eq:apap-third-momentn}
    \end{align}
which proves the claim of Theorem~\ref{thm:higherOrderEstimates}.
\end{proof}

\section{The Gibbs state part V: 2-pdm and particle number variances}\label{sec:Gibbs-state-V}
In this section we prove pointwise bounds for the integral kernel of the 2-pdm of the Gibbs state $G_{\beta,N}$ in Fourier space (also called the four-point correlation function) as well as bounds for the variances of $\mathcal{N}_0$ and $\mathcal{N}_+$. Our analysis is based on the trace norm bound for $G_{\beta,N}$ in \eqref{eq:traceNormBound} and the moment bounds in Theorem~\ref{thm:higherOrderEstimates}.

\subsection{Pointwise bounds for the 2-pdm in Fourier space}
We start our discussion of pointwise bounds for the integral kernel of the 2-pdm of the Gibbs state with the computation of that of the state $\Gamma_{\beta,N}$ in \eqref{eq:referenceState}.
\subsubsection*{The kernel of the 2-pdm of $\Gamma_{\beta,N}$} 
To compute the four-point correlation function of $\Gamma_{\beta,N}$ we need to distinguish the cases $N_0 \geq N^{2/3}$ and $N_0 < N^{2/3}$ and we start with the former. Using $a_0 | z \rangle  = z | z \rangle$ with the coherent state $|z \rangle$ in \eqref{eq:coherentstate}, the Wick rule (the state $G(z)$ in \eqref{eq:referenceState-intro1} is quasi free), and Lemma~\ref{lem:1pdmAndPairingFunction}, we find
\begin{align} \label{eq:2pdm-tG-1}
\Tr[ a_0^* a_0^* a_0  a_0  \Gamma_{\beta,N}] &= \int_{\mathbb{C}} \langle z , a_0^* a_0^* a_0 a_0 z \rangle g^{\mathrm{BEC}}(z) \de z = \int_{\mathbb{C}} |z|^4 g^{\mathrm{BEC}}(z) \de z, \nn \\
\Tr[ a_0^* a_p^* a_0  a_p  \Gamma_{\beta,N}] &= \int_{\mathbb{C}} \langle z , a_0^* a_0 z \rangle  \Tr_+ [ a_p^* a_p G^{\mathrm{Bog}}(z)]   g^{\mathrm{BEC}}(z) \de z =  \gamma_p \int_{\mathbb{C}}  |z|^2 g^{\mathrm{BEC}}(z) \de z = \gamma_p \widetilde{N}_0,\nn\\
\Tr [ a_0^* a^*_0 a_pa_{-p}  \Gamma_{\beta,N}] &= \int_{\mathbb{C}} \Tr_+ [a_p a_{-p} G^{\mathrm{Bog}}(z)] \ \overline{z}^2 g^{\mathrm{BEC}}(z) \de z =  \alpha_p \int_{\mathbb{C}} |z|^2 g^{\mathrm{BEC}}(z)  \de z = \alpha_p \widetilde{N}_0, \nn\\
\Tr [ a_p^* a^*_q a_{r}a_{s}  \Gamma_{\beta,N} ] &= \int_{\mathbb{C}} \Tr_+ [ a_p^* a^*_q a_ra_{s} G^{\mathrm{Bog}}(z)) ] g^{\mathrm{BEC}}(z) \de z \nn\\
&=  \int_{\mathbb{C}} \Big( \Tr_+ [a_p^* a^*_q G^{\mathrm{Bog}}(z))] \Tr_+ [a_r a_{s} G^{\mathrm{Bog}}(z) ] + \Tr_+ [ a_p^* a_r G^{\mathrm{Bog}}(z) ] \Tr_+ [ a^*_q a_{s} G^{\mathrm{Bog}}(z)] \nn \\
&\hspace{1.5cm}+ \Tr_+ [ a_p^* a_{s} G^{\mathrm{Bog}}(z)] \Tr [ a^*_q a_r G^{\mathrm{Bog}}(z)] \Big) g^{\mathrm{BEC}}(z) \de z\nn\\
&=  \delta_{p,-q}\delta_{r,-s} \alpha_p \alpha_r  + \delta_{p,r}\delta_{q,s} \gamma_p \gamma_q +\delta_{p,s}\delta_{q,r} \gamma_p\gamma_q,
\end{align}
for all $p,q,r,s\in \Lambda_+^*$. All other expectations involving only one or three creation and annihilation operators with zero momentum equal zero. Since $\Gamma_{\beta,N}$ is translation-invariant, the above expectations are the only ones that do not vanish.

If $N_0<N^{2/3}$ we have $\widetilde \Gamma = G_{\beta,N}^{\mathrm{id}}$ with the Gibbs state $G_{\beta,N}^{\mathrm{id}}$ of the ideal gas in \eqref{eq:GibbsStateIdealGas}. An application of the Wick rule therefore shows
 \begin{align}\label{eq:2pdm-tG-2}
\Tr (a_0^* a^*_0 a_pa_{-p}  \Gamma_{\beta,N}) &=0, \quad \forall p \in \Lambda_+^*, \nn\\
\Tr (a_p^* a^*_q a_{r}a_{s}  \Gamma_{\beta,N}) &=   \Tr [a_p^* a_r G_{\beta,N}^{\mathrm{id}}] \Tr [a^*_q a_{s} G_{\beta,N}^{\mathrm{id}}] + \Tr [a_p^* a_{s} G_{\beta,N}^{\mathrm{id}}] \Tr[a^*_q a_r G_{\beta,N}^{\mathrm{id}}] \nn \\
&=  \frac{\delta_{p,r}\delta_{q,s} +\delta_{p,s}\delta_{q,r}}{(e^{\beta (p^2- \mu_0(\beta,N))}-1)(e^{\beta (q^2- \mu_0(\beta,N))}-1)} ,\quad \forall p,q,r,s\in  \Lambda^*
\end{align}
with $\mu_0(\beta,N)$ in \eqref{eq:idealgase1pdmchempot}. All  expectations involving only one or three creation or annihilation operators with zero momentum equal zero.

\subsubsection*{The kernel of the 2-pdm of $G_{\beta,N}$} 
In the next step we use the trace-norm bound in \eqref{eq:traceNormBound} and the moment bounds in Theorem~\ref{thm:higherOrderEstimates} to relate the correlation functions of the states $G_{\beta,N}$ and $\Gamma_{\beta,N}$. It is not difficult to check that the state $\Gamma_{\beta,N}$ also satisfies the bounds in Theorem~\ref{thm:higherOrderEstimates}. 

In the following we assume that $p,q,r,s\in \Lambda^*_+$. An applications of Lemma~\ref{lem:fundamentalLemmaForPWBounds} shows
\begin{align}\label{eq:2pdm-00pp-GtG}
\Tr [ a_0^* a_0 a_p^* a_p (G_{\beta,N} - \Gamma_{\beta,N} ) ] &\le 2 \sqrt{ \Tr [ (a_0^* a_0 a_p^* a_p)^2 (G_{\beta,N} + \Gamma_{\beta,N} ) ]} \sqrt{\Tr |G_{\beta,N} - \Gamma_{\beta,N} |}  \\
&\le \sqrt{2} \Big( \Tr [ \mathcal{N}^4  (G_{\beta,N} + \Gamma_{\beta,N}) ] \Big)^{1/4} \Big( \Tr [ (a_p^*a_p)^4 (G_{\beta,N} + \Gamma_{\beta,N}) ] \Big)^{1/4} \|G_{\beta,N} - \Gamma_{\beta,N} \|_1^{1/2}. \nn 
\end{align}
We use Theorems~\ref{thm:norm-approximation} and \ref{thm:higherOrderEstimates} and \eqref{eq:2pdm-tG-1} to check that this implies
\begin{equation}
    \Tr [a_0^* a_0 a_p^* a_p  G_{\beta,N} ]  = \Tr [a_0^* a_0 a_p^* a_p  \Gamma_{\beta,N}] + O(N^{5/3-1/96} ) = \gamma_p \widetilde{N}_0(\beta,N) + O(N^{5/3-1/96})
    \label{eq:2pdmAndi1}
\end{equation}
if $\kappa > \frac{1}{4 \pi}[\upzeta(3/2)]^{2/3}$. If $\kappa < \frac{1}{4 \pi}[\upzeta(3/2)]^{2/3}$ we use Theorem~\ref{thm:norm-approximation}, \eqref{eq:B-G-G'-2} in Lemma~\ref{lem:fundamentalLemmaForPWBounds} with $\theta = 3/2$, part~(d) of Theorem~\ref{thm:higherOrderEstimates}, and the Cauchy--Schwarz inequality to see that
\begin{align}
    |\Tr [a_0^* a_0 a_p^* a_p  (G_{\beta,N}-\Gamma_{\beta,N}) ] | &\lesssim \big( \Tr[ (a_0^* a_0)^3 (G_{\beta,N}+\Gamma_{\beta,N})] \ \Tr[ (a_p^* a_p)^3 (G_{\beta,N}+\Gamma_{\beta,N})] \ \Vert G_{\beta,N}-\Gamma_{\beta,N} \Vert_1 \big)^{1/3} \nn \\
    &\lesssim N^{-1/144},
\end{align}
and hence
\begin{align}
    \Tr [a_0^* a_0 a_p^* a_p  G_{\beta,N} ] = \Tr [a_0^* a_0 a_p^* a_p G_{\beta,N}^{\mathrm{id}} ] + O(N^{-1/144}) = \frac{N_0(\beta,N)}{\exp(\beta(p^2-\mu_0(\beta,N)))-1}  + O(N^{-1/144} ).
    \label{eq:2pdmAndi2}
\end{align}
To obtain the last equality, we also used \eqref{eq:2pdm-tG-2}. 

Similarly, we check that
\begin{align}
    |\Tr [ a_0^* a^*_0 a_p a_{-p}  (G_{\beta,N} -  \Gamma_{\beta,N})]| \leq& 2 \sqrt{\Tr[ a^*_p a^*_{-p} a_0 a_0 a_0^* a_0^* a_p a_{-p} (G_{\beta,N} + \Gamma_{\beta,N})]} \sqrt{\Vert G_{\beta,N} - \Gamma_{\beta,N} \Vert_1} \nn \\
    \lesssim& \Big( \Tr [ \mathcal{N}^4  (G_{\beta,N} + \Gamma_{\beta,N}) ] \Big)^{1/4} \Big( \Tr [ (a_p^*a_p)^4 (G_{\beta,N} + \Gamma_{\beta,N}) ] \Big)^{1/8} \nn \\ 
    & \times \Big( \Tr [ (a_{-p}^*a_{-p})^4 (G_{\beta,N} + \Gamma_{\beta,N}) ] \Big)^{1/8} \|G_{\beta,N} - \Gamma_{\beta,N} \|_1^{1/2} \nn \\
    & \lesssim N^{5/3 - 1/96}
    \label{eq:2pdmAndi3}
\end{align}
holds if $\kappa > \frac{1}{4 \pi}[\upzeta(3/2)]^{2/3}$ and that
\begin{equation}
    |\Tr [ a_0^* a^*_0 a_p a_{-p}  (G_{\beta,N} -  \Gamma_{\beta,N})]| \lesssim N^{-1/144}    
    \label{eq:2pdmAndi4}
\end{equation}
$\kappa < \frac{1}{4 \pi}[\upzeta(3/2)]^{2/3}$. With \eqref{eq:2pdm-tG-1} and \eqref{eq:2pdm-tG-2} we conclude that
\begin{equation}
    \Tr [ a_0^* a^*_0 a_p a_{-p} G_{\beta,N} ] = \widetilde{N}_0(\beta,N) \alpha_p + O(N^{5/3-1/96}) 
    \label{eq:2pdmAndi5}
\end{equation}
if $\kappa > \frac{1}{4 \pi}[\upzeta(3/2)]^{2/3}$ and 
\begin{equation}
    \Tr [ a_0^* a^*_0 a_p a_{-p} G_{\beta,N} ] = O(N^{-1/144}) 
    \label{eq:2pdmAndi6}
\end{equation}
if $\kappa < \frac{1}{4 \pi}[\upzeta(3/2)]^{2/3}$.

The bound in the fourth line of \eqref{eq:2-pdmNonCondensedPhase} has been proved in \eqref{eq:lowerBoundFE5}. The bounds for all other matrix elements follow from similar arguments. It remains to compute the variances of $\mathcal{N}_0$ and $\mathcal{N}_+$ in the state $G_{\beta,N}$.

\subsection{\texorpdfstring{Variance of $\mathcal{N}_+$}{varianceN+}}

We first consider the variance of $\mathcal{N}_+$ and we start with its computation in the state $\Gamma_{\beta,N}$. If $N_0 \geq N^{2/3}$ we have
\begin{equation}
    \Tr[ \mathcal{N}_+ \Gamma_{\beta,N} ] = \int_{\mathbb{C}} \Tr_+[ \mathcal{N}_+ G^{\mathrm{Bog}}(z) ] g^{\mathrm{BEC}}(z) \de z = \sum_{p \in \Lambda_+^*} \gamma_p,
    \label{eq:2pdmAndi9}
\end{equation}
and the variance reads
\begin{equation}
    \Tr[ \mathcal{N}^2_+ \Gamma_{\beta,N} ] - \left( \Tr[ \mathcal{N}_+ \Gamma_{\beta,N} ] \right)^2 = \sum_{p \in \Lambda_+^*} \left( \alpha_p^2 + \gamma_p^2 + \gamma_p \right).
    \label{eq:2pdmAndi10}
\end{equation}
If $N_0 < N^{2/3}$ all expectations are given by that of the Gibbs state of the ideal gas $G^{\mathrm{id}}_{\beta,N}$ and we find
\begin{equation}
    \Tr[ \mathcal{N}_+ \Gamma_{\beta,N} ] = \Tr [ \cN_+G_{\beta,N}^{\mathrm{id}} ] = \sum_{p \in \Lambda^*_+} \frac{1}{\exp(\beta(p^2 - \mu_0))-1} 
    \label{eq:2pdmAndi11}
\end{equation}
as well as
\begin{equation}
    \Tr[ \mathcal{N}^2_+ \Gamma_{\beta,N} ] - \left( \Tr[ \mathcal{N}_+ \Gamma_{\beta,N} ] \right)^2 = \sum_{p \in \Lambda_+^*} \frac{1}{\exp(\beta(p^2 - \mu_0))-1}  \left( 1+ \frac{1}{\exp(\beta(p^2 - \mu_0))-1} \right).
    \label{eq:2pdmAndi12}
\end{equation}

We recall the notation $N_+^{\mathrm{G}}(\beta,N) = \Tr[\mathcal{N}_+ G_{\beta,N}]$. An application of Lemma~\ref{lem:fundamentalLemmaForPWBounds} shows
\begin{equation}
    |\Tr[ (\mathcal{N}_+ - N_+^{\mathrm{G}})^2 (G_{\beta,N} - \Gamma_{\beta,N} ) ]| \leq 2 \sqrt{\Tr[ (\mathcal{N}_+ - N_+^{\mathrm{G}})^4 (G_{\beta,N} + \Gamma_{\beta,N} ) ]} \sqrt{\Vert G_{\beta,N} - \Gamma_{\beta,N} \Vert_1}.
    \label{eq:2pdmAndi13}
\end{equation}
We use \eqref{eq:varianceNPlus-higher} to bound the expectation of the operator $(\mathcal{N}_+ - N_+^{\mathrm{G}})^4$ in the state $G_{\beta,N}$ by a constant times $N^{8/3}$. We also need the same bound with $G_{\beta,N}$ replaced by $\Gamma_{\beta,N}$. From \eqref{eq:bound1pdm12} (without the logarithmic factor) we know that
\begin{equation}\label{eq:N+-tN+-new}
    | \Tr[ \mathcal{N}_+ G_{\beta,N} - \widetilde{N}_+(\beta,N) | \lesssim N^{2/3-1/48}
\end{equation}
holds with $\widetilde{N}_+(\beta,N) = \Tr [\mathcal{N}_+ \Gamma_{\beta,N}]$. Using this, we write ($A = O(N^{2/3-1/48} )$, $\Delta \mathcal{N}_+ = \mathcal{N}_+ - \widetilde{N}_+$)
\begin{align}
    \Tr[(\mathcal{N}_+ - N_+^{\mathrm{G}})^4 \Gamma_{\beta,N}] &= \Tr[(\Delta \mathcal{N}_+ + A)^4 \Gamma_{\beta,N}] = \sum_{k=0}^4 \binom{4}{k} A^{4-k} \Tr[ (\Delta \mathcal{N}_+)^k  \Gamma_{\beta,N}]  \nn \\
    &= \Tr [(\Delta \mathcal{N}_+)^4  \Gamma_{\beta,N}] + \sum_{k=1}^3 \binom{4}{k} A^{4-k} \Tr[ (\Delta \mathcal{N}_+)^k \Gamma_{\beta,N}].
    \label{eq:2pdmAndi14}
\end{align}
With the Wick rule we check that $|\Tr[(\Delta \mathcal{N}_+)^{k} \Gamma_{\beta,N}]| \lesssim N^{2k/3}$ holds for $k=1,2,3,4$. We use this bound in \eqref{eq:2pdmAndi14} to show that 
\begin{equation}
    \Tr[(\mathcal{N}_+ - N_+^{\mathrm{G}})^4 \Gamma_{\beta,N}] \lesssim N^{8/3}.
    \label{eq:2pdmAndi15}
\end{equation}
In particular, the right-hand side of \eqref{eq:2pdmAndi13} is bounded by a constant times $N^{4/3-1/96}$ and we have
\begin{equation}
    \Tr[ (\mathcal{N}_+ - N_+^{\mathrm{G}})^2 G_{\beta,N} ] = \Tr[ (\mathcal{N}_+ - N_+^{\mathrm{G}})^2 \Gamma_{\beta,N} ] + O(N^{4/3-1/96}).
    \label{eq:2pdmAndi16}
\end{equation}

Combining \eqref{eq:2pdmAndi16} and the bound $|N_+^{\mathrm{G}}-\widetilde{N}_+| \lesssim N^{2/3-1/48}$ from \eqref{eq:N+-tN+-new} we obtain 
\begin{equation}
    \Tr[ (\mathcal{N}_+ - N_+^{\mathrm{G}})^2 G_{\beta,N} ] = \Tr[ (\mathcal{N}_+ - \widetilde{N}_+)^2 \Gamma_{\beta,N} ] + O(N^{4/3-1/96}).
    \label{eq:2pdmAndi17}
\end{equation}
The term involving $\Gamma_{\beta,N}$ on the right-hand side equals \eqref{eq:2pdmAndi10} if $\kappa > \frac{1}{4 \pi}[\upzeta(3/2)]^{2/3}$ and \eqref{eq:2pdmAndi12} if $\kappa < \frac{1}{4 \pi}[\upzeta(3/2)]^{2/3}$. It remains to compute the variance of $\mathcal{N}_0$ in the state $G_{\beta,N}$.

\subsection{\texorpdfstring{Variance of $\mathcal{N}_0$}{varianceN0}}

We start with the case $\kappa > \frac{1}{4 \pi}[\upzeta(3/2)]^{2/3}$. Here we know from Theorem~\ref{thm:particleNumberDistributionBEC2} that the variance of $\mathcal{N}_0$ in the state $G_{\beta,N}$ is of order $N^{5/3}$ provided. Since $ \Tr [ (\cN_+ - N_+)^2 G_{\beta,N} ] \sim N^{4/3}$ in this parameter regime we know that the particle number variance in the entire system is of order $N^{5/3}$ and comes from the condensate. To compute the variance of the number of particles in the condensate, we can therefore simply compute the variances of $\mathcal{N}$ and estimate the difference. 

In the following we denote  by $\mathbf{Var}_{\beta,N}(A)$ the variance of the operator $A$ in the state $G_{\beta,N}$. By $\mathbf{Cov}_{\beta,N}(A,B)$ we denote for two operators $A,B$ with $[A,B] = 0$ the related covariance. We have 
\begin{equation}
    \mathbf{Var}_{\beta,N}(\mathcal{N}) = \mathbf{Var}_{\beta,N}(\mathcal{N}_0) + \mathbf{Var}_{\beta,N}(\mathcal{N}_+) + 2 \mathbf{Cov}_{\beta,N}(\mathcal{N}_0,\mathcal{N}_+).
    \label{eq:2pdmAndi18}
\end{equation}
Using this and $|\mathbf{Cov}_{\beta,N}(\mathcal{N}_0,\mathcal{N}_+)| \leq \sqrt{ \mathbf{Var}_{\beta,N}(\mathcal{N}_0) \mathbf{Var}_{\beta,N}(\mathcal{N}_+)}$, we conclude that
\begin{align}
    \frac{1}{1+\epsilon} \left[ \mathbf{Var}_{\beta,N}(\mathcal{N}) - \left( 1 + \epsilon^{-1} \right) \mathbf{Var}_{\beta,N}(\mathcal{N}_+) \right] &\leq \mathbf{Var}_{\beta,N}(\mathcal{N}_0) \nn \\
    &\leq \frac{1}{1-\epsilon} \left[ \mathbf{Var}_{\beta,N}(\mathcal{N}) + \left( 1 +\epsilon^{-1} \right) \mathbf{Var}_{\beta,N}(\mathcal{N}_+) \right]
    \label{eq:2pdmAndi19}
\end{align}
holds for any $0 < \epsilon < 1$. Eq.~\eqref{eq:2pdmAndi17} can be used to estimate the variance of $\mathcal{N}_+$ in the above equation. It therefore remains to compute the variance of $\mathcal{N}$ in the state $G_{\beta,N}$.

To that end, we apply a Griffith argument and perturb $\hat{v}(0)$. We choose $\delta \in \mathbb{R}$ such that $|\delta| < \hat{v}(0)$ holds. An application of the lower bound for the free energy in \eqref{eq:finalLowerBound14} with $\hat{v}(p)$ replaced by $\hat{v}_{\delta}(p) = \hat{v}(p) + \delta \ \delta_{p,0}$ and $F^{\mathrm{BEC}}$ replaced by $F_{\mathrm{c}}^{\mathrm{BEC}}$ shows
\begin{equation}
    \mathcal{F}(G_{\beta,N}) + \frac{\delta}{2N} \sum_{u,v \in \Lambda^*} \Tr[ a_u^* a_v^* a_u a_v G_{\beta,N} ] \geq F^{\mathrm{Bog}}(\beta,N) + \frac{(\hat{v}(0) + \delta) N}{2} + F_{\mathrm{c},\delta}^{\mathrm{BEC}}(\beta,N_0(\beta,N)) - C N^{1/2}.
    \label{eq:2pdmAndi20}
\end{equation}
Here $F_{\mathrm{c},\delta}^{\mathrm{BEC}}(\beta,N_0(\beta,N))$ denotes the free energy in \eqref{eq:GibbsDistributionDiscrete} with $\hat{v}_{\delta}(p)$. To see that $F_{\mathrm{c}}^{\mathrm{BEC}}$ can be kept in the lower bound, we refer to the arguments below \eqref{eq:finalLowerBound12}. An application of part~(a) of Lemma~\ref{prop:FreeEnergyBEC} shows
\begin{equation}
    F_{\mathrm{c},\delta}^{\mathrm{BEC}}(\beta,N_0(\beta,N)) = \frac{1}{2 \beta} \ln \left( \frac{ (\hat{v}(0) +\delta ) \beta  }{2 \pi N} \right) + O\left( \exp\left(- c N^{1/6} \right) \right).
    \label{eq:2pdmAndi21}
\end{equation}
In combination, \eqref{eq:2pdmAndi20}, \eqref{eq:2pdmAndi21}, and the upper bound in \eqref{eq:UpperBoundAndi10} with $F^{\mathrm{BEC}}$ replaced by $F_{\mathrm{c}}^{\mathrm{BEC}}$ (that is, we do not use \eqref{eq:UpperBoundAndi9} in the upper bound) imply
\begin{equation}
     \frac{\delta}{2N} \Tr[ \mathcal{N}^2 G_{\beta,N} ] \geq \frac{\delta N}{2} + \frac{1}{2 \beta} \ln\left( \frac{\hat{v}(0) + \delta}{\hat{v}(0)} \right) - C N^{1/2}.
     \label{eq:2pdmAndi22}
\end{equation}
Accordingly, we have 
\begin{align}
    \Tr[\mathcal{N}^2 G_{\beta,N}] - N^2 &\geq \frac{N}{\beta \delta} \ln\left( \frac{\hat{v}(0) + \delta}{\hat{v}(0)} \right) - \frac{C N^{3/2}}{\delta} \quad \text{ if } \delta > 0 \quad \text{ and } \nn \\ 
    \Tr[\mathcal{N}^2 G_{\beta,N}] - N^2 &\leq \frac{N}{\beta \delta} \ln\left( \frac{\hat{v}(0) + \delta}{\hat{v}(0)} \right) + \frac{C N^{3/2}}{\delta} \quad \text{ if } \delta < 0. 
    \label{eq:2pdmAndi23}
\end{align}
Choosing $\delta = \pm N^{-1/12}$, we obtain
\begin{equation}
    \Tr[\mathcal{N}^2 G_{\beta,N}] - N^2 = \frac{N}{\beta \hat{v}(0)} + O(N^{5/3-1/4}).
    \label{eq:2pdmAndi24}
\end{equation}

To complete the argument we combine \eqref{eq:2pdmAndi17}, \eqref{eq:2pdmAndi19} with $\epsilon = N^{-1/6}$, and \eqref{eq:2pdmAndi24}, which gives
\begin{equation}
    \Tr[\mathcal{N}_0^2 G_{\beta,N}] - (\Tr[\mathcal{N}_0 G_{\beta,N}])^2 = \frac{N}{\beta \hat{v}(0)} + O(N^{5/3-1/6}).
    \label{eq:2pdmAndi25}
\end{equation}
In the last step we use \eqref{eq:1pdmBEC} to replace $\Tr[\mathcal{N}_0 G_{\beta,N}]$ by $\widetilde{N}_0(\beta,N)$. The final result reads
\begin{equation}
    \Tr[\mathcal{N}_0^2 G_{\beta,N}] = \widetilde{N}_0(\beta,N) + \frac{N}{\beta \hat{v}(0)} + O(N^{5/3-1/12}).
    \label{eq:2pdmAndi26}
\end{equation}
It remains to consider the case $\kappa < \frac{1}{4 \pi}[\upzeta(3/2)]^{2/3}$.  

Here we have $\Gamma_{\beta,N} = G_{\beta,N}^{\mathrm{id}}$. Applications of \eqref{eq:traceNormBound} (without the logarithmic factor), Lemma~\ref{lem:fundamentalLemmaForPWBounds} with $\theta = 3/2$, and part~(d) of Theorem~\ref{thm:higherOrderEstimates} (this bound also holds for $G_{\beta,N}^{\mathrm{id}}$) show 
\begin{equation}
    | \Tr[ (a^*_0 a_0)^2 (G_{\beta,N} - G_{\beta,N}^{\mathrm{id}}) ] \lesssim \big( \Tr[(a_0^* a_0)^3 (G_{\beta,N} + G_{\beta,N}^{\mathrm{id}})] \big)^{2/3} \ \Vert G_{\beta,N} - G_{\beta,N}^{\mathrm{id}} \Vert_1^{1/3} \lesssim N^{-1/144}.
    \label{eq:2pdmAndi27}
\end{equation}
But this implies
\begin{equation}
    \Tr[(a_0^* a_0)^2 G_{\beta,N}] = \Tr[(a_0^* a_0)^2 G^{\mathrm{id}}_{\beta,N}] + O(N^{-1/144}) = 2 N_0^2(\beta,N) + N_0(\beta,N) + O(N^{-1/144})
    \label{eq:2pdmAndi28}
\end{equation}
and ends our discussion of the variance of $\mathcal{N}_0$ in the state $G_{\beta,N}$. Theorem~\ref{thm:main2} is proved.

\vspace{0.5cm}

\appendix

\begin{center}
\huge \textsc{--- Appendix ---}
\end{center}

\section{Properties of the condensate functionals} \label{app:effcondensate}
In this section we collect properties of the discrete and the continuous versions of our effective condensate functional. The first statement provides us with bounds for the chemical potential related to the continuous condensate functional. The proofs of parts (a) and (b) can be found in \cite[Lemma~C.1]{BocDeuSto-24}, the proof of part (c) is a straightforward adaption of the related proof in \cite[Lemma~B.1.3]{CapDeu-23} (one needs to send the particle number cut-off to infinity).

\begin{lemma}
	\label{lem:ChemPotBECCont}
	We consider the limit $N \to \infty$, $\beta N^{2/3} \to \kappa \in (0,\infty)$. Let $g$ be the Gibbs distribution in \eqref{eq:GibbsDistributionDiscrete} and assume that $\int_{\mathbb{C}} |z|^2 g(z) \de z = M(N)$ with a sequence $M(N)$ of positive numbers. The chemical potential $\mu$ related to $g$ satisfies the following statements for a given $\epsilon > 0$:
	\begin{enumerate}[label=(\alph*)]
		\item If $M \gtrsim N^{5/6 + \epsilon}$ then there exists a constant $c>0$ such that
		\begin{equation}
			\left| \mu - \frac{\hat{v}(0) M}{N} \right| \lesssim \exp\left( - c N^{\epsilon} \right).
			\label{eq:LemChemPotBEC1}
		\end{equation}
		\item If $M \lesssim N^{5/6 - \epsilon}$ then we have
		\begin{equation}
			\left| \mu + \frac{1}{\beta M} \right| \lesssim \frac{N^{-2 \epsilon}}{\beta M}.
			\label{eq:LemChemPotBEC2}
		\end{equation}
		\item For any $M = M(N)$, we have
		\begin{equation}
			| \mu | \lesssim \left( \frac{1}{\beta M} + \frac{1}{\sqrt{\beta N}} + \frac{M}{N} \right).
			\label{eq:LemChemPotBEC3}
		\end{equation} 
	\end{enumerate}
\end{lemma}

In the next lemma we investigate the free energy of our continuous effective condensate theory. 

\begin{lemma}
	\label{prop:FreeEnergyBEC}
	We consider the limit $N \to \infty$, $\beta N^{2/3} \to \kappa \in (0,\infty)$. The following statements hold for given $\epsilon > 0$:
	\begin{enumerate}[label=(\alph*)]
		\item Assume that $M \gtrsim N^{5/6 + \epsilon}$. There exists a constant $c>0$ such that 
		\begin{equation}
			F_{\mathrm{c}}^{\mathrm{BEC}}(\beta,M) = \frac{1}{2 \beta} \ln \left( \frac{ \hat{v}(0) \beta  }{2 \pi N} \right) + O\left( \exp\left(- c N^{\epsilon} \right) \right).
			\label{eq:FreeEnergyBECInteractingLimit_app}
		\end{equation}
		\item Assume that $M \lesssim N^{5/6 - \epsilon}$. Then 
		\begin{equation}
			F_{\mathrm{c}}^{\mathrm{BEC}}(\beta,M) = - \frac{1}{\beta} \ln(M) - \frac{1}{\beta} + O\left( N^{2/3 - 2 \epsilon} \right)
			\label{eq:FreeEnergyBECNonInteractingLimit_app}
		\end{equation}
		holds. In particular, $F_{\mathrm{c}}^{\mathrm{BEC}}(\beta,M)$ is independent of $\hat{v}(0)$ at the given level of accuracy. 
        \item Assume that $M \gtrsim N^{5/6}$. Then 
        \begin{equation}
            \left| F_{\mathrm{c}}^{\mathrm{BEC}}(\beta,M) + \frac{5}{6 \beta} \ln(N) \right| \lesssim N^{2/3}.
            \label{eq:FreeEnergyBECApriori1}
        \end{equation}
        \item Assume that $M \lesssim N^{5/6}$. Then
        \begin{equation}
            \left| F_{\mathrm{c}}^{\mathrm{BEC}}(\beta,M) +  \frac{1}{\beta} \ln(M) \right| \lesssim N^{2/3}.
            \label{eq:FreeEnergyBECApriori2}
        \end{equation}
	\end{enumerate} 
\end{lemma}
\begin{proof}
    The proof of the first two items can be found in \cite[Lemma~C.2]{BocDeuSto-24}. Note, however, that our definition of the free energy and that in the reference differ by the additive constant $\hat{v}(0) M^2/(2N)$. 
    
    Let us prove part~(c). We use the notations $h = \hat{v}(0)/(2N) > 0$, $\sigma = \mu \sqrt{\beta/(4h)}$. Inspection of the proof of \cite[Lemma~C.1]{BocDeuSto-24} shows that $M \gesssim N^{5/6}$ is equivalent to $\sigma \gesssim -1$. We have
    \begin{equation}
        \int_{\mathbb{C}} \exp(-\beta (h |z|^4 - \mu |z|^2)) \de z = \int_0^{\infty} \exp(-\beta (h x^2 - \mu x) \de x = \frac{\exp(\sigma^2)}{\sqrt{\beta h}} \int_{-\sigma}^{\infty} \exp(-x^2) \de x.
        \label{eq:FreeEnergyBECApriori21}
    \end{equation}
    To obtain this result, we wrote the two-dimensional integration over $\mathbb{C}$ with respect to the measure $\de z = \de x \de y / \pi$ in polar coordinates $(r, \phi)$, integrated out $\phi$, and afterwards introduced the variable $x = r^2$. For the free energy this implies 
    \begin{align}
        F_{\mathrm{c}}^{\mathrm{BEC}}(\beta,M) &= -\frac{\sigma^2}{\beta} + \frac{\ln(\beta h)}{2 \beta} + \mu M - h M^2 + O(N^{2/3}) \nonumber \\
        &=  - \frac{5}{6 \beta} \ln(N) + \frac{\hat{v}(0)}{2 N} \left( M - \frac{\mu N}{\hat{v}(0)} \right)^2 + O(N^{2/3}). 
        \label{eq:FreeEnergyBECApriori22}
    \end{align}
    To obtain a bound for the second term on the right-hand side of \eqref{eq:FreeEnergyBECApriori22} we note that
    \begin{equation}
        M = \frac{\int_{\mathbb{C}} |z|^2  \exp(-\beta (h |z|^4 - \mu |z|^2)) \de z}{\int_{\mathbb{C}} \exp(-\beta (h |z|^4 - \mu |z|^2)) \de z} = \frac{\mu}{2 h} + \frac{1}{\sqrt{\beta h}} \frac{\int_{-\sigma}^{\infty} x \exp(-x^2) \de x }{\int_{-\sigma}^{\infty} \exp(-x^2) \de x} = \frac{\mu N}{\hat{v}(0)} + O( N^{5/6} ).
        \label{eq:FreeEnergyBECApriori23}
    \end{equation}
    In combination, \eqref{eq:FreeEnergyBECApriori22} and \eqref{eq:FreeEnergyBECApriori23} show
    \begin{equation}
        F_{\mathrm{c}}^{\mathrm{BEC}}(\beta,M) = - \frac{5}{6 \beta} \ln(N) + O(N^{2/3}),
        \label{eq:FreeEnergyBECApriori24}
    \end{equation}
    which proves part~(c). It remains to prove part~(d) of Lemma~\ref{prop:FreeEnergyBEC}.

    We define the chemical potential $\widetilde{\mu}$ by 
    \begin{equation}
        M = \frac{\int_{\mathbb{C}} |z|^2  \exp(\beta \widetilde{\mu} |z|^2)) \de z}{\int_{\mathbb{C}} \exp(\beta \widetilde{\mu} |z|^2) \de z} = \frac{1}{-\beta \widetilde{\mu}}.
        \label{eq:FreeEnergyBECApriori25}
    \end{equation}
    Let $\varrho(z)$ be a probability distribution on $\mathbb{C}$ with $\int_{\mathbb{C}} |z|^2 \varrho(z) \de z = M$. Using $M \lesssim N^{5/6}$ and $\hat{v}(0) > 0$, we check that 
    \begin{align}
        \mathcal{F}^{\mathrm{BEC}}_{\mathrm{c}}(\rho) - \frac{M^2}{2 \hat{v}(0)} \geq -\frac{1}{\beta} \ln\left( \int_{\mathbb{C}} \exp(\beta \widetilde{\mu}|z|^2) \de z \right) + \widetilde{\mu} M - C N^{2/3} = -\frac{1}{\beta} \ln(M) - \frac{1}{\beta} - C N^{2/3}.
        \label{eq:FreeEnergyBECApriori26}
    \end{align}

    To prove an upper bound for $F_{\mathrm{c}}^{\mathrm{BEC}}(\beta,M)$ we use the distribution
    \begin{equation}
        \varrho_0(z) = \frac{\exp(\beta \widetilde{\mu} |z|^2))}{\int_{\mathbb{C}} \exp(\beta \widetilde{\mu} |z|^2) \de z}
        \label{eq:FreeEnergyBECApriori27_0}
    \end{equation}
    as a trial state. A straightforward computation that uses $M \lesssim N^{5/6}$ shows
    \begin{equation}
        \mathcal{F}^{\mathrm{BEC}}_{\mathrm{c}}(\rho_0) \leq - \frac{1}{\beta} \ln(M) + C N^{2/3}.
        \label{eq:FreeEnergyBECApriori27}
    \end{equation}
    In combination, \eqref{eq:FreeEnergyBECApriori26} and \eqref{eq:FreeEnergyBECApriori27} prove part~(d).
\end{proof}

In the following lemma we prove a statement that quantifies how a change in the expected particle number of the continuous effective condensate theory changes the chemical potential.

\begin{lemma}
\label{lem:perturbationTheoryChemicalPotential}
    We consider the limit $N \to \infty$, $\beta/ \beta_{\mathrm{c}}\to \kappa \in (0,\infty)$ with $\beta_{\mathrm{c}}$ in \eqref{eq:crittemp}, and assume $\hat{v}(0) > 0$. Let $\mu$ and $\widetilde{\mu}$ be the chemical potentials leading to an expected number of $M$ and $\widetilde{M}$ in the continuous effective condensate theory defined in \eqref{eq:condensatefunctionalcontinuous}. Then we have
    \begin{equation}
        | \mu - \widetilde{\mu} | \lesssim \frac{|M-\widetilde{M}|}{N} f(\xi)
        \label{eq:perturbationTheoryChemicalPotential1_0}
    \end{equation}
    with
    \begin{equation}
        f(x) = \begin{cases} 1 & \text{if } x \geq -1 \\ x^2 & \text{if } x < -1 \end{cases} \quad \text{ and } \quad \xi = \sqrt{\frac{\beta N}{4 \hat{v}(0)}} \min\{ \mu , \widetilde{\mu}  \}.
        \label{eq:perturbationTheoryChemicalPotential2_0}
    \end{equation}
\end{lemma}
\begin{proof}
    We recall the definitions of $h$ and $\sigma$ above \eqref{eq:FreeEnergyBECApriori21}. A short computation shows 
    \begin{equation}
        M = \frac{ \int_{\mathbb{C}} |z|^2 \exp( -\beta ( h |z|^4 - \mu |z|^2 ) ) \de z }{\int_{0}^{\infty} \exp( -\beta ( h |z|^4 - \mu |z|^2 ) ) \de z } = \frac{\int_{0}^{\infty} x \exp( -\beta ( h x^2 - \mu x ) ) \de x}{\int_{0}^{\infty} \exp( -\beta ( h x^2 - \mu x ) ) \de x} = \frac{\mu}{2 h} + \frac{ g\left( -\sqrt{2} \sigma \right) }{\sqrt{2 \beta h}}
        \label{eq:perturbationTheoryChemicalPotential3_0}
	\end{equation}
    with 
    \begin{equation}
        g(x) = \frac{1}{ \sqrt{2} \exp(x^2/2) \int_{(x/\sqrt{2})}^{\infty} \exp(-t^2) \de t}.
        \label{eq:perturbationTheoryChemicalPotential4}
    \end{equation}
    To obtain the first equality, we applied the same coordinate transformation as in \eqref{eq:FreeEnergyBECApriori21}. The distribution of the particle number in the continuous effective condensate theory is called a truncated Gaussian. 

    In the following we assume without loss of generality that $M \geq \widetilde{M}$, and hence $\mu \geq \widetilde{\mu}$. We also denote by $\widetilde{\sigma}$ the parameter $\sigma$ when $\mu$ is replaced by $\widetilde{\mu}$. Using \eqref{eq:perturbationTheoryChemicalPotential4}, a first order Taylor expansion, and $g'(x) = g(x)[g(x)-x]$, we check that there exists a number $a \in \{ t \sqrt{2} \sigma + (1-t) \sqrt{2} \widetilde{\sigma} \ | \ t \in [0,1] \}$ such that
    \begin{equation}
        0 \leq \mu-\widetilde{\mu} +\frac{ g( -\sqrt{2} \sigma ) - g(-\sqrt{2} \widetilde{\sigma}) }{\sqrt{\beta/(2h))}} = (\mu - \widetilde{\mu}) \left[ 1 - g(-a) (g(-a) + a) \right] = 2h( M - \widetilde{M} )
        \label{eq:perturbationTheoryChemicalPotential5}
    \end{equation}
    holds. The function $g$ is strictly positive, continuous, behaves as $\exp(-x^2/2)/\sqrt{2}$ for $x \to -\infty$, and as 
    \begin{equation}
        g(x) = \frac{x}{1 - \frac{1}{x^2} + \frac{6}{x^4} + O(x^{-6})}
        \label{eq:perturbationTheoryChemicalPotential6}
    \end{equation}
    for $x \to \infty$, see \cite[Eq. 7.1.23]{AbraStegun1972}. Using this, we check that the function $G(x) = 1 - g(x)[g(x)-x]$ is strictly positive, continuous, goes to $1$ for $x \to -\infty$, and behaves as
    \begin{equation}
        G(x) = \frac{7}{x^2} + O(x^{-4})
        \label{eq:perturbationTheoryChemicalPotential7}
    \end{equation}
    for $x \to \infty$. 

    Let us first assume that $a \geq -1$. In this case \eqref{eq:perturbationTheoryChemicalPotential5} and the properties of the function $G$ imply
    \begin{equation}
        0 \leq \mu - \widetilde{\mu} \lesssim \frac{M - \widetilde{M}}{N}.
        \label{eq:perturbationTheoryChemicalPotential8}
    \end{equation}
    If $a < -1$ we have
    \begin{equation}
        0 \leq \mu - \widetilde{\mu} \lesssim a^2 \frac{M - \widetilde{M}}{N} \leq \xi^2 \frac{M - \widetilde{M}}{N}.
        \label{eq:perturbationTheoryChemicalPotential9}
    \end{equation}
    In combination, \eqref{eq:perturbationTheoryChemicalPotential8} and \eqref{eq:perturbationTheoryChemicalPotential9} prove \eqref{eq:perturbationTheoryChemicalPotential1_0}.
\end{proof}

In the next lemma we quantify how a change in the expected particle number of the continuous effective condensate theory affects the free energy. Its proof uses Lemma~\ref{lem:perturbationTheoryChemicalPotential}. 

\begin{lemma}\label{lem:replaceParticleNumberInCondensateEnergy}
    We consider the limit $N \to \infty$, $\beta/ \beta_{\mathrm{c}}\to \kappa \in (0,\infty)$ with $\beta_{\mathrm{c}}$ in \eqref{eq:crittemp}, and assume $\hat{v}(0) > 0$. Let $M(N), \widetilde{M}(N)$ be two functions that satisfy $M,\widetilde{M} \gesssim N^{5/6 - \epsilon}$ with some $0 < \epsilon < 1/6$ and $| M - \widetilde{M} | \lesssim a_N N^{2/3}$ with some $0 \leq \delta < 1/3$. Then there exists a constant $c>0$ such that 
    \begin{equation}
        | F_{\mathrm{c}}^{\mathrm{BEC}}(\beta,M) - F_{\mathrm{c}}^{\mathrm{BEC}}(\beta,\widetilde{M}) | \lesssim f^{\mathrm{BEC}}(M,\widetilde{M}) \label{eq:perturbationTheoryChemicalPotential1}
    \end{equation}
    with
    \begin{equation}
       f^{\mathrm{BEC}}(M,\widetilde{M}) = \begin{cases}
            \exp(-cN^{\epsilon}) & \ \text{ if } \ \ M,\widetilde{M} \gesssim N^{5/6 + \epsilon} \ \text{ with } \ \epsilon > 1/24, \\
            a_N N^{1/2 + 3 \epsilon} & \ \text{ if } \ \ N^{5/6-\epsilon} \lesssim M, \widetilde{M} \lesssim N^{5/6 + \epsilon} \ \text{ with } \ 0 < \epsilon \leq 1/24, \\
            N^{2/3 - \epsilon} + a_N N^{1/2 + \epsilon} & \ \text{ if } \ \ M,\widetilde{M} \lesssim N^{5/6-\epsilon} \ \text{ with } \ 1/24 < \epsilon < 1/6. \\
        \end{cases}
        \label{eq:perturbationTheoryChemicalPotential1b}
    \end{equation}
    In the second parameter regime we also require $a_N N^{3 \epsilon}  \lesssim N^{1/6}$.
\end{lemma}
\begin{proof}
    We recall the definitions of $h$ and $\sigma$ above \eqref{eq:FreeEnergyBECApriori21}. We also denote by $\widetilde{\sigma}$ the parameter $\sigma$ when $\mu$ is replaced by $\widetilde{\mu}$. A straightforward computation shows that
    \begin{equation}
        F_{\mathrm{c}}^{\mathrm{BEC}}(\beta,M) = \frac{1}{2 \beta} \left[ 1 -  g^2\left(-\sqrt{2} \sigma \right)  \right] -\frac{1}{\beta} \ln \left( \sqrt{2 e} \int_{-\sigma}^{\infty} \exp(-t^2) \de t \right)
        \label{eq:perturbationTheoryChemicalPotential1b_0}
    \end{equation}
    with $g$ in \eqref{eq:perturbationTheoryChemicalPotential4}. In particular,
    \begin{equation}
        F_{\mathrm{c}}^{\mathrm{BEC}}(\beta,M) - F_{\mathrm{c}}^{\mathrm{BEC}}(\beta,\widetilde{M}) = \frac{1}{2 \beta} \left[ g^2\left(-\sqrt{2} \sigma \right) - g^2\left(-\sqrt{2} \widetilde{\sigma} \right) \right] - \frac{1}{\beta} \ln \left( \frac{ \int_{-\sigma}^{\infty} \exp(-t^2) \de t }{ \int_{-\widetilde{\sigma}}^{\infty} \exp(-t^2) \de t } \right).
        \label{eq:perturbationTheoryChemicalPotential2}
    \end{equation}
    
    Let us first assume that $M,\widetilde{M} \in [N^{5/6 - \epsilon}, N^{5/6+\epsilon}]$. Inspection of the proof of \cite[Lemma~C.1]{BocDeuSto-24} shows that this is equivalent to either $-N^{-1/6+\epsilon} \lesssim \mu, \widetilde{\mu} \lesssim N^{-1/6 + \epsilon}$ or $-N^{\epsilon} \lesssim \sigma, \widetilde{\sigma} \lesssim N^{\epsilon}$. Using these bounds, Lemma~\ref{lem:perturbationTheoryChemicalPotential}, the formula for the derivative of $g$ above \eqref{eq:perturbationTheoryChemicalPotential5}, and the asymptotic behavior of $g$ in and above \eqref{eq:perturbationTheoryChemicalPotential6}, we check that the absolute value of the first term in \eqref{eq:perturbationTheoryChemicalPotential2} is bounded by a constant times $a_N N^{1/2 + 3 \epsilon}$. With the same ingredients we see that the absolute value of the second term is bounded by a constant times $a_N N^{1/3 + 3 \epsilon}$, and hence
    \begin{equation}
        | F_{\mathrm{c}}^{\mathrm{BEC}}(\beta,M) - F_{\mathrm{c}}^{\mathrm{BEC}}(\beta,\widetilde{M}) | \lesssim a_N N^{1/2 + 3 \epsilon}.
        \label{eq:perturbationTheoryChemicalPotential3}
    \end{equation}

    To obtain \eqref{eq:perturbationTheoryChemicalPotential1}, we use \eqref{eq:perturbationTheoryChemicalPotential3} if $N^{5/6-\epsilon} \lesssim M, \widetilde{M} \lesssim N^{5/6 + \epsilon}$ with $0 < \epsilon \leq 1/24$. The bound for the parameter range $M,\widetilde{M} \gesssim N^{5/6 + \epsilon}$ with $\epsilon > 1/24$ follows from part~(a) of Lemma~\ref{prop:FreeEnergyBEC} and that for $M,\widetilde{M} \lesssim N^{5/6-\epsilon}$ with $1/24 < \epsilon < 1/6$ from part~(b) of the same lemma.
\end{proof}

The next lemma allows us to compare the chemical potentials related to the discrete and the continuous versions of our effective condensate functional.

\begin{lemma}\label{lem:CompChemPotentials}
We consider the limit $N \to \infty$, $\beta N^{2/3} \to \kappa \in (0,\infty)$. Let $\widetilde{\mu}$ and $\mu$ be the chemical potentials leading to the expected number of $M(N) \geq 0$ particles in the discrete and the continuous effective condensate theories defined in \eqref{eq:condensatefunctional} and \eqref{eq:condensatefunctionalcontinuous}, respectively. We assume $\hat{v}(0) > 0$ and that either $M(N) \gesssim N^{2/3}$ or $|\widetilde{\mu}| \lesssim 1$ holds. Then we have
	\begin{equation}
		\left|  \widetilde{\mu}(\beta,M) - \mu(\beta,M) \right| \lesssim \beta | \mu(\beta,M) |.
		\label{eq:DifferenceChemPots}
	\end{equation}
\end{lemma}
\begin{proof}
	We recall the notation $h = \hat{v}(0)/(2N) > 0$. Our choices for $\widetilde{\mu}$ and $\mu$ imply the identity
	\begin{equation}
		\frac{\sum_{n=0}^{\infty} n \exp( -\beta ( h n^2 - \widetilde{\mu} n ) )}{\sum_{n=0}^{\infty} \exp( -\beta ( h n^2 - \widetilde{\mu} n ) )} = \frac{ \int_{\mathbb{C}} |z|^2 \exp( -\beta ( h |z|^4 - \mu |z|^2 ) ) \de z }{\int_{0}^{\infty} \exp( -\beta ( h |z|^4 - \mu |z|^2 ) ) \de z } = \frac{\int_{0}^{\infty} x \exp( -\beta ( h x^2 - \mu x ) ) \de x}{\int_{0}^{\infty} \exp( -\beta ( h x^2 - \mu x ) ) \de x}.
		\label{eq:DifferenceChemPots1}
	\end{equation}
	To obtain this result, we applied the same coordinate transformation as in \eqref{eq:FreeEnergyBECApriori21}. 
	
	For a differentiable function $f:\mathbb{R} \to \mathbb{R}$ the Euler--Maclaurin formula reads
	\begin{equation}
		\sum_{k=0}^K f(k) = \int_0^K f(x) \de x + \frac{f(K)-f(0)}{2} - \int_0^K f'(x) P_1(x) \de x.
		\label{eq:EulerMaclaurin}
	\end{equation}
	Here $P_1$ denotes a periodized Bernoulli function that satisfies $|P_1(x)| \leq 1$ for all $x \in \mathbb{R}$, see e.g. \cite[Section~23]{AbraStegun1972}. We apply the Euler--Maclaurin formula and take the limit $K \to \infty$ to show that
	\begin{align}
		\sum_{n=0}^{\infty} n^p \exp( -\beta ( h n^2 - \widetilde{\mu} n ) ) =& \int_{0}^{\infty} x^p \exp( -\beta ( h x^2 - \widetilde{\mu} x ) ) \de x - \frac{1}{2} - p \int_0^{\infty} x^{p-1} \exp( -\beta ( h x^2 - \widetilde{\mu} x ) ) P_1(x) \de x \nonumber \\
		&+ \beta \int_0^{\infty} x^{p} ( 2hx - \widetilde{\mu} ) \exp( -\beta ( h x^2 - \widetilde{\mu} x ) ) P_1(x) \de x 
		\label{eq:DifferenceChemPots2}
	\end{align}
	holds for all $p \in \mathbb{R}$. Next, we insert this identity for $p=0,1$ on the left-hand side of \eqref{eq:DifferenceChemPots1} and find that it equals
	\begin{equation}
		\frac{Z_1(\widetilde{\mu})}{Z_0(\widetilde{\mu})} \frac{1+2\beta h W_2(\widetilde{\mu})/Z_1(\widetilde{\mu}) -\beta \widetilde{\mu} W_1(\widetilde{\mu})/Z_1(\widetilde{\mu}) - W_0(\widetilde{\mu})/Z_1(\widetilde{\mu}) - 1/(2 Z_1(\widetilde{\mu}))}{1 + 2 \beta h W_1(\widetilde{\mu})/Z_0(\widetilde{\mu}) - \beta \widetilde{\mu} W_0(\widetilde{\mu})/Z_0(\widetilde{\mu}) - 1/(2 Z_0(\widetilde{\mu}))},
		\label{eq:DifferenceChemPots3}
	\end{equation}
	where we used the notation
	\begin{equation}
		Z_p(\widetilde{\mu}) = \int_0^{\infty} x^p \exp( -\beta ( h x^2 - \widetilde{\mu} x ) ) \de x.
		\label{eq:DifferenceChemPots4}
	\end{equation}
	If $P_1(x)$ appears additionally in the integral we denoted the same object by $W_p$.
	
	To show that the second factor in \eqref{eq:DifferenceChemPots3} can be approximated by $1$, we now derive upper and lower bounds for the functions $Z_p(\widetilde{\mu})$. Since $P_1$ is a bounded function all upper bounds hold similarly for $W_p(\widetilde{\mu})$. Note that lower bounds for this function are not needed. A change of coordinates allows us to write
	\begin{equation}
		Z_p(\widetilde{\mu}) = \exp\left( \frac{\beta \widetilde{\mu}^2}{4h}\right) \frac{1}{\sqrt{\beta h}}  \int_{-\frac{\widetilde{\mu} \sqrt{\beta}}{2\sqrt{h}}}^{\infty} \left( \frac{x}{\sqrt{\beta h}} + \frac{\widetilde{\mu}}{2h} \right)^p \exp( -x^2 ) \de x.
		\label{eq:DifferenceChemPots5}
	\end{equation}
	If $\widetilde{\mu} \geq 0$ we have the upper bound
	\begin{align}
		Z_p(\widetilde{\mu}) \leq \exp\left( \frac{\beta \widetilde{\mu}^2}{4h}\right) \frac{1}{\sqrt{\beta h}}  \int_{-\infty}^{\infty} \left( \frac{|x|}{\sqrt{\beta h}} + \frac{\widetilde{\mu}}{2h} \right)^p \exp( -x^2 )  \de x \lesssim_p \exp\left( \frac{\beta \widetilde{\mu}^2}{4h}\right) \frac{1}{\sqrt{\beta h}} \left[ \left( \frac{1}{\sqrt{\beta h}} \right)^p + \left ( \frac{\widetilde{\mu}}{h}\right)^p \right]
		\label{eq:DifferenceChemPots6}
	\end{align}
	and the lower bound
	\begin{equation}
		Z_p(\widetilde{\mu}) \geq \exp\left( \frac{\beta \widetilde{\mu}^2}{4h}\right) \frac{1}{\sqrt{\beta h}}  \int_{0}^{\infty} \left( \frac{x}{\sqrt{\beta h}} + \frac{\widetilde{\mu}}{2h} \right)^p \exp( -x^2 )  \de x \gtrsim \exp\left( \frac{\beta \widetilde{\mu}^2}{4h}\right) \frac{1}{\sqrt{\beta h}} \left[ \left( \frac{1}{\sqrt{\beta h}} \right)^p + \left ( \frac{\widetilde{\mu}}{h}\right)^p \right]. 
		\label{eq:DifferenceChemPots7}
	\end{equation}
	Next, we consider $\widetilde{\mu} < 0$ and focus on the cases $p \in \{ 0,1,2 \}$. 

	From \cite[Eq.~7.1.13]{AbraStegun1972} we know that
	\begin{equation}
		\frac{1}{x + \sqrt{x^2 + 2}} < \exp\left( x^2 \right) \int_x^{\infty} \exp\left( -t^2 \right) \de t \leq \frac{1}{x + \sqrt{x^2 + 4/\pi}}.
		\label{eq:DifferenceChemPots8}
	\end{equation}
	It is elementary to check that for $a>0$ we have
	\begin{align}
		\int_a^{\infty} x \exp(-x^2) \de x &= \frac{1}{2} \exp(-a^2), \nonumber \\
		\int_a^{\infty} x^2 \exp(-x^2) \de x &= \frac{1}{2} \int_a^{\infty} \exp(-x^2) \de x + \frac{a}{2} \exp(-a^2).
		\label{eq:DifferenceChemPots9}
	\end{align} 
	In combination, \eqref{eq:DifferenceChemPots8} and \eqref{eq:DifferenceChemPots9} imply 
	\begin{equation}
		Z_0(\widetilde{\mu}) \sim \frac{1}{\sqrt{\beta h}} \frac{1}{1 -\widetilde{\mu} \sqrt{\beta/h}}, \qquad Z_1(\widetilde{\mu}) \sim \frac{1}{\beta h} \quad \text{ and } \quad Z_2(\widetilde{\mu}) \lesssim \frac{1}{(\beta h)^{3/2}} \left( 1 - \widetilde{\mu} \sqrt{\beta/h} \right).
		\label{eq:DifferenceChemPots10}
	\end{equation}
	Let us recall here that $a \sim b$ iff $a \lesssim b$ and $b \lesssim a$, see  Section~\ref{sec:intro}.
	
	Finally, we use the bounds in \eqref{eq:DifferenceChemPots6}, \eqref{eq:DifferenceChemPots7} and \eqref{eq:DifferenceChemPots10} to show that the second factor in \eqref{eq:DifferenceChemPots3} equals $1 + O(\beta (1+|\widetilde{\mu}|))$. In combination with \eqref{eq:DifferenceChemPots1} and \eqref{eq:DifferenceChemPots3}, we conclude that
	\begin{equation}
		\frac{Z_1(\widetilde{\mu})}{Z_0(\widetilde{\mu})} \left[ 1 + O(\beta(1+|\widetilde{\mu}|)) \right] = \frac{Z_1(\mu)}{Z_0(\mu)}.
		\label{eq:DifferenceChemPots14}
	\end{equation} 
	Before we can use this bound, we need some a-priori information on $\widetilde{\mu}$ showing that $|\beta \widetilde{\mu}| = o(1)$ (this is necessary in case of the assumption $M(N) \gesssim N^{2/3}$). Let us first assume that $-\beta \widetilde{\mu} \to c > 0$ as $N \to \infty$. Using dominated convergence, we check that this implies $M \lesssim 1$. This contradicts $M \gtrsim N^{2/3}$, and hence we have $|\beta \widetilde{\mu}| = o(1)$ for negative $\widetilde{\mu}$. Next we assume $\widetilde{\mu} > 0$ and $\widetilde{\mu}/2h \in \mathbb{N}$. We have
	\begin{equation}
		M(N) \geq \frac{\sum_{n=0}^{\infty} n \exp( -\beta ( h n^2 - \widetilde{\mu} n ) )}{\sum_{n=0}^{\infty} \exp( -\beta  ( n^2 - \widetilde{\mu} n )}	\geq \frac{\widetilde{\mu}}{2h}	\frac{\sum_{n\geq  \widetilde{\mu}/(2h)}^{\infty} \exp( -\beta h ( n - \widetilde{\mu}/(2h) )^2 )}{\sum_{n=-\infty}^{\infty} \exp( -\beta h ( n - \widetilde{\mu}/(2h) )^2 )} = \frac{\widetilde{\mu}}{4h} = N \frac{\widetilde{\mu}}{2 \hat{v}(0)},
		\label{eq:DifferenceChemPots15_0}
	\end{equation} 
	and hence $\widetilde{\mu} \leq 2 \hat{v}(0)$. In combination, the above considerations and the fact that the left-hand side of \eqref{eq:DifferenceChemPots1} is strictly monotone increasing in $\widetilde{\mu}$, allow us to conclude that $|\beta \widetilde{\mu}| \leq o(1) $. In the next step, we use this information and \eqref{eq:DifferenceChemPots14} to derive a bound for the absolute value of the difference of the two chemical potentials $\widetilde{\mu}$ and $\mu$. 
	
	To that end, we first introduce the function
	\begin{equation}
		\Theta(x) \coloneqq \frac{ 1 + \sqrt{\pi} x \exp(x^2) \mathrm{erfc}(- x) }{\exp(x^2) \mathrm{erfc}(- x)},
		\label{eq:DifferenceChemPots15}
	\end{equation}
	where $\mathrm{erfc}(x) = (2/\sqrt{\pi}) \int_x^{\infty} \exp(-t^2) \de t$, as well as the notations $\widetilde{\eta} = \widetilde{\mu} \sqrt{\beta/(4h)}$ and $\eta =\mu \sqrt{\beta/(4h)}$. A short computation, compare with \cite[Eqs.~C.3--C.6]{BocDeuSto-24}, shows that \eqref{eq:DifferenceChemPots15} can be written in terms of $\Theta$ as 
	\begin{equation}
		\Theta(\widetilde{\eta}) \left[ 1 + O(\beta(1+|\widetilde{\mu}|) \right] = \Theta(\eta).
		\label{eq:DifferenceChemPots16}
	\end{equation}
	It is straightforward to check that the function $\Theta$ is strictly positive for all $x \in \mathbb{R}$, strictly monotone increasing, continuous, satisfies $\Theta(x) \simeq \sqrt{\pi} x$ for $x \to \infty$, and $\Theta(x) \simeq \sqrt{\pi} /(2|x|)$ for $x \to -\infty$. Moreover, $\Theta'$ is also continuous and has the asymptotic behavior $\Theta'(x) \simeq \sqrt{\pi}$ for $x \to \infty$ and $\Theta'(x) \simeq \sqrt{\pi}/(2 x^2)$ for $x \to -\infty$. Using the strict monotonicity and the continuity of $\Theta$, we check that $\widetilde{\eta} \simeq \eta$ holds. 
	
	To obtain our final bounds, we apply the first order Taylor expansion 
	\begin{equation}
		\Theta(\widetilde{\eta}) =  \Theta(\eta) + \int_0^1 \Theta'(\widetilde{\eta} + t (\eta - \widetilde{\eta})) \de t (\widetilde{\eta} - \eta)
		\label{eq:DifferenceChemPots17}
	\end{equation}
	and conclude that
	\begin{equation}
		| \widetilde{\eta} - \eta | \lesssim \frac{\beta (1+|\widetilde{\mu}| ) | \Theta(\eta) |   }{\int_0^1 \Theta'(\widetilde{\eta} + t (\eta - \widetilde{\eta})) \de t}.
		\label{eq:DifferenceChemPots18}
	\end{equation}
	Let us distinguish two cases and first assume that $\widetilde{\eta}, \eta \gtrsim - 1$. The asymptotic behavior of $\Theta$ and $\Theta'$ described above implies
	\begin{equation}
		| \widetilde{\eta} - \eta | \lesssim \beta (1+|\widetilde{\mu}| ) |\widetilde{\eta}| \leq \beta (1+|\widetilde{\mu}| ) (|\eta| + |\eta - \widetilde{\eta}|), 
		\label{eq:DifferenceChemPots19}
	\end{equation}
	and hence $| \widetilde{\mu} - \mu | \lesssim \beta (1+|\widetilde{\mu}| ) |\mu|$. Using this bound and $|\mu| \lesssim 1$, we find $| \widetilde{\mu} - \mu | \lesssim \beta |\mu|$, which proves the claim in this case. Next, we assume $\widetilde{\eta}, \eta \lesssim- 1$. In this case, the bound
	\begin{equation}
		\int_0^1 \Theta'(\widetilde{\eta} + t (\eta - \widetilde{\eta})) \de t \gtrsim \int_0^1 \frac{1}{(\widetilde{\eta} + t (\eta - \widetilde{\eta}))^2} \de t = \frac{1}{\widetilde{\eta} \eta}
		\label{eq:DifferenceChemPots20}
	\end{equation}
	and \eqref{eq:DifferenceChemPots18} imply
	\begin{equation}
		| \widetilde{\eta} - \eta | \lesssim \beta (1+|\widetilde{\mu}| ) |\widetilde{\eta} | \leq \beta (1+|\widetilde{\mu}| ) |\eta | + \beta (1+|\widetilde{\mu}| ) |\widetilde{\eta} - \eta|.
		\label{eq:DifferenceChemPots21}
	\end{equation}
	Above we already argued that this implies $| \widetilde{\mu} - \mu | \lesssim \beta |\mu |$, and hence the claim of Lemma~\ref{lem:CompChemPotentials} is proved. 
\end{proof}

In our analysis we need to compare the minima of our discrete and continuous free energy functionals. The bound we apply is proved in the following lemma. 

\begin{lemma}\label{lem:comparisonContinuousDiscreteCondensateFreeEnergy}
    We consider the limit $N \to \infty$, $\beta N^{2/3} \to \kappa \in (0,\infty)$. Let $M(N)$ be a sequence of nonnegative real numbers that satisfies $N^{2/3} \lesssim M(N) \lesssim N$. Then we have
    \begin{equation}
	   | F_{\mathrm{c}}^{\mathrm{BEC}}(\beta,M) - F^{\mathrm{BEC}}(\beta,M) | \lesssim N^{1/3}.
	   \label{eq:ComparisonCondensateFreeEnergies}
    \end{equation}
\end{lemma}
\begin{proof}
    We recall the definition of $\mu^{\mathrm{BEC}}$ in \eqref{eq:GibbsdistributiomDiscrete} and denote by $\mu$ the chemical potential leading to an expected number of $M(N)$ particles in the continuous effective condensate theory in \eqref{eq:condensatefunctionalcontinuous}. From the assumption $M \gesssim N^{2/3}$ and part~(c) of Lemma~\ref{lem:ChemPotBECCont} we know that $|\mu| \lesssim 1$. The same assumption and an application of Lemma~\ref{lem:CompChemPotentials} show $| \mu - \mu^{\mathrm{BEC}} | \lesssim \beta | \mu| $. In particular, $| \mu - \mu^{\mathrm{BEC}} | \lesssim \beta$ and $| \mu^{\mathrm{BEC}} | \lesssim 1$.
    
In the first step, we replace $\mu$ by $\mu^{\mathrm{BEC}}$ in $F_{\mathrm{c}}^{\mathrm{BEC}}(\beta,M)$ and we start with a lower bound:
\begin{align}
	F_{\mathrm{c}}^{\mathrm{BEC}}(\beta,M) &= -\frac{1}{\beta} \ln \left( \int_{\mathbb{C}} \exp\left( - \beta \left( \frac{\hat{v}(0)}{2N} |z|^4 - \mu |z|^2 \right) \right) \de z \right)  + \mu M - \frac{\hat{v}(0) M^2}{2N} \label{eq:UpperBoundAndi5} \\
    &\leq -\frac{1}{\beta} \ln\left( \int_{\mathbb{C}} \exp\left(-\beta\left( \frac{\hat{v}(0)}{2N} |z|^4 - \mu^{\mathrm{BEC}} |z|^2 \right) \right) \de z \right) + \mu^{\mathrm{BEC}}  M - \frac{\hat{v}(0) M^2}{2N} + (\mu^{\mathrm{BEC}} - \mu ) ( \widetilde{M} - M ), \nonumber
\end{align}
where 
\begin{equation}
	\widetilde{M} = \frac{\int_{\mathbb{C}} |z|^2 \exp\left(-\beta\left( \frac{\hat{v}(0)}{2N} |z|^4 - \mu^{\mathrm{BEC}} |z|^2 \right) \right) \de z}{\int_{\mathbb{C}}  \exp\left(-\beta\left( \frac{\hat{v}(0)}{2N} |z|^4 - \mu^{\mathrm{BEC}} |z|^2 \right) \right) \de z}.
	\label{eq:UpperBoundAndi6}
\end{equation}
To obtain this bound, we used that the first term after the smaller or equal sign is concave in $\mu$, and that its first derivative with respect to $\mu^{\mathrm{BEC}}$ equals $-\widetilde{M}$. If we use $| \mu^{\mathrm{BEC}} - \mu | \lesssim \beta$ and assume that $\widetilde{M} \lesssim N$, we see that the last term on the right-hand side of \eqref{eq:UpperBoundAndi5} is bounded from above by a constant times $N^{1/3}$. The bound $\widetilde{M} \lesssim N$ follows from $|\mu^{\mathrm{BEC}}| \lesssim 1$, \eqref{eq:DifferenceChemPots6}, \eqref{eq:DifferenceChemPots7}, and \eqref{eq:DifferenceChemPots10}. 

To obtain the reverse inequality we again use the convexity of the first term after the smaller or equal sign \eqref{eq:UpperBoundAndi5} and Lemma~\ref{lem:CompChemPotentials}, which gives
\begin{align}
	F_{\mathrm{c}}^{\mathrm{BEC}}(\beta,M) \geq -\frac{1}{\beta} \ln\left( \int_{\mathbb{C}} \exp\left(-\beta\left( \frac{\hat{v}(0)}{2N} |z|^4 - \mu^{\mathrm{BEC}} |z|^2 \right) \right) \de z \right) + \mu^{\mathrm{BEC}} M - \frac{\hat{v}(0) M^2}{2N}.
    \label{eq:UpperBoundAndi6b}
\end{align}

In the second step, we apply the same coordinate transformation as in \eqref{eq:FreeEnergyBECApriori21} and use \eqref{eq:DifferenceChemPots2} to write
\begin{align}
	\int_{\mathbb{C}}  \exp\left(-\beta\left( \frac{\hat{v}(0)}{2N} |z|^4 - \mu^{\mathrm{BEC}} |z|^2 \right) \right) \de z =&  \sum_{n=0}^{\infty} \exp\left(-\beta\left( \frac{\hat{v}(0)}{2N} n^2 - \mu^{\mathrm{BEC}} n \right) \right) + \frac{1}{2} \label{eq:UpperBoundAndi7}  \\
	& - \beta \int_0^{\infty} (\hat{v}(0)x/N - \mu^{\mathrm{BEC}}) \exp\left(-\beta\left( \frac{\hat{v}(0)}{2N} x^2 - \mu^{\mathrm{BEC}} x \right) \right) P_1(x) \de x. \nonumber
\end{align}
Here $P_1$ denotes a periodized Bernoulli function that satisfies $|P_1(x)| \leq 1$ for all $x \in \mathbb{R}$. Using additionally $|\mu^{\mathrm{BEC}}| \lesssim 1$, \eqref{eq:DifferenceChemPots6}, \eqref{eq:DifferenceChemPots7}, and \eqref{eq:DifferenceChemPots10} it is straightforward to check that
\begin{equation}
	\frac{1}{\beta} \left| \ln\left( \int_{\mathbb{C}}  \exp\left(-\beta\left( \frac{\hat{v}(0)}{2N} |z|^4 - \mu^{\mathrm{BEC}} |z|^2 \right) \right) \de z \right) - \ln\left( \sum_{n=0}^{\infty} \exp\left(-\beta\left( \frac{\hat{v}(0)}{2N} n^2 - \mu^{\mathrm{BEC}} n \right) \right) \right) \right| \lesssim 1. 
	\label{eq:UpperBoundAndi8}
\end{equation}
When we put \eqref{eq:UpperBoundAndi5} and \eqref{eq:UpperBoundAndi8} together this proves the claim of the lemma.
\end{proof}

For reference in the main text we also state the following corollary for the grand potential of our effective condensate theory. Its proof is given in \eqref{eq:UpperBoundAndi8}.

\begin{corollary}
\label{cor:ComparisonCondensateGrandPotentials}
    We consider the limit $N \to \infty$, $\beta/ \beta_{\mathrm{c}}\to \kappa \in (0,\infty)$ with $\beta_{\mathrm{c}}$ in \eqref{eq:crittemp} and assume that $| \mu | \lesssim 1$. Then we have
    \begin{equation}
	   -\frac{1}{\beta} \ln \left( \int_{\mathbb{C}} \exp\left( - \beta \left( \frac{\hat{v}(0)}{2N} |z|^4 - \mu |z|^2 \right) \right) \de z \right) \leq -\frac{1}{\beta} \ln\left( \sum_{n=0}^{\infty} \exp\left(-\beta\left( \frac{\hat{v}(0)}{2N} n^2 - \mu n \right) \right) \right) + C.
	   \label{eq:ComparisonCondensateGrandPotentials}
    \end{equation}
\end{corollary}

The next two lemmas provide us with bounds for the variance and the moment generating function of our continuous effective condensate theory that are needed in Section~\ref{sec:asymptoticsCondensateDistribution}. We start with the bounds for the variance.
\begin{lemma} \label{lem:asymptoticsVariances}
    We consider the limit $N \to \infty$, $\beta N^{2/3} \to \kappa \in (0,\infty)$. Let $M(N)$ be a sequence of nonnegative real numbers that satisfies $0 \leq M(N) \lesssim N$. Let $g$ be given as in \eqref{eq:GibbsDistributionDiscrete} with $\widetilde{N}_0 = M(N)$ and denote its variance by $\textbf{Var}(\beta,N)$. 
    \begin{enumerate}[label=(\alph*)]
        \item If $M(N) \gg N^{5/6}$ the variance satisfies 
        \begin{equation}
            \lim_{N \to \infty} \frac{\beta \hat{v}(0) \textbf{Var}(\beta,N)}{N} = 1.
            \label{eq:varianceBounda}
        \end{equation}
        \item If $M(N) = t N^{5/6}$ with some fixed $t \in \mathbb{R}$ then the parameter $\sigma = \mu \sqrt{\beta N/(2\hat{v}(0))}$ does not depend on $N$ and we have
        \begin{equation}
            \lim_{N \to \infty} \frac{\beta \hat{v}(0) \textbf{Var}(\beta,N)}{2N} = \frac{\int_{-\sigma}^{\infty} x^2 \exp(-x^2) \de x}{\int_{-\sigma}^{\infty} \exp(-x^2) \de x} - \left( \frac{\int_{-\sigma}^{\infty} x \exp(-x^2) \de x}{\int_{-\sigma}^{\infty} \exp(-x^2) \de x} \right)^2.
            \label{eq:varianceBoundb}
        \end{equation}
        \item If $ 1 \ll M(N) \ll N^{5/6}$ we have
        \begin{equation}
            \lim_{N \to \infty} (-\beta \mu)^2 \textbf{Var}(\beta,N) = \lim_{N \to \infty} \frac{\textbf{Var}(\beta,N)}{[M(N)]^2} = 1.
            \label{eq:varianceBoundc}
        \end{equation}
    \end{enumerate}
\end{lemma}
\begin{proof}
    We recall the notation $h = \hat{v}(0)/(2N)$. The $p$-th moment of $g$ reads
    \begin{equation}
        \int_0^{\infty} x^p g(\sqrt{x}) \de x = \frac{\int_{-\sigma}^{\infty} \left( x/\sqrt{\beta h} + \frac{\mu}{2h} \right)^p \exp(-x^2) \de x}{ \int_{-\sigma}^{\infty} \exp(-x^2) \de x },
        \label{eq:varianceBound1}
    \end{equation}
    and hence
    \begin{equation}
        \beta h \textbf{Var}(\beta,N) = \frac{\int_{-\sigma}^{\infty} x^2 \exp(-x^2) \de x}{\int_{-\sigma}^{\infty} \exp(-x^2) \de x} - \left( \frac{\int_{-\sigma}^{\infty} x \exp(-x^2) \de x}{\int_{-\sigma}^{\infty} \exp(-x^2) \de x} \right)^2.
        \label{eq:varianceBound2}
    \end{equation}
    Using \eqref{eq:varianceBound1} with $p = 1$ one easily checks that $\sigma \to +\infty$ if $M(N) \gg N^{5/6}$. In this limit the right-hand side of \eqref{eq:varianceBound2} converges to $1/2$, which proves \eqref{eq:varianceBounda}. Similarly, we see that $\sigma$ does not depend on $N$ provided $M(N) = t N^{5/6}$ with some fixed $t \in \mathbb{R}$, and hence \eqref{eq:varianceBoundb} holds. 
    
    Finally, we consider the case $1 \ll M(N) \ll N^{5/6}$, where we have $\sigma \ll - 1$. Here we use the identity $\beta \mu = 2 \sigma \sqrt{\beta h}$ to write
    \begin{equation}
        \int_0^{\infty} x^p g(\sqrt{x}) \de x = \left(\frac{1}{-\beta \mu} \right)^p \frac{\int_{0}^{\infty} x^p \exp(-x-x^2/(4 \sigma^2)) \de x}{ \int_{0}^{\infty} \exp(-x-x^2/(4 \sigma^2)) \de x }.
        \label{eq:varianceBound3}
    \end{equation}
    For the variance this implies
    \begin{equation}
        (-\beta\mu)^2 \mathbf{Var}(\beta,N) = \frac{\int_{0}^{\infty} x^2 \exp(-x-x^2/(4 \sigma^2)) \de x}{ \int_{0}^{\infty} \exp(-x-x^2/(4 \sigma^2)) \de x } - \left( \frac{\int_{0}^{\infty} x \exp(-x-x^2/(4 \sigma^2)) \de x}{ \int_{0}^{\infty} \exp(-x-x^2/(4 \sigma^2)) \de x } \right)^2.
        \label{eq:varianceBound4}
    \end{equation}
    Since $\sigma \ll -1$ an application of the dominated convergence theorem shows
    \begin{equation}
        \lim_{N \to \infty} (-\beta\mu)^2 \mathbf{Var}(\beta,N) = \frac{\int_{0}^{\infty} x^2 \exp(-x) \de x}{ \int_{0}^{\infty} \exp(-x) \de x } - \left( \frac{\int_{0}^{\infty} x \exp(-x) \de x}{ \int_{0}^{\infty} \exp(-x) \de x } \right)^2 = 1.
        \label{eq:varianceBound5}
    \end{equation}
    In combination with the fact that $M(N) \simeq (-\beta\mu)^{-1}$ in this parameter regime, \eqref{eq:varianceBound5} proves \eqref{eq:varianceBoundc}.
\end{proof}

\begin{lemma} \label{lem:BoundCenteredCondensateDistribution}
    We consider the limit $N \to \infty$, $\beta N^{2/3} \to \kappa \in (0,\infty)$. Let $M(N)$ be a sequence of nonnegative real numbers that satisfies $N^{2/3} \leq M(N) \leq N$. We recall the definition of the random variable $\widetilde{X}_{\beta,N}$ in \eqref{eq:widetildeX} with $\widetilde{N}_0 = M(N)$ and assume that $0 < \lambda_0 < 1$. There exists a constant $C > 0$, which is independent of $N$, such that
    \begin{equation}
        \lim_{N \to \infty} \mathbf{E}\left( \exp(\lambda |\widetilde{X}_{\beta,N}|) \right) \leq C
        \label{eq:BoundCenteredCondensateDistributiona}
    \end{equation}
    holds for every $0 < \lambda < \lambda_0$.
\end{lemma}
\begin{proof}
    We have
    \begin{equation}
        \mathbf{E}(\exp(\lambda | \widetilde{X}_{\beta,N}|)) = \frac{\int_0^{\infty} \exp\left(\lambda \left| \frac{x-M(N)}{\sqrt{\mathbf{Var}(\beta,N)}} \right| - \beta \left( h x^2 - \mu x \right) \right) \de x}{\int_0^{\infty} \exp\left( - \beta \left( h x^2 - \mu x \right) \right) \de x}
        \label{eq:BoundCenteredCondensateDistribution1}
    \end{equation}
    with the variance $\mathbf{Var}(\beta,N)$ of the random variable $X_{\beta,N}$ defined below \eqref{eq:condensateDistribution11}. We first consider the case $M(N) \gg N^{5/6}$, where we have $\sigma \gg 1$. We write \eqref{eq:BoundCenteredCondensateDistribution1} as 
    \begin{equation}
        \frac{ \int_{-\sigma}^{\infty} \exp\left(\lambda \left| \frac{x+\sqrt{\beta h} ( \mu/(2h) - M(N))}{\sqrt{\beta h} \sqrt{\mathbf{Var}(\beta,N)}} \right| - x^2 \right) \de x}{\int_{-\sigma}^{\infty} \exp\left( - x^2 \right) \de x}.
        \label{eq:BoundCenteredCondensateDistribution2}
    \end{equation}
    From \eqref{eq:FreeEnergyBECApriori23} we know that
    \begin{equation}
        \lim_{N \to \infty} \sqrt{\beta h} \left( \frac{\mu}{2h} - M(N) \right) = - \frac{\int_{-\infty}^{\infty} x \exp(-x^2) \de x }{\int_{-\infty}^{\infty} \exp(-x^2) \de x} = 0.
        \label{eq:BoundCenteredCondensateDistribution3}
    \end{equation}
    Moreover, an application of part~(a) of Lemma~\ref{lem:asymptoticsVariances} shows $\lim_{N \to \infty} \sqrt{\beta h } \sqrt{\mathbf{Var}(\beta,N)} = 1/\sqrt{2}$ and we conclude that
    \begin{equation}
        \lim_{N \to \infty} \mathbf{E}(\exp(\lambda | \widetilde{X}_{\beta,N}|)) = \frac{1}{\sqrt{\pi}} \int_{-\infty}^{\infty} \exp\left(\lambda \sqrt{2} \left| x \right| - x^2 \right) \de x.
        \label{eq:BoundCenteredCondensateDistribution4}
    \end{equation}

    If $M(N) = t N^{5/6}$ with $t \in \mathbb{R}$ we know that $\sigma \in \mathbb{R}$ is fixed. In this case we still use \eqref{eq:BoundCenteredCondensateDistribution2} but replace \eqref{eq:BoundCenteredCondensateDistribution3} by
    \begin{equation}
        -\lim_{N \to \infty} \sqrt{\beta h} \left( \frac{\mu}{2h} - M(N) \right) = \frac{\int_{-\sigma}^{\infty} x \exp(-x^2) \de x }{\int_{-\sigma}^{\infty} \exp(-x^2) \de x} \eqqcolon A.
        \label{eq:BoundCenteredCondensateDistribution5}
    \end{equation}
    We also use part~(b) of Lemma~\ref{lem:asymptoticsVariances} to compute $\lim_{N \to \infty} \sqrt{\beta h } \sqrt{\mathbf{Var}(\beta,N)} \eqqcolon B$. We highlight that $B$ is given by the square root of the right-hand side of \eqref{eq:varianceBoundb}. In combination, these considerations imply
    \begin{equation}
        \lim_{N \to \infty} \mathbf{E}(\exp(\lambda | \widetilde{X}_{\beta,N}|)) = \frac{ \int_{-\sigma}^{\infty} \exp\left(\lambda \left| \frac{x - A}{B} \right| - x^2 \right) \de x }{\int_{-\sigma}^{\infty} \exp\left(- x^2 \right) \de x}.
        \label{eq:BoundCenteredCondensateDistribution6}
    \end{equation}

    Finally, we consider the case $N^{2/3} \leq M(N) \ll N^{5/6}$, where $\sigma \ll -1$. Here we argue as in \eqref{eq:varianceBound3} to write \eqref{eq:BoundCenteredCondensateDistribution1} as 
    \begin{equation}
        \frac{\int_0^{\infty} \exp\left(-x - \frac{x^2}{4 \sigma^2} + \lambda \left| \frac{x + \beta \mu M}{-\beta \mu \sqrt{\mathbf{Var}(\beta,N)}} \right| \right) \de x}{\int_0^{\infty} \exp\left(-x - \frac{x^2}{4 \sigma^2} \right) \de x}.
        \label{eq:BoundCenteredCondensateDistribution7}
    \end{equation}
    When we apply part~(c) of Lemma~\ref{lem:asymptoticsVariances} and dominated convergence on the right-hand side of \eqref{eq:BoundCenteredCondensateDistribution7}, we find
    \begin{equation}
        \lim_{N \to \infty} \mathbf{E}(\exp(\lambda | \widetilde{X}_{\beta,N}|)) = \frac{\int_0^{\infty} \exp\left(-x + \lambda \left| x-1 \right| \right) \de x}{\int_0^{\infty} \exp\left(-x \right) \de x}.
        \label{eq:BoundCenteredCondensateDistribution8}
    \end{equation}
    
    To prove the claim of the lemma, we choose $C$ as the maximum of the right-hand sides of \eqref{eq:BoundCenteredCondensateDistribution4}, \eqref{eq:BoundCenteredCondensateDistribution6}, and \eqref{eq:BoundCenteredCondensateDistribution8} with $\lambda = \lambda_0 < 1$. 
\end{proof}

In the last lemma in Appendix~\ref{app:effcondensate} we compute the characteristic function of $\widetilde{X}_{\beta,N}$ defined in \eqref{eq:condensateDistribution15} in the three parameter regimes that have been considered already in the preceding two lemmas. Also this lemma finds application in Section~\ref{sec:asymptoticsCondensateDistribution}.
\begin{lemma} \label{lem:lemmaCharacteristicFunction}
    We consider the limit $N \to \infty$, $\beta N^{2/3} \to \kappa \in (0,\infty)$. Let $M(N)$ be a sequence of nonnegative real numbers that satisfies $N^{2/3} \leq M(N) \leq N$. Let $\widetilde{X}_{\beta,N}$ be the random variable defined in \eqref{eq:condensateDistribution15} with $\widetilde{N}_0 = M(N)$ and let 
    \begin{equation}
        \phi_{\beta,N}(t) = \mathbf{E}( e^{\mathrm{i} t \widetilde{X}_{\beta,N}} )
        \label{eq:lemmaCharacteristicFunctiona}
    \end{equation}
    be its characteristic function.  
    \begin{enumerate}[label=(\alph*)]
        \item If $M(N) \gg N^{5/6}$ the characteristic function satisfies 
        \begin{equation}
            \lim_{N \to \infty} \phi_{\beta,N}(t) = \exp\left( -t^2/2 \right).
            \label{eq:lemmaCharacteristicFunctionb}
        \end{equation}
        The right-hand side is the characteristic function of a standard normal distribution.
        \item If $M(N) = t N^{5/6}$ with some fixed $t \in \mathbb{R}$ the parameter $\sigma = \mu \sqrt{\beta N/(2\hat{v}(0))}$ does not depend on $N$ and we have
        \begin{equation}
        \lim_{N \to \infty} \phi_{\beta,N}(t) = \frac{ \int_{\frac{-\sigma-A}{B}}^{\infty}\exp\left(\mathrm{i}t x - \left( x B + A \right)^2 \right) \de x }{\int_{\frac{-\sigma-A}{B}}^{\infty} \exp\left(-\left( x B + A \right)^2 \right) \de x}.
        \label{eq:lemmaCharacteristicFunctionc}
        \end{equation}
        Here $A$ is given as in \eqref{eq:BoundCenteredCondensateDistribution5} and $B$ as below that equation. The right-hand side is the characteristic function of the probability distribution
        \begin{equation}
            f_{\sigma,A,B}(x) = \frac{\exp\left(-\left( x B + A \right)^2 \right)}{\int_{\frac{-\sigma-A}{B}}^{\infty} \exp\left(-\left( x B + A \right)^2 \right) \de x}
            \label{eq:lemmaCharacteristicFunctiond}
        \end{equation}
        on the interval $[(-\sigma-A)/B,\infty)$.
        \item If $ N^{2/3} \leq M(N) \ll N^{5/6}$ we have
        \begin{equation}
            \lim_{N \to \infty} \phi_{\beta,N}(t) = \frac{e^{-\mathrm{i} t}}{1-\mathrm{i}t}. 
            \label{eq:lemmaCharacteristicFunctione}
        \end{equation}
        The right-hand side is the characteristic function of the probability distribution $f(x) = \exp(-(1+x))$ on the interval $[-1,\infty)$.
    \end{enumerate}
\end{lemma}
\begin{proof}
The proof of the above lemma is, except for some straightforward computations, almost literally the same as that of Lemma~\ref{lem:BoundCenteredCondensateDistribution}, and therefore omitted.
\end{proof}
\section{Properties of the effective chemical potential}
\label{app:effectiveChemicalPotential}
In this section we investigate equation \eqref{eq:GrantCanonicalEffectiveIddealGasChemPot} for the effective chemical potential appearing in our statements for the grand potential. The first lemma guarantees the existence of a unique solution and provides a priori bounds.
\begin{lemma}\label{lem:effectiveChemicalPotentialAPriori}
    The following three statements hold:
    \begin{enumerate}[label=(\alph*)]
    \item Assume that $\beta,\eta, \hat{v}(0) > 0$ and $\mu \in \mathbb{R}$. The equation
    \begin{equation}
        \sum_{p \in \Lambda^*} \frac{1}{e^{\beta(p^2 - \widetilde{\mu})}-1} = \frac{(\mu - \widetilde{\mu})\eta}{\hat{v}(0)}
        \label{eq:GrantCanonicalEffectiveIddealGasChemPotv2}
    \end{equation}
    for $\widetilde{\mu}$ admits a unique solution in the set $(-\infty,0)$.  
    \item We consider the limit $\eta \to \infty$, $\beta/\beta_{\mathrm{c}}(\eta) \to \kappa  \in (0,\infty)$ with $\beta_{\mathrm{c}}$ in \eqref{eq:crittemp}. We assume that $\mu$, which may depend on $\eta$, satisfies $-\eta^{2/3} \lesssim \mu \lesssim 1$. There exists a constant $c>0$ such that the unique solution to \eqref{eq:GrantCanonicalEffectiveIddealGasChemPotv2} satisfies 
    \begin{equation}
        c \leq \mu - \widetilde{\mu} \leq c^{-1}.
        \label{eq:muMinusMu0Bound}
    \end{equation}
    Moreover, if $\mu \geq 0$ then $-\widetilde{\mu} \lesssim 1$ and if $\mu < 0$ we have $-\widetilde{\mu} \lesssim \eta^{2/3}$. 
    \item Under the assumptions stated in part~(b) there exists a constant $c > 0$ such that
    \begin{equation}
        c \eta \leq \sum_{p \in \Lambda^*} \frac{1}{e^{\beta(p^2 - \widetilde{\mu})}-1} \leq c^{-1} \eta.
    \end{equation}
    \end{enumerate}
\end{lemma}

\begin{proof}
    For $\widetilde{\mu} \in (-\infty,0)$ we define the function
    \begin{equation}
        f(\widetilde{\mu}) = \sum_{p \in \Lambda^*} \frac{1}{e^{\beta(p^2 - \widetilde{\mu})}-1} +  \frac{(\widetilde{\mu} - \mu)\eta}{\hat{v}(0)}.
        \label{eq:EffectiveChemicalPOtential1}
    \end{equation}
    It is not difficult to check that $f$ is continuous, strictly monotone increasing, and satisfies $f(\widetilde{\mu}) \to -\infty$ for $\widetilde{\mu} \to - \infty$ and $f(\widetilde{\mu}) \to +\infty$ for $\widetilde{\mu} \to 0$. This implies part~(a) and it remains to prove parts~(b) and (c).

    With \eqref{eq:GrantCanonicalEffectiveIddealGasChemPotv2}, $\widetilde{\mu} < 0$ and by interpreting the relevant over $p \neq 0$ sum as a Riemann, integral it is not difficult to see that
    \begin{equation}
        \mu - \widetilde{\mu} \lesssim \frac{1}{\eta} \left[ \frac{1}{-\beta \widetilde{\mu}} + \frac{C}{\beta^{3/2}} \right] \lesssim \frac{1}{-\eta^{1/3} \widetilde{\mu}} + 1.
        \label{eq:EffectiveChemicalPOtential2}    
    \end{equation}
    If $\mu \geq 0$ the above inequality implies $-\widetilde{\mu} \lesssim 1$. In the case $\mu < 0$, we additionally use $\mu \gesssim - \eta^{2/3}$ and find $-\widetilde{\mu} \lesssim \eta^{2/3}$.

    We use $-\widetilde{\mu} \lesssim \eta^{2/3}$ again and drop the term in the sum with $p = 0$ to check that
    \begin{equation}
        \mu - \widetilde{\mu} \gesssim \frac{1}{\eta \beta^{3/2}}  \gesssim 1.
        \label{eq:EffectiveChemicalPOtential3} 
    \end{equation}
    To derive an upper bound for $\mu - \widetilde{\mu}$ in the case $\mu < 0$, we combine \eqref{eq:EffectiveChemicalPOtential3} and \eqref{eq:EffectiveChemicalPOtential2} as follows:
    \begin{equation}
        \mu - \widetilde{\mu} \lesssim \frac{1}{-\eta^{1/3} \widetilde{\mu}} + 1 \lesssim \frac{1}{\eta^{1/3}(1-\mu)} + 1 \leq \frac{1}{\eta^{1/3}} + 1 \lesssim 1.
        \label{eq:EffectiveChemicalPOtential4} 
    \end{equation}
    This proves part~(b).

    Part~(c) of Lemma~\ref{lem:effectiveChemicalPotentialAPriori} follows from \eqref{eq:GrantCanonicalEffectiveIddealGasChemPotv2} and \eqref{eq:muMinusMu0Bound}. 
    \end{proof}

    In the second lemma we perturb $\hat{v}(0)$ and study by how much the solution $\widetilde{\mu}$ to \eqref{eq:GrantCanonicalEffectiveIddealGasChemPotv2} changes. 
    \begin{lemma}\label{lem:effectiveChemicalPotentialPerturbation2}
        We consider the limit $\eta \to \infty$, $\beta/\beta_{\mathrm{c}}(\eta) \to \kappa  \in (0,\infty)$ with $\beta_{\mathrm{c}}$ in \eqref{eq:crittemp}. We assume that $\mu$, $\delta$, which may depend on $\eta$, satisfy $-\eta^{2/3} \lesssim \mu \lesssim 1$ and $0 < \delta \lesssim 1$. Let $\mu_1$ be the solution to \eqref{eq:GrantCanonicalEffectiveIddealGasChemPotv2} and let $\mu_2$ be the solution to the same equation with $\hat{v}(0)$ replaced by $\hat{v}(0) + \delta$. Then we have
        \begin{equation}
            0 < \mu_1 - \mu_2 \lesssim \delta.
            \label{eq:EffectiveChemicalPOtential5} 
        \end{equation}
    \end{lemma}
    \begin{proof}
        We define $\Delta \mu = \mu_1 - \mu_2$. Let us subtract the equation for $\mu_2$ from the equation for $\mu_1$. This and a first order Taylor expansion yield the identity
        \begin{equation}
            \Delta \mu \left[ \sum_{p \in \Lambda^*} \frac{1}{4 \sinh^2\left( \frac{\beta(p^2 - \xi) }{2} \right) } + \frac{\eta}{\hat{v}(0)} \right] = \frac{(\mu - \mu_2)\eta}{\hat{v}(0)} \frac{\delta}{\hat{v}(0) + \delta}.
            \label{eq:EffectiveChemicalPOtential6} 
        \end{equation}
        From part~(b) of Lemma~\ref{lem:effectiveChemicalPotentialAPriori} we know that $\mu - \mu_2 > 0$ and we conclude $\mu_1-\mu_2 > 0$. The term in the backets on the left-hand side of \eqref{eq:EffectiveChemicalPOtential6} is bounded from below by a constant times $\eta $. Another application of part~(b) of Lemma~\ref{lem:effectiveChemicalPotentialAPriori} shows that the right-hand side of \eqref{eq:EffectiveChemicalPOtential6} is bounded from above by a constant times $\eta \delta$. In combination, these two bounds prove the claim. 
    \end{proof}
\section{Properties of the ideal Bose gas}
\label{app:idealGas}
In this section we prove bounds related to the ideal Bose gas that are used in the main text. The first lemma provides a bound relating the expected number of particles in the condensate of Bose gases governed by two different one-particle Hamiltonians.
\begin{lemma}
    \label{lem:idealGasWithVariable1ParticleHamiltonian}
    We consider the limit $N \to \infty$, $\beta N^{2/3} \to \kappa \in (0,\infty)$. We choose $h$ as 
    \begin{equation}
    h = \sum_{p\in \Lambda^*} h(p) | \varphi_p \rangle \langle \varphi_p | \quad \text{ with } \quad h(p) = p^2 + \lambda f(p).
    \label{eq:bound1pdm1b}
\end{equation}
Here $\varphi_p(x) = e^{\mathrm{i} p \cdot x}$ and $f : \Lambda^* \to \mathbb{R}$ is a bounded function that satisfies $f(0) = 0$. The absolute value of the parameter $\lambda \in \mathbb{R}$ is chosen small enough such that $h$ satisfies \eqref{eq:generalizedOneParticleHamiltonian}. By $\mu_0(\beta,N,\lambda)$ and $N_0(\beta,N,\lambda)$ we denote chemical potential and the expected numbers of particles in the condensate of the ideal gas related to $h$, respectively. We have
    \begin{equation}
        | \mu_0(\beta,N,\lambda) - \mu_0(\beta,N,0) | \lesssim | \lambda | \ \Vert f \Vert_{\infty} \quad \text{ and } \quad | N_0(\beta,N,\lambda) - N_0(\beta,N,0) | \lesssim \frac{|\lambda| \ \Vert f \Vert_{\infty}}{\beta}.
        \label{eq:appIdealGas2}
    \end{equation}
\end{lemma}
\begin{proof}
    We apply a first order Taylor expansion to see that
    \begin{align}
        0 &= \sum_{p \in \Lambda^*} \left( \frac{1}{\exp(\beta(h(p) - \mu_0(\beta,N,\lambda)))-1} - \frac{1}{\exp(\beta(p^2 - \mu_0(\beta,N,0)))-1} \right) \nonumber \\
        &= -\int_0^1 \sum_{p \in \Lambda^*} \frac{\lambda f(p) + \mu_0(\beta,N,\lambda) - \mu_0(\beta,N,0) }{\sinh^2 \left( \frac{\beta(t(h(p)-\mu_0(\beta,N,\lambda)) + (1-t)(p^2-\mu_0(\beta,N,0)))}{2} \right)} \de t
        \label{eq:appIdealGas3a}
    \end{align}
    holds. Using \eqref{eq:appIdealGas3a} we obtain $|\mu_0(\beta,N,\lambda) - \mu_0(\beta,N,0)| \lesssim | \lambda | \ \Vert f \Vert_{\infty}$, which proves the first bound in \eqref{eq:appIdealGas2}. Next, we write
    \begin{align}
        N_0(\beta,N,0) - N_0(\beta,N,\lambda) &= \sum_{p \in \Lambda^*_+} \left( \frac{1}{\exp(\beta(h(p) - \mu_0(\beta,N,\lambda)))-1} - \frac{1}{\exp(\beta(p^2 - \mu_0(\beta,N,0)))-1} \right) \nonumber \\
        &= -\int_0^1 \sum_{p \in \Lambda_+^*} \frac{\beta (\lambda f(p) + \mu_0(\beta,N,\lambda) - \mu_0(\beta,N,0)) }{\sinh^2 \left( \frac{\beta(t(h(p)-\mu_0(\beta,N,\lambda)) + (1-t)(p^2-\mu_0(\beta,N,0)))}{2} \right)} \de t.
        \label{eq:appIdealGas3}
    \end{align}
    The absolute value of the second term on the right-hand side is bounded by a constant times $\beta^{-1}(|\lambda| \ \Vert f \Vert_{\infty} + | \mu_0(\beta,N,\lambda) - \mu_0(\beta,N,0) |)$, and hence 
    \begin{equation}
        | N_0(\beta,N,\lambda) - N_0(\beta,N,0)  | \lesssim \frac{|\lambda| \ \Vert f \Vert_{\infty}}{\beta},
        \label{eq:appIdealGas4}
    \end{equation} 
    which proves the second claim.
\end{proof}

The next lemma provides us with bounds for the first derivative of the chemical potential with respect to a perturbation of the one-particle Hamiltonian.

\begin{lemma}
\label{lem:derivativesChemPot}
    We consider the limit $N \to \infty$, $\beta N^{2/3} \to \kappa \in (0,\infty)$. The one-particle Hamiltonian $h$ is chosen as in Lemma~\ref{lem:idealGasWithVariable1ParticleHamiltonian} above and $\mu_0(\beta,N,\lambda)$ denotes again the chemical potential that leads to an expected number of $N$ particles in the ideal gas governed by the one-particle Hamiltonian $h$. We have
    \begin{equation}
        \left| \frac{\partial \mu_0(\beta,N,\lambda)}{\partial \lambda} \right| \lesssim \Vert f \Vert_{\infty}.  
         \label{eq:appIdealGas6}
    \end{equation}
\end{lemma}
\begin{proof}
    The chemical potential $\mu_0(\beta,N,\lambda) < 0$ is defined as the unique solution to the equation
    \begin{equation}
        N = \sum_{p \in \Lambda^*} \frac{1}{\exp(\beta(p^2 + \lambda f(p) - \mu_0(\beta,N,\lambda)))}.
        \label{eq:appIdealGas7}
    \end{equation}
    Differentiation of both sides of this equation with respect to $\lambda$ gives
    \begin{equation}
        \frac{\partial \mu_0(\beta,N,\lambda)}{\partial \lambda} \sum_{p \in \Lambda^*} \frac{1}{\sinh^2\left( \frac{\beta(p^2 + \lambda f(p) - \mu(\beta,N,\lambda))}{2} \right)}  =  \sum_{p \in \Lambda^*} \frac{f(p)}{\sinh^2\left( \frac{\beta(p^2 + \lambda f(p) - \mu(\beta,N,\lambda))}{2} \right)}. 
        \label{eq:appIdealGas8}
    \end{equation}
    We take the absolute value on both sides and find
    \begin{equation}
        \left| \frac{\partial \mu_0(\beta,N,\lambda)}{\partial \lambda} \right| \lesssim \sup_{p \in \Lambda^*} | f(p) |,
        \label{eq:appIdealGas9}
    \end{equation}
    which proves the claim. 
\end{proof}
In the last lemma in this section we derive a bound for the derivative of the chemical potential with respect to the particle number.

\begin{lemma}
\label{lem:BoundDerivativeCHemPotWrtN}
    We consider the limit $N \to \infty$, $\beta N^{2/3} \to \kappa \in (0,\infty)$. The chemical potential $\mu_0$ of the ideal gas in \eqref{eq:idealgase1pdmchempot} satisfies the bound
    \begin{equation}
        0 \leq \frac{\partial \mu_0(\beta,N)}{\partial N}  \lesssim \frac{1}{\beta N_0^2(\beta,N) + N^{2/3}}.
    \end{equation}
\end{lemma}
\begin{proof}
    Differentiation of both sides of \eqref{eq:idealgase1pdmchempot} with respect to $N$ gives
    \begin{equation}
        1 = \frac{\partial \mu_0(\beta,N)}{\partial N} \sum_{p \in \Lambda^*} \frac{\beta}{4 \sinh^2\left( \frac{\beta(p^2 - \mu_0(\beta,N)) }{2} \right) }
    \end{equation}
    and implies the bound
    \begin{equation}
        0 \leq \frac{\partial \mu_0(\beta,N)}{\partial N} \lesssim \frac{1}{1/(\beta \mu_0^2) + \beta^{-1}}.
    \end{equation}
    In combination with the identity $\mu_0 = -(1/\beta) \ln(1+N_0^{-1}) \sim - 1/(\beta N_0)$ this proves the claim.
\end{proof}
\section{Bounds related to a perturbed Bogoliubov Hamiltonian}
\label{app:perturbedBogoliubov}
In this section we prove a lemma that is used in Section~\ref{sec:1pdm} to compute the 1-pdm of the Gibbs state $G_{\beta,N}$ in \eqref{eq:relativeEntropyBound}. To not interrupt the main line of the argument there, we state and prove it here.

\begin{lemma}
    \label{lem:perturbedBogoliubov}
    We consider the limit $N \to \infty$, $\beta N^{2/3} \to \kappa \in (0,\infty)$. The one-particle Hamiltonian $h$ is chosen as in Lemma~\ref{lem:idealGasWithVariable1ParticleHamiltonian} above. Let $\gamma_p(\lambda,\mu)$ be given in as in \eqref{eq:gammap} with $p^2$ replaced by $h(p)$ and $\mu_0(\beta,N)$ replaced by $\mu < 0$. We have
    \begin{equation}
    \left| \frac{\partial \gamma_p }{\partial \mu} \right| \lesssim \frac{1}{\beta p^4} 
    \label{eq:appPerturbedBogA13}
\end{equation}
    as well as
    \begin{equation}
    \left| \frac{\partial \gamma_p(\lambda,\mu) }{\partial \lambda} \right| \lesssim \frac{\Vert f \Vert_{\infty}}{\beta p^4} \quad \text{ and } \quad \left| \frac{\partial^2 \gamma_p }{\partial \lambda^2} \right| \lesssim \frac{\Vert f \Vert_{\infty}^2}{\beta p^4}.
    \label{eq:appPerturbedBogA1}
\end{equation}
\end{lemma}
 \begin{proof}

We recall the definitions of $\epsilon(p)$, $u_p$, $v_p$, and $\gamma^{\mathrm{Bog}}$ in \eqref{eq:BogoliubovDispersion}, \eqref{eq:coefficientsBogtrafoAndi}, and \eqref{eq:gammapBog}, respectively. In all definitions we replace $p^2$ by $h(p)$ and $\mu_0(\beta,N)$ by $\mu$. 

The first derivative of $\epsilon(p)$ with respect to $\mu$ satisfies
\begin{equation}
    \frac{\partial \epsilon(p)}{\partial \mu} = \frac{-1}{\epsilon(p)} \left( p^2 + \lambda f(p) - \mu + \hat{v}(p) N_0/N \right),\quad \left| \frac{\partial \epsilon(p)}{\partial \mu} \right| \lesssim 1. 
    \label{eq:appPerturbedBog1}
\end{equation}
We also check that
\begin{equation}
    \frac{\partial \epsilon(p)}{\partial \lambda} = \frac{f(p)}{\epsilon(p)} \left( p^2 + \lambda f(p) - \mu + \hat{v}(p) N_0/N \right), \quad \left| \frac{\partial \epsilon(p)}{\partial \lambda} \right| \lesssim \Vert f \Vert_{\infty}.
    \label{eq:appPerturbedBog1b}
\end{equation}
For the second derivative of the dispersion relation with respect to $\lambda$, we have 
\begin{align}
    \frac{\partial^2 \epsilon(p)}{\partial \lambda^2} =& \frac{f^2(p)}{\epsilon(p)} \left[ 1 - \frac{1}{\epsilon^2(p)} \left( p^2 + \lambda f(p) - \mu + \hat{v}(p) N_0/N \right)^2 \right],\quad    \left| \frac{\partial^2 \epsilon(p)}{\partial \lambda^2} \right| \lesssim \frac{\Vert f \Vert_{\infty}^2}{p^2}. 
    \label{eq:appPerturbedBog3}
\end{align}
The derivative of $\gamma_p^{\mathrm{Bog}}$ with respect to $\mu$ satisfies 
\begin{equation}
    \frac{\partial \gamma_p^{\mathrm{Bog}}}{\partial \mu} = \frac{-\beta}{8 \sinh^2\left( \frac{\beta\epsilon(p)}{2} \right)} \frac{\partial \epsilon(p)}{\partial \mu},\quad \left| \frac{\partial \gamma_p^{\mathrm{Bog}}}{\partial \mu} \right| \lesssim \frac{1}{\beta p^4}.
    \label{eq:appPerturbedBog5}
\end{equation}
Here the bound in \eqref{eq:appPerturbedBog5} follows from  \eqref{eq:appPerturbedBog1}. Moreover, for its derivative with respect to $\lambda$ we have   
\begin{equation}
    \frac{\partial \gamma_p^{\mathrm{Bog}}}{\partial \lambda} = \frac{-\beta}{8 \sinh^2\left( \frac{\beta\epsilon(p)}{2} \right)} \frac{\partial \epsilon(p)}{\partial \lambda},\quad \left| \frac{\partial \gamma_p^{\mathrm{Bog}}}{\partial \lambda} \right| \lesssim \frac{\Vert f \Vert_{\infty}}{\beta p^4}.
    \label{eq:appPerturbedBog5b}
\end{equation}
The second derivative with respect to $\lambda$ satisfies 
\begin{equation}
    \frac{\partial^2 \gamma_p^{\mathrm{Bog}}}{\partial \lambda^2} = \frac{-\beta}{8 \sinh^2\left( \frac{\beta \epsilon(p)}{2} \right)} \frac{\partial^2 \epsilon}{\partial \lambda^2} + \frac{\beta^2 \cosh\left( \frac{\beta \epsilon(p)}{2} \right)}{4 \sinh^3\left( \frac{\beta \epsilon(p)}{2} \right)} \left( \frac{\partial \epsilon}{\partial \lambda} \right)^2,\quad \left| \frac{\partial^2 \gamma_p^{\mathrm{Bog}}}{\partial \lambda^2} \right| \lesssim \frac{\Vert f \Vert^2_{\infty}}{\beta p^4}
    \label{eq:appPerturbedBog7}
\end{equation}
where we used \eqref{eq:appPerturbedBog1b} and \eqref{eq:appPerturbedBog3} to obtain this bound. We also have
\begin{equation}
    \frac{\partial v_p}{\partial \mu} = - u_p \frac{N_0 \hat v(p)}{2N} \left( \frac 1 {p^2 + \lambda f(p) - \mu} - \frac{1}{p^2 + \lambda f(p) - \mu + 2 \hat{v}(p) N_0/N} \right),\quad \left| \frac{\partial v_p}{\partial \mu} \right| \lesssim \frac{1}{p^2}
    \label{eq:appPerturbedBog9}
\end{equation}
and
\begin{equation}
    \frac{\partial v_p}{\partial \lambda} = u_p \frac{N_0 f(p)\hat v(p)}{2N} \left( \frac 1 {p^2 + \lambda f(p) - \mu} - \frac{1}{p^2 + \lambda f(p) - \mu + 2 \hat{v}(p) N_0/N} \right),\quad \left| \frac{\partial v_p}{\partial \lambda} \right| \lesssim \frac{\Vert f \Vert_{\infty}}{p^2}.
    \label{eq:appPerturbedBog9b}
\end{equation}
\begin{align}
    \frac{\partial^2 v_p}{\partial^2 \lambda} &= v_p  \left( \frac{N_0 f(p)\hat v(p)}{2N} \left( \frac 1 {p^2 + \lambda f(p) - \mu} - \frac{1}{p^2 + \lambda f(p) - \mu + 2 \hat{v}(p) N_0/N} \right) \right)^2 \nonumber\\
    &\quad + u_p \frac{N_0 f(p)\hat v(p)}{2N} \left( \frac {-f(p)} {p^2 + \lambda f(p) - \mu} + \frac{f(p)}{p^2 + \lambda f(p) - \mu + 2 \hat{v}(p) N_0/N} \right),\quad \left| \frac{\partial^2 v_p}{\partial \lambda^2} \right| \lesssim \frac{\Vert f \Vert_{\infty}^2}{p^4}.
    \label{eq:appPerturbedBog11}
\end{align}

Finally, we consider $\gamma_p$, whose first derivative with respect to $\mu$ satisfies 
\begin{equation}
    \frac{\partial \gamma_p }{\partial \mu} = 2 v_p \frac{\partial v_p}{\partial \mu} \left( 1 + 2 \gamma_p^{\mathrm{Bog}} \right) + \left( 1 + 2 v_p^2 \right) \frac{\partial \gamma_p^{\mathrm{Bog}}}{\partial \mu},\quad \left| \frac{\partial \gamma_p }{\partial \mu} \right| \lesssim \frac{1}{\beta p^4},
    \label{eq:appPerturbedBog13}
\end{equation}
which proves \eqref{eq:appPerturbedBogA13}. Here we used  \eqref{eq:boundvp}, \eqref{eq:appPerturbedBog5}, and \eqref{eq:appPerturbedBog9}. For 
its first derivative with respect to $\lambda$, we see that   
\begin{equation}
    \frac{\partial \gamma_p }{\partial \lambda} = 2 v_p \frac{\partial v_p}{\partial \lambda} \left( 1 + 2 \gamma_p^{\mathrm{Bog}} \right) + \left( 1 + 2 v_p^2 \right) \frac{\partial \gamma_p^{\mathrm{Bog}}}{\partial \lambda},\quad \left| \frac{\partial \gamma_p }{\partial \lambda} \right| \lesssim \frac{\Vert f \Vert_{\infty}}{\beta p^4},
    \label{eq:appPerturbedBog13b}
\end{equation}
where we used \eqref{eq:boundvp}, \eqref{eq:appPerturbedBog5b}, and \eqref{eq:appPerturbedBog9b}. The second derivative of $\gamma_p$ with respect to $\lambda$ is given by
\begin{equation}
    \frac{\partial^2 \gamma_p }{\partial \lambda^2} = 2 \left[ \left( \frac{\partial v_p}{\partial \lambda} \right)^2 + v_p \frac{\partial^2 v_p}{\partial \lambda^2} \right] \left( 1 + 2 \gamma_p^{\mathrm{Bog}} \right) + 8 v_p  \frac{\partial v_p}{\partial \lambda} \frac{\partial \gamma_p^{\mathrm{Bog}}}{\partial \lambda} + \left( 1 + 2 v_p^2 \right) \frac{\partial^2 \gamma_p^{\mathrm{Bog}}}{\partial \lambda^2}.
    \label{eq:appPerturbedBog14}
\end{equation}
We use \eqref{eq:boundvp}, \eqref{eq:appPerturbedBog5b}, \eqref{eq:appPerturbedBog7}, \eqref{eq:appPerturbedBog9b}, and \eqref{eq:appPerturbedBog11} to check that
\begin{equation}
    \left| \frac{\partial^2 \gamma_p }{\partial \lambda^2} \right| \lesssim \frac{\Vert f \Vert_{\infty}^2}{\beta p^4}
    \label{eq:appPerturbedBog15}
\end{equation}
holds. Eqs. \eqref{eq:appPerturbedBog13b} and \eqref{eq:appPerturbedBog15} prove \eqref{eq:appPerturbedBogA1} and therewith the claim.
\end{proof}

\bibliographystyle{siam}

\vspace{0.5cm} 

\noindent (Andreas Deuchert) Department of Mathematics, Virginia Tech \\ 
225 Stanger Street, Blacksburg, VA 24060-1026, USA \\ 
E-mail address: \texttt{andreas.deuchert@vt.edu} \\

\vspace{0.2cm} 

\noindent (Phan Thành Nam) Department of Mathematics, LMU Munich \\
Theresienstrasse 39, 80333 Munich, Germany \\
Email address: \texttt{nam@math.lmu.de} \\

\vspace{0.2cm} 

\noindent (Marcin Napiórkowski) Department of Mathematical Methods in Physics, Faculty of Physics, University of Warsaw \\
Pasteura 5, 02-093 Warszawa, Poland \\
Email address: \texttt{marcin.napiorkowski@fuw.edu.pl}

\end{document}